\documentclass[aps,prx,reprint,superscriptaddress,nofootinbib,notitlepage]{revtex4-2}
\pdfoutput=1

\usepackage{graphicx}
\usepackage{amsmath,amsthm,amssymb,amsfonts,setspace}
\usepackage{bm}
\usepackage{mathtools}
\usepackage{braket}
\usepackage{float}
\usepackage{color}
\usepackage[usenames,dvipsnames]{xcolor}
\usepackage{hyperref}  
\usepackage{cleveref} 
\usepackage{wrapfig}
\usepackage{tikz}
\usepackage[left=1in,right=1in,top=1in,bottom=1in]{geometry}
\usepackage{thmtools, thm-restate}
\usepackage{makecell}

\hypersetup{colorlinks=true,urlcolor=[rgb]{0,0,0.5},citecolor=[rgb]{0,0.5,0},linkcolor=[rgb]{0,0.5,0}}

\usepackage{graphicx}
\graphicspath{ {./Images/}} 

 \theoremstyle{plain}
 \newtheorem{fact}{Fact}
 \theoremstyle{plain}
 \newtheorem{lem}{Lemma}
 \theoremstyle{plain}
 \newtheorem{thm}{Theorem}
 \theoremstyle{plain}
 
   \theoremstyle{plain}
 \newtheorem{prop}{Proposition}
 \theoremstyle{plain}
 \newtheorem{corr}{Corollary}
 \theoremstyle{plain}
 
 \theoremstyle{remark}
 \newtheorem*{rem*}{Remark}
	\theoremstyle{remark}
	 \newtheorem{rem}{Remark}
  \theoremstyle{plain}
 \newtheorem{Definition}{Definition}	 
   \theoremstyle{plain}
 \newtheorem{Conjecture}{Conjecture}
    \theoremstyle{plain}
 \newtheorem{res}{Result}

\newcommand{\ep}{\epsilon}
\newcommand{\e}{\mathrm{e}}
\newcommand{\ot}{\otimes}
\newcommand{\ii}{\mathrm{i}} 

\renewcommand{\exp}{\mathrm{exp}}

\DeclareMathOperator{\tr}{tr}

\renewcommand{\H}{\mathcal{H}} 
\newcommand{\Lmax}{\mathcal{L}^{\mathrm{max}}} 
\newcommand{\Lgen}{\mathcal{L}^{\mathrm{gen}}}

\newcommand{\Hfer}{\mathcal{H}_f}
\newcommand{\Hfree}{\mathcal{H}^{+}_{\mathrm{Fock}}}
\newcommand{\Hfock}{\mathcal{H}_\mathrm{Fock}}

\newcommand{\Hpas}{\mathcal{H}_\mathrm{pas}} 
\newcommand{\Hact}{\mathcal{H}_\mathrm{act}} 

\newcommand{\tHpas}{\tilde{\mathcal{H}}_\mathrm{pas}} 
\newcommand{\tHact}{\tilde{\mathcal{H}}_\mathrm{act}} 

\newcommand{\U}{\mathrm{U}} 
\newcommand{\SU}{\mathrm{SU}} 
\newcommand{\SO}{\mathrm{SO}} 

\renewcommand{\u}{\mathfrak{u}} 
\newcommand{\so}{\mathfrak{so}} 
\newcommand{\g}{\mathfrak{g}} 

\newcommand{\C}{\mathbb{C}} 
\newcommand{\R}{\mathbb{R}} 
\newcommand{\supp}{\mathrm{supp}} 

\newcommand{\I}{\mathbb{I}} 
\renewcommand{\P}{\mathbb{P}} 
\newcommand{\E}{\mathbb{E}} 

\renewcommand{\O}{\mathcal{O}} 
\newcommand{\N}{\mathcal{N}} 

\newcommand{\Q}{\mathcal{Q}} 
\newcommand{\PP}{\mathcal{P}} 

\newcommand{\sharP}{\#\mathrm{P}} 

\renewcommand{\ket}[1]{\left| #1 \right>} 
\renewcommand{\bra}[1]{\left< #1 \right|} 
\renewcommand{\braket}[2]{\langle #1 | #2 \rangle} 
\newcommand{\ketbra}[2]{\left| #1 \rangle\langle #2 \right|} 

\newcommand{\n}{\mathbf{n}} 

\newcommand{\x}{\mathbf{x}} 
\newcommand{\y}{\mathbf{y}} 
\newcommand{\z}{\mathbf{z}} 
\renewcommand{\v}{\mathbf{v}} 
\newcommand{\w}{\mathbf{w}} 

\newcommand{\muu}{\bm{\mu}} 

\renewcommand{\SS}{\mathbb{S}} 

\newcommand{\D}{\mathcal{D}} 
\renewcommand{\L}{\mathcal{L}} 

\newcommand{\Pfer}{\mathbb{P}_{\mathrm{pas}}} 
\newcommand{\Pflo}{\mathbb{P}_{\mathrm{act}}} 

\newcommand{\inrho}{\rho_{\mathrm{in}}} 
\newcommand{\inpsi}{\Psi_{\mathrm{in}}} 
\newcommand{\psiQUAD}{\ket{\Psi_4}} 

\newcommand{\X}{\mathcal{X}} 
\newcommand{\Y}{\mathcal{Y}} 
\newcommand{\Z}{\mathcal{Z}} 

\newcommand{\A}{\mathcal{A}} 
\newcommand{\B}{\mathcal{B}} 
\renewcommand{\S}{\mathcal{S}} 

\newcommand{\lnull}{l_{\mathrm{N}}} 
\newcommand{\lbin}{l_{\mathrm{B}}} 
\newcommand{\lfix}{l_{\mathrm{F}}} 

\newcommand{\loct}{l_{\mathrm{oct}}} 
\newcommand{\lempty}{l_{\mathrm{empty}}} 
\newcommand{\lquad}{l_{\mathrm{quad}}} 

\newcommand{\Cstate}{\mathcal{C}_\mathrm{in}} 
\newcommand{\Ca}{\mathcal{C}_{\mathrm{A}8}}
\newcommand{\Span}{\mathrm{span}} 

\newcommand{\Pipas}{\Pi_\mathrm{pas} }
\newcommand{\Piact}{\Pi_\mathrm{act} }

\newcommand{\Gpas}{\mathcal{G}_\mathrm{pas}} 
\newcommand{\Gact}{\mathcal{G}_\mathrm{act}} 

\newcommand{\nupas}{\nu_\mathrm{pas}} 
\newcommand{\nuact}{\nu_\mathrm{act}} 

\global\long\global\long\global\long\def\bra#1{\mbox{\ensuremath{\langle#1|}}}
\global\long\global\long\global\long\def\ket#1{\mbox{\ensuremath{|#1\rangle}}}

\global\long\global\long\global\long\def\kb#1#2{\mbox{\ensuremath{\ensuremath{\ensuremath{|#1\rangle\!\langle#2|}}}}}

\global\long\global\long\global\long\def\SET#1#2{\mbox{\ensuremath{\ensuremath{\left\lbrace\left. #1\ \right| #2\right\rbrace }}}}

\newcommand\dgg{^{\dagger}}
\newcommand{\id}{\mathbb{I}}
\newcommand\abs[1]{\left|#1\right|}
\newcommand{\norm}[1]{\left\lVert#1\right\rVert}
\newcommand\av[1]{\left\langle #1 \right\rangle}
\newcommand\poly[1]{\mathrm{poly}(#1)}

\newcommand{\bPh}{\boldsymbol{\phi}}
\newcommand{\bVar}{\boldsymbol{\varphi}}
\newcommand{\oX}{\tilde{X}}

\newcommand{\Cpas}{C_\mathrm{pas}}
\newcommand{\Cact}{C_\mathrm{act}}
\newcommand{\Cpasval}{5.7}
\newcommand{\Cactval}{16.2}

\begin{document} 	
	\title{Fermion Sampling: a robust quantum computational advantage \\ scheme  using fermionic linear optics and magic input states}

\author{Micha\l\ Oszmaniec}
\affiliation{ 
Center for Theoretical Physics, Polish Academy of Sciences,\\ Al. Lotnik\'ow 32/46, 02-668
Warszawa, Poland}	
\email{oszmaniec@cft.edu.pl}
\author{Ninnat Dangniam}
\affiliation{ 
Center for Theoretical Physics, Polish Academy of Sciences,\\ Al. Lotnik\'ow 32/46, 02-668
Warszawa, Poland}
\author{Mauro E.S. Morales}
\affiliation{Centre for Quantum Computation and Communication Technology, \\ Centre for Quantum Software and Information, University of Technology Sydney, NSW 2007, Australia}
\author{Zolt\'an Zimbor\'as}
\affiliation{Wigner Research Centre for Physics, H-1121, Budapest, Hungary}
\email{zimboras.zoltan@wigner.hu}
\affiliation{MTA-BME Lend\"ulet Quantum Information Theory Research Group Budapest, Hungary} 
\affiliation{Mathematical Institute, Budapest University of Technology and Economics,\\ H-1111, Budapest, Hungary}
 	
\begin{abstract}

Fermionic Linear Optics (FLO) is a restricted model of quantum computation which in its original form is known to be efficiently classically simulable. We show that, when initialized with suitable input states, FLO circuits can be used to demonstrate quantum computational advantage  with  strong hardness guarantees.
Based on this, we propose a  quantum advantage scheme which is a fermionic analogue of Boson Sampling: Fermion Sampling with magic input states.

We consider in parallel two classes of circuits: particle-number conserving (passive)  FLO and active FLO that preserves only fermionic parity and is closely related to Matchgate circuits introduced by Valiant. Mathematically, these classes of circuits can be understood as fermionic representations of the  Lie groups $\U(d)$ and $\SO(2d)$. This observation allows us to prove our main technical results. We first show anticoncentration for probabilities in random FLO circuits of both kind. Moreover, we prove robust average-case hardness of computation of probabilities.  To achieve this, we adapt the worst-to-average-case reduction based on Cayley transform, introduced recently by Movassagh \cite{movassagh_quantum_2020}, to  representations of low-dimensional Lie groups. 
Taken together, these findings provide hardness  guarantees comparable to the paradigm of Random Circuit Sampling.

Importantly, our scheme has also a potential for experimental realization. Both passive and active FLO circuits are relevant for quantum chemistry and many-body physics and have been already implemented in proof-of-principle experiments  with superconducting qubit architectures. Preparation of the desired quantum input states can be obtained by a simple quantum circuit acting independently on disjoint blocks of four qubits and using 3 entangling gates per block. We also argue that due to the structured nature of FLO circuits, they can be efficiently certified using resources scaling polynomially with the system size.

\end{abstract}

\maketitle	
\vspace*{-3mm}
\section{Introduction}

Universal fault-tolerant quantum computers are expected to exceed capabilities of classical computers in many applications including optimization problems, simulation of many-body quantum systems, machine learning and code-breaking. However, practical requirements for implementations of quantum algorithms   generally require the noise level to be below a certain stringent threshold and an encoding of logical qubits into a large number of physical qubits \cite{gidney2019factor}.
Despite the impressive progress made along the road to realize a large-scale fault tolerant quantum computer as shown in the proof-of-principle demonstrations of error correction \cite{error-correction-demo2016} and fault tolerance  \cite{fault-tolerant-demo2020}, what we have at present and in the near future are NISQ devices \cite{Preskill2018NISQ}: noisy, intermediate-scale quantum processors having of the order of tens or hundreds of qubits. 



The paradigm of quantum computational advantage \cite{Lund2017,SuprRev2017} (also known as supremacy) aims to develop schemes showing computational advantage of restricted-purpose quantum machines under minimal  theoretical assumptions while minimising hardware requirements. Importantly, given the current status of complexity theory, a rigorous separation of the power of quantum and classical computers cannot be made without plausible assumptions such as the non-collapse of the polynomial hierarchy (a weaker version of $\mathrm{P \neq NP}$). Current quantum advantage schemes are usually based on the problem of sampling, i.e., the task of generating samples of a distribution generated by a given quantum circuit or a specifically tuned devices (see, however, \cite{Bravyi2018advantage} for an alternative proposal involving \emph{relation problems} that challenges classical computers in the playground of shallow circuits). The first candidate  for demonstration of quantum computational supremacy was Boson Sampling \cite{Aaronson2013} that proposed to sample from photonic networks that were initialized in single-photon photon states of several modes. Subsequent sampling proposals include the IQP sampling \cite{Bremner2011,Bremner2016}, the Random Circuit Sampling \cite{Boixo2016,bouland_complexity_2019,movassagh_quantum_2020}, the quantum Fourier sampling \cite{fefferman2016fourier} and many other schemes \cite{Morimae2017,bermejo-vega_architectures_2018} that usually operate on multiqubit systems that undergo evolutions under restricted gate sets (there exist, however,  other proposals that use Gaussian states \cite{hamilton2017gaussian, lund2014boson} or atomic systems \cite{bermejo-vega_architectures_2018,haferkamp_closing_2019}). 

\begin{figure*}[t!]
    \centering
    \includegraphics[scale=0.25]{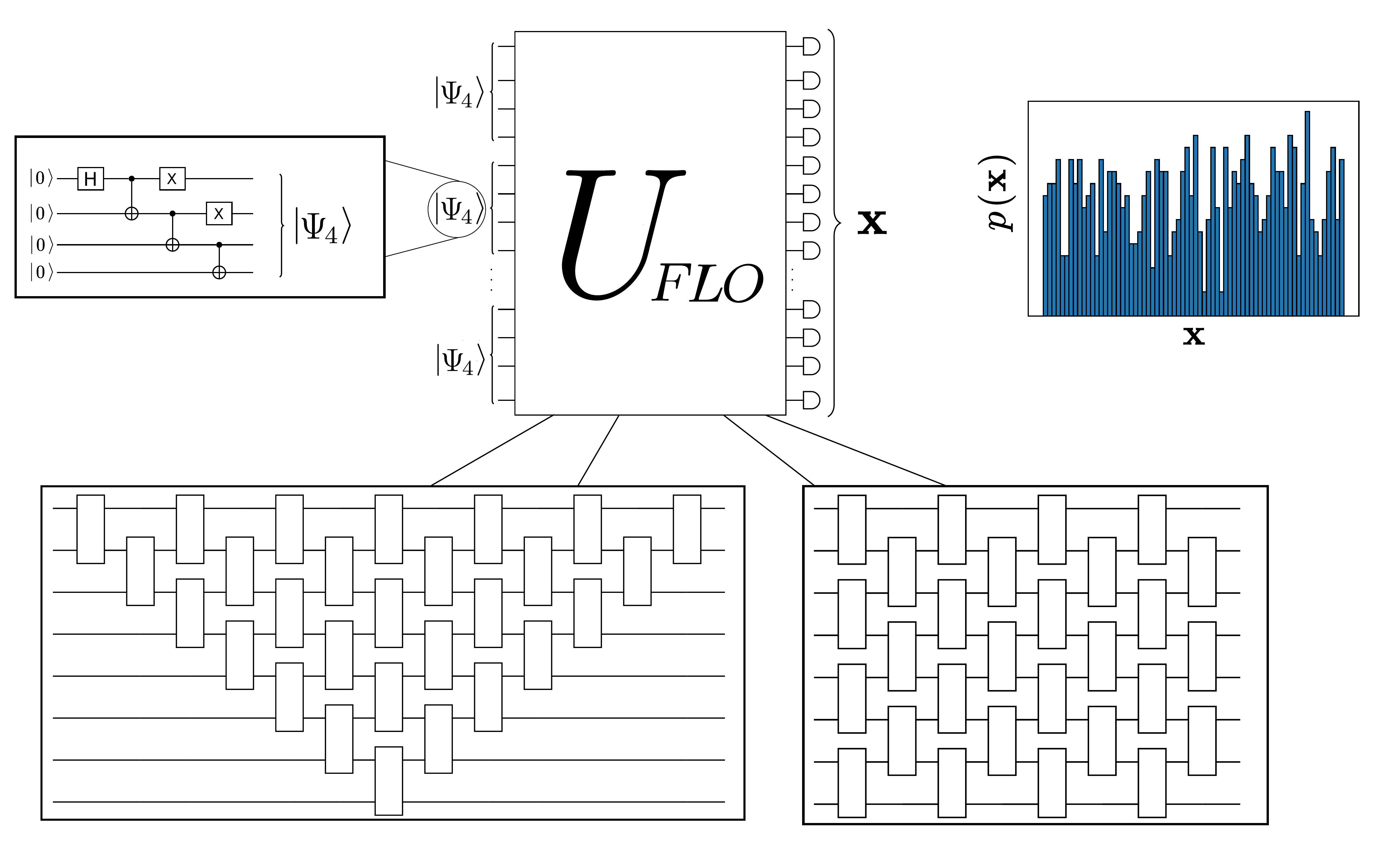}
    \caption{The setup considered in our work. We run an FLO circuit $U_{FLO}$ (passive or active) with input state $\ket{\inpsi} =\psiQUAD^{\ot N}$ and sample bitstrings $\x$ with the probability distribution $p(\x)$ induced by the circuit. Using Jordan-Wigner transformation that encodes fermions in qubits, the state $\psiQUAD$ can be easily prepared as shown in the inset to the left.   The decomposition of the circuits into elementary gate set can be realized by the fermionic analogues of existing layouts for linear optical networks \cite{ReckUniform2016,Reck1994} as discussed in Appendix~\ref{app:decomp} .
    \label{fig:magic-input-and-decomp}}
\end{figure*}

Random circuit sampling (RCS) is the task of sampling from the output distribution of a randomly selected quantum circuit. RCS has been recently experimentally demonstrated in a system of  53 superconducting  qubits arranged in the planar layout, and using random two-qubit-local circuits of depth 20  \cite{Suprem2019}. Beside its experimental feasibility in current NISQ architectures, this quantum advantage proposal also enjoys strong hardness guarantees based on two technical results available for random quantum circuits: a worst-to-average-case reduction of the hardness of computing the output probabilities \cite{bouland_complexity_2019,movassagh_quantum_2020} and anticoncentration \cite{hangleiter_anticoncentration_2018,Harrow2018}.

Independently of RCS, the original proposal of Boson Sampling received a lot of interest  due to the experimental progress in the field of integrated photonics \cite{brod2019photonic}. {Currently, the state-of-the-art experiments involve 14 indistinguishable photons in 20 modes \cite{wang2019boson}, while a recent work \cite{GaussBSExperiment2020} reported demonstration of Gaussian Boson Sampling (i.e., a variant of Boson Sampling with Gaussian input states) using 50 squeezed states at the input of a 100 mode photonic network.} While worst-to-average-case reduction (at least for the exact computation of the probabilities) was available for the Boson Sampling \cite{Aaronson2013}, the anticoncentration property remains unproven for this scheme.
Interestingly, neither of the theoretical guarantees are currently in place for Gaussian Boson Sampling.

In this work, we propose a quantum advantage scheme based on a fermionic analogue of Boson Sampling: Fermion Sampling with magic input states. In our scheme a suitable input state  $\ket{\inpsi}$ in $d=4N$  fermionic modes is transformed via Fermionic Linear Optical (FLO) transformation  $V$, and is measured using particle-number resolving detectors (see Fig.  \ref{fig:magic-input-and-decomp}). We consider in parallel two classes of circuits: particle-number conserving (passive)  FLO and active FLO that preserves only the fermionic parity \cite{terhal_classical_2002,knill_fermionic_2001} and is closely related  to Matchgate circuits introduced by Valiant \cite{valiant_quantum_2002}. Mathematically, these classes of circuits can be understood as fermionic representations of the  Lie groups $\U(d)$ and $\SO(2d)$. This observation allows us to prove our main technical results. We first show anticoncentration for probabilities in random FLO circuits of both kind. Moreover, we prove robust average-case hardness of computation of probabilities. To achieve this we adapt the worst-to-average-case reduction based on Cayley transform \cite{movassagh_quantum_2020} to our scenario, when instead of the defining representation of the unitary group one considers higher dimensional representations of low-dimensional Lie groups. Taken together, these findings  give hardness  guarantees matching that of the paradigm of RCS \cite{bouland_complexity_2019,movassagh_quantum_2020}. We also argue that, due to the structural properties of FLO gates, one can efficiently certify them with resource scaling polynomially with the system size.

  We argue that our scheme is feasible to realize experimentally. While experimental realization of linear-optical transformation in systems of \emph{real} fermionic systems is usually hard because of Coulomb interaction (see however \cite{bocquillon2014electron}), we make use of the fact that this class of operations is  relevant for performing  quantum chemistry and many-body simulations on a quantum computer  \cite{chemistryGOOGLE2020, fermPASSlayout2018, fermACTpassLAYOUT2018, fermGAUSSlayout2019}. Specifically, after a standard Jordan-Wigner encoding of qubits into fermions, our sampling proposal becomes readily implementable by restricted set of gates and layouts native to superconducting qubit architecture used in simulations of quantum chemistry \cite{ContGatesets2020}. Compared to the RCS implementation of arbitrary FLO transformation requires depth scaling proportional to $N$. In this encoding magic input states can be prepared using 3 entangling gates per each and every disjoint block of four qubits and particle-number measurements are realized via standard computational basis measurement (see Figure~\ref{fig:magic-input-and-decomp}).  {While using circuits of linear depth might seem challenging at first sight, this requirement follows solely form our proof techniques that guarantee anticoncentration when sampling from uniform distribution on passive and active FLO circuits. However, we give numerical evidence outcome probabilities corresponding to active FLO circuits anticoncentrate in much smaller depth (see Figure \ref{fig:anticon-sim}). In fact, it is plausible that active FLO circuits anticoncentrate in \emph{logarithmic} depth, just like random quantum circuits formed from universal gates, as proved by the recent work by Dalzell \emph{et. al.} \cite{logANTICONCENTRATION}  }

\vspace*{-2mm}

\subsection*{Significance of results and relation to prior work} 

\vspace*{-1mm}

\subsubsection*{Relevance of the technical results for hardness of Fermion Sampling}

A first step in establishing hardness of any quantum advantage proposal is  showing hardness of sampling up to relative error. Sampling in relative error refers to the task in which a classical computer, given a classical description of the quantum process of interest (e.g., input states, arrangement of gates in the device, etc.),  is challenged to efficiently sample from the probability distribution $\{q_\x\}$ that for every output $\x$ satisfies $|p_\x -q_\x|\leq \alpha p_\x$ , where $\{p_\x\}$ is the true probability distribution produced by the device and $\alpha>0$ is a constant. To establish hardness of sampling up to some relative error, it suffices to show that certain probabilities  produced by the device are $\sharP$-hard to compute\footnote{Informally, $\sharP$ is the complexity class of counting solutions to problems that can be efficiently verified. For more details, see \cite{aroraComplexity}.}, assuming that polynomial hierarchy does not collapse. This hardness of quantum probabilities  for specific circuits  and outcomes (i.e., in the worst case) was proven long ago for circuit built from universal gates \cite{terhal2004adptive,fenner1999determining}. For non universal models of quantum computation a standard technique for establishing  $\sharP$-hardness of computation of probabilities is based on showing that a particular non-universal model becomes universal when postselection is allowed \cite{Aaronson2013,Bremner2011,farhi_QAOA_2016,brod_complexity_2015,bouland_conjugated_2018,gao_ising_2017,bermejo-vega_architectures_2018}. Relative error approximation is however too strong to be a reasonable notion of approximation from the  physical perspective. This is because even a very small amount of experimental noise can render very large relative error.

A more realistic notion of approximate sampling is based on additive error \cite{Aaronson2013,pashayan_estimation_2020} in which classical computer is supposed to efficiently produce samples from probability distribution $\{q_\x\}$ satisfying $\sum_{\x} |p_\x -q_\x|\leq \ep$, where $\ep$ is the error parameter. Establishing hardness for additive error approximate sampling is however much more challenging then in the case of relative error. Assuming non-collapse of the polynomial hierarchy, the currently existing techniques \cite{Aaronson2013,Bremner2016,pashayan_estimation_2020} establish this hardness using Stockmayer approximate counting algorithm \cite{Stockmeyer1985} and by relying on two technical properties of a given quantum advantage proposal:  (i) anticoncentration of outcome probabilities, and (ii)  $\sharP$-hardness of relative error approximate computation of outcome probabilities  on average. Anti-concentration  refers to property that probability amplitudes $p_\x (V)$ are typically not too small, compered to their average value, for random circuits $V$ defining a given quantum advantage proposal. Anticoncentration property has been shown in several schemes \cite{Bremner2016,Morimae2017,bouland_conjugated_2018,Harrow2018,haferkamp_closing_2019}, while for others, including  Fourier Sampling \cite{fefferman2016fourier} and Boson Sampling \cite{Aaronson2013} it remains unproven. On the other hand average-case $\sharP$- hardness of relative error approximate computation of $p_\x (V)$ has not been proven to date for the existing quantum advantage proposals. There however exist intermediate results that support it in the form of average-case $\sharP$-hardness of \emph{exact} computation of $p_\x (V)$ for Boson Sampling \cite{Aaronson2013}, RCS \cite{bouland_complexity_2019} and related schemes \cite{haferkamp_closing_2019}. These works adopt the polynomial interpolation technique from \cite{Aaronson2013} and to prove worst-to-average-case hardness reduction. This reduction has been recently improved by Movassagh \cite{movassagh_quantum_2020} for RCS who showed that it is average-case $\sharP$-hard to approximate $p_\x (V)$ in additive error $\exp(-\Theta(N^{4.5}))$, where $N$ is the number of qubits.

In this work, in order to justify computational hardness of the proposed Fermion Sampling scheme we prove the following results

\begin{itemize}
    \item[i] Anticoncentration of probabilities $p_\x(V)$ in the output of the scheme for both passive and active FLO circuits initialized in magic states. 
     \item[ii] Robust worst-to-average-case hardness reduction for computation of probabilities for passive and active FLO circuits initialized in magic states up to error $\exp(-\Theta(N^{6}))$.
\end{itemize}

\begin{table*} 
\begin{center}
\includegraphics[width=\linewidth]{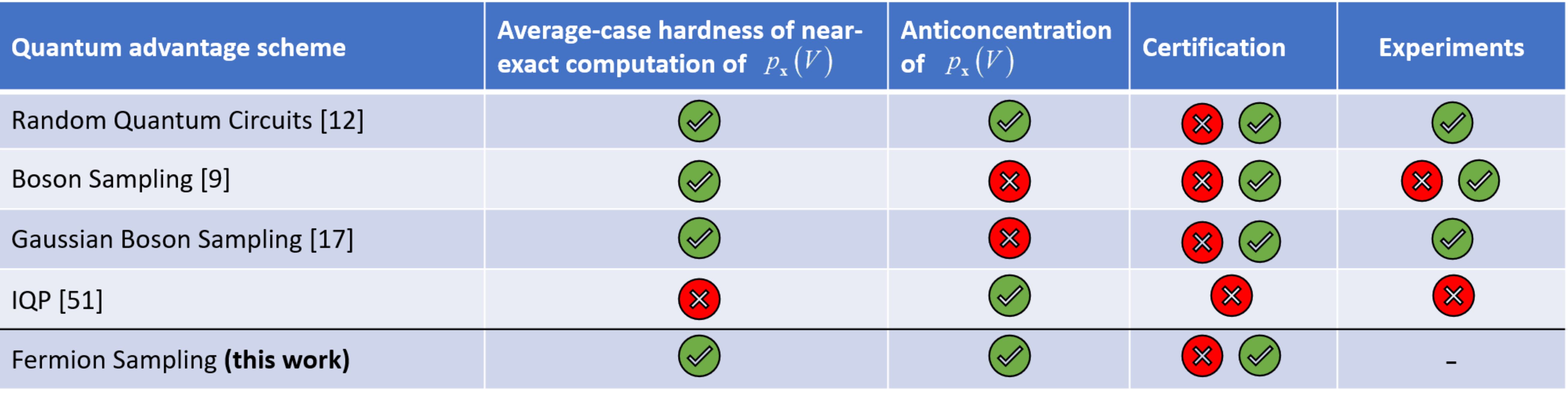}
\caption{ {Summary of results concerning some of the sampling schemes proposed for quantum supremacy. The first column lists the different sampling schemes: Random Circuit Sampling (RCS) \cite{Boixo2016,Hangleiter2018,bouland_complexity_2019,Harrow2009design,Suprem2019,Wu2021}, Boson Sampling (BS) \cite{Aaronson2013,Bouland2021,wang2019boson,brod2019photonic}, Gaussian Boson Sampling (GBS) \cite{hamilton2017gaussian,GaussBSExperiment2020,GaussBSExperiment2021,deshpande2021GBS}, Instantaneous Quantum Polynomial (IQP) \cite{Shepherd2009,Bremner2016} and Fermion Sampling. The second column indicate whether  average-case hardness for near-exact computation of output probabilities has been proven. The third column indicates if the circuits involved in the scheme fulfill the anticoncentration property, the fourth column shows if the scheme has some certification procedure for the circuits involved and the final column indicates if there are experiments which realize the sampling scheme. Entries with both check and crosses mean that partial results exist.}  \label{tab:sampling_schemes} }
\end{center}
\end{table*}
    
Instrumental to our proofs is the fact that active and passive FLO circuits are representations of the low-dimensional (of dimensions scaling polynomially with the number of fermionic modes $d$ ) Lie groups $\U(d)$ and $\SO(2d)$ respectively.   For the anticoncentration property, we do not use the 2-design property (which is not satisfied for FLO unitaries), but instead prove it relying on specific group-theoretic properties of FLO circuits. For the worst-to-average case reduction, we follow the state-of-the-art technique by Movassagh \cite{movassagh_quantum_2020}, which utilizes Cayley path to construct a low-degree rational interpolation between the worst-case and average-case circuits, while generalizing it in two significant directions. First, while the interpolation in \cite{movassagh_quantum_2020} is performed directly using physical circuits, ours is performed at the level of group elements which are then represented as circuits (see Fig.~\ref{fig:cayley-path}) 
Secondly, while \cite{movassagh_quantum_2020} applies the interpolation to local one- and two-qubit gates that constitute the circuit, we directly apply it to a global circuit while maintaining the low-degree nature of the rational functions, which is required for the robust reduction.

These results put Fermion Sampling at the comparable level as RCS \cite{bouland_complexity_2019,movassagh_quantum_2020} in terms of state-of-the-art hardness guarantees, surpassing that of Boson Sampling. The advantage of our scheme compared to RCS is that FLO circuits can be efficiently certified due to their low-dimensional structural properties.  The apparent disadvantage is the size of the required circuits -  RCS can be implemented in depth $\sqrt{N}$  \cite{Boixo2016,Harrow2018} while our scheme requires depth of the structured circuit scaling like $N$. {A more general comparison of our Fermionic sampling scheme with other previous schemes in the literature is given in Table \ref{tab:sampling_schemes}. We compare the schemes in terms of available results in the literature regarding proofs of average-case hardness, anticoncentration and experimental results implementing the schemes or similar results.}

\subsubsection*{Comparison with Boson Sampling}

 Boson Sampling \cite{Aaronson2013}, the first quantum advantage proposal based on sampling, relies on the fact that the probability amplitudes of indistinguishable bosons initially prepared in a Fock state and passing through linear-optical network, can be expressed via matrix permanents. Computation of permanent is know to be $\sharP$- hard in the worst case \cite{aroraComplexity}. In contrast, the analogous amplitudes for fermions are given by the determinant, which can be computed efficiently. Physically, this difference in complexity can be attributed to the fact that bosonic Focks states are \emph{non-Gaussian} bosonic states, while their fermionic counterparts are in fact fermionic Gaussian states \cite{bravyi_fermionic_2002}. Thus, to  make a closer the analogy with Boson Sampling, we define our Fermion Sampling using \emph{non-Gaussian} input states $\ket{\inpsi}=\psiQUAD^{\ot N}$, where $\psiQUAD=\frac{1}{\sqrt{2}}(\ket{0011}+\ket{1100})$.  This state can be prepared easily on a quantum computer but at the same time can be expressed as an exponential sum of orthogonal Fock states. This is sufficient to guarantee hardness of the corresponding probability amplitudes. It was shown by Ivanov and Gurvits \cite{Ivanov2017,Ivanov2020}  that if $\ket{\inpsi}$ is transformed via particle-number preserving (passive) FLO transformation, the probability amplitudes are related to \emph{mixed discriminants} of matrices, which is known to be $\sharP$-hard, because they can be efficiently reduced to permanent. In the context of active FLO transformations, auxiliary states  $\psiQUAD$ are known to promote this class of transformations to universality \cite{bravyi_universal_2006} (see also \cite{hebenstreit2019}), which can be used to show $\sharP$-hardness of probabilities arising form active FLO circuits initialized with such non-Gaussian states.   We conclude the comparison with boson sampling by clarifying the role of the measurements used. Our proposal uses fermionic particle-number measurements which are themselves fermionic Gaussian. This differentiates Fermion Sampling from Boson Sampling schemes. This includes Gaussian Boson Sampling \cite{hamilton2017gaussian} in which  bosonic squeezed states (that are bosonic Gaussian) are transformed using linear optics, and finally measured using (non-Gaussian) particle-number detectors. In that proposal non-Gaussian character of the particle-number measurement is crucial for hardness \cite{Rahimi-Keshari2016}.

\begin{figure*}[ht]
    \centering
    \includegraphics[scale=0.7]{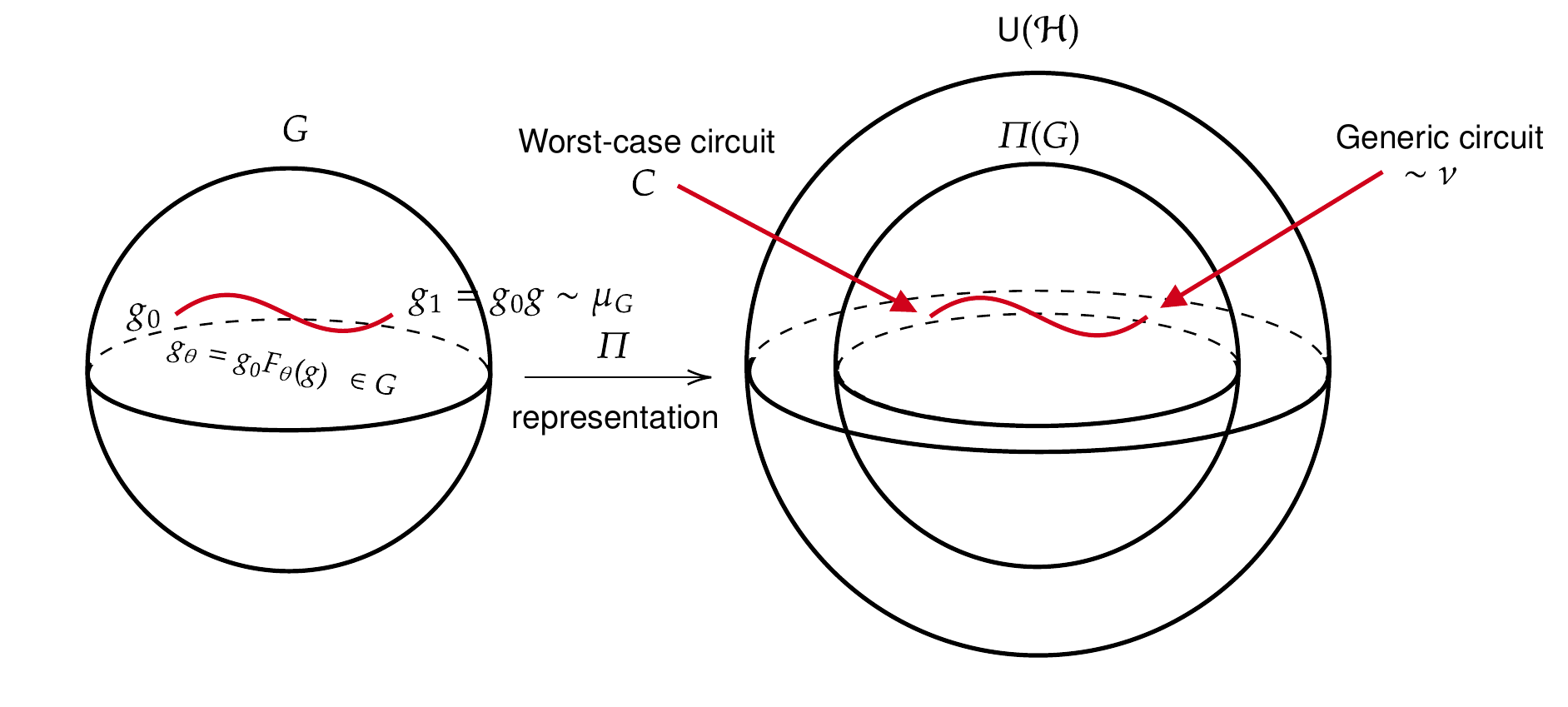}
    \caption{
    $\Pi$ is a group representation from $G$ to (a subgroup $\Pi(G)$ of) the group $\U(\H)$ of all quantum circuits on Hilbert space $\H$ (typically of exponential dimension).
    The Cayley path $g_\theta = g_o F_\theta(g)\in G$ gives a rational interpolation between a fixed element $g_0$ and $g_1 = g_0 g$. This gives rise to a rational interpolation between circuits $C\coloneqq \Pi(g_0)$ and  $\Pi(g_0)\Pi(g)=\Pi(g_0 g)$. To carry out worst-to-average-case reduction we would consider $g_0$ to be group element corresponding to the worst-case circuit $C$ while $g$ will be chosen to be a \emph{generic} element of the Lie group $G$.
    \label{fig:cayley-path} 
    }
\end{figure*}
\vspace*{-4mm}

\subsubsection*{Comparison to Existing Benchmarking Protocols}

{A convincing demonstration of quantum supremacy requires a means to build confidence that the output $q(\x)$ of the quantum device is close to the ideal distribution $p(\x)$. Such verification can be done with differing levels of efficiency (or inefficiency) depending on how much one can assume or trust the correct functioning of the device.
Of the highest standard (requiring minimal assumptions) of such verification is to certify whether $q(\x)=p(\x)$ or $\sum_{\x} |p_\x -q_\x| > \epsilon $ for some small $\epsilon>0$---that is, we are ruling out \emph{all} adversarial distributions that are $\epsilon$-away from $p(\x)$---using only the classical output statistics of the sampling device. Building on a classical result on identity testing of probability distributions, Ref.~\cite{hangleiter2019sample} shows that such form of stringent, device-independent certification is infeasible for most prominent quantum supremacy distributions, requiring exponentially many samples.}

{In reality, however, the experimenters do have prior knowledge about the functioning of various components of their device and the model of the physical noise that degrades the computation, which prompts one to move away from the minimal assumptions. If one insists on only making use of the classical output statistics, benchmarking protocols of the cross-entropy type \cite{Boixo2016,bouland_complexity_2019,Suprem2019}, allows one to rule out distributions that are otherwise ideal but are corrupted by depolarizing noise with few samples. 
(Ref.~\cite{bouland_complexity_2019} derives an entropic condition on distributions that are ruled out by the cross entropy difference proposed in \cite{Boixo2016}.)
Thus, some amount of effort in the analysis of the Google's experiment \cite{Suprem2019} is dedicated to benchmarking the error model and validating their assumptions.
The downside is that this type of measures requires one to actually compute the ideal quantum supremacy probabilities, which are presumed to be extremely hard to compute for when the number of qubits is large. 
For Boson Sampling, a weaker form of certification based on state discrimination \cite{aaronson2014uniform} can be made fully efficient (that is, efficient in both the number of samples and computation time), but can only certify against a fixed adversarial distribution (for example the uniform distribution).}

{Assuming the ability to change input states or measurement settings (trusted preparations and measurements), direct certifications for several quantum advantage architectures can be devised that are fully efficient. One such line of works 
\cite{Hangleiter2018,haferkamp_closing_2019,chabaud2020efficient} is based on the idea of fidelity witness of output states. }

{In our work, we offer a way to characterize FLO circuits assuming the ability to prepare certain product input states and perform trusted Pauli measurements.
The certification is indirect as it only certifies a component of the experiment (the circuit) but not the quantum supremacy distribution itself. However, the protocol suggests that a direct certification of Fermion Sampling can be developed similar to that of Boson Sampling as our scheme is analogous to the characterization of linear optical networks \cite{Rahimi-Keshari2016} whereas
no such scheme exists in the present for Random Circuit Sampling.}


\subsubsection*{Relation to Fermionic Quantum Computation} 

Quantum computing with (active) FLO circuits has received significant attention over the years. While FLO circuits with unentangled input states and measurements
are efficiently simulable classically, they constitute a ``maximally classical" subset of quantum circuits in the sense that an addition of any non-FLO unitary allows one to reach any unitary on the relevant Hilbert space \cite{OZ2017}.
Thus, similar to Clifford circuits, FLO circuits with additional resources constitute an interesting model of universal quantum computing \cite{bravyi_fermionic_2002,bravyi_universal_2006}.
Here we review the most important results about computational power of FLO circuits and their extensions. 
Common notions of simulation in the literature fall into two classes of \emph{strong} and \emph{weak} simulations:
strong simulation refers to the ability to compute the the marginal probability of any chosen outcome, whereas weak simulation refers to the ability to sample from the output probability distribution. 
Strong classical simulability of FLO circuits can be traced back to the work of Valiant \cite{valiant_quantum_2002} in which he introduced so-called matchgates for the purpose of studying algorithms for graphs. Assuming computational-basis input states and measurements, circuits of nearest neighbor (n.n) matchgates in 1D layout 
can be strongly simulated on a classical computer in polynomial time, even with adaptive measurement in the computational basis \cite{terhal_classical_2002}.
Soon after the introduction of matchgates, their classical simulability was connected to exact solvability in physics as n.n matchgate circuits can be mapped to evolutions of non-interacting fermions via the Jordan-Wigner transformation \cite{terhal_classical_2002,knill_fermionic_2001}
and extended to classical simulability of dissipative FLO and non-unitary matchgates \cite{bravyi_lagrangian_2004,DissipativeFLO}. The geometric locality restriction is non-trivial as matchgate computation becomes quantum universal when the n.n condition is lifted \cite{Jozsa2008} or the linear chain of qubits is replaced by more general graphs \cite{brod_geometries_2012}.

When one considers arbitrary product input states, adaptive computation using FLO circuits with such inputs can be simulated classically \cite{brod_efficient_2016}. This is in striking contrast to the case of Clifford circuits in which supplying single-qubit magic states and adaptive measurement in the computational basis suffices for universal quantum computation \cite{bravyi2005magic}. Since every fermionic state (or qubit state with a fixed parity) of fewer than 4 qubits is Gaussian  \cite{bravyi2005capacity,melo_power_2013}, FLO circuits with computational-basis measurement must be supplied with at least four-qubit magic input state to attain universality. The first example of such state is $\ket{a_8} = \frac{1}{\sqrt{2}}(\ket{0000}+\ket{1111})$ (which can be converted to $\psiQUAD$ used in our scheme, see the proof of Lemma \ref{lem:actPROJfinal}) was introduced in \cite{bravyi_universal_2006} along with the corresponding state-injection scheme for universal quantum computation using Ising anyons. Much more general result was established in \cite{hebenstreit2019} where it was showed that all non-Gaussian states, when supplied in multiple copies,  allow one to perform universal quantum computation. Finally, weak classical simulability of FLO circuits with noisy magic input states was studied in \cite{melo_power_2013,oszmaniec_classical_2014},

Alternative to magic input states, adding an arbitrary non-FLO gate \cite{bravyi_fermionic_2002,brod_extending_2011,OZ2017} (see also \cite{zimboras2014}), or entangled measurements such as non-destructive parity measurement \cite{bravyi_fermionic_2002,beenakker2004charge}, also allows one to perform universal quantum computation.
When the final measurement is restricted to only one qubit line and no adaptive measurement is allowed during the computation, the circuits are classically simulable in the strong sense even with magic input states. This result was first proven for an arbitrary product input state in \cite{Jozsa2008} and observed to generalized to any product of $O(\log m)$-qubit states in \cite{hebenstreit_computational_2020}.
Recently comprehensive investigation of the complexity landscape of FLO circuits with auxiliary resources are given in \cite{hebenstreit_computational_2020}, which investigated the hardness  of FLO circuits depending on: (i)  whether the input is a product state or copies of entangled magic states, (ii) whether adaptive measurements are allowed,  (iii) whether the final measurement is performed only on a single qubit or on all qubits. It was established there that strong simulation of FLO circuits is $\sharP$-hard in all cases considered except ones that are already known to be classically simulable \cite{terhal_classical_2002,Jozsa2008,brod_efficient_2016}. Using the standard postselection argument \cite{Bremner2011}, it is possible to show that  weak simulation FLO circuits with magic input states and no adaptive measurement implies collapse of the polynomial hierarchy (see Appendix \ref{sec:sharP_shallow}).  This scenario coincides with one of the settings considered in this work. However, in this work we are concerned with establishing hardness of Fermion Sampling up to additive error which, as explained earlier, is a property much harder to establish.

\emph{Organization of the paper.}--  First, in Section \ref{sec:notations} we lay out basic notations and concepts, focusing mostly on the fermionic context. Then in Section \ref{sec:results_significance}  we formally define our quantum advantage proposal and give a high-level overview of our results and their significance. We also present there arguments in favour of experimental feasibility of our scheme. In Section \ref{sec:concl} we discuss possible applications of our work and present future research directions.  In the subsequent Section \ref{sec:anticoncentration} we prove that output probabilities  of FLO circuits initialized in suitable magic states anticoncentrate for generic active and passive  FLO circuits. These results, together with known \cite{Ivanov2017} worst-case $\sharP$ hardness of probability distributions,  is then used in Section \ref{sec:hardness-sampling} to prove hardness of approximate Fermion Sampling.   Section \ref{sec:cayley} is devoted to the quantitative analysis of the Cayley path transformation for unitary and orthogonal groups. 
In Section \ref{sec:worst-to-avg-reduction} we use technical results form the two preceding parts to prove worst-to-average-case reduction for hardness of computing probabilities in our quantum advantage scheme. \ In the final Section \ref{sec:cert} we show that in the Jordan-Wigner encoding an unknown FLO unitary can be efficiently certified using resources scaling polynomially with the number of fermionic modes. {The Appendix consists of seven parts and contains auxiliary technical results. In Appendix~\ref{app:decomp}, we describe in detail the decomposition of an FLO circuits into elementary one- and two-qubit gates. In Appendix \ref{app:hardnessOFsampling} we provide a formal proof of hardness of approximate sampling from FLO circuits, utilizng the anticoncentration property from Section \ref{sec:anticoncentration}. In subsequent Appendix \ref{app:TVbound} we prove bounds on TV distance between Haar measure and its Cayley-path deformation for classical groups $\U(d)$ and $\SO(2d)$. In part \ref{sec:degree} we upper bound the degrees of polynomials associated to FLO circuits.}  In subsequent part \ref{sec:Projectors}, we give details of the computations needed in the proof of anticoncentration of our results.
In Appendix~\ref{app:effTOM}, we prove a lemma concerning the stability of the FLO representations (an analogue of the stability result proved standard Boson Sampling \cite{arkhipov2015bosonsampling}), which is used in the tomography scheme of FLO unitaries.  Finally, in  Appendix~\ref{sec:sharP_shallow} we prove $\sharP$-Hardness of probabilities in shallow depth active FLO circuits.

\section{Notation and basic concepts} \label{sec:notations}

In this section we describe main concepts and notation needed in the paper. Specifically, we will introduce the language of second quantization, vital for describing fermionic systems. We will define passive and active fermionic linear optical circuits. Finally, we survey Jordan-Wigner transformation which allows to implement fermionic systems and associated unitaries acting on them in terms of spin systems and standard quantum circuits. All these ingredients allow us to formally define our scheme for attaining quantum computational  advantage with FLO circuits.

Let $\H$ be a finite-dimensional Hilbert space. Normalized vectors in this space will be denoted by $\ket{\Psi},\ket{\Phi}$ etc. Such normalized vectors give rise to  pure states, i.e., rank 1 nonnegative operators on $\H$. For the sake of brevity we will use the notation  $\Psi=\ketbra{\Psi}{\Psi},\Phi=\ketbra{\Phi}{\Phi}$ etc. We will use the symbol $\D(\H)$ for the set of all (possibly mixed) quantum states on  $\H$. Finally, by $\U(\H)$ we denote group of unitary operators on $\H$. We will be consider a system of fermions with single-particle Hilbert space being $\C^d$. The Hilbert space associated to this system is a $d$ mode Fock space
\begin{equation}
  \Hfock(\C^d)=\bigoplus_{n=0}^d \bigwedge^n (\C^d)\ , \label{eq:Fock_space}
\end{equation}
where $\bigwedge^n(\C^d)$, is the totally anti-symmetric subspace of $(\C^d)^{\otimes n}$  describing states consisting of exactly $n$ fermions, and $\bigwedge^0 (\C^d) =\Span_\C(\ket{0_F})$, where $\ket{0_F}$ is the Fock vacuum. Any basis $\lbrace \ket{1},\ket{2},\ldots,\ket{d}\rbrace$ of single-particle Hilbert space  defines a family of creation and annihilation operators acting on $  \Hfock(\C^d)$:  $f_j^\dagger$ and $f_j$, respectively, where $j=1,2,\ldots,d$. These operators satisfy canonical anti-commutation relations $\{f_j,f_k^\dagger\}\equiv f_j  f_k^\dagger+ f_k^\dagger f_j = \delta_{j,k}$ and $\{f_j,f_k\}=\{f_j^\dagger,f_k^\dagger\}=0$, with $\delta_{j,k}$ being the Kronecker symbol.

Given this set of creation and annihilation operators, it is natural to introduce the so-called Fock basis states, which forms a basis of $\Hfock(\C^d)$, as  
\begin{equation}
| \x \rangle:= (f_1^{\dagger})^{x_1}(f_2^{\dagger})^{x_2}\cdots  (f_d^{\dagger})^{x_d}| 0_F\rangle
\end{equation}
for any $\x \in \{0,1 \}^{d}$, where we used the notation $(f_1^{\dagger})^0 = \id$.
Throughout the paper, we will denote the set $\{1,\ldots, d \}$ as $[d]$.
Given an arbitrary subset $\X \subset [d]$, it will also be useful to introduce the notation $ \ket{\X} $
for the Fock basis state $\ket{\x}$ with $x_j=1 $ if $j \in \X$ and $x_j =0$ otherwise. We will also use $\binom{\X}{k}$ to denote the  collection of subsets of finite set $\X$ of size  $k$ (we shall assume the convention $\binom{\X}{k}=\emptyset $ if $|\X|<k$.

Considering the direct sum decomposition  of the Fock space into fixed particle number subspaces in Eq.~\eqref{eq:Fock_space}, a specific  Fock basis state $| \X \rangle$
(with $|\X| =n $)
is an element of the $n$-particle subspace $\bigwedge^n(\C^d) $.
Note that since   $\bigwedge^n(\C^d) $ can be regarded as a subspace of $ (\C^d)^{\otimes n}  $, it is natural to consider such a Fock basis state $| \X \rangle $ (where $\X = \{a_1, \ldots a_n \}$ with $a_i < a_j$ if $i < j$) as an element in  $(\C^d)^{\otimes n} $ which is given by the formula
\begin{equation}
\begin{split}
  | \X \rangle &= |a_1 \rangle \wedge |a_2 \rangle \wedge \ldots \wedge | a_n \rangle  \\
  &= \frac{1}{\sqrt{n!}} \sum_{i_1, \ldots i_n=1}^n \epsilon_{i_1, i_2, \ldots, i_n} |a_{i_1}\rangle \otimes | a_{i_2} \rangle \otimes \cdots | a_{i_n} \rangle.
\end{split}
\end{equation}
Here and throughout the paper we will use the generalized Levi-Civita symbol, i.e., for any string of positive integers $(k_1, k_2,\ldots  , k_n )$ with $k_i \ne k_j$ if $i \ne j$ we define $\epsilon_{k_{1}, k_{2}, \ldots, k_{n}}=(-1)^{p}$, where $p$ is the parity of the permutation $\pi$ for which $k_{\pi(i)} < k_{\pi{(j)}}$ if $i <j$ and $p$ is its parity, while $\epsilon_{k_{1}, k_{2}, \ldots, k_{n}}=0$ if some of the entries in  $(k_1, k_2,\ldots  , k_n )$ coincide. 


A {\it passive fermionic linear optical transformation} on the $n$-particle subspace $\bigwedge^n(\C^d) \subset (\C^d)^{\otimes n}$ is given as a transformation $U^{\otimes n}$ restricted from being  $(\C^d)^{\otimes n} \to (\C^d)^{\otimes n} $ function to being a  $\bigwedge^n(\C^d) \to \bigwedge^n(\C^d) $  map. 
Passive FLO can be understood abstractly as the irreducible representation of the low-dimensional  symmetry group  $\U(d)$ in the Hilbert space $\bigwedge^n(\C^d)$
\begin{align}\label{eq:pasREP}
    \Pipas: \U(d) &\longrightarrow \U \left( \bigwedge^n(\C^d) \right), \\
    U &\longmapsto  \left. U^{\ot n}\right|_{\bigwedge^n(\C^d)}.
\end{align}
That is, we get a representation of $\U(d)$ on a fixed particle fermionic subspace. A useful equivalent definition is that for any $U=e^{iK}\in \SU(d)$, $\Pipas$ is the restriction of the Fock state unitary $e^{i/2 \sum_{n,m} K_{nm}f^\dagger_{n}f^{\phantom{dagger}}_m}$ to the subspace $\bigwedge^n(\C^d)$.

An important concept when discussing passive FLO transformations are Slater determinant states.  These are states of the form $| \Psi \rangle = |\xi_1\rangle |\wedge \xi_2 \rangle \wedge \ldots \wedge |\xi_n \rangle$, where 
$\{|\xi_i \rangle \}_{i=1}^n \subset \C^d$
is a set of orthonormal vectors  of the one-particle Hilbert space $\C^d$.
By definition, Fock basis states are special cases of Slater determinant states.
And passive FLO transformations act transitively on the set of Slater determinant states.
The overlap between any two Slater determinant states, $|\Psi \rangle = |\xi_1\rangle |\wedge |\xi_2 \rangle \wedge \ldots \wedge |\xi_n \rangle$ and $|\Phi \rangle = |\phi_1\rangle \wedge |\phi_2 \rangle \wedge \ldots \wedge |\phi_n \rangle$, can be expressed by the simple determinant formula
\begin{equation}
    \langle \Psi | \Phi \rangle = \det C \, ,  \; \; 
    C_{i,j}=\langle \xi_i| \phi_j\rangle. \label{eq:det-inner}
\end{equation}

A standard way to measure fermionic systems is to perform particle number measurement i.e  a projective measurement in the Fock basis basis $\ket{\x}$ defined previously. Upon obtaining measurement result labelled by  $\X$, numbers  $x_i$ have the interpretation of number of particles in mode $i$.

Let us next introduce the self-adjoint Majorana mode operators
\begin{equation}
\label{def:majorana_mode_ops}
m_{2j-1} = f_j+f_j^\dagger\, ,\qquad 
m_{2j}  = -i\, (f_j-f_j^\dagger),
\end{equation}
with anti-commutation relations $\{m_j,m_k\}=2\,\id \delta_{j,k}$. These operators define  parity operator $Q=i^d\prod_{i=1}^{2d} m_i$ in $\Hfock(C^d)$. The subspace of $\Hfock(C^d)$ that corresponds to eigenvalue of $+1$ of $Q$ is spanned by Fock states $\ket{\n}$ having even number of particles. In what follows we shall refer it to this vector space as positive parity subspace and denote it by $\Hfock^+(\C^d)$. Majorana operators allow also to active fermionic linear optical transformations. We say that a fermionic unitary $U$ is \emph{free}, \emph{Gaussian}, or \emph{linear-optical}, 
if it can be written as an exponential of a quadratic Hamiltonian, i.e., $U=\mathrm{e}^{i H}$,  where 
\begin{equation}
\label{eq:H}
H=\frac{i}{4}\sum_{j,k=1}^{2d} A_{j,k} \, m_j\,  m_k,
\end{equation}
and $A=-A^T\in\mathbb{R}^{2d\times 2d}$. Active FLO transformation form a group which can be conveniently understood in terms of (projective)  representation of the $\SO(2d)$ group\footnote{The FLO operators themselves form a group isomorphic to the universal cover of $\SO(2d)$, called  Spin$(2d)$.}: 
\begin{align}\label{eq:actREP}
    \Piact: \SO(2d) &\longrightarrow \U(  \Hfock(\C^d) )\, \\
   O  & \longmapsto\  \exp\left( \frac{1}{4} \sum_{i,j=1}^{2d} \left[\log(O)\right]_{ij} m_i m_j\right),
\end{align}
Often the restriction of $\Piact$ to the positive parity subspace $\Hfock^+(\C^d)$ is considered, we will not use a new symbol for this restricted representation rather simply refer to it by writing $\Piact: \SO(2d) \to \U(  \Hfock^+(\C^d) )$. Pure positive-parity Gaussian states are defined as pure states  the form $\Psi=\Pi(
O)\ketbra{0_F}{O_F} \Pi(O)^\dagger$, for $O\in SO(2d)$. In other words pure positive-parity Gaussian fermionic states are states that can be generated from the vacuum by active FLO transformations. Similarly, it is possible to define negative-parity pure fermionic Gaussian states as states generated by active FLO from, say, a Fock state with a single excitation.

 If we look at the action on the operators, we get an actual (i.e., non-projective) representation. In particular, a single Majorana operator evolves under an active FLO transformation as follows
\begin{align}
\label{eq:mode_transform}
 U^\dagger \, m_j\,U =\sum_{k=1}^{2d} O_{jk}\, m_k ,
\end{align}
where $U = \e^{-i Ht}$ with $H=\frac{i}{4}\sum_{j,k=1}^{2d} A_{j,k}m_j m_k$ and and $O = e^{-A}\in \SO(2d)$.

We will also use the notation $\Gpas$ and $\Gact$ to denote respectively passive and active fermionic linear optical gates. The name comes from the fact that these gates transform single creation/majorana operators to linear combination of creation/majorana operators, respectively.  

An important ingredient when discussing how to implement FLO transformations on qubit systems is the Jordan-Wigner transformation, that provides an equivalence between fermion and qubit systems through the
unitary mapping
$\mathcal{V}_{\textrm{JW}}:  \Hfock(\C^d) \to (\C^{2})^{\otimes d}$
given as the following mapping between 
\begin{align}
     &\mathcal{V}_{\textrm{JW}} \left( (f_1^{\dagger})^{x_1}(f_2^{\dagger})^{x_2}\cdots  (f_d^{\dagger})^{x_d}| 0_F\rangle \right)
     = \bigotimes_{p=1}^d \ket{x_p}
\end{align}
for all $\x=(x_1,x_2, \ldots, x_d) \in \{0,1 \}^{d}$.
which in turn induces an isomorphic mapping between majorana and spin operators
\begin{align}
    m_{2p-1} \mapsto \; \mathcal{V}_{\textrm{JW}}\,  m_{2p-1} \,  \mathcal{V}^{\dagger}_{\textrm{JW}}  &= Z_1 \cdots Z_{p-1} X_{p}, 
    \\
     m_{2p} \mapsto
      \; \mathcal{V}_{\textrm{JW}}\,  m_{2p} \,  \mathcal{V}^{\dagger}_{\textrm{JW}} 
      &= Z_1 \cdots Z_{p-1} Y_{p}, 
\end{align}
where $p\in[d]$.
To make the connection between fermions and qubit systems even more transparent, one often introduces the {\it occupation number notation} for vectors in $  \Hfock(\C^d) $ as $| \x \rangle:= (f_1^{\dagger})^{x_1}(f_2^{\dagger})^{x_2}\cdots  f_d^{\dagger})^{x_d}| 0_F\rangle$ for any $\x \in \{0,1 \}^{d}$. As the $|\x \rangle$ vectors are mapped via the Jordan-Wigner transformation to the computational basis states of , they are also called the {\it fermionic computational basis states} in  $\Hfock(\C^d)$.

Since groups  $\U(d)$ and $\SO(2d)$ are compact groups (for comprehensive introduction to the theory of Lie groups and their representations, see \cite{WalachBook,HallGroups}), each  possesses a unique normalized integration measure invariant under any group translation called Haar measure. We will donate this measure by  $\mu_G$ for the $G$ one of the symmetry groups above. Invariance of $\mu_G$ means that any measurable subset $A \subset G$ and any $h\in G$, we have that
\begin{align}
    \mu(hA) = \mu(Ah) = \mu(A)\ .
\end{align}
The above condition to the level of expectation values (averages) reads
\begin{align}
    \int_G d\mu (g) f(gh) = \int_G d\mu(g) f(hg) = \int_G d\mu(g) f(g)\ .
\end{align}
where $f$ is any integrable function on $G$ and $h\in G$.  We will denote by $\nupas$ the distribution of the unitaries $V=\Pipas(U)$, where $\U\sim \mu_{\U(d)}$ and by $\nuact$ distribution of the unitaries $\Piact(O)$, where $O\sim\mu_{\SO(2d)}$. In order to keep the notation compact we will suppress the dependance of these measures on $d$ and $n$ (values of these parameters will be implied from the context).

Finally, we use the following notation to denote growth of functions: Let $f$ and $g$ be two positive-valued functions. We write $f=O(g)$ iff $\lim_{x\to\infty} f(x)/g(x) < \infty$ and $f=o(g)$ iff $\lim_{x\to\infty} f(x)/g(x) = 0$.


\section{Main results}\label{sec:results_significance}

In this part we formally define our scheme for demonstration of quantum computational advantage and present informally main results of this work. In the end we comment on the practical feasibility of our quantum advantage scheme.

Having reviewed the basic concepts needed, we are now ready to formally introduce our quantum advantage proposal, which is illustrated in Fig.~\ref{fig:magic-input-and-decomp}. 
We have a system of $d=4N$ fermionic modes. The input state of the scheme is an $N$-fold tensor product  the  non-Gaussian magic state  $\psiQUAD=(\ket{1100}+\ket{0011})/\sqrt{2}$, i.e.,
\begin{equation}\label{eq:InputState}
    \ket{\inpsi}= \psiQUAD^{\otimes N} \, .
\end{equation}
Note that states equivalent to this has been used in FLO computation schemes in \cite{bravyi_universal_2006, Ivanov2017, hebenstreit2019}. After the initialization, a generic FLO operation is applied either respecting the particle-number conservation (passive scheme) or not (active scheme). Any FLO unitary can be decomposed into two-qubit FLO gates of linear depth either in diamond, triangle or brickwall layouts, see Figs.~\ref{fig:magic-input-and-decomp} and \ref{fig:brick_triangle}. The choice of the FLO operation $V$ is  done via the probability distributions $\nupas$ and $\nuact$  induced from the Haar measures on $\U(d)$ and $\SO(2d)$, respectively (see Section.~\ref{sec:notations}). 

For a type of particular type of FLO circuit the computational task we address is ability is to  sample from  the output distribution
\begin{equation}\label{eq:sample}
    p_{\x}(V,\inpsi) = \abs{\av{\x|V|\inpsi}}^2\ ,
\end{equation}
where output bitstring satisfy $|\x|=2N$ and  $|\x|$-even for passive FLO and active FLO respectively.  This computational task will be referred to as \emph{Fermion Sampling}. We prove four main technical results that underpin the hardness of Fermion Sampling.

The first result is anticoncentration for FLO circuits. Informally speaking it states that for the considered familly of circuits and fixed output $\x$ values  $\abs{\av{\x|V|\inpsi}}^2$ are typically not much smaller compared to they average average value.

\begin{res}[Anticoncentration for generic FLO circuits]\label{res:anticoncentration}
Let $\nu=\nupas$ or $\nu=\nuact$ be uniform distribution over passive and respectively active FLO circuit acting on $4N$ fermion modes. Let $\inpsi$ be the input state to our quantum advantage proposal. Then, there exist a constant $C \ge 1$ such that for every outcome $\x$ and for every $\alpha\in[0,1]$ 
\begin{equation}\label{eq:anticoncentrationPGEN}
    \Pr_{V\sim \nu} \left( p_{\x}(V,\inpsi) > \frac{\alpha}{|\H|} \right) > \frac{(1-\alpha)^2}{C} \ ,
\end{equation}
where $\H=\bigwedge^{2N}(\C^{4N})$ for passive FLO and $\H=\Hfock^+(\C^{4N})$ for active FLO.  
\end{res}
The formal version of this result is given in Theorem \ref{thm:anticoncentrationFLO}. It is important to emphasize that in the course of the proof of the this results we do not use the property of gate sets of interest forming an (approximate) $2-$design \cite{hangleiter_anticoncentration_2018}. In fact, it can be proved that measures $\nupas,\nuact$ do not form a projective $2-$design. We perform the proof of anticoncentration by heavily using group-theoretical techniques and particular properties of fermionic representations of symmetry groups $\U(d)$ and $\SO(2d)$.

In line with standard methodology based on Stockmeyer's algorithm \cite{Stockmeyer1985}), anticoncentration, and hiding property   we reduce approximate sampling from $\{ p_{\x}(V,\inpsi)\}$ to approximate computation of particular probability $p_{\x_0}(V,\inpsi)$ (see Theorem \ref{thm:sample-computing}). This allows us to prove hardness of approximate Fermion Sampling in Theorem \ref{th:SAMPLhar}  by conjecturing non collapse of the polynomial hierarchy (cf. Conjecture \ref{conj:PHnoncollapse}) and average-case hardness of computation approximate computation  of $p_{\x_0}(V,\inpsi)$  in relative error (cf. Conjecture \ref{conj:AVERhar}).

\begin{rem}
It is important to stress that in the passive FLO case our anticoncentration results \emph{do not} follow from anticoncentration results for the determinant proved in \cite{Aaronson2013}. The reason is that our probability amplitudes can be only expressed via determinants of submatrices of $U\in\U(d)$ (cf. Section \ref{sec:degree}). Also, these submatrices cannot be approximated via Gaussian matrices (we work in the regime in which number of modes  $d$ is comparable to the total number of particles $n$.

\begin{figure*}
    \centering
    \includegraphics[scale=0.55]{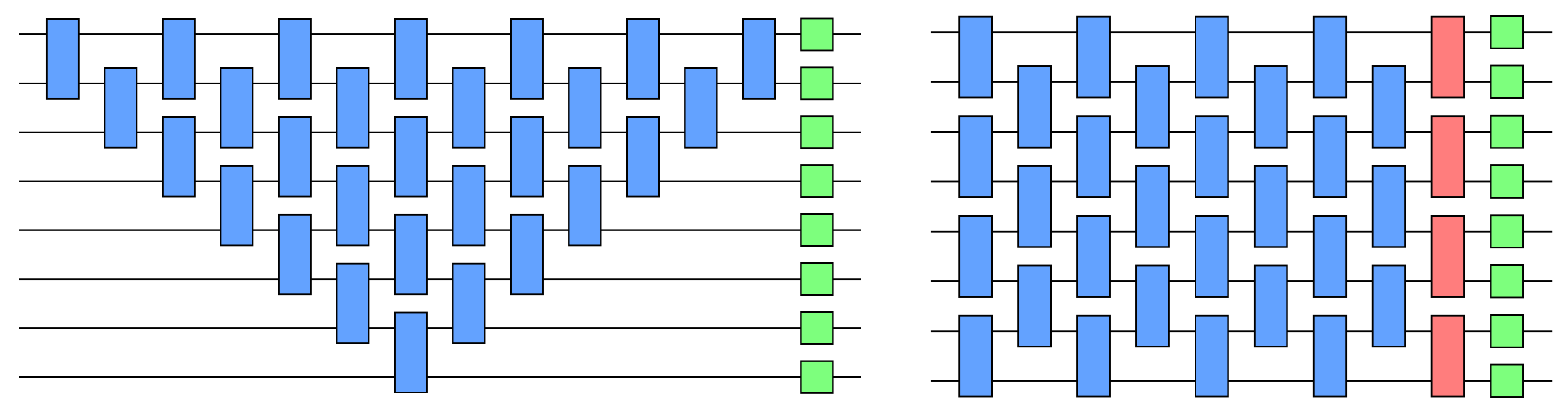}
    \caption{Circuit layouts implementing arbitrary passive and active FLO transformations. These layouts are based on the decomposition of arbitrary elements of the $\U(d)$ and $\SO(2d)$ groups into a sequence of nearest-neighbor Givens rotations and a diagonal matrix. The depicted two-qubit gates in the passive FLO case are of the type $D_{\mathrm{pas}}(\alpha_1, \alpha_2)$ (see Eq.~\eqref{eq:D_pas}), while the single-qubit gates are $Z$-rotations. The  two-qubit gates in the active FLO case are of the type $D_{\mathrm{act}}(\{ \beta_i \})$ (see Eq.~\eqref{eq:D_act}) and the single-qubit gates are Pauli unitaries. The extra layer of red colored two-qubit gates are only needed in the active case. The decomposition of the two-qubit gates $D_{\mathrm{pas}}(\alpha_1, \alpha_2)$ and $D_{\mathrm{act}}(\{ \beta_i \})$ into native gates of superconducting qubit architectures are provided in Fig.~\ref{fig:theta}.}
    \label{fig:brick_triangle}
\end{figure*}

\end{rem}

To give evidence for Conjecture \ref{conj:AVERhar} we prove two  worst-to-average-case reductions that allow us to prove weaker versions of approximate hardness result

\begin{res}[Worst-case to average reduction for exact computation of probabilities]\label{res:exactRED} Let $\nu=\nupas$ or $\nu=\nuact$ be uniform distribution over passive and respectively active FLO circuit acting on $4N$ fermion modes. Let $\inpsi$ be the input state to our quantum advantage proposal.
Let $V_0$ be a FLO gate (either active or passive) such that $p_{\x_0}(V_0),\inpsi)$ is $\sharP$-hard to compute (see Remark \ref{rem:exact_sharPhard}). It is then $\sharP$-hard to compute values of $p_{\x_0}(V,\inpsi)$ with probability greater than $\frac{3}{4}+\frac{1}{\poly{N}}$ over the choice of $V\sim\nu$.
\end{res}

\begin{res}[Worst-case to average reduction for approximate computation of probabilities]
Under the notation given in Result \ref{res:exactRED} we have that it is $\sharP$-hard to approximate probability $p_{\x_0}(V,\inpsi)$ to within accuracy $\ep=\exp(-\Theta(N^6))$ with probability greater than $1-o(N^{-2})$ over the choice of $V\sim\nu$.
\end{res}
{Note that above results are incomparable in the sense that Result 2 cannot be decuced form 3 and vice versa. This is because they make statements about hardness of outcome probabilities for different fractions of circuits.}
Formal versions of the above are given in  Theorems \ref{thm:average-hardness-exact-prob} and \ref{thm:average-hardness-approx-prob}. To obtain the above result we generalize the method developed recently by Movassagh  \cite{movassagh_quantum_2020} in the context of random quantum circuits. The key technical ingredient a Cayley path, which gives rational interpolation between quantum circuits. We realize that, for the purpose of the two reductions given above, it is possible to apply it directly on one the level of the Lie group underlying a particular class of FLO transformations ($\U(d)$ and $\SO(2d)$ for passive and active FLO respectively). We then use the fact that fermionic representations can be realized low degree of polynomials in entries of matrices of appropriate symmetry groups. This observation allows us to adapt the the reduction method of Movassagh with relative ease.

{
\begin{rem}
In the course of the proof of the above result we have realised a technical issue in Movassagh \cite{movassagh_quantum_2020}. Correction of the proof gives worse then claimed tolerance for error $\ep=\exp(-\theta(N^{4.5}))$ (for the Google layout), which is still better then the one claimed here. On the other hand, application of our reduction in conjunction with recent improvements over the Paturi lemma that appeared in the paper by Bouland \emph{et. al} \cite{Bouland2021,Kondo2021_robustness} (published after the completion of this work) boost error tolerance of our scheme to $\ep=\exp(-\theta(N^{2}\log(N)))$.
\end{rem}}



Finally, the experimental feasibility of our proposal is further increased by the fact that due to the structure of FLO circuits, they can be efficiently certified using resources scaling polynomially with the system size. 

\begin{res}[Efficient tomography of FLO circuits]
Let $V$ be an unknown active FLO circuit on a system of $d$ qubits that encodes $d$ fermionic modes. Assume we have access to computational basis measurements  and single qubit gates. Then $V$ can be estimated  up to accuracy $\ep$  in the diamond norm by repeating $r\approx \frac{d^3}{\ep^2}$ rounds of experiments, each involving $O(d^2)$ independent single qubit state preparations and single Pauli measurements at the end of the circuit.
\end{res}

The rigorous formulation of the above, together with the explicit protocol for carrying out tomography, is provided in Section \ref{sec:cert}. Importantly, our method avoids exponential scaling inherent to the general multiqubit tomography protocols. Moreover, it can can be also viewed as a fermionic analogue of the certification methods developed previously in the context of photonics and bosonic linear optics \cite{Lobino2008,OptQPT2011,LinOptTom2013}.

{\bf Implementation of the scheme}

\begin{figure*}
      \centering
    \includegraphics[scale=0.65]{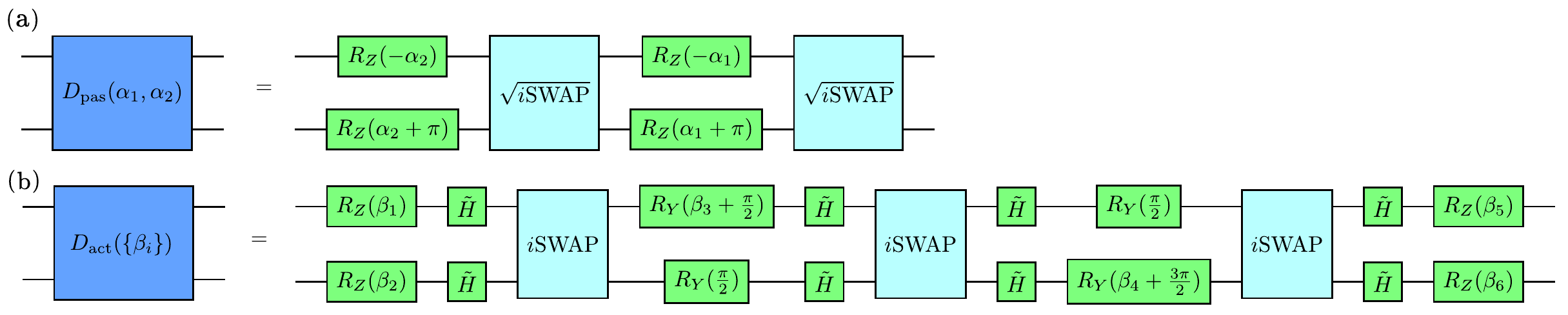}
    \caption{Decomposition of (a) the Givens rotation in the passive FLO setting, $D_{\textrm{pas}}(\alpha_1, \alpha_2)$,  in terms of $\sqrt{i\textsc{swap}}$ gates and (b) the merged Givens rotations in the active FLO setting, $D_{\textrm{act}}(\{ \beta_i \})$, in terms of $i\textsc{swap}$ gates. $R_{W}( \alpha)$ (with $W \in \{X, Y, Z\}$) denotes the one-qubit rotation gate $\e^{i \alpha W}$. The gates $\tilde{H}$ are defined by the relations $\tilde{H}Z\tilde{H}=Y$ and $\tilde{H}Y\tilde{H}=Z$.
    }
    \label{fig:theta}
\end{figure*}

It is important to stress that our proposal has a strong potential for experimental realization, e.g., on quantum processors with superconducting  qubit architectures. The actual implementation should be feasible already on near-term quantum devices, as the construction of parametric programmable passive linear optical circuits, due to their relevance in Quantum Chemistry, has been already experimentally demonstrated on Google's Sycamore quantum processor \cite{chemistryGOOGLE2020}.

The preparation of the input fermionic magic state $\psiQUAD^{\otimes N}$, vital to our proposal, can be performed by applying on the  computational basis state $\ket{0}^{\otimes 4N}$ a  simple constant depth circuit consisting of 3 CNOTs and 3 one-qubit gates per quadruple blocks of qubits as shown in Fig.~\ref{fig:magic-input-and-decomp}. 
One can implement an arbitrary passive FLO (or {\it basis rotation} in the Quantum Chemistry lingo) in linear depth  using only nearest neighbor gates and assuming a  minimal  linearly connected architecture \cite{fermPASSlayout2018, fermACTpassLAYOUT2018}. 
Two such layouts are depicted in Fig.~\ref{fig:brick_triangle}. In terms of two-qubit gates, the triangle layout has a depth of $d{-}1$, while the depth of the brickwall layout is only $d/2$.
These circuits are analogous to the layouts of Boson Sampling circuits \cite{Reck1994, ReckUniform2016} and are based on decomposing a unitary $U \in \U(d)$ into  individual Givens rotations, which we describe in Appendix~\ref{app:decomp}. In the passive FLO case the two qubit gates have the form:
\begin{equation} \label{eq:D_pas}
    D_{\mathrm{pas}}(\alpha_1, \alpha_2) = (\e^{-i \alpha_1 Z_1/2}\e^{i \alpha_1 Z_2/2}) \; \e^{i \alpha_2 (X_1X_2 + Y_1Y_2)/2} ,
\end{equation}
and the final one qubit gates are $Z$ rotations. The triangle and the brickwall layout can also be used to decompose an arbitrary active FLO operation \cite{fermACTpassLAYOUT2018, fermGAUSSlayout2019}, however, in this case the two-qubit gates will have a more complicated structure (as they arise from merging several Givens rotations): 
\begin{equation} \label{eq:D_act}
\begin{split}
    D_{\mathrm{act}}(\{ \beta_i \}) &= 
    (\e^{i \beta_5  Z_1/2}\e^{i \beta_6  Z_2/2}) \,
      \e^{i (\beta_3 X_1X_2 + \beta_4 Y_1Y_2)/2} \,  \\
      &\times(\e^{i \beta_1  Z_1/2}  \e^{i \beta_2  Z_2/2}) 
    ,
\end{split}
\end{equation}
and the one-qubit unitaries at the end of the circuit are either  Pauli matrices or identities. The derivation of these statements are given in Appendix~\ref{app:decomp},
they are based on the decomposition of arbitrary elements of the $\U(d)$ and $\SO(2d)$ groups into a sequence of nearest-neighbor Givens rotations and a diagonal matrix. The passive FLO representation of the Givens rotations and of the diagonal matrix are then translated to two-qubit gates of the type $D_{\mathrm{pas}}(\alpha_1, \alpha_2)$ and a series single-qubit $Z$ rotations in a layout the depicted Fig.~\ref{fig:brick_triangle}.  In the active FLO case several represented Givens rotations are merged into two-qubit gates of type $D_{\mathrm{act}}(\{\beta_i\})$ and the  single-qubit unitaries at the end of the circuit are Pauli gates.

In the experimental demonstration of  programmable passive FLO transformations by the Google team \cite{chemistryGOOGLE2020}, the  native gates of the Sycamore processors, the $\sqrt{i\textsc{swap}}$ gates and single-qubit $Z$ rotations, were used. The $i\textsc{swap}$ and $\sqrt{i\textsc{swap}}$ gates, defined as
\begin{equation}
     i\textsc{swap}= \begin{pmatrix}
1 & 0 & 0 & 0 \\
0 & 0 & -i & 0 \\
0 & -i & 0 &  0\\
0 & 0 & 0 & 1
\end{pmatrix} \ , 
\sqrt{i\textsc{swap}}=\begin{pmatrix}
1 & 0 & 0 & 0 \\
0 & \frac{1}{\sqrt{2}} & \frac{-i}{\sqrt{2}} & 0 \\
0 & \frac{-i}{\sqrt{2}} & \frac{1}{\sqrt{2}} &  0\\
0 & 0 & 0 & 1 
\end{pmatrix},
\end{equation}
were exactly introduced in quantum computing as standard gates because they are native in superconducting qubit architectures \cite{schuch2003natural}.
It is important to note that these gates are actually FLO gates. This lucky coincidence  supports the feasibility of our proposal, since the Givens rotations for passive FLO, i.e., the two-qubit gates $D_{\mathrm{pas}}(\alpha_1, \alpha_2)$ in the layouts of Figs.~\ref{fig:brick_triangle}, can be decomposed into two $\sqrt{i\textsc{swap}}$ gates and four single qubit $Z$-rotations as depicted in part (a) of Fig.~\ref{fig:theta}.  The two-qubit gates used in the active FLO setup, $D_{\mathrm{act}}(\{ \beta_i\})$, can be decomposed into 3 $i\textsc{swap}$ gates shown in part (b) of Fig.~\ref{fig:theta}. 

{
In Result \ref{res:anticoncentration} the anticoncentration of passive and active FLO circuits is presented. This result gives for active FLO random circuits an estimation of the constant $\Cact$ which is computed by bounding the expectation values from the Paley-Zygmund inequality $\frac{(\E X)^2}{\E X^2}$, with $X=|\bra{\x}V\ket{\inpsi}|^2$ where $V$ is the random circuit. As shown in Figure \ref{fig:brick_triangle}, we have assumed that the random circuits have linear depth with respect to the number of modes, nonetheless it is still possible that the anticoncentration property is obtained for lower depths. Note that at shallow depth the so called property of data hiding is also present, this property states that given $|\bra{\x}V\ket{\inpsi}|^2$ and a fixed state $\ket{\x_0}$ there is a FLO circuit $V_{\x}$ such that $\abs{\bra{\x} V\ket{\inpsi}}^2 = \abs{\bra{\x_0} V_{\x}\ket{\inpsi}}^2$. For more detail of the data hiding property see Lemma \ref{lem:hidingPROP}. To test that anticoncentration is obtained at shallow depth, we have simulated noiseless Fermion sampling experiments with up to 28 qubits using the IBM Qiskit simulator \cite{Qiskit}. As in Figure~\ref{fig:magic-input-and-decomp}, we have initialized the input with a fixed number of quadruples $\psiQUAD$. We performed simulations for $1,2,\dots, 7$ quadruples, for each of these, we computed the output probability of a fixed output string for an average of $26000$ circuits for each number of quadruples. The circuits used have the brickwall architecture from Figure \ref{fig:brick_triangle}. From each random circuit the output probability of a fixed output string was computed. With the obtained output probabilities, the expectation values in the Paley-Zygmund inequality $\frac{(\E X)^2}{\E X^2}$ were computed. In Figure \ref{fig:anticon-sim} we plot for each fixed number of quadruples the minimum depth required for the ratio of expectation values to surpass a certain threshold and compare to the case where the depth is linear with respect to the number of modes. The simulations suggest that the depths required to obtain anticoncentration are shallow as the number of modes is increased, wherein the case for $7$ initial quadruples and a threshold of $0.4$ the number of layers required is $7$. Moreover, in the random circuit sampling set up it has been found that anticoncentration is obtained in logarithmic depth, our simulations show that this could also be the case in our setup.}

\begin{figure}
    \centering
    \includegraphics[scale=0.6]{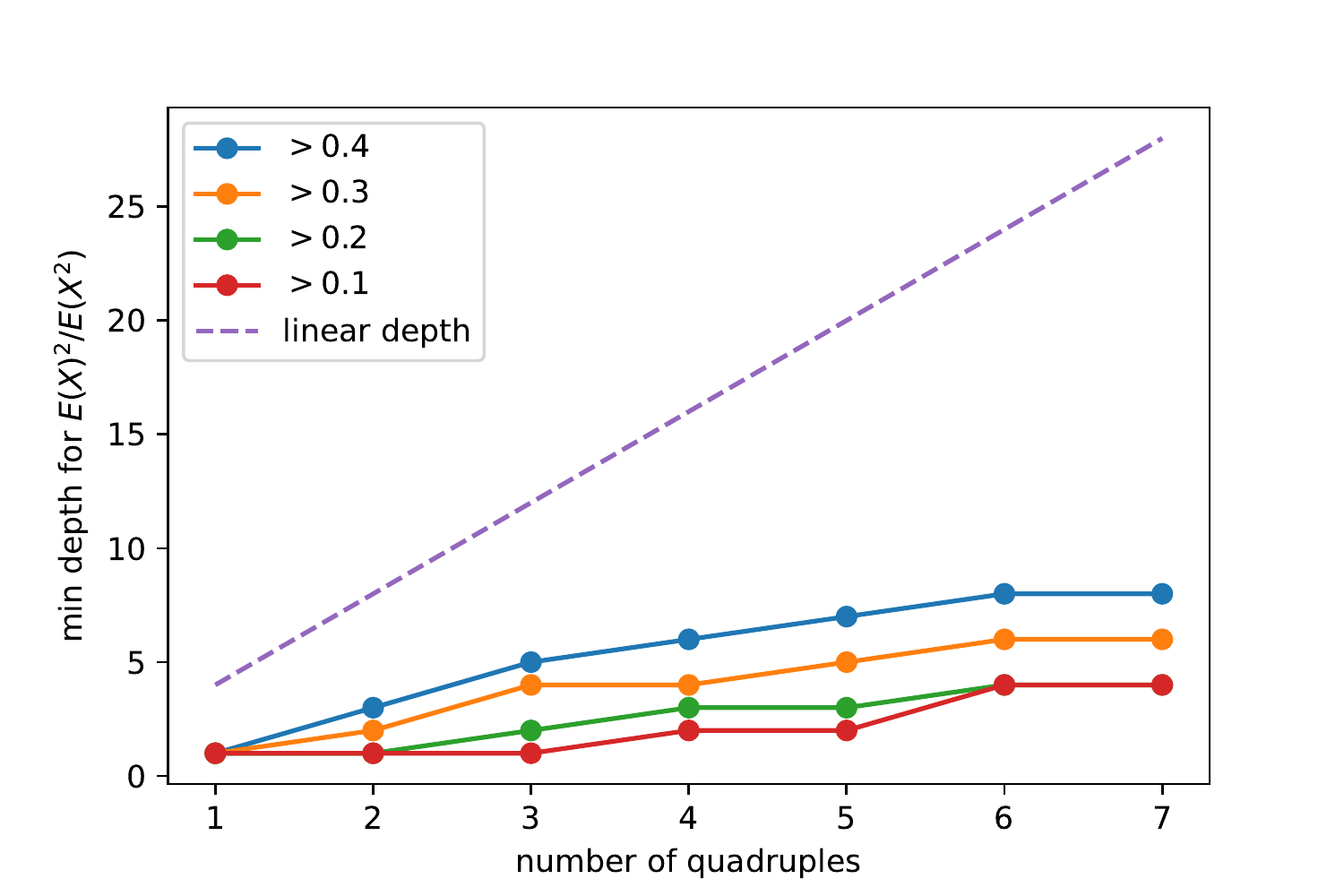}
    \caption{{Plots for the minimum depth required in circuits of a fixed number of quadruples for the ratio $\frac{(\E X)^2}{\E X^2}$, with $X=|\bra{\x}V\ket{\inpsi}|^2$ ($V$ is the random circuit and $\ket{\inpsi}$ is the quadruple input) to surpass a threshold defined in the legend. For comparison the linear depth in the number of modes is shown. The data suggests that the depths required to obtain the anticoncentration property scale sublinearly.}} 
    \label{fig:anticon-sim}
\end{figure}




\section{Discussion and Open problems}\label{sec:concl}

We believe that our results and techniques used to establish them will be of relevance also for problems not directly related to Fermion Sampling.  
The first group of potential applications is related to the structure of our quantum advantage scheme. As discussed in the introduction, quantum advantage proposals are typically not constructed because of their practical usefulness. However, recently several proposals for  applications of  Boson Sampling \cite{Aaronson2013} and Gaussian Boson Sampling \cite{hamilton2017gaussian, lund2014boson} have been suggested. These include combinatorial optimization problems \cite{arrazola2018using, arrazola2018quantum}, calculation of Franck–Condon profiles for vibronic spectra \cite{huh2015boson, huh2017vibronic},    molecular docking \cite{banchi2020molecular} and machine learning using graph kernels \cite{schuld2020measuring}.        { Admittedly a feasible scalability is unlikely for those applications that should produce as an answer  one specific bitstring (or a few specific bitstrings), as the output probabilities are generically exponentially suppressed. However, different type of applications than those, like dense graph sampling or using graph samples for graph-kernel methods, might turn out to be useful.}  It should be mentioned that all of those application are based on the fact that this type of sampling, unlike the generic Random  Circuit Sampling, have a very structured nature, as
they are based that one can sample with probabilities proportional to the permanents and Hafnians of certain matrices constructed from the circuit description. In particular, most of the mentioned applications use the fact that one can sample with probabilities proportional to the permanents and Hafnians of certain matrices constructed from the circuit description. These polynomial functions of matrix entries encode interesting properties, e.g., for adjacency matrices of graphs they provide the number of cycle covers and perfect matchings of the graph, respectively. In our proposal  similar polynomials appear when describing the sampling probabilities: the mixed discriminant \cite{Ivanov2020} and their generalizations (e.g., mixed Pfaffian \cite{ikai2011theory}), which also encode important graph properties. Thus, we have good reasons to believe that our proposal, beside providing a robust computational advantage setup, also can be used for other algorithms with interesting applications.

 All of those application are based on the fact that this type of sampling, unlike the generic Random  Circuit Sampling, have a very structured nature, as
they are based that one can sample with probabilities proportional to the permanents and Hafnians of certain matrices constructed from the circuit description. In particular, most of the mentioned applications use the fact that one can sample with probabilities proportional to the permanents and Hafnians of certain matrices constructed from the circuit description. These polynomial functions of matrix entries encode interesting properties, e.g., for adjacency matrices of graphs they provide the number of cycle covers and perfect matchings of the graph, respectively. In our proposal  similar polynomials appear when describing the sampling probabilities: the mixed discriminant \cite{Ivanov2020} and their generalizations (e.g., mixed Pfaffian \cite{ikai2011theory}), which also encode important graph properties. Thus, we have good reasons to believe that our proposal, beside providing a robust computational advantage setup, also can be used for other algorithms with interesting applications.

We also want to emphasize the universality of our techniques for establishing anticoncentration and  worst-to-average-case reduction for structured random circuits.  Our anticoncentration result exploits group-theoretic properties of fermionic circuits and can be likely generalized to other scenarios where low-dimensional group structure appears. Moreover, our generalization of the worst-to-average-case reduction of Movassagh \cite{movassagh_quantum_2020} can be applied to any sampling problem provided outcome probabilities can  be interpreted as polynomials on low-dimensional groups (Cayley transformation that underlines the reduction can be defined on arbitrary Lie groups).  For example, one can view Boson Sampling in the first quantization picture where the group of linear optical networks acts on the totally symmetric subspace of $N$ qudits, where $d$ is the number of modes (or other representation when bosons are partially distinguishable \cite{moylett2018quantum}), as opposed to the totally antisymmetric representation for fermions. 
   
We conclude this part with stating a number of interesting problems that require further study.

 \emph{Role of non-Gaussianity for hardness and anticoncentration}.
    In practice magic states are not perfect and a high level of noise can bring these states into the convex hull of Gaussian states \cite{bravyi_universal_2006,melo_power_2013,oszmaniec_classical_2014}, which we know are efficiently simulable classically under FLO evolutions and Gaussian measurements. How much noise does our hardness result tolerate? The same question applies to anticoncentration of output probabilities, in which case we do not yet have a proof that FLO circuits with Gaussian inputs do not anticoncentrate, but we have numerical evidence that proofs based on the Paley-Zygmund inequality do not work (see Remark~\ref{rem:gaussian-anti-open-problem}). 
    
 \emph{Fermion Sampling with less magic}:
    In our scheme, all the input qubit lines are injected with magic states.
    Do the hardness result and anticoncentration hold if we use only, say, $O(\log m)$ magic states? 
    
\emph{Algorithms for classical simulation.}
    Devising algorithms to approximately simulate FLO circuits with magic input states on average would not only lead to useful applications (for example in the context of quantum chemistry), but is also vital to understand the complexity landscape of random FLO circuits.
    For RCS scheme employed in Google's experiment with qubits placed on a 2D grid, advances in classical simulation techniques imposed a limit on the robustness of the average-case hardness that can be achieved with the current worst-to-average-case reduction that is agnostic to the circuit architecture and depth \cite{Napp2020}.
    
 \emph{Tomography and certification of FLO circuits and Fermion Samping. }
    In this work, we only gave an efficient method to estimate an unknown FLO circuit (A related benchmarking of FLO circuits was recently proposed in \cite{helsen2020matchgateRB}.). It is interesting to extend our scheme beyond unitary circuits and to devise a method for which sample complexity and number of experimental settings exhibit an optimal scaling with the system size.  
    Additionally, with further assumption on how the quantum device operates (e.g. the noise model), is there a simple diagnostic tool for Fermion Sampling similar to cross-entropy benchmarking for RCS \cite{Boixo2016}?

\section{Anticoncentration of FLO circuits}\label{sec:anticoncentration}

In this section we prove that outcome probabilities in  fermionic circuits initialized in the state $\inpsi$ anticoncentrate for Haar random fermionic linear optical circuits. We prove anticoncentration for both passive and active fermionic linear optics. Our proof is based on interpretation of these circuits in terms of representation of group $\U(d)$ and $\SO(2d)$, where $d$ is the number of fermionic modes used. 

Let $\H$ be a Hilbert space and let $\lbrace\ket{\x}\rbrace$ be a fixed (computational) basis of $\H$. For $V\in \U(\H)$ and a pure state $\ket{\Psi}$.  In what follows we will denote by $p_{\x}(V,\Psi)$ the probability of obtaining outcome $\x$ on some input state $\ket{\Psi}$ on which unitary $V$ was applied. Born rule implies:
\begin{equation}\label{eq:bornRULE}
    p_{\x}(V,\Psi) =|\bra{\x}V\ket{\Psi}|^2\ .
\end{equation}
In what follows we restrict our attention to $\H=\bigwedge^{2N}(\C^{4N})$ (for passive FLO) and $\H=\Hfock^+(\C^{4N})$ (for active FLO. Moreover for $\x\in\lbrace0,1\rbrace^{4N}$ vectors $\ket{\x}$  will denote standard Fock states (cf. Section \ref{sec:notations}). In both of the cases considered the set of allowed $\x$ is different (See Theorem~\ref{thm:anticoncentrationFLO} for more details).

\begin{Definition}[Anticoncentration of ensemble of unitary matrices]\label{def:anticoncentration}
Let $\nu$ be an ensemble (probability distribution) of unitary matrices $\U(\H)$. We say that that $\nu$ exhibits anticoncentration on input state $\ket{\Psi}$ if and only if for every outcome $\x$ of computational basis measurement 
\begin{equation}\label{eq:anticoncentrationDEF}
    \Pr_{V\sim \nu} \left(p_{\x}(V,\Psi) > \frac{\alpha}{|\H|} \right) > \beta\ ,
\end{equation}
where $\alpha,\beta$ are positive constants.
\end{Definition}
\begin{rem}
In this work  we will be concerned with families of probability distributions that are defined on  Hilbert spaces of increasing dimension, parametrized by the total number of fermionic modes $d$. In this context, motivated by structure of the proof of hardness of sampling (see Theorem \ref{th:SAMPLhar})  we will be interested in cases when $\alpha,\beta=\Theta(1)$ i.e are independent on $|\H|$. 
\end{rem}

Below we state our main result regarding anticoncentration of fermionic linear circuits initialized in the tensor product of Fermionic magic states    
\begin{equation}
     \ket{\inpsi}= \psiQUAD^{\ot N},
\end{equation}
where $\psiQUAD=\frac{1}{\sqrt{2}}(\ket{0011} + \ket{1100})$. 
Note that $\ket{\inpsi}\in\bigwedge^{2N}(\C^{4N})$.

\begin{thm}[Anticoncentration for fermionic linear optical circuits initialized in product of magic states]\label{thm:anticoncentrationFLO}
Let $\Hpas=\bigwedge^{2N}(\C^{4N})$ and let $\Hact=\Hfock^+(\C^{4N})$ be Hilbert spaces describing $2N$ Fermions in $4N$ modes and positive parity Fermions in $4N$ modes. Let $\Gpas$ and $\Gact$ be respectively passive and active FLO transformations acting on the respective Hilbert spaces and distributed according to the uniform measures  $\nupas$ and $\nuact$ (see Section \ref{sec:notations}). Let $\ket{\inpsi}$ be the initial state to which both families of circuits are applied.  Then, for every $\x$ of Hamming weight $|\x|=2N$ we have
\begin{equation}\label{eq:anticoncentrationPASSIVE}
    \Pr_{V\sim \nupas} \left( p_{\x}(V,\inpsi) > \frac{\alpha}{|\Hpas|} \right) > \frac{(1-\alpha)^2}{C_\mathrm{pas}}\ ,
\end{equation}
where $\Cpas = \Cpasval$ and $|\Hpas|=\binom{4N}{2N}$. Moreover, for every $\x$ with even Hamming weight we have
\begin{equation}\label{eq:anticoncentrationACTIVE}
    \Pr_{V\sim \nuact} \left( p_{\x}(V,\inpsi) > \frac{\alpha}{|\Hact|} \right) > \frac{(1-\alpha)^2}{C_\mathrm{act}}\ ,
\end{equation}
where $\Cact=\Cactval$ and $|\Hact|=2^{4N}/2$. 
\end{thm}

\begin{proof}
In order to prove Eq. \eqref{eq:anticoncentrationPASSIVE} and Eq. \eqref{eq:anticoncentrationPASSIVE} we start with a standard tool used when proving anticoncentration - the  Paley-Zygmund inequality. It states that for arbitrary nonnegative bounded random variable $X$ and for $0<\alpha<1$, we have
\begin{equation}
	\Pr_X (X > \alpha \E X)  \ge (1-\alpha)^2 \frac{(\E X)^2}{\E X^2}\  .
\end{equation}

We use this bound for $X=|\bra{\x}V\ket{\inpsi}|^2$, where $V\sim \nupas$ or $V\sim \nupas$. Recall that, as explained in the Section \ref{sec:notations}, linear circuits $\Gpas$ and $\Gact$  can be understood in terms of representations of symmetry groups  $\U(d)$ and $\SO(2d)$. Haar measures on these symmetry groups induce uniform distributions on the $\Gpas$ and $\Gact $. Therefore for $k=1,2$ we have
\begin{equation}\label{eq:momentHAAR}
  \underset{V\sim\nu}{\E} \left[ p_{\x}(V,\inpsi)^k \right] {=} \int_{G}d\mu(g) \left[ \tr(\ketbra{\x}{\x} \Pi(g)\inpsi\Pi(g)^\dag) \right]^k ,
\end{equation}
where $\mu$ is the Haar measure on a Lie group $G$, and $\Pi$ is a unitary representation of $G$ in a suitable Hilbert space $\H$. The case of passive FLO corresponds to $G=\U(4N)$ and $\Pi=\Pipas$ while for active FLO we have $G=\SO(8N)$ and  $\Pi=\Piact$ (c.f.  Eq.\eqref{eq:pasREP} and Eq. \eqref{eq:actREP}). Both groups are irreducibly represented in Hilbert spaces $\Hact$ and $\Hpas$ by virtue of Schur lemma unitaries $\Pi(g)$ forming a $1$-design. Consequently
\begin{equation}\label{eq:firstMOMENT}
\begin{split}
     \underset{V\sim\nu}{\E} &\left[ p_{\x}(V,\inpsi) \right] \\  &= \int_{G}d\mu(g) \left[ \tr(\ketbra{\x}{\x} \Pi(g)\inpsi\Pi(g)^\dag) \right]  \\
     &= \tr(\ketbra{\x}{\x} \int_{G}d\mu(g)\left[ \Pi(g)\inpsi\Pi(g)^\dag \right] ) 
     =  \frac{1}{|\H|}\ ,
\end{split}
\end{equation}
where in the last equality we used the $1$-design property and the fact that $\ket{\x}\in \H$ is a normalized vector. Computation of the second moment can be greatly simplified by the usage of group theory. Let us first rewrite $\underset{V\sim\nu}{\E} \left[ p_{\x}(V,\inpsi)^2 \right]$ in the form convenient for computation:
\begin{equation}\label{eq:secondMoment}
\begin{split}
    \underset{V\sim\nu}{\E} &\left[ p_{\x}(V,\inpsi)^2 \right] \\
    &= \int_{G}d\mu(g) \left[ \tr(\ketbra{\x}{\x}^{\ot 2}  \Pi(g)^{\ot 2}  \inpsi^{\ot 2} (\Pi(g)^\dag)^{\ot 2}  \right] \\
    &= \tr( A_{\Pi,G}\inpsi \ot \inpsi)\ ,
\end{split}
\end{equation}
where, due to unitarity of representation $\Pi$, and invariance of Haar measure under transformation $g\mapsto g^{-1}$  
\begin{equation}\label{eq:almostPROJECTOR}
    A_{\Pi,G}= \int_{G}d\mu(g) \left[\Pi(g)^{\ot 2} \ketbra{\x}{\x}^{\ot 2} (\Pi(g)^\dag)^{\ot 2}  \right]\ .
\end{equation}
Operator $ A_{\Pi,G}$ acts on two copies of the original Hilbert space, $\H\ot \H$ and is a manifestly $G$ invariant in the sense that for all $g$ we have $[A_{G,\Pi},\Pi(g)^{\ot 2}]=0$. The integration in \eqref{eq:almostPROJECTOR} can be carried out explicitly because objects in question have very specific properties that are rooted in the fact that Fock states constitute generalized coherent states of the considered representations of $\U(4N)$ and $\SO(8N)$ (cf. Remark \ref{rem:coherentSTATESsemisimple} for more details). Let $\ket{\Psi}$ be a fixed pure state in $\H$ and let $\ket{\x}$ be a fixed fermionic Fock state belonging to appropriate Hilbert space $\H$
\begin{equation}\label{eq:CoherentSTATErep}
\exists g\in G\ \text{s.t. } \Psi=\Pi(g) \ketbra{\x}{\x} \Pi(g)^\dag \iff \ket{\Psi}^{\ot 2} \in \tilde{\H}\ ,
\end{equation}
where $\tilde{\H}\subset \H\ot\H$ is the carrier space of certain unique irreducible representation of $G$. In other words, it appears as one of the irreducible representations in the decomposition of the space $\H\ot \H$, where $G$ is represented via $g\mapsto\Pi(g)^{\ot 2}$. Let $\tilde{\P}$ be the orthonormal projector onto $\tilde{\H}\subset\H\ot \H$. Due to property \eqref{eq:CoherentSTATErep} we get that $\supp(A_{\Pi,G})\subset\tilde{\H}$. Combining this with $G$-invariance of $A_{\Pi,G}$ we get, using Schur lemma, that $A_{\Pi,G}$ must be proportional to $\tilde{\P}$. The proportionality constant follows easily from normalization of $A_{\Pi,G}$. Putting this all together we obtain
\begin{equation}\label{eq:projectorAVER}
    A_{\Pi,G} =\frac{1}{|\tilde{\H}|}\tilde{\P}\ .
\end{equation}
Inserting the expressions for the first and second moments to Payley-Zygmund inequality we get
\begin{equation} \label{eq:anticoncABS}
	\Pr_{V\sim \nu}\left( p_{\x}(V,\inpsi)  > \frac{\alpha}{|\H|}\right)  \ge (1-\alpha)^2 \frac{|\tilde{\H}|}{|\H|^2 }\frac{1}{\tr(\tilde{\P} \inpsi\ot \inpsi )}\  .
\end{equation}
From the above expression it is clear that anticoncentration is controlled by: (i) the ratio of $\frac{|\tilde{\H}|}{|\H|^2 }$ and (ii) expectation value $\tr(\tilde{\P} \inpsi\ot \inpsi)$. We will give explicit forms of the projectors $\tilde{\P}$, as well-as the dimensions $|\tilde{\H}|$ for both passive and active FLO in Lemmas \ref{lem:PROJpassFERMpurities} and \ref{lem:PROJactive} in the Appendix \ref{app:comp}. From the expressions given there we obtain \footnote{To simplify expressions for active FLO, we use the bound $\binom{4N}{2N} \geq \frac{2^{8N}}{\sqrt{\pi 4N}}$. } 
\begin{equation}
    \frac{|\tHpas|}{|\Hpas|^2} \geq \frac{1}{N}\ ,\  \frac{|\tHact |}{|\Hact|^2} \geq  \frac{1}{\sqrt{\pi N}}\ .
\end{equation}
For passive FLO this gives (recall that $|\x|=2N$ and $\Hpas=\bigwedge^{2N}(\C^{4N})$)  
\begin{equation}\label{eq:semiAntiPASSIVE}
\begin{split}
    	\Pr_{V\sim \nupas}&\left( p_{\x}(V,\inpsi) 
    	> \frac{\alpha}{\binom{4N}{2N}}\right) \\ 
    	&\ge (1-\alpha)^2 \frac{1}{N}\frac{1}{\tr(\Pfer \inpsi\ot \inpsi )}\ .
\end{split}
\end{equation}
Similarly, for active FLO we obtain (recall that $|\x|$ is even and $\Hact=\Hfock^+(\C^{4N})$) we have
\begin{equation}\label{eq:semiAntiactive}
\begin{split}
    	\Pr_{V\sim \nuact}&\left( p_{\x}(V,\inpsi)  > \frac{\alpha}{2^{4N-1}}\right) \\  &\ge (1-\alpha)^2 \frac{1}{\sqrt{\pi N}}\frac{1}{\tr(\Pflo \inpsi\ot \inpsi )}\ .
\end{split}
\end{equation}
In order to complete the proof we need the following inequalities 
\begin{equation}\label{eq:EPupperbound}
 \tr(\Pfer \inpsi\ot \inpsi ) \leq \frac{\Cpas}{N}\ \ , \ \    \tr(\Pflo \inpsi\ot \inpsi ) \leq \frac{\Cact}{\sqrt{\pi N}}\ .
\end{equation}
Proof of the above relies on the explicit form of the projectors as well as some combinatorial  considerations. The details are given in the appendix (see specifically Lemma \ref{lem:passPROJfinal} for  of passive FLO and Lemma \ref{lem:actPROJfinal} for active FLO ). 
\end{proof}

\begin{rem}
In the course of proving \eqref{eq:EPupperbound} in Appendices we arrive at upper bounds on $ \tr(\Pfer \inpsi\ot \inpsi )$ and $ \tr(\Pflo \inpsi\ot \inpsi )$  that are efficiently computable as a function of $N$.  The numerics shown in Figure \ref{fig:pasFINbound} and Figure \ref{fig:actFINbound} strongly suggest that the bounds provided by the values $\Cpas=\Cpasval$ and $\Cact=\Cactval$ given here are not tight and can be improved by better proof techniques, specifically up to $\Cpas \leq 2.4$ and $\Cact \leq 2.7$.
\end{rem}

\begin{rem}\label{rem:coherentSTATESsemisimple}
The existence of projector $\tilde{\P}$ such that equivalence in Eq. \eqref{eq:CoherentSTATErep} holds follows form the the group-theoretical characterizations of of Slater determinants as well as pure fermionic Gaussian states with positive parity. Namely, these classes of states constitute examples the so-called \emph{generalized coherent states} of simple, compact and connected Lie groups ($\SU(d)$ and $\mathrm{Spin}(2d)$) that are irreducibly represented in the appropriate Hilbert space ($\bigwedge^n(\C^d)$ and $  \Hfock^+(\C^d)$ respectively). The fact that such classes of states can be characterized via the quadratic condition $A\ket{\Psi}^{\ot 2}=0$ is as known result in algebraic geometry \cite{WLichtenstein1982}. This was translated to the quantum information language in \cite{KusBengtsson2009}, later rephrased as \eqref{eq:CoherentSTATErep}  and used to characterize correlations in systems consisting of indistinguishable particles \cite{MOuniversalFRAME,TypicalityMO2014}  (see also Chapter 3 of \cite{OszmaniecPhd}). An equivalent characterization of pure fermionic gaussian states was also independently  discovered by Bravyi in the context of fermionic quantum information \cite{bravyi_lagrangian_2004} (see also \cite{melo_power_2013}). 
\end{rem}

\begin{rem}\label{rem:gaussian-anti-open-problem}
A curious reader may wonder whether anticoncentration holds also if FLO circuits that are initialized in free Gaussian states $\Psi_{gauss}$ (with fixed number of particles for passive FLO).  For such states $\tr(\tilde{\P} \, \Psi_{\mathrm{gauss}}\ot \Psi_{\mathrm{gauss}}) =1$ and therefore for such we cannot get strong anticoncentration inequalities using Eq. \eqref{eq:anticoncABS}. We have also tried to use higher moments in conjugation with Payley-Zygmund inequality but this did not. We leave the question whether FLO circuits anticoncentrate when acting on Gaussian states as an open problem. 
\end{rem}

\section{Hardness of sampling}\label{sec:hardness-sampling}

In this part we use anticoncentratio of FLO circuits and standard complexity-theoretic conjectures to prove classical hardness for sampling from FLO circuits initialized by magic states. We adopt to the fermionic setting standard techniques \cite{Bremner2016,Morimae2017,bermejo-vega_architectures_2018,pashayan_estimation_2020} that use the anticoncentration property to prove hardness of sampling based on conjectures about hardness of approximation of probability amplitudes $p_\x(V,\inpsi)$ to within relative error.

We start with a formal  definition of a sampling problem defined by FLO circuits initialized in magic input states.

\begin{Definition}[Fermion Sampling task]\label{def:sampling-task}
Let $\Hpas=\bigwedge^{2N}(\C^{4N})$ and let $\Hact=\Hfock^+(\C^{4N})$ be Hilbert spaces describing $2N$ Fermions in $4N$ modes and positive parity Fermions in $4N$ modes. Let $\Gpas$ and $\Gact$ be passive and active FLO transformation. Let $V$ be an  FLO circuit on the Hilbert space $\H_{\mathrm{pas}}$ or $\H_{\mathrm{act}}$ and let $p(V)$ denote probability distribution $p_\x(V,\inpsi)$. Given a description of $V$, sample from a probability distribution $q(V)$ that is $\epsilon$-close to $p(V,\Psi_{\mathrm{in}})$ in $l_1$-norm (twice the total variation distance)
\begin{align}
    \norm{p(V)-q(V)}_{1} = \sum_{\x} \abs{p_{\x}(V)-q_{\x}(V)} \le \epsilon,
\end{align}
in time $\poly{N}$. 
\end{Definition}

\begin{rem}
It is more convenient to use $l_1$-norm in place of the TVD as it appears more directly in the proof of Theorem \ref{thm:sample-computing}.
\end{rem}

It was realized in \cite{Aaronson2013,Bremner2016} that, by virtue of Stockmeyer's theorem,  the hardness of classically sampling from $p_{\x}(V,\Psi)$ up to an additive error  is connected to the hardness of computing $p_{\x}(V,\Psi)$ for most instances of $\x$ and $U$. 
In particular, the existence of a classical machine that performs the sampling task implies average-case approximation in a low level of the complexity class called the polynomial hierarchy.  To prove this fact,
we start by defining the notion of approximating in the average-case. 

\begin{Definition}
An algorithm $\O$ is said to give an $(\eta,\delta)$-multiplicative approximate of $q_{\z}$ on average over the probability distribution $\mathcal{P}$ of inputs $\z$ iff $\O$ outputs $\O_{\z}$ such that 
\begin{align}
    \Pr_{\z \sim \mathcal{P}} \left[\abs{\O_{\z} - q_{\z}} \le \eta q_{\z} \right] \ge 1-\delta \ .
\end{align}
\end{Definition}

\begin{rem}\label{rem:circuit-outcome-joint-prob}
For applications to hardness of sampling, $\z$ will generally be a tuple of inputs $(V,\x)$, an FLO circuit and a measurement outcome. Correspondingly, $\mathcal{P}$ will be the joint probability distribution $V \sim \nupas$ and $\x \sim \mathrm{unif}(\Hpas)$ in the case of passive FLO (resp. $V \sim \nuact$ and $\x \sim \mathrm{unif}(\Hact)$ in the case of active FLO), where $\x \sim \mathrm{unif}(\H)$ is the uniform distribution of outcomes restricted to the Hilbert space $\H$.
\end{rem}

We now prove the hiding property \cite{Aaronson2013,Bremner2016,bouland_complexity_2019}, of FLO circuits. This will would allow us to focus on hardness of particular outcome probability. 

\begin{lem}[Hiding property for FLO]\label{lem:hidingPROP}
Consider a fixed state $\ket{\x_0}\in \Hpas$ ($\Hact$ resp.) then for any $V$ passive FLO (active FLO resp.) and $\ket{\x}\in\Hpas$ ($\Hact$ resp.) there is a passive (active) FLO $V_\x$ such that $\abs{\bra{\x} V\ket{\inpsi}}^2 = \abs{\bra{\x_0} V_{\x}\ket{\inpsi}}^2$.
\end{lem}

\begin{proof}
It is enough to show that given $\x$ there is $V_\x$ passive (active) FLO s.t. $V_\x\ket{\x_0}=\ket{\x}$ up to a global phase. In the passive case this is achieved with gates implementing fermionic swaps $U^{[i,j]}$ such that $U^{[i,j]} f_i^\dag U^{[i,j]\dag} = f_j^\dag$ and $U^{[i,j]} f_j^\dag U^{[i,j]\dag} = f_i^\dag$, the order in which they are applied is defined by $\ket{\x}$. The same can be accomplished in active FLO case with operators $-\ii m_{2i}m_{2i+1}$ changing the number of fermions (but not parity) and quasi braiding operators  $U^{(p,q)}$ to exchange the majorana operators to the corresponding places. The quasi braidings act on majorana operators as $ U^{(p,q)} m_p (U^{(p,q)})^\dagger =  m_q$, $ U^{(p,q)} m_q (U^{(p,q)})^\dagger =  m_p$ and $ U^{(p,q)} m_x (U^{(p,q)})^\dagger =  m_x$ when $x\neq p,q$.

\end{proof}


An additional ingredient required for a quantum sampling advantage is anti-concentration which states that most output probabilities of a random circuit are sufficiently big so that the approximation error to computing the probabilities is small relative to the probabilities being computed. Both average-case hardness and anti-concentration provide robustness of the sampling task to noise.

\begin{thm}[From approximate sampling to approximately computing probabilities]\label{thm:sample-computing} 
   
   Let $\Hpas=\bigwedge^{2N}(\C^{4N})$ and let $\Hact=\Hfock^+(\C^{4N})$ be Hilbert spaces describing $2N$ Fermions in $4N$ modes and positive parity Fermions in $4N$ modes.  Consider in parallel passive FLO circuits and active FLO circuits acting on the input state $\ket{\inpsi}$. If there is a classical algorithm $\mathcal{C}$ that performs Fermion Sampling as described in Definition \ref{def:sampling-task} with the $l_1$-error $1/(64C)$, where $C$ is the constant $\Cpas=\Cpasval$ (resp. $\Cact=\Cactval$) appearing in the anticoncentration condition for passive FLO circuits (resp. active FLO circuits) in Theorem \ref{thm:anticoncentrationFLO}.

    Then there is an algorithm in $\mathrm{BPP^{NP}}$ that approximates the probability $p_{\x_0}(V,\inpsi)$ for an arbitrary but fixed fiducial outcome $\x_0$ up to multiplicative error $1/4 + o(1)$ on $1/(8C)$ fraction of FLO circuits drawn from the distribution $\nu = \nupas$ for passive FLO circuits (resp. $\nuact$ for active FLO circuits.) 
\end{thm}

    The proof of the above theorem is given in the Appendix and follows the standard reduction based on Stockmayer alghorithm  \cite{Bremner2016,pashayan_estimation_2020}. Alternatively, one could arrive at a similar result in two steps: first showing that a classical approximate sampler implies approximations up to an additive error $\epsilon/|\H|$, where $\epsilon$ is the TV distance achieved in the sampling task in the polynomial hierarchy,  then showing that anticoncentration improves the approximations to multiplicative ones \cite{bouland_complexity_2019}. The alternative proof may be beneficial when anticoncentration does not hold or is undesirable, for example, when anticoncentration renders (black box) certification of quantum advantage infeasible \cite{hangleiter2019sample}.

Armed with Theorem~\ref{thm:sample-computing}, we now state the other conjectures needed before proving the  hardness of sampling.

\begin{Conjecture}[Average-case of approximating probabilities on FLO circuits initialized in $\ket{\inpsi}$  ]\label{conj:AVERhar}

Computing a $(1/4+o(1), 1/(8C))$-multiplicative approximate to $p_{\x_0}(V,\inpsi)$ for $1/(8C)$ fraction of $V$ sampled from the Haar distribution $\nu$ is $\sharP$-hard.
($C=C_{\mathrm{pas}}, \nu=\nupas$ for passive FLO circuits and $C=C_{\mathrm{act}}, \nu=\nuact$ for active FLO circuits)
\end{Conjecture}

\begin{Conjecture}\label{conj:PHnoncollapse}
The polynomial hierarchy does not collapse. 
\end{Conjecture}

\begin{rem}\label{rem:exact_sharPhard}
The motivation for Conjecture \ref{conj:AVERhar} comes from the fact that computing exactly the probabilities is $\sharP$-hard, this can be seen by writing the output as a polynomial as in Lemma \ref{lem:ivanov_amplitude}, it has been shown that computing a permanent exactly reduces to computing this polynomial and it is know that computing the permanent exactly is $\sharP$-hard. 
\end{rem}

\begin{thm}[Hardness of sampling from FLO circuits initialized in $\ket{\inpsi}$ ]
\label{th:SAMPLhar} 
If Conjectures \ref{conj:AVERhar} and \ref{conj:PHnoncollapse} are true, then there is no efficient classical algorithm that can approximately sample with $l_1$-error $1/(64\Cpas)$ (resp. $1/(64\Cact)$) from output probability distributions induced by passive (resp. active) FLO circuits with the input given by $\ket{\inpsi}$.
\end{thm}

\begin{proof}
 By Theorem \ref{thm:sample-computing}, if there were an approximate sampler with respect to passive (resp. active) FLO circuits with input $\ket{\inpsi}$, then there would exist a algorithm $\mathrm{BPP}^{\mathrm{NP}}$ that $(1/4+o(1), 1/(8C))$-multiplicative approximates $p_{\x_0}(V,\inpsi)$ in for $1/(8C)$ fraction of passive (resp. active) FLO circuits. Where $C=C_{\mathrm{pas}}$ in the passive case and $C=C_{\mathrm{act}}$ in the active. By Conjecture \ref{conj:AVERhar} this is a  $\sharP$-hard problem. It is known \cite{lautemann1983} that $\mathrm{BPP}$ is inside the third level of the polynomial hierarchy, i.e., $\mathrm{BPP}^{\mathrm{NP}}\subseteq \Sigma_3$. By a well known result of Toda \cite{Toda1991} $\mathrm{PH} \subseteq \mathrm{P}^{\sharP}$ and thus $\mathrm{PH}\subseteq \Sigma_3$.
\end{proof}





\section{Cayley Path for unitary and orthogonal groups}\label{sec:cayley}

In this section, following \cite{movassagh_quantum_2020}, we introduce a rational interpolation between 
elements of the low-dimensional symmetry groups underlying  FLO transformations. In what follows  by $G$ we will denote either of the Lie group $\U(d)$ ore $\SO(2d)$.  The rational interpolation is constructed from the Cayley transform, which is a  rational mapping form the Lie algebra $\g$ into the corresponding group $G$.  
For both groups we give upper bounds for the total variation distance (TVD) between the Haar measure $\mu_G$ on $G$ and and its \emph{deformations} $\mu^\theta_G$ obtained via Cayley path. These bounds imply TVD  bounds between distributions of the corresponding FLO circuits. This and other technical results  established bellow will be called upon in the proof of the worst-to-average-case reduction in Section~\ref{sec:worst-to-avg-reduction}.

The Lie algebras $\u(d)$ and $\so(2d)$ of $\U(d)$ and $\SO(2d)$ are defined to be
\begin{align}
    \u(d) &= \{X \in \C^{d\times d}|X\dgg = -X\}, \\
    \so(2d) &=  \{X \in \R^{2d \times 2d}|X^{T} = -X\},
\end{align}
where $X^T$ denotes the transpose of the matrix $X$. 

\begin{rem}
We do not use the physicists' convention which requires that elements of Lie algebra $X$ satisfy $\exp(i\theta X) \in G$. Therefore, In particular, here $\u(d)$ (resp. $\so(d)$) consists of skew-Hermitian (resp. antistmmetric) matrices. 
\end{rem}

Every element $X\in\g$ defines a one-parameter path in $G$: $\lbrace\exp(\theta X)\rbrace_{\theta\in\R}$, via the exponential map, $\exp:\g\rightarrow G$. Both orthogonal and unitary groups are compact and connected, Therefore exponential map is surjective and can be used to parametrize $G$, and provides an interpolation between any two group elements. However, the interpolation is not polynomial in nature, and while it is possible to truncate the power series of $\exp$ to obtain a polynomial interpolation \cite{bouland_complexity_2019}, the resulting interpolation represents circuits that are not unitary cf.\cite{movassagh_quantum_2020}).

To remedy this \cite{movassagh_quantum_2020} employs an algebraic Cayley transformation between $\u(d)$ and $\U(d)$. This transformation can be however defined more generally as a mapping between Lie algebra ans the corresponding Lie group \cite{GW}. For our needs it is enough to consider the case of unitary and special orthogonal groups. 
\begin{Definition}
Let $G$ be $\U(d)$ or $\SO(2d)$, and let $\g$ denotes its Lie algebra. The Cayley transform is a mapping $f:\g\rightarrow G$ defined via
\begin{align}\label{eq:cayley-transform}
	f(X) = (\id-X)(\id+X)^{-1}\ . 
\end{align}
\end{Definition}

It is easy to see that the image of $f(\g)$ equals a dense subset $\tilde{G}=\SET{g\in G}{\lbrace-1\rbrace\notin\mathrm{sp}(g) }$ consisting of elements of $G$ (i.e unitary or orthogonal matrices) that do not have $-1$ in their spectrum. On $\tilde{G}$ the inverse of $f$ is well-defined. Specifically,  $f^{-1}:\tilde{G}\rightarrow \g$ is given by
\begin{equation}\label{eq:inverseCayley}
	f^{-1}(g) = (\id-g)(\id+g)^{-1}\ ,
\end{equation}
where $g \in \tilde{G}$. This explicit form of the inverse map can be verified directly from the definition of $f$. 
Cayley map defines a path deformation between $g_0\in G$ and $g_0f(X)$ as follows (see Fig.~\ref{fig:Cpath})  
\begin{figure}\label{fig:Cpath}
    \centering
    \includegraphics[scale=0.4]{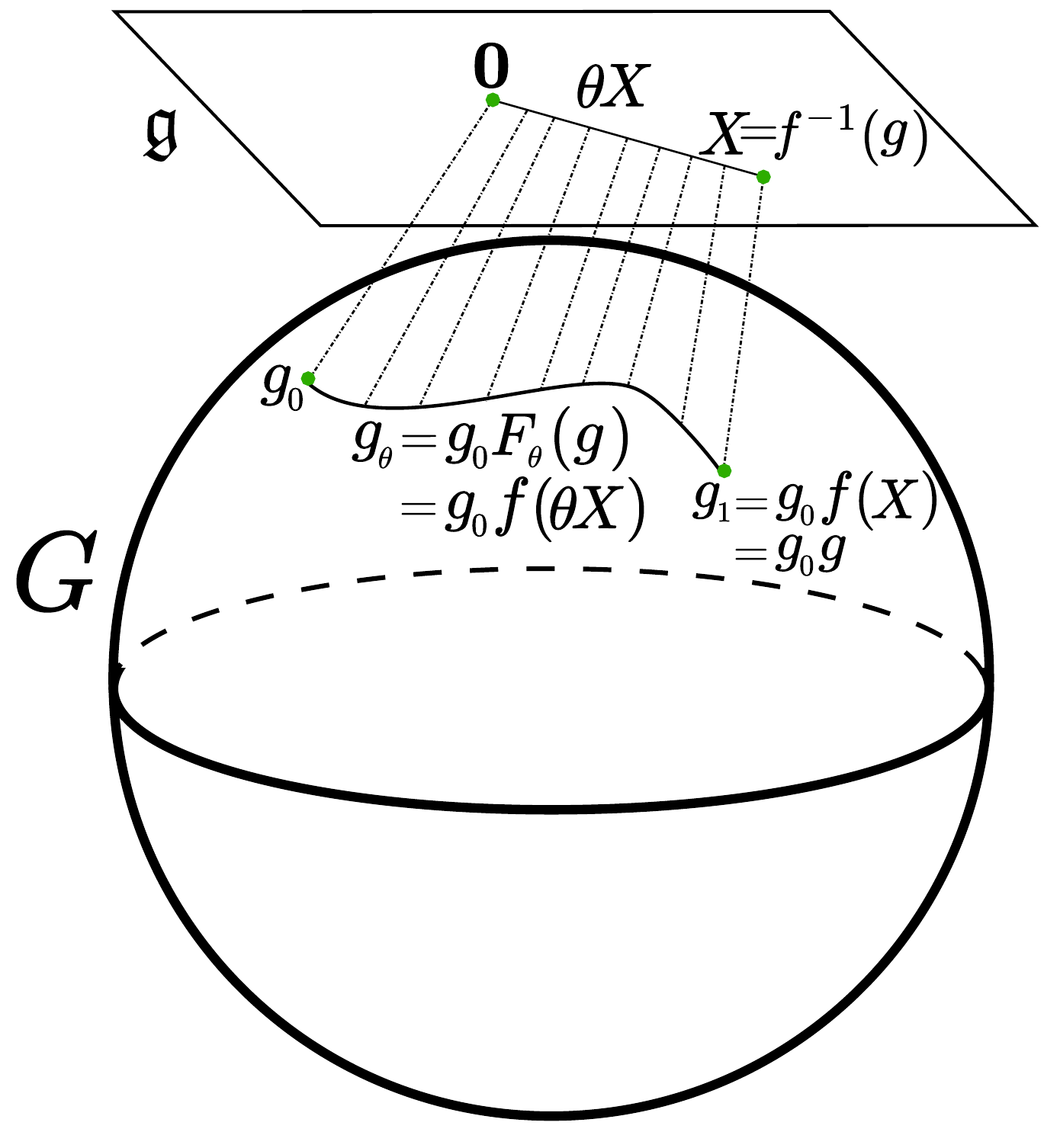}
    \caption{Path deformation defined by the Cayley map in Eq.~\eqref{eq:cayley-transform}. A path is induced between element $g_0\in G$ and $g_0 g$ by taking $X = f^{-1}(g)\in \mathfrak{g}$ and considering the perturbation $g_\theta = g_0 f(\theta X)$. 
    }
\end{figure}

Cayley map can be used to define a rational interpolation between arbitrary group elements. To this end consider first the map $F_{\theta} : \tilde{G}\rightarrow G$, given by
\begin{equation}\label{eq:deform2}
    F_\theta (g)= f(\theta f^{-1}(g)),\ \theta\in[0,1]\ .
\end{equation}
The above mapping can be evaluated explicitly (note that elements of the considered Lie groups are normal matrices and therefore functional calculus can be performed effectively in the same way as if we were dealing with functions of a real variable):
\begin{equation}\label{eq:deformTRUE}
    F_\theta (g) = \frac{(1-\theta)\id +(1+\theta) g}{ (1+\theta)\id +(1-\theta)g }\ ,\ \theta\in[0,1]\ .  
\end{equation}
For both orthogonal and unitary operators we have $\|g\|\ = 1$. Therefore for $\theta\in(0,1]$ the denominator of \eqref{eq:deformTRUE}  does not vanish and therefore we can use \eqref{eq:deformTRUE} to define  $F_\theta$ to be a function defined on whole $G$, while for any $g \in G$ we get that $\lim_{\theta \to 0} F_\theta(g)  =\id$.
Therefore for $\theta\in[0,1]$ the denominator of \eqref{eq:deformTRUE}  does not vanish and therefore we can use \eqref{eq:deformTRUE} to define  $F_\theta$ to a be a function defined on whole $G$. Importantly, for the fixed input as $\theta$ goes from $0$ to $1$ we move on a rational path form the identity $\I$ to $g$. Consequently the path
\begin{equation}\label{eq:deform}
	g_\theta = g_0 F_\theta(g),\ \theta\in[0,1]\ . 
\end{equation}
 a rational interpolation between a fixed group element $g_0$ (which can correspond, for example, to a worst-case $\sharP$-hard FLO circuit) and a completely generic group element $g_0 g$. 

It is important to note that both $f(\theta X)$ and $X$ be simultaneously brought into a block diagonal form by conjugation by elements of the group: $M\mapsto g M g^{-1}$. It follows from the fact that $f(\theta X)$ is simply a function of $X$  and than the transformation properties of elements of $\g$ under the conjugation by elements of $G$. For the case of $X\in\u(d)$ we have an elementary fact from linear algebra that there exist a unitary $U\in\U(d)$ such that 
\begin{equation}
    U X U^\dagger = \sum_{j=1}^d \phi_j X_j\ ,
\end{equation}
where $X_j = i\ketbra{j}{j}$. Similarly, for any  $X\in\so(2d)$  there exist $O\in\SO(2d)$ sush that
\begin{equation}
    O X O^T = \sum_{j=1}^d \phi_j \oX_j \ ,
\end{equation}
where $\oX_j = \ketbra{2j}{2j-1} - \ketbra{2j-1}{2j}$ is the generator for the $j$th block. These statements have analogues on the level of elements of the group. Every unitary $U$ can be transformed into a diagonal form
\begin{align}\label{eq:cartan-basis-U(d)}
    \mathrm{diag}(e^{i\phi_1},e^{i\phi_2},\dots,e^{i\phi_d})
  =
    \exp\left(\sum_{j=1}^d \phi_j X_j\right)\ .
\end{align}
For elements $\SO(2d)$, the block diagonalization amounts to the geometric fact that any $2d$-dimensional rotation can be decomposed into $d$ independent planar rotations of the form $\exp\left(\sum_{j=1}^d \phi_j \oX_j\right)$. 

\raggedbottom
The following lemma, which we prove in Appendix \ref{app:TVbound}, shows that for $1-\theta \leq o(\frac{1}{d^2})$ the distribution of elements of the group $g$, and $g_\theta =g_0 F_\theta(g)$, where $g\sim\mu_G$, are close in total variation distance.

\begin{lem}[TV distance between the Haar measure in $G$ and its $\theta$-deformation]\label{lem:TVdistanceGroup}
Let $G$ be equal to $\U(d)$ or  $SO(2d)$. Let $g_0\in G$ be a fixed element in $G$. Let $g\sim \mu_G$ an let $g_\theta=g_0 F_\theta(g)$, for $\theta\in[0,1]$ and $F_\theta:G\rightarrow G$ defined in \eqref{eq:deformTRUE}. Let now $\mu_G^\theta$ denotes the induced measure according to which $g_\theta$ is distributed. Assume furthermore that $\theta\in [1-\Delta,1]$, for $\Delta>0$. We then have 
\begin{equation}
\begin{split}
    \norm{\mu_{\U(d)} -\mu^\theta_{\U(d)}}_{\mathrm{TVD}} &\leq   d^2 \Delta/2 ,\ \\ \norm{\mu_{\SO(2d)} -\mu^\theta_{\SO(2d)}}_{\mathrm{TVD}} &\leq  d^2 \Delta/2 \ .
\end{split}
\end{equation}
\end{lem}
\begin{rem}
A similar analysis was carried out in \cite{movassagh_quantum_2020} for the case of unitary group $\U(d)$. There however considerations were carried out for $d=O(1)$. This was justified because gates in question were only single and two qubit gates. The above Lemma can be viewed as an extension of the analysis given there in the sense of allowing arbitrary relation between $d$ and $\Delta$.
\end{rem}

The robustness of the quantum supremacy claim will be tied directly to the degree of the rational functions that interpolate between quantum circuits (Appendix~\ref{sec:degree}). 
Here we give the explicit rational functions and their degrees in the Cayley-path interpolation $g_{\theta}=F_{\theta}(g)$ at the group level \eqref{eq:deformTRUE} in $\U(d)$ and $\SO(2d)$.
(A similar result for $\U(d)$ was derived in  \cite{movassagh_quantum_2020}.)

In the case of $\U(d)$, since $g$ can always be diagonalized by some element $h$: $hgh^{-1} = \sum_{j=1}^d e^{i\phi_j}\ketbra{j}{j}$, we have that 
\begin{align}
    g_{\theta}  &= \sum_{j=1}^d \frac{(1-\theta)+(1+\theta)e^{i\phi_j }}{(1+\theta)+(1-\theta)e^{i\phi_j }} g_0 h^{-1} \ketbra{j}{j}h \\
    &= \sum_{j=1}^d \frac{1 + i\theta\tan(\phi_j/2)}{1 - i\theta\tan(\phi_j/2)} g_0 h^{-1}\ketbra{j}{j}h \\
    &= \frac{1}{\Q_g(\theta)} \sum_{j=1}^d P_j(\theta) g_0 h^{-1}\ketbra{j}{j}h 
    \eqqcolon \frac{\PP_{g_0,g}(\theta)}{\Q_g(\theta)}, \label{eq:RATpoly_passive}
\end{align}
where
\begin{align}
    \Q_g(\theta) &= \prod_{j=1}^d (1 - i\theta\tan(\phi_j/2)), \label{eq:Q-U(d)} \\
    P_j (\theta) &= (1 + i\theta\tan(\phi_j/2)) \prod_{\substack{1\le k \le d \\ k\neq j}} (1 - i\theta\tan(\phi_k/2)) \label{eq:P-U(d)}
\end{align}
are both polynomials of degree $d$ in $\theta$, and $\mathcal{P}_{g_0,g}(\theta)$ is a formal polynomial that depends on the matrices $g$ and $g_0$. 

The same calculation applies to the case $\SO(2d)$, except that now each eigenspace is two-dimensional and spanned by $\id_j = \ketbra{2j-1}{2j-1} + \ketbra{2j}{2j}$ and $\oX_j =  \ketbra{2j}{2j-1} - \ketbra{2j-1}{2j}$. Again, let $hgh^{-1} = \sum_{j=1}^d (\cos\phi_j\id_j + \sin\phi_j\oX_j)$, and one has
\begin{align}
\begin{split}
    g_{\theta} 
    &= g_0  \sum_{j=1}^d 
    [1+\cos\phi_j + \theta^2(1-\cos\phi_j)]^{-1} \\
    &\times \left\{
    \left[1+\cos\phi_j-\theta^2(1-\cos\phi_j)\right]\id_j + 2\theta\sin\phi_j h^{-1} \oX_j h \right\} 
\end{split}
\\
    &= g_0 \sum_{j=1}^d \frac{[1-\theta^2\tan^2(\phi_j/2)]\id_j + 2\theta\tan(\phi/2) h^{-1} \oX_j h}{1+\theta^2\tan^2(\phi_j/2)} \label{eq:simplFORM}\\
    &= \frac{1}{\Q_g(\theta)} \sum_{j=1}^d 
    \left( 
     P_j^{\mathrm{diag}}(\theta) g_0\id_j + P_j^{\mathrm{off}}(\theta) g_0 h^{-1}\oX_j h
     \right) \label{eq:RATpoly_active} \\
     &\eqqcolon \frac{\mathcal{P}_{g_0,g}(\theta)}{\Q_g(\theta)}\ ,
\end{align}
where in \eqref{eq:simplFORM} we divided both the numerator and the denominator by $1+\cos\phi_j$ and
\begin{align}
    \Q_g(\theta) &= \prod_{j=1}^d (1+\theta^2\tan^2(\phi_j/2))\ , \label{eq:Q-SO(2d)} \\
    \begin{split}
    P_j^{\mathrm{diag}}(\theta) &= (1+\theta^2\tan^2(\phi_j/2))^2) \\ 
    &\times \prod_{\substack{1\le k \le d \\ k\neq j}} (1+\theta^2\tan^2(\phi_j/2))^2)\ ,
    \end{split} \\
    P_j^{\mathrm{off}}(\theta) &= 2\theta\tan(\phi_j/2) \prod_{\substack{1\le k \le d \\ k\neq j}} (1+\theta^2\tan^2(\phi_j/2))^2)
\end{align}
are polynomials in $\theta$ of degree $2d$, $2d$, and $2d-1$ respectively, and $\mathcal{P}_{g_0,g}(\theta)$ is a formal polynomial that depends on the matrices $g$ and $g_0$.

Below we give a lower bound for $\Q_g(\theta)$ to assure that the rational function does not blow up, and an upper bound for generic $g\in G$, which will be crucial for a robust reduction in Section~\ref{sec:worst-to-avg-reduction}. 
Note that the coefficients of the polynomial $\Q_g(\theta)$ depends only on generalized eigenvalues of $g$ ($e^{i\phi_j}$ in the unitary case and $\cos\phi_j$, $\sin\phi_j$ in the orthogonal case) and hence $Q(\theta)$ can be pre-computed in time polynomial in $d$ by diagonalizing $g$, computing each $\tan(\phi_j/2)$ which is just an algebraic function of $e^{i\phi_j}$, and computing the final result.

\begin{lem}\label{lem:DenominatorBOunds}
Let $\Q_g(\theta)$ be the polynomial in defined in \eqref{eq:Q-U(d)} for $G=\U(d)$ and in  \eqref{eq:Q-SO(2d)} for $G=\SO(2d)$. Let now $\tilde{\Delta}>0$. Then  we have the following inequalities  
\begin{align} 
 \Pr_{g\sim\mu_{\U(d)}} \left(   \abs{\Q_g(\theta)}^2 \le \left[1+\left(\frac{\theta\pi}{\tilde{\Delta}}\right)^2 \right]^{d}\ \right)\geq 1- d\frac{\tilde{\Delta}}{\pi} \ , \label{eq:denominatroUPPERboundUd}  \\
\Pr_{g\sim\mu_{\SO(2d)}} \left(    \abs{\Q_g(\theta)}^2 \le \left[1+\left(\frac{\theta\pi}{\tilde{\Delta}}\right)^{2} \right]^{2d} \right) \geq  1- d\frac{\tilde{\Delta}}{\pi} \ . \label{eq:denominatroUPPERboundSO}
\end{align}
In addition, for all $g$, $|Q_g(\theta)|^2 \ge 1$ for both $U(d)$ and $\SO(2d)$.
\end{lem}

\begin{proof}
Since $g\in G$ is Haar distributed, every generalised eigenphase $\phi_j$ is distributed uniformly on the interval $[-\pi,\pi]$ \cite{SzarekBook}. Therefore, for every $j$ we have
\begin{align}
    \Pr_{g\sim\mu_G} \left(\phi_j \in [\pi-\tilde{\Delta},\pi] \cup [-\pi,-\pi+\tilde{\Delta}] \right)
    = \frac{\tilde{\Delta}}{\pi}.
\end{align}
It is easy to verify that for $\phi_j \in [-\pi+\tilde{\Delta},\pi-\tilde{\Delta}]$ we have  $\abs{\tan(\phi_j/2)} \leq \pi/\tilde{\Delta}$. Using the union bound over different $\phi_j$, $j\in[d]$ we obtain that with probability at least $1- d\frac{\tilde{\Delta}}{\pi}$,
\begin{equation}
   \abs{\tan(\phi_j/2)} \le  \pi/\tilde{\Delta} \ \text{for all } j\in[d]\ . 
\end{equation}
Using the definition of polynomials $\Q_g(\theta)$ from Eq. \eqref{eq:Q-U(d)}  and  Eq. \eqref{eq:Q-SO(2d)}, we obtain the claimed inequalities \eqref{eq:denominatroUPPERboundUd} and \eqref{eq:denominatroUPPERboundSO}. 
\end{proof}

\begin{rem}
We believe that the inequalities stated in Lemma~\ref{lem:DenominatorBOunds} can be greatly improved by the usage of more sophisticated techniques from random matrix theory. However, for our purposes these crude estimates are sufficient (see proof of Theorem \ref{thm:average-hardness-approx-prob}).
\end{rem}

\section{Robust average-case hardness of output probabilities of fermionic circuits}\label{sec:worst-to-avg-reduction}

In this part we give strong evidence for the Conjecture \ref{conj:AVERhar} used to show classical hardness of sampling from fermionic linear circuits initialized in $\ket{\inpsi}$ (cf. Theorem \ref{th:SAMPLhar}). There we conjectured that it is $\sharP$-hard to approximate probabilities $p_\x(V,\inpsi)=|\bra{\x}V\ket{\inpsi}|^2$ of \emph{generic} FLO circuits initialized in $\ket{\inpsi}$ to relative error. To support the conjecture we prove weaker theorems showing average-case $\sharP$ hardness of exact computation of $p_\x(V,\inpsi)$ (Theorem \ref{thm:average-hardness-exact-prob}) and extend it further to average-case $\sharP$-hardness of approximating $p(\x|V,\inpsi)$ up to error $\epsilon=\exp(-\Theta(N^6))$ (Theorem \ref{thm:average-hardness-approx-prob}), where $N$ is the number of states $\psiQUAD$ used.


To establish this, we combine previously known worst-case hardness results (see discussion in Section \ref{sec:hardness-sampling}) and generalize  to our setting  the rational interpolation method based on Cayley transform introduced recently by Movassagh \cite{movassagh_quantum_2020}.
We also use the fact that both passive FLO circuits  ($\Gpas$) as well as  active FLO ($\Gact$) images of  representations $\Pi$ of low-dimensional symmetry groups $G$ (equal to $\U(d)$ or $\SO(2d)$). This implies, according to Lemma~\ref{lem:degreesOFRAT} in  Appendix~\ref{sec:degree}, that, when evaluated on the Cayley path $g_\theta$, the circuit gives rise to outcome probabilities described as rational functions of low degree. This low-degree structure allows a worst-to-average-case reduction for the family of circuits considered. Specifically, we reduce the problem of computation of the worst-case probability $p(\x|V_0,\inpsi)$ to computing  $p(\x|V,\inpsi)$, for a typical passive or active FLO circuit $V$. {On a high level, the reductions proceeds as follows 
\begin{itemize}
    \item[(i)] We use the  Cayley path to construct a low-degree rational interpolation $g_\theta$ between the worst-case element $g_0$ of $G$ and a generic element $g_{\theta=1} =g\sim \mu_G$.
    \item[(ii)] The outcome probabilities $p(\x|\Pi(g_{\theta}),\inpsi)$ are low degree rational functions of $\theta$ such that the denominator is efficiently computable and it is not too large for most group elements as shown in Corollary \ref{cor:DenominatorCircuitBOunds}.
    \item[(iii)] We then assume the existence of an oracle that can evaluate (exactly or with some precision) the outcome probabilities $p(\x|\Pi(g_{\theta}),\inpsi)$ for certain fraction of FLO circuits. We show that by evaluating the probability outputs of circuits corresponding to group elements $g_\theta$ with $\theta$ close to $1$ and by the bound for the denominator in the rational polynomial representing the probabilities, we can recover up to some small error the output probabilities of the original worst-case  $g_0\in G$ using polynomial or rational  interpolation. On the technical side this is obtained by appllying  rational Berlekamp-Welch theorem (Theorem \ref{th:BWrational}) for the exact case and  Paturi's lemma (cf. Lemma ~\ref{lem:PaturiLem}) in the approximate case  and Theorem~\ref{thm:poly-eqspace}.
    \item[(iv)] Thus, if we can compute (or approximate) the outcome probabilities for a finite fraction of circuits, then we can also compute (approximate) the outcome probabilities for the worst-case circuit. This means that $\sharP$-hardness of the worst-case implies that approximating output probabilities for a finite fraction of circuits is also  $\sharP$-hard.
\end{itemize}}

{
We first state and prove result about average-case hardness of exact computation of $p(\x|V,\inpsi)$. As stated above we will make use of the following result that guarantees that it is possible to recover an unknown rational function $F(\theta)$ from a set of its values at different points, even if some of the evaluation are erroneous. } 

\begin{thm}[Berlekamp-Welch for rational functions \cite{movassagh_quantum_2020}]\label{th:BWrational} Let $R(\theta)$ be a  rational function of degree $\deg(R)=(d_1,d_2)$. A set of points  $\S=\lbrace (\theta_1,r_1),(\theta_2,r_2),\ldots,(\theta_L,r_L)\rbrace$ specifies $R(\theta)$ uniquely  provided $L>d_1+d_2+2t$, where 
\begin{equation}
    \left| \SET{i\in[L]}{R(\theta_i)\neq r_i} \right|\leq t\ .
\end{equation}
Moreover, $F(\theta)$ can be recovered in polynomial time in $L$ and $\deg(R)$, when $\S$ is given. 
\end{thm}

Recall that in our quantum advantage scheme we have $d=4N$, $\ket{\inpsi}=\psiQUAD^{\ot N}$ , for $\psiQUAD=(\ket{0011} +\ket{1100})/\sqrt{2}$ (therefore for the case of passive FLO $n=2N$). Let now $g_0\in G$ be an element of the symmetry group such that $p_{\x_0}(V_0,\inpsi)$ is $\sharP$-hard to compute, where $V_0=\Pi(g_0)$ and $\x_0$ is the specific output state. We use a Cayley path interpolation between $g_0$ and Haar-random elements from $G$
\begin{equation}\label{eq:interTW}
    g_\theta=g_0 F_\theta(g) \ ,\  g\sim \mu_G\ .
\end{equation}
Let $\mu^\theta_G$ be the distribution of $g_\theta$ obtained in this way. In Lemma \ref{lem:TVdistanceGroup} we proved bounds for the TV distances $\|\mu_G-\mu_{G}^{\theta} \|_\mathrm{TVD}$. These bounds can be directly translated on the level of the corresponding circuits. Indeed, let  $V_\theta= \Pi(g_\theta)$  and let $\nu^\theta_G$
denote the distribution of the corresponding quantum circuits obtained by appropriate representation $\Pi$ of $G$. Since distribution of the Haar random FLO circuits $\nupas,\nupas$  are obtained in exactly the same way  we get from the monotonicity of TV distance (cf. Section \ref{sec:notations}).
\begin{equation}\label{eq:TVschemeBOUNDS}
\begin{split}
    \norm{\nupas -\nupas^\theta}_{\mathrm{TVD}} &\leq   8 N^2\Delta \ ,\\   
    \norm{\nuact -\nuact^\theta}_{\mathrm{TVD}} &\leq  8 N^2 \Delta \ ,
\end{split}
\end{equation}    
where $\theta\in[1-\Delta,1]$. Finally, from Lemma \ref{lem:degreesOFRAT}  we know that probabilities  $R(\theta) =\tr(\kb{\x_0}{\x_0} \Pi(g_\theta) \rho \Pi(g_\theta)^\dag)$ are rational functions of the deformation parameter $\theta$ of degrees 
\begin{equation}\label{eq:finalDEGREES}
\begin{split}
\text{Passive FLO: }&\ \    \deg (R) =  (16 N^2, 16 N^2)\ \ , \\
\text{Active FLO: }&\ \    \deg (R) =  (32 N^2,32 N^2)\ ,
\end{split}
\end{equation} 
and act that passive and active Fermionic linear circuits  
gives an average to worst-case reduction for outcome probabilities generated by fermionic circuits. 


\begin{thm}[Average-case $\sharP$-hardness of computation of outcome probabilities of FLO circuits]\label{thm:average-hardness-exact-prob}
 Let $\ket{\x_0}$ belong to the suitable Hilbert space $\Hpas=\bigwedge^{2N}(\C^{4N})$ for passive FLO and $\Hact=\Hfock^+(\C^{4N})$, respectively. Then it is $\sharP$-Hard to compute $p_{\x_0}(V,\inpsi)=\left|\bra{\x_0}V \ket{\inpsi}\right|^2$ with probability $\alpha>\frac{3}{4}+\delta$, $\delta=\frac{1}{\poly{N}}$, over the the uniform distribution of circuits: $V\sim \nupas$ for passive FLO and $V\sim \nuact$ for active FLO. \end{thm}

\begin{rem} Due to hiding property (see Lemma \ref{lem:hidingPROP} both active and passive FLO gates can permute between possible output Fock states $\ket{\x}$ in  $\Hpas=\bigwedge^{2N}(\C^{4N})$ and $\Hact=\Hfock^+(\C^{4N})$, respectively. Therefore, using the invariance of the Haar measure on $G$, we can transform $\x_0$ above into any other output $\x$ satisfying $|\x|=2N$ (for passive FLO) and $|\x|$ even (for active FLO).  

\end{rem}

\begin{proof}
We first fix the symmetry group $G$ describing a class of FLO circuits. The proof is virtually identical for both $G=\U(4N)$ and $G=\SO(8N)$. Suppose that $\O$ is an oracle that given a description of $\Pi(g)$ computes $\left|\bra{\x_0}\Pi(g) \ket{\inpsi}\right|^2$ with high probability, i.e.,
\begin{equation}\label{eq:oractleWCR}
    \Pr_{g\sim \mu_G} \left[ \O(\Pi(g))= \left|\bra{\x_0}\Pi(g) \ket{\inpsi}\right|^2 \right] >\alpha\ .
\end{equation}
The uniform distribution of FLO circuits $\nu_G$  is obtained by setting $V=\Pi(g)$, where $g\sim \mu_G$  (recall that $\nu_G=\nupas$ for $G=\U(4N)$ and $\nu_G=\nuact$ for $G=\SO(8N)$).  Therefore Eq. \eqref{eq:oractleWCR} is equivalent to
\begin{equation}\label{eq:oractleWCR2}
    \Pr_{V\sim \nu_G} \left[ \O(V)= \left|\bra{\x_0} V\ket{\inpsi}\right|^2 \right] >\alpha\ .
\end{equation}
In what follows we argue that oracle $\O$ can be used to compute the $\sharP$-hard probability in polynomial time. The argument presented bellow follows steps from worst-to-average-case reduction for permanents of Gaussian matrices from  \cite{Aaronson2013}, and its modification that involving rational interpolation from \cite{movassagh_quantum_2020}.  We consider a rational path interpolation   ${g}_\theta=g_0 F_\theta(g)$ between worst-case $g_0$ and =$g'=g_0 g$, where $g$ is chosen according to Haar measure on $G$. We will call $\O$ on $L$ distinct FLO circuits  $\Pi(g_{\theta_1}),\Pi(g_{\theta_2}),\ldots,\Pi(g_{\theta_L})$, where $\theta_i\in[1-\Delta,1]$, and the parameter $\Delta$ will be chosen later. We use evaluations $\lbrace\O(\Pi(g_{\theta_i}))\rbrace_{i=1}^L$ to efficiently reconstruct the rational function using Berlekamp-Welch algorithm for rational functions   $R(\theta)= \left|\bra{\x_0} \Pi(g_\theta)\ket{\inpsi}\right|^2$ (cf. Theorem \ref{th:BWrational}). If the reconstruction is successful, evaluation of $R(\theta)$ at $\theta=0$ gives us the $\sharP$-hard probability $R(0)= \left|\bra{\x_0} V_0\ket{\inpsi}\right|^2$ (we used here $V_0=\Pi(g_0)$).

To assess the success probability with which the above scheme evaluates $\left|\bra{\x_0} V_0\ket{\inpsi}\right|^2$ correctly we first bound the success probability with which oracle $\O$ computes the value of $p_{\x_0}(\Pi(g_\theta),\inpsi)$ correctly. Using variational characterization of TV distance and bounds from Eq.\eqref{eq:TVschemeBOUNDS} we obtain
\begin{equation}
\begin{split}
\Pr_{V\sim \nu_G} &\left[ \O(V)= \left|\bra{\x_0} V\ket{\inpsi}\right|^2 \right] \\ 
&-\Pr_{V\sim \nu^\theta_G} \left[ \O(V)= \left|\bra{\x_0} V\ket{\inpsi}\right|^2 \right] \leq C N^2 \Delta\ .
\end{split}
\end{equation}
Combining the above with \eqref{eq:oractleWCR2} we get
\begin{equation}
    \Pr_{V\sim \nu^\theta_G} \left[ \O(V)= \left|\bra{\x_0} V\ket{\inpsi}\right|^2 \right] \geq \alpha -  C N^2 \Delta\ .
\end{equation}
Or equivalently
\begin{equation}\label{eq:succBOUND}
    \Pr_{g\sim \mu_G} \left[ \O(\Pi(g_\theta))= \left|\bra{\x_0} \Pi(g_\theta)\ket{\inpsi}\right|^2 \right] \ge \alpha -  C N^2  \Delta\  .
\end{equation}
According to rational Berlekamp-Welch algorithm the number of evaluations $L$ of a rational function $R(\theta)$ that allows to reconstruct it despite heaving at most $t$ incorrect evaluations has to satisfy $L>d_1+d_2 +2 t$. Note that in the considered case $d_1 + d_2=\Theta(N^2)$  (cf. \eqref{eq:finalDEGREES}). The probability of having number of errors that exceeds the bound allowing for reconstruction of $R(\theta)$ can be estimated using Markov inequality applied for the random variable counting the number of invalid evaluations of the oracle
\begin{equation}
    t(g) = \left|\SET{\theta_i}{   \O(\Pi(g_{\theta_i}))\neq|\bra{\x_0}\Pi(g_{\theta_i})\ket{\inpsi}|^2,\ i\in[L]}\right|. 
\end{equation}
From the definition of $t$ and the inequality \eqref{eq:succBOUND} it follows that $\E_{g\sim \mu_G} t(g) \leq [1-\alpha + CN^2 \Delta]  L$. Using this estimate in Markov inequality (recall that by assumption  $\alpha>\frac{3}{4}+\delta$, for $\delta=\frac{1}{\poly{N}}$) we get
\begin{equation}
\begin{split}
    \Pr_{g\sim \mu_G}\left[ t(g)>\frac{L-d_1-d_2}{2}\right] &\leq \frac{[1-\alpha + CN^2 \Delta ] L }{\frac{L-d_1-d_2}{2}} \\
    &\leq \frac{\frac{1}{4}-\delta + CN^2\Delta }{\frac{1}{2}-\frac{d_1+d_2}{2L}}\ .
\end{split}
\end{equation}
By choosing $\Delta$ and $L$ such that $C N^2 \Delta \leq \frac{\delta}{2}$ and $\frac{d_1+d_2}{2L} \leq  \frac{\delta}{4}$ (this can be done with $\Delta=\frac{1}{\poly{N}}$ and $L=\poly{N}$ because $d_1+d_2=\Theta(N^2)$), we obtain 
\begin{equation}
    \Pr_{g\sim \mu_G}\left[ t(g)>\frac{L-d_1-d_2}{2}\right]  \leq \frac{\frac{1}{4} -\frac{\delta}{2}}{\frac{1}{2} -\frac{\delta}{4}} \leq \frac{1}{2}- \frac{\delta}{4}\ .
\end{equation}
The leftmost part of the above inequality is the probability of failure of our protocol. Therefore, since $\delta=\frac{1}{\poly{N}}$, we can repeat the procedure polynomially many  number of times, for different choices of $\Pi(g)$, compute $R{\Pi(g)}(0)$ each time, and output the majority vote. The probability of successfully computing the right result (i.e., $\left|\bra{\x_0} V_0\ket{\inpsi}\right|^2$) can be made exponentially close to $1$ in this way.   
\end{proof}

We proceed with proving the robust version of the above result. To this end we shift to polynomial interpolation because much more is known about its robustness to errors. To phrase our problem using polynomials we first note that the rational function $R_{g_0,g}=\left|\bra{\x_0} \Pi(g_\theta)\ket{\inpsi}\right|^2$, where $g_\theta =g_0 F_\theta(g)$ can be written as 
\begin{equation}\label{eq:Rdef}
    R_{g_0,g}(\theta)= \frac{D_{g_0,g}(\theta)}{Q_g(\theta)}\ ,
\end{equation}
where for both groups $D_{g_0,g},Q_g$ are real polynomials of degrees $D_{g_0,g}=d_1 =\Theta(N^2)$, $Q_g=d_2=\Theta(N^2)$ (cf Lemma \ref{lem:degreesOFRAT}). Moreover, the denominator $Q_g(\theta)$ can be computed efficiently in $N$ (see Eq. \eqref{eq:denominatorsCIRCUITS}) given a classical description of $g$. Hence, we have
\begin{lem} \label{lem:equivRAT}
 Let $R_{g_0,g}(\theta)$ be defined as in \eqref{eq:Rdef} (for fixed $g_0,g\in G$, where $G=\U(4N)$ or $G=\SO(8N)$). Then complexity of computation of $R_{g_0,g}(\theta)$ and $D_{g_0,g}(\theta)$ is equivalent up to $\Theta(N^2)$ overhead.
\end{lem}
The above allows us to use, following \cite{movassagh_quantum_2020,bouland_complexity_2019}, techniques of polynomial interpolation in order to estimate the hard probability $R_{g_0,g}(0)$. We will now state two results form this domain that will be used later in the robust version of the worst-to-average-case reduction given above.

\begin{lem}[Paturi lemma \cite{Paturi1992}]\label{lem:PaturiLem}
Let $P(\theta)$ be a polynomial of degree $k$ and suppose that for $|P(\theta)|\leq \epsilon$ for $\theta\in[1-\Delta,1]$, $\Delta\in(0,1]$. Then 
\begin{equation}\label{eq:PaturiLEMMA}
    P(0)\leq \epsilon\,\exp(4k(1+\Delta^{-1}))\ .
\end{equation}
\end{lem}
\begin{rem}
The above lemma is usually presented in a slightly different form in which the assumption $|P(\theta)|\leq\epsilon$ for $\theta\in[-\Delta,\Delta]$ ($\Delta>0$) is used to establish  $P(0)\leq \epsilon\,\exp(2k(1+\Delta^{-1}))$. Our result can be deduced from the former via simple affine change of variables $\theta\mapsto\theta'=- \frac{2}{2-\Delta}\theta + 1$. 
\end{rem}
\bigskip
\begin{thm}[Values of polynomials bounded at equally spaced points \cite{PolyGrowth}]\label{thm:poly-eqspace}
Let $\theta_i$, $i=1,\ldots,L$ be a collection of $L$ equally spaced points in the interval $[1-\Delta,1]$, $\Delta\in(0,1)$. Let $P(\theta)$ be a polynomial of degree $k$. Assume that for every $i$, $|P(\theta_i)|\leq \epsilon$. Then there exist absolute constants $a,b>0$ such that
\begin{equation}\label{eq:valuesBOUND}
    \max_{\theta\in[1-\Delta,1]} |P(\theta)|\leq \epsilon\,  \exp\left(b\frac{k^2}{L} + a\right)\ .
\end{equation}
\end{thm}
\begin{rem}
The problem of bounding values of polynomials that are bounded on uniformly spaced interval has a long history and there were more recent developments in this topic (see for example \cite{PolyGrowth2}). However for our purposes the above result by Coppersmith and Rivlin is sufficient.
\end{rem}

Before we formulate our result and prove our main theorem we need one more technical ingredient. Informally speaking, since $Q_g(\theta)$ appears in the denominator of \eqref{eq:Rdef} we need ensure that values of $Q_g(\theta)$ are not too large for typical values of $g$. This is achieved by combining Lemma \ref{lem:DenominatorBOunds} and explicit formulas for $Q_g(\theta)$ given in \eqref{eq:denominatorsCIRCUITS} we obtain 

\begin{corr}\label{cor:DenominatorCircuitBOunds}
Let $g\in G$ and let $Q_g(\theta)$ be the polynomial in defined in \eqref{eq:denominatorsCIRCUITS} for $G=\U(d)$  in $G=\SO(2d)$. Assume that $n=2N$ , $d=4N$. Let now $\tilde{\Delta}>0$. We then have the following inequalities  
\begin{align} 
 \Pr_{g\sim\mu_{\U(d)}} \left(   Q_g(\theta) \le \left[1+\left(\frac{\theta\pi}{\tilde{\Delta}}\right)^2 \right]^{16 N^2}\ \right)\geq 1- 4N\frac{\tilde{\Delta}}{\pi} \ ,  \\
\Pr_{g\sim\mu_{\SO(2d)}} \left(    Q_g(\theta) \le \left[1+\left(\frac{\theta\pi}{\tilde{\Delta}}\right)^{2} \right]^{32 N^2} \right) \geq  1- 4N\frac{\tilde{\Delta}}{\pi} \ . 
\end{align}
\end{corr}

Combining all technical ingredients  stated above we are in the position to prove our main result.

\begin{thm}[Average-case $\sharP$-hardness of approximation outcome probabilities of FLO circuits]\label{thm:average-hardness-approx-prob}
Let $V_0$ be a FLO circuit such that computing $p_{\x_0}(V_0,\inpsi)=\left|\bra{\x_0}V_0 \ket{\inpsi}\right|^2$ is $\sharP$-Hard, where $V_0$ is element of either passive or active FLO circuits and the output Fock state $\ket{\x_0}$ belongs to the suitable Hilbert space $\Hpas=\bigwedge^{2N}(\C^{4N})$ for passive FLO and $\Hact=\Hfock^+(\C^{4N})$, respectively.

Let $\epsilon=\exp(-\Theta(N^6))$. Then it is $\sharP$-Hard to compute  $p_{\x_0}(V,\inpsi)=\left|\bra{\x_0}V \ket{\inpsi}\right|^2$ to accuracy $\epsilon$ with probability $\alpha>1-\delta$, $\delta=o(N^{-2})$, over the the uniform distribution of circuits: $V\sim \nupas$ for passive FLO and $V\sim \nuact$ for active FLO.
\end{thm}

\begin{rem} Using the same arguments as in remark below Theorem \ref{thm:average-hardness-exact-prob} we can transform $\x_0$ above into any other output $\x$ satisfying $|\x|=2N$ (for passive FLO) and $|\x|$ even (for active FLO).  
\end{rem}

\begin{proof}
We first fix the symmetry group $G$ describing a class of FLO circuits. The uniform distribution of FLO circuits $\nu_G$  is obtained by by setting $V=\Pi(g)$, where $g\sim \mu_G$  (recall that $\nu_G=\nupas$ for $G=\U(4N)$ and $\nu_G=\nuact$ for $G=\SO(8N)$). 
We start with an oracle $\O$ that given a classical description of $V=\Pi(g)$, is able to approximately compute $p_{\x_0}(V,\inpsi)=\left|\bra{\x_0}\Pi(g) \ket{\inpsi}\right|^2$,
\begin{equation}\label{eq:oractleAPPROX1}
    \Pr_{g\sim \mu_G} \left[\left| \O(\Pi(g))- |\bra{\x_0}\Pi(g) \ket{\inpsi}|^2\right|\leq \epsilon \right]  >1-\delta\ .
\end{equation}
Equivalently, we have 
\begin{equation}\label{eq:oractleAPPROX2}
    \Pr_{V\sim \nu_G} \left[\left| \O(V)- |\bra{\x_0}V \ket{\inpsi}|^2\right|\leq \epsilon \right]  >1-\delta\ .
\end{equation}
For a generic Haar random $g\in G$ we again consider a rational path $g_\theta=g_0 F_\theta(g)$ between $g_0g$ and $g_0$, where $g_0$ is an element of the group corresponding to the worst-case circuit. Recall that by $\mu^\theta_G$  we denoted the distribution of $g_\theta$ for $g\sim\mu_G$.  We will now query oracle $\O$ multiple times on $g_{\theta_i}$, where $\theta_i$ are $L$ equally distributed points in the interval $[1-\Delta,1]$, for $\Delta>0$ to be set latter. 
By using the variational characterization of TV distance and Eq.~\eqref{eq:TVschemeBOUNDS}, we obtain that for every  $\theta_i\in[1-\Delta,1]$
\begin{equation}\label{eq:oractleAPPROX3}
\begin{split}
    \Pr_{g\sim \mu_G} &\left[\left| \O(\Pi(g_{\theta_i}))- |\bra{\x_0}\Pi(g_{\theta_i}) \ket{\inpsi}|^2\right|\leq \epsilon \right]  \\
    &>1-\delta - 8\Delta N^2\ ,
\end{split}
\end{equation}
 Let now $D_{g_0,g}(\theta)$ be a polynomial of degree $\deg(D_{g_0,g})=\Theta(N^2)$ that we defined \eqref{eq:Rdef}. Recall that the denominator of $R_{g_0,g}(\theta)$, $Q_g(\theta)$ can be computed efficiently (cf. Lemma~\ref{lem:equivRAT}). Therefore we can use $\O$ to construct an oracle $\tilde{\O}$ that computes approximations of values of polynomial $D_{g_0,g}$ at point $\theta_i$ with potentially high probability over the choice of $g$
\begin{equation}
\begin{split}
    \Pr_{g\sim \mu_G} &\left[| \tilde{\O}(\Pi(g_{\theta_i}))- D_{g_0,g}(\theta_i) | \leq \epsilon Q_g(\theta_i) \right] \\
    &>1-\delta - 8\Delta N^2\ ,
\end{split}
\end{equation}
We now use Corollary \ref{cor:DenominatorCircuitBOunds} to bound $Q_g(\theta)$ in the above expression: 
\begin{equation}
     \Pr_{g\sim\mu_G} \left[   Q_g(\theta) \le \exp\left( A \log\left(\frac{1}{\tilde{\Delta}}\right) N^2  \right) \right] \geq 1- 4N\frac{\tilde{\Delta}}{\pi} \label{eq:denominatroUPPERboundG}\ ,
\end{equation}
where we assumed  $\tilde{\Delta}\in(0,1)$ and $A$ a positive numerical constant mildly depending on the group $G$. Using  the bound $\Pr(X\cap Y) \geq \Pr(X)+\Pr(Y)-1$ we obtain
\begin{equation}\label{eq:oractleAPPROX4}
\begin{split}
    \Pr_{g\sim \mu_G} &\left[| \tilde{\O}(\Pi(g_{\theta_i})){-} D_{g_0,g}(\theta_i)|\leq \epsilon\, \exp\left( A \log\left(\frac{1}{\tilde{\Delta}}\right) N^2  \right) \right] \\ &>1-\delta - 8\Delta N^2\ - 4N\frac{\tilde{\Delta}}{\pi} .
\end{split}
\end{equation}

We finally use union bound lower to bound the probability that $\tilde{\O}$ is successful \emph{for all} $L$ equally spaced $\theta_i$ in $[1-\Delta,1]$:
\begin{align}
     \Pr_{g\sim \mu_G} &\left[\forall \theta_i \ | \tilde{\O}(\Pi(g_{\theta_i})){-} D_{g_0,g}(\theta_i)|\leq \epsilon\, \exp( A \log(\tfrac{1}{\tilde{\Delta}}) N^2  ) \right]
   \nonumber \\ &>1-L(\delta + 8\Delta N^2\ + 4N\frac{\tilde{\Delta}}{\pi}) .
    \label{eq:oractleAPPROX5}
\end{align}
If $L\approx \deg(D_{g_0,g})=\Theta(N^2)$ the above evaluations of $\tilde{\O}$ can be used to recover polynomial $\tilde{P}_{g_0,g}$ passing through points $(\theta_i, \tilde{\O}(\Pi(g_{\theta_i})))$ and having identical degree to $D_{g_0,g}$. By \eqref{eq:oractleAPPROX4} and results of Coppersmith and Rivlin stated in \cite{PolyGrowth} we know that (note that we set $L\approx\deg(D_{g_0,g})=\Theta(N^2)$)
\begin{equation}
\begin{split}
    \max_{\theta\in[1-\Delta, 1} &\left|\tilde{P}_{g_0,g}(\theta) - D_{g_0,g}(\theta)\right| \\
    &\leq \epsilon\, \exp\left( A \log\left(\frac{1}{\tilde{\Delta}}\right) N^2\right)  \exp(\Theta(N^2)) \\
    &= \ep\, \exp\left(\Theta(N^2) \log\left(\frac{1}{\tilde{\Delta}}\right)  \right)\ .
\end{split}
\end{equation}
Recall that by assumption and definition of Cayley path  $D_{g_0,g}(0)$ encodes a (rescaled) $\sharP$-hard probability amplitude. Using Paturi lemma for the polynomial $\tilde{D}_{g_0,g}(\theta) - D_{g_0,g}(\theta)$ we finally obtain 
\begin{equation}
\begin{split}
    \Big\vert&\tilde{D}_{g_0,g}(0) - D_{g_0,g}(0)\Big\vert \\
    &\leq \ep\, \exp\left(\Theta(N^2)\log\left(\frac{1}{\tilde{\Delta}}\right)+\Theta(N^2)(1{+}\Delta^{-1})\right). 
\end{split}
\end{equation}

To sum up, the initially assumed oracle $\O$ allows us to construct an efficient algorithm $\A$ that approximately computes $\sharP$-hard quantity  $D_{g_0,g}(0)=Q_g(\theta) \left|\bra{\x_0}\Pi(g_0) \ket{\inpsi}\right|^2$:

\begin{equation}\label{eq:oractleAPPROX6}
\begin{split}
    \Pr_{g\sim \mu_G} &\left[ \left| \A(\Pi(g))- D_{g_0,g}(0)|\right|\leq \tilde{\ep} \right]  \\
    &>1- B N^2 \left(\delta + 8\Delta N^2\ + 4N\frac{\tilde{\Delta}}{\pi} \right) ,
\end{split}
\end{equation}
where $\tilde{\ep}=\ep\, \exp\left(\Theta(N^2)\log\left(\frac{1}{\tilde{\Delta}}\right)+\Theta(N^2)(1+\Delta^{-1})\right)$, and $B>0$ is a numerical constant. Success probability of the protocol to exceeds $\frac{1}{2}$  with the following scaling
\begin{equation}\label{eq:DELTASset}
    \Delta=\Theta(N^{-4}) \ , \; \tilde{\Delta}= \Theta(N^{-3})\ .
\end{equation}

From the result of \cite{Dyer2000} we have $\sharP$ hardness guarantees up to constant multiplicative error. Since for $\sharP$-hard quantity this such error implies additive error of magnitude at most $2^{-\Theta(N)}$. Therefore by setting $\tilde{\ep}\leq 2^{-\Theta(N)}$ which, by the virtue of Eq.\eqref{eq:DELTASset} corresponds to scaling of the original error  $\ep=\exp(-\Theta(N^6))$  allows to to extrapolate to the hardness neighbored.

\end{proof}

\begin{rem}
In the course of the proof of the above result we have realised an inadequate usage of the oracle in the reduction by Movassagh \cite{movassagh_quantum_2020} (the author assumed that the oracle works as in \eqref{eq:oractleAPPROX3} but without the neccesary dependence on $\Delta$ . Correction of the proof seems to give in that case worse than claimed tolerance for error $\ep=\exp(-\theta(N^{4.5})$ (for the Google layout), which is still better then the one claimed here.
\end{rem}

\section{Efficient tomography of fermionic linear optics}\label{sec:cert} 

\begin{figure}
    \includegraphics[scale=0.6]{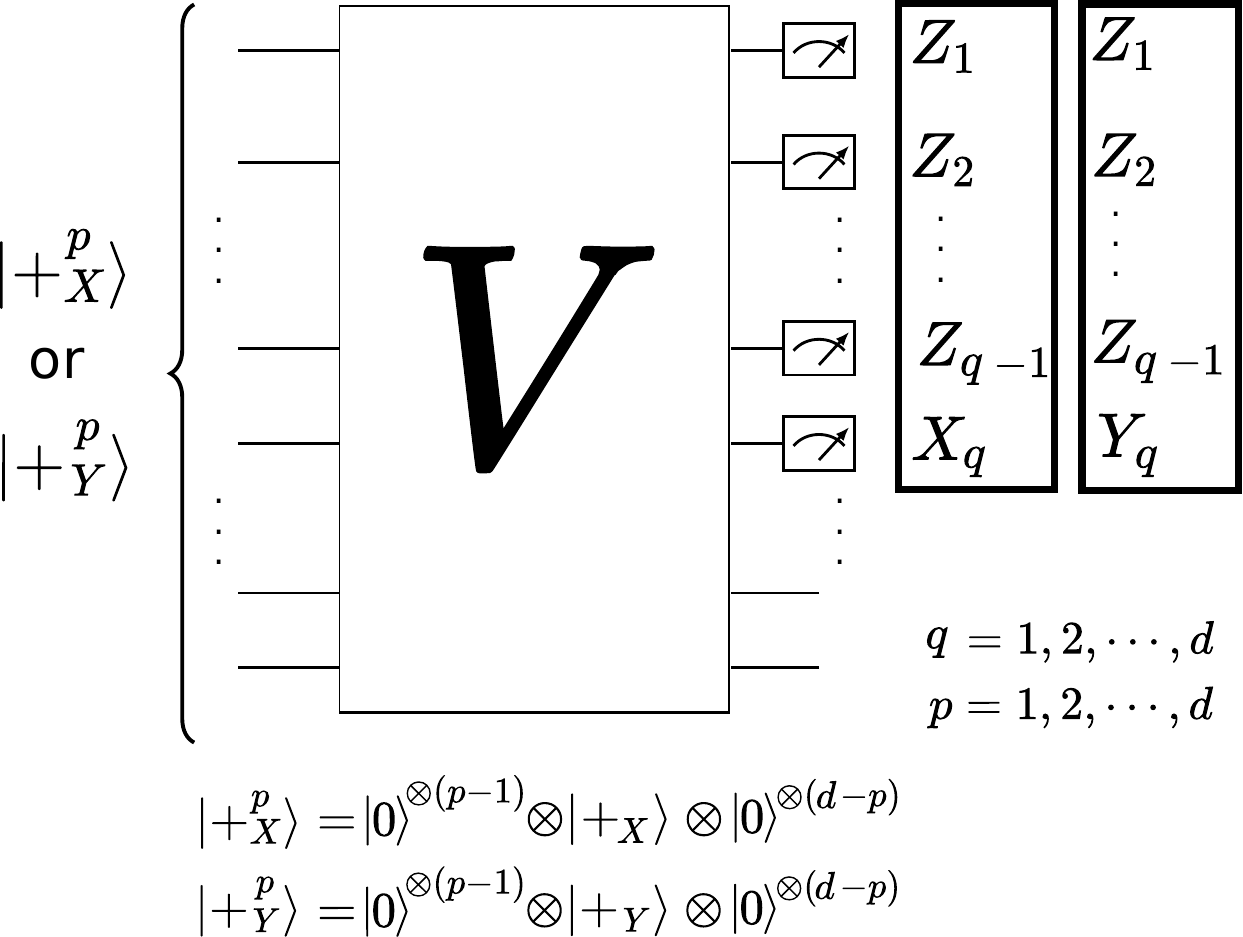}
    \caption{A graphical presentation of the tomography protocol of an active FLO circuit $V$. A single step of the protocol consists of (i) preparation of $2d$ input states  $\ket{+^{p}_X}$ and $\ket{+^{p}_Y}$ ($p=1,\ldots,d$),  (ii) transformation of the states via the circuit $V$ and (iii)  for each of the $2d$ states measuring the operators $Z_1 Z_2 \cdots Z_{q-1} X_{q}$ and $Z_1 Z_2 \cdots Z_{q-1} Y_{q}$ ($q=1,\ldots,d$). These operations are then repeated multiple times in order to gather sufficient statistics necessary to reconstruct the orthogonal matrix $O\in\SO(2d)$ that defines the unitary channel $\Phi_V$ associated to  $V=\Piact(O)$. }\label{fig:tomog}
\end{figure}

The tomography and certification of gates, i.e., the task of ensuring that the correct unitary was implemented, is vital for near-term quantum devices. However, it 
is often an inherently challenging problem due to exponential scaling of the number of parameters describing a general multiqubit quantum operation \cite{eisert2020quantum}. Here we show that the structure of FLO unitaries allows us to perform  their tomography efficiently using resources scaling only polynomially with the system size. As passive fermionic gates form a subset of active FLO circuits, we focus only on the tomography of the latter ones, since from this also the tomography of passive circuit follows.


We will use here again the Jordan-Wigner mapping between $d$ qubit system and fermionic Fock space with $d$ physical modes (see Section \ref{sec:notations}), and define the following $2d$ pure states:
\begin{align}
    \begin{split}
    \ket{+^{p}_X}  & = \id^{\otimes (p-1)} \otimes H \otimes \id^{\otimes (d-p)}  |0 \rangle^{\otimes d} \\
    &= |0 \rangle^{\otimes (p-1)} \otimes | +_X \rangle \otimes | 0 \rangle^{\otimes (d-p)} \ ,\end{split}\\
    \begin{split}
    \ket{+^{p}_Y} & = \id^{\otimes (p-1)} \otimes \tilde{H} \otimes \id^{\otimes (d-p)}  |0 \rangle^{\otimes d} \\
    &=|0 \rangle^{\otimes (p-1)} \otimes | +_Y \rangle \otimes | 0 \rangle^{\otimes (d-p)} \ , 
    \end{split}
\end{align}
where $p=1, \ldots, d$.


In terms of majorana operators, one can write the density matrices of these states  as

\begin{widetext}
\begin{align}
  \rho_{2p-1} &= \ketbra{+^{p}_X}{+^{p}_X}  = 
    \prod_{q=1}^{p-1}\left( \frac{\id + i m_{2q-1}m_{2q} }{2} \right) 
    \left( \frac{\id + \prod_{q=1}^{p-1}(i m_{2q-1}m_{2q})
    m_{2p-1}}{2}\right) 
    \prod_{q=p+1}^d\left( \frac{\id + i m_{2q-1}m_{2q} }{2} \right)
   \ ,  
\\
    \rho_{2p} &= \ketbra{+^{p}_Y}{+^{p}_Y} = 
    \prod_{q=1}^{p-1}\left( \frac{\id + i m_{2q-1}m_{2q} }{2} \right) 
    \left( \frac{\id + \prod_{q=1}^{p-1}(i m_{2q-1}m_{2q}) m_{2p}}{2}\right) 
    \prod_{q=p+1}^d\left( \frac{\id + i m_{2q-1}m_{2q} }{2} \right) \ ,
\end{align}
\end{widetext}
where, as before, $p=1, \ldots, d$. Expanding these density matrices in terms of majorana monomials (given in Eq.~\eqref{eq:major-mon}), we observe that for an arbitrary $\rho_x$ ($x=1, \ldots , 2d$)
there is only one majorana monomial of degree 1 appearing, namely $m_x$. 
Thus, considering the FLO evolved states $V \rho_x V^{\dagger}$, the degree 1 majorana terms will be of the form (see Eq.~\eqref{eq:mode_transform}) $Vm_xV^{\dagger}= \sum_{y=1}^{2d} O_{yx} m_y$, where $O\in\SO(2d)$ is the orthogonal matrix that encodes the FLO circuit $V$. In order to obtain arbitrary element of the orthogonal matrix $O_{yx}$, one needs only to insert the state $\rho_x$, evolve it with the FLO unitary $V$, and then measure the expectation value of $m_y$:
\begin{equation}\label{eq:expVALUEort}
    O_{yx}= \tr( m_y V \rho_y V^{\dagger} )\ .
\end{equation}
Measuring the expectation value of $m_{2q-1}$ and $m_{2q}$ amounts to measuring $Z_1 \cdots Z_{q-1} X_q$ and $Z_1 \cdots Z_{q-1} Y_q$, respectively. These can all be done, after a single layer of local base change operations, through usual computational basis measurements. The graphical presentation of our tomography scheme is given in Fig.~\ref{fig:tomog}.
The following theorem show that the construction outlined above allows to recover an unknown FLO circuit $V$ efficiently in $d$, both  in terms of the number of different setups needed for the implementation as well as in terms of sample complexity. Importantly, our results give rigorous recovery guarantees in the diamond norm, despite the presence of statistical fluctuations.

\begin{thm}[Efficient tomography of active FLO unitary channels]\label{th:EFtomography}
Let $V$ be an unknown active FLO circuit acting on $d$ qubits. Consider the following estimation protocol using the states $\rho_x$ and observables $m_y$ ($x,y \in [2d]$) and comprising of  $r$ independent experimental rounds.  A single experimental round, say the $k$'th, consists of the following routine:
\begin{itemize}
    \item For every pair $(x, y) \in[2d]^{\times 2}$: (i) prepare $\rho_x$ as input state; (ii) evolve $\rho_x$ via the circuit $V$; (iii) measure $V\rho_x V^\dagger$ using $m_y$ obtaining outcome $m^{(k)}_{yx}\in\{-1, 1\}$.
\end{itemize}
The outcomes of the $k$'th round are gathered in the  $2d\times2d$
matrix  $M^{(k)}$ with entries $m^{(k)}_{yx}$. After $r$ rounds, define $\hat{M}_r\coloneqq\frac{1}{r}\sum_{k=1}^r M^{(k)}$ as the sample average of matrices $M^{(k)}$. Then, let $\hat{O}_r\in \SO(2d)$ be defined as the orthogonal matrix appearing in the polar decomposition of $\hat{M}_r$ (i.e., $\hat{M}_r=  O_r P $, where $P$ is a semidefinite real $2d\times 2d$ matrix). Finally, set $\hat{V}\coloneqq\Piact(\hat{O}_r)$ as the estimator of the circuit $V$ after $r$ rounds of the protocol. 

Assume that all routines in the protocol are implemented perfectly. Furthermore, let $\delta\in(0,1)$ be fixed and let $\Phi_{V}$ and $\Phi_{\hat{V}}$ be the unitary channels defined by the active FLO circuits $V$ and $\hat{V}$, respectively. Then, for the number of rounds satisfying 
\begin{equation}\label{eq:SamplecomplexityTomography}
    r\geq \frac{28 d^3}{\ep^2} \log\left(\frac{4d}{\delta}\right) \ ,
\end{equation}
the protocol outputs an FLO circuit $\hat{V}$ such that  $\| \Phi_{V} - \Phi_{\hat{V}} \|_{\Diamond}\leq \ep$  with probability at least $1-\delta$.

\end{thm}

\begin{rem}
We believe that it possible to improve the sampling complexity and the number of quantum circuits needed for the tomography of an unknown FLO unitary $V$. Moreover, we expect that our proof technique can also be used for the quantum process tomography of general fermionic Gaussian channels.
\end{rem}

There are three key difficulties that need to be circumvented in order to establish the above result. The first one is related to the fact that, by the virtue of \eqref{eq:expVALUEort}, the protocol estimates an orthogonal matrix $O\in\SO(2d)$ not the circuit $V$ or the associated $d$-qubit channel $\Phi_V$. The following Lemma, proved in Appendix \ref{app:effTOM}, allows us to connect operator-norm distance between elements of the orthogonal group  with the diamond norm between the corresponding quantum channels (this result can be viewed as a fermionic version of the analogous stability result proved by Arkhipov for standard Boson Sampling \cite{arkhipov2015bosonsampling}) .

\begin{lem}[Stability of the active FLO representation]\label{lem:stabAFLO}
Consider two elements of the orthogonal group, $O, O' \in \SO(2d)$, and let $V$ and $V'$ be the corresponding active FLO unitaries, i.e., $V = \Piact(O)$ and $V' = \Piact(O')$. Furthermore, let $\Phi_{V}$ and $\Phi_{V'}$ be the unitary channels defined by $V$ and $V'$, respectively. Then the following inequality is satisfied  
\begin{equation}
   \| \Phi_{V} - \Phi_{V'} \|_{\Diamond} 
   \le 2 d \| O - O' \|.
\end{equation}
\end{lem}

The second technical issue arises because the sample-average matrices $\hat{M}_s$ appearing in the protocol are not necessarily orthogonal. For this reason we use the (real) polar decomposition in order to get an orthogonal matrix from $\hat{M}_s$.  The Lemma below gives an upper bound for the possible operator-norm error that can result from this procedure. 

\begin{lem}[Operator-norm stability of the real polar decomposition \cite{PolarStab}]\label{lem:PolarDecomp}
Let $O$ be orthogonal matrix $n\times n$. Let $\Delta $ be $n \times n$ real matrix such that $\|\Delta  \|\leq 1$. Let $O_\Delta$ be the orthogonal transformation appearing in the polar decomposition of $O+\Delta A$ (i.e. $O+\Delta  =O_{O+\Delta } P$ for a semidefinite real matrix $P$). We then have the following inequality
\begin{equation}\label{eq:stabPolar}
    \|O-O_\Delta \| \leq \| \Delta\|\ .
\end{equation}
\end{lem}

The above lemma follows as a direct corollary of Theorem 2.3 in \cite{PolarStab}.

The last technical ingredient needed for the proof of Theorem \ref{th:EFtomography} is the following matrix concentration bound, which allows to control the magnitude of statistical fluctuations incurred in our scheme. 

\begin{lem}[Matrix Bernstein inequality \cite{TroopMatrixConcentration}]\label{lem:matrixBernstein} 
Let $S^{(1)},\ldots, S^{(r)}$ be independent, centered real $n\times n$ random matrices with uniformly bounded operator norm, i.e., for all $k\in[r]$
\begin{equation}
   \E S^{(k)} =0\ ,\ \|S^{(k)}\|\leq L\ .
\end{equation}
Assume furthermore that the entries of each $S^{(k)}$ are independently distributed  with a variance upper bounded by a constant, $\mathrm{Var}(S^{(k)}_{ij})\leq c$. 

We then have the following concentration inequality valid for arbitrary $\tau>0$
\begin{equation}\label{eq:MatrixBernstein}
    \Pr\left(\left\|\frac{1}{r}\sum_{k=1}^r S^{(k)}  \right\| \geq \tau  \right) \leq 2 n\, \exp\left(-\frac{r \tau^2}{2(n c +\frac{L}{3} \tau ) }\right)\ .
\end{equation}

\end{lem}
     
A more general version of the above inequality (that does not require independently distributed entries of matrices $S^{(k)}$) can be found  in Theorem 1.6.2 from  \cite{TroopMatrixConcentration}.

\begin{proof}[Proof of Theorem \ref{th:EFtomography}]
Let us remark first that our tomography protocol was defined such that the matrices $M^{(k)}$ originating form different rounds $k$ are independent from each other, and  for fixed $k$ also their entries $m_{yx}^{(k)}$ are  independent. Furthermore, by virtue of Eq.~\eqref{eq:expVALUEort}, we have
\begin{equation}
    \E m_{yx}^{(k)} = O_{yx}\ , 
\end{equation}
where $O\in\SO(2d)$ is an orthogonal matrix corresponding to the circuit $V$. We now apply Lemma~\ref{lem:matrixBernstein} to the sequence of $2d\times2d$ matrices $\Delta^{(k)}\coloneqq M^{(k)}-O$. From definition matrix elements of $\Delta^{(k)}$ satisfy $|\Delta^{(k)}_{yx}|\leq 2 $. From this and the fact that $m_{yx}^{(k)}\in\{-1,1\}$ it easily follows that 
\begin{equation}
    \|\Delta^{(k)} \| \leq 4 d\ ,\  \mathrm{Var}(\Delta^{(k)}_{yx}) \leq 1\ . 
\end{equation}
Inserting these estimates in Eq.~\eqref{eq:MatrixBernstein} (and noting that $n=2d$) gives 
\begin{equation}
    \Pr\left(\left\| \frac{1}{r} \sum_{k=1}^r \Delta^{(k)} \right\| \geq \tau  \right) \leq 4 d\, \exp\left(-\frac{r\tau^2}{4d(1  +\frac{2}{3} \tau ) }\right) \ .
\end{equation}
Recalling that $\hat{M}_r=\frac{1}{r}\sum_{k=1}^r M^{(k)}$, using the definition of $\Delta^{(k)}$, and assuming that $\tau<1$ (in what follows we will see that we can indtroduce this constraint without the loss of generality)  we obtain 
\begin{equation}
    \Pr\left(\left\|  \hat{M}_r - O \right\| \leq \tau  \right) \geq 1 - 4d\, \exp\left(-\frac{r \tau^2}{7d }\right)\ .
\end{equation}

We therefore know that, provided $r$ is high enough, the sample average $\hat{M}^{(k)}$ approximates matrix $O$ in operator norm. Applying Lemma~\ref{lem:PolarDecomp} to $\hat{O}_r$, i.e., to the orthogonal part of the polar decomposition of $\hat{M}_r$ (this corresponds to setting $\Delta=  \hat{M}_r-O $ in \eqref{eq:stabPolar}), we obtain 

\begin{equation}
    \Pr\left(\left\|  \hat{O}_r - O \right\| \leq \tau  \right) \geq 1 - 4d\, \exp\left(-\frac{r \tau^2}{7d }\right)\ .
\end{equation}

Recalling that $\hat{V}=\Piact(\hat{O})$ and $V=\Piact(O)$ and invoking Lemma~\ref{lem:stabAFLO} we finally arrive to

\begin{equation}\label{eq:finalPROBbound}
    \Pr\left(\left\|  \Phi_{\hat{V}} - \Phi_{V} \right\|_\Diamond \leq 2 d\, \tau \right) \geq 1 - 4 d\, \exp\left(-\frac{r \tau^2}{7d }\right) .
\end{equation}
We conclude the proof by setting $\ep\coloneqq 2 d\, \tau$ and noting that \eqref{eq:SamplecomplexityTomography} follows from requiring that right-hand side of Eq. \eqref{eq:finalPROBbound} is larger than $1-\delta$.

\end{proof}

{\bf Acknowledgements}\\[3mm] We would like to thank Ramis Movassagh for explanations regarding his work. MO and ND acknowledge support by the Foundation for Polish Science through the TEAM-NET project (contract no. POIR.04.04.00-00-17C1/18-00). 
 ZZ was supported by the NKFIH through the Quantum Technology National Excellence Program (project no. 2017-1.2.1-NKP-2017-00001), the grants K124152, K124176, KH129601, K120569  and the Quantum Information National Laboratory of Hungary.

\bibliography{fermADV}

\begin{thebibliography}{117}%
\makeatletter
\providecommand \@ifxundefined [1]{%
 \@ifx{#1\undefined}
}%
\providecommand \@ifnum [1]{%
 \ifnum #1\expandafter \@firstoftwo
 \else \expandafter \@secondoftwo
 \fi
}%
\providecommand \@ifx [1]{%
 \ifx #1\expandafter \@firstoftwo
 \else \expandafter \@secondoftwo
 \fi
}%
\providecommand \natexlab [1]{#1}%
\providecommand \enquote  [1]{``#1''}%
\providecommand \bibnamefont  [1]{#1}%
\providecommand \bibfnamefont [1]{#1}%
\providecommand \citenamefont [1]{#1}%
\providecommand \href@noop [0]{\@secondoftwo}%
\providecommand \href [0]{\begingroup \@sanitize@url \@href}%
\providecommand \@href[1]{\@@startlink{#1}\@@href}%
\providecommand \@@href[1]{\endgroup#1\@@endlink}%
\providecommand \@sanitize@url [0]{\catcode `\\12\catcode `\$12\catcode
  `\&12\catcode `\#12\catcode `\^12\catcode `\_12\catcode `\%12\relax}%
\providecommand \@@startlink[1]{}%
\providecommand \@@endlink[0]{}%
\providecommand \url  [0]{\begingroup\@sanitize@url \@url }%
\providecommand \@url [1]{\endgroup\@href {#1}{\urlprefix }}%
\providecommand \urlprefix  [0]{URL }%
\providecommand \Eprint [0]{\href }%
\providecommand \doibase [0]{https://doi.org/}%
\providecommand \selectlanguage [0]{\@gobble}%
\providecommand \bibinfo  [0]{\@secondoftwo}%
\providecommand \bibfield  [0]{\@secondoftwo}%
\providecommand \translation [1]{[#1]}%
\providecommand \BibitemOpen [0]{}%
\providecommand \bibitemStop [0]{}%
\providecommand \bibitemNoStop [0]{.\EOS\space}%
\providecommand \EOS [0]{\spacefactor3000\relax}%
\providecommand \BibitemShut  [1]{\csname bibitem#1\endcsname}%
\let\auto@bib@innerbib\@empty
\bibitem [{\citenamefont {{Movassagh}}(2019)}]{movassagh_quantum_2020}%
  \BibitemOpen
  \bibfield  {author} {\bibinfo {author} {\bibfnamefont {R.}~\bibnamefont
  {{Movassagh}}},\ }\bibfield  {title} {\bibinfo {title} {{Quantum supremacy
  and random circuits}},\ }\href@noop {} {\bibfield  {journal} {\bibinfo
  {journal} {arXiv e-prints}\ ,\ \bibinfo {eid} {arXiv:1909.06210}} (\bibinfo
  {year} {2019})},\ \Eprint {https://arxiv.org/abs/1909.06210}
  {arXiv:1909.06210 [quant-ph]} \BibitemShut {NoStop}%
\bibitem [{\citenamefont {{Gidney}}\ and\ \citenamefont
  {{Eker{\r{a}}}}(2019)}]{gidney2019factor}%
  \BibitemOpen
  \bibfield  {author} {\bibinfo {author} {\bibfnamefont {C.}~\bibnamefont
  {{Gidney}}}\ and\ \bibinfo {author} {\bibfnamefont {M.}~\bibnamefont
  {{Eker{\r{a}}}}},\ }\bibfield  {title} {\bibinfo {title} {{How to factor 2048
  bit RSA integers in 8 hours using 20 million noisy qubits}},\ }\href@noop {}
  {\bibfield  {journal} {\bibinfo  {journal} {arXiv e-prints}\ ,\ \bibinfo
  {eid} {arXiv:1905.09749}} (\bibinfo {year} {2019})},\ \Eprint
  {https://arxiv.org/abs/1905.09749} {arXiv:1905.09749 [quant-ph]} \BibitemShut
  {NoStop}%
\bibitem [{\citenamefont {Ofek}\ \emph {et~al.}(2016)\citenamefont {Ofek},
  \citenamefont {Petrenko}, \citenamefont {Heeres}, \citenamefont {Reinhold},
  \citenamefont {Leghtas}, \citenamefont {Vlastakis}, \citenamefont {Liu},
  \citenamefont {Frunzio}, \citenamefont {Girvin}, \citenamefont {Jiang} \emph
  {et~al.}}]{error-correction-demo2016}%
  \BibitemOpen
  \bibfield  {author} {\bibinfo {author} {\bibfnamefont {N.}~\bibnamefont
  {Ofek}}, \bibinfo {author} {\bibfnamefont {A.}~\bibnamefont {Petrenko}},
  \bibinfo {author} {\bibfnamefont {R.}~\bibnamefont {Heeres}}, \bibinfo
  {author} {\bibfnamefont {P.}~\bibnamefont {Reinhold}}, \bibinfo {author}
  {\bibfnamefont {Z.}~\bibnamefont {Leghtas}}, \bibinfo {author} {\bibfnamefont
  {B.}~\bibnamefont {Vlastakis}}, \bibinfo {author} {\bibfnamefont
  {Y.}~\bibnamefont {Liu}}, \bibinfo {author} {\bibfnamefont {L.}~\bibnamefont
  {Frunzio}}, \bibinfo {author} {\bibfnamefont {S.}~\bibnamefont {Girvin}},
  \bibinfo {author} {\bibfnamefont {L.}~\bibnamefont {Jiang}}, \emph {et~al.},\
  }\bibfield  {title} {\bibinfo {title} {Extending the lifetime of a quantum
  bit with error correction in superconducting circuits},\ }\href@noop {}
  {\bibfield  {journal} {\bibinfo  {journal} {Nature}\ }\textbf {\bibinfo
  {volume} {536}},\ \bibinfo {pages} {441} (\bibinfo {year}
  {2016})}\BibitemShut {NoStop}%
\bibitem [{\citenamefont {{Egan}}\ \emph {et~al.}(2020)\citenamefont {{Egan}},
  \citenamefont {{Debroy}}, \citenamefont {{Noel}}, \citenamefont {{Risinger}},
  \citenamefont {{Zhu}}, \citenamefont {{Biswas}}, \citenamefont {{Newman}},
  \citenamefont {{Li}}, \citenamefont {{Brown}}, \citenamefont {{Cetina}} \emph
  {et~al.}}]{fault-tolerant-demo2020}%
  \BibitemOpen
  \bibfield  {author} {\bibinfo {author} {\bibfnamefont {L.}~\bibnamefont
  {{Egan}}}, \bibinfo {author} {\bibfnamefont {D.~M.}\ \bibnamefont
  {{Debroy}}}, \bibinfo {author} {\bibfnamefont {C.}~\bibnamefont {{Noel}}},
  \bibinfo {author} {\bibfnamefont {A.}~\bibnamefont {{Risinger}}}, \bibinfo
  {author} {\bibfnamefont {D.}~\bibnamefont {{Zhu}}}, \bibinfo {author}
  {\bibfnamefont {D.}~\bibnamefont {{Biswas}}}, \bibinfo {author}
  {\bibfnamefont {M.}~\bibnamefont {{Newman}}}, \bibinfo {author}
  {\bibfnamefont {M.}~\bibnamefont {{Li}}}, \bibinfo {author} {\bibfnamefont
  {K.~R.}\ \bibnamefont {{Brown}}}, \bibinfo {author} {\bibfnamefont
  {M.}~\bibnamefont {{Cetina}}}, \emph {et~al.},\ }\bibfield  {title} {\bibinfo
  {title} {{Fault-Tolerant Operation of a Quantum Error-Correction Code}},\
  }\href@noop {} {\bibfield  {journal} {\bibinfo  {journal} {arXiv e-prints}\
  ,\ \bibinfo {eid} {arXiv:2009.11482}} (\bibinfo {year} {2020})},\ \Eprint
  {https://arxiv.org/abs/2009.11482} {arXiv:2009.11482 [quant-ph]} \BibitemShut
  {NoStop}%
\bibitem [{\citenamefont {Preskill}(2018)}]{Preskill2018NISQ}%
  \BibitemOpen
  \bibfield  {author} {\bibinfo {author} {\bibfnamefont {J.}~\bibnamefont
  {Preskill}},\ }\bibfield  {title} {\bibinfo {title} {Quantum {C}omputing in
  the {NISQ} era and beyond},\ }\href
  {https://doi.org/10.22331/q-2018-08-06-79} {\bibfield  {journal} {\bibinfo
  {journal} {{Quantum}}\ }\textbf {\bibinfo {volume} {2}},\ \bibinfo {pages}
  {79} (\bibinfo {year} {2018})}\BibitemShut {NoStop}%
\bibitem [{\citenamefont {Lund}\ \emph {et~al.}(2017)\citenamefont {Lund},
  \citenamefont {Bremner},\ and\ \citenamefont {Ralph}}]{Lund2017}%
  \BibitemOpen
  \bibfield  {author} {\bibinfo {author} {\bibfnamefont {A.~P.}\ \bibnamefont
  {Lund}}, \bibinfo {author} {\bibfnamefont {M.~J.}\ \bibnamefont {Bremner}},\
  and\ \bibinfo {author} {\bibfnamefont {T.~C.}\ \bibnamefont {Ralph}},\
  }\bibfield  {title} {\bibinfo {title} {Quantum sampling problems,
  bosonsampling and quantum supremacy},\ }\href
  {https://doi.org/10.1038/s41534-017-0018-2} {\bibfield  {journal} {\bibinfo
  {journal} {npj Quantum Information}\ }\textbf {\bibinfo {volume} {3}},\
  \bibinfo {pages} {15} (\bibinfo {year} {2017})}\BibitemShut {NoStop}%
\bibitem [{\citenamefont {{Harrow}}\ and\ \citenamefont
  {{Montanaro}}(2017)}]{SuprRev2017}%
  \BibitemOpen
  \bibfield  {author} {\bibinfo {author} {\bibfnamefont {A.~W.}\ \bibnamefont
  {{Harrow}}}\ and\ \bibinfo {author} {\bibfnamefont {A.}~\bibnamefont
  {{Montanaro}}},\ }\bibfield  {title} {\bibinfo {title} {{Quantum
  computational supremacy}},\ }\href {https://doi.org/10.1038/nature23458}
  {\bibfield  {journal} {\bibinfo  {journal} {\nat}\ }\textbf {\bibinfo
  {volume} {549}},\ \bibinfo {pages} {203} (\bibinfo {year} {2017})},\ \Eprint
  {https://arxiv.org/abs/1809.07442} {arXiv:1809.07442 [quant-ph]} \BibitemShut
  {NoStop}%
\bibitem [{\citenamefont {Bravyi}\ \emph {et~al.}(2018)\citenamefont {Bravyi},
  \citenamefont {Gosset},\ and\ \citenamefont
  {K{\"o}nig}}]{Bravyi2018advantage}%
  \BibitemOpen
  \bibfield  {author} {\bibinfo {author} {\bibfnamefont {S.}~\bibnamefont
  {Bravyi}}, \bibinfo {author} {\bibfnamefont {D.}~\bibnamefont {Gosset}},\
  and\ \bibinfo {author} {\bibfnamefont {R.}~\bibnamefont {K{\"o}nig}},\
  }\bibfield  {title} {\bibinfo {title} {Quantum advantage with shallow
  circuits},\ }\href {https://doi.org/10.1126/science.aar3106} {\bibfield
  {journal} {\bibinfo  {journal} {Science}\ }\textbf {\bibinfo {volume}
  {362}},\ \bibinfo {pages} {308} (\bibinfo {year} {2018})}\BibitemShut
  {NoStop}%
\bibitem [{\citenamefont {Aaronson}\ and\ \citenamefont
  {Arkhipov}(2013)}]{Aaronson2013}%
  \BibitemOpen
  \bibfield  {author} {\bibinfo {author} {\bibfnamefont {S.}~\bibnamefont
  {Aaronson}}\ and\ \bibinfo {author} {\bibfnamefont {A.}~\bibnamefont
  {Arkhipov}},\ }\bibfield  {title} {\bibinfo {title} {The computational
  complexity of linear optics},\ }\href
  {https://doi.org/10.4086/toc.2013.v009a004} {\bibfield  {journal} {\bibinfo
  {journal} {Theory of Computing}\ }\textbf {\bibinfo {volume} {4}},\ \bibinfo
  {pages} {143} (\bibinfo {year} {2013})}\BibitemShut {NoStop}%
\bibitem [{\citenamefont {Bremner}\ \emph {et~al.}(2011)\citenamefont
  {Bremner}, \citenamefont {Jozsa},\ and\ \citenamefont
  {Shepherd}}]{Bremner2011}%
  \BibitemOpen
  \bibfield  {author} {\bibinfo {author} {\bibfnamefont {M.~J.}\ \bibnamefont
  {Bremner}}, \bibinfo {author} {\bibfnamefont {R.}~\bibnamefont {Jozsa}},\
  and\ \bibinfo {author} {\bibfnamefont {D.~J.}\ \bibnamefont {Shepherd}},\
  }\bibfield  {title} {\bibinfo {title} {Classical simulation of commuting
  quantum computations implies collapse of the polynomial hierarchy},\ }\href
  {https://doi.org/10.1098/rspa.2010.0301} {\bibfield  {journal} {\bibinfo
  {journal} {Proc. R. Soc. A}\ }\textbf {\bibinfo {volume} {467}},\ \bibinfo
  {pages} {459} (\bibinfo {year} {2011})}\BibitemShut {NoStop}%
\bibitem [{\citenamefont {Bremner}\ \emph {et~al.}(2016)\citenamefont
  {Bremner}, \citenamefont {Montanaro},\ and\ \citenamefont
  {Shepherd}}]{Bremner2016}%
  \BibitemOpen
  \bibfield  {author} {\bibinfo {author} {\bibfnamefont {M.}~\bibnamefont
  {Bremner}}, \bibinfo {author} {\bibfnamefont {A.}~\bibnamefont {Montanaro}},\
  and\ \bibinfo {author} {\bibfnamefont {D.}~\bibnamefont {Shepherd}},\
  }\bibfield  {title} {\bibinfo {title} {Average-case complexity versus
  approximate simulation of commuting quantum computations},\ }\href
  {https://doi.org/10.1103/PhysRevLett.117.080501} {\bibfield  {journal}
  {\bibinfo  {journal} {Phys. Rev. Lett.}\ }\textbf {\bibinfo {volume} {117}},\
  \bibinfo {pages} {080501} (\bibinfo {year} {2016})}\BibitemShut {NoStop}%
\bibitem [{\citenamefont {{Boixo}}\ \emph {et~al.}(2018)\citenamefont
  {{Boixo}}, \citenamefont {{Isakov}}, \citenamefont {{Smelyanskiy}},
  \citenamefont {{Babbush}}, \citenamefont {{Ding}}, \citenamefont {{Jiang}},
  \citenamefont {{Bremner}}, \citenamefont {{Martinis}},\ and\ \citenamefont
  {{Neven}}}]{Boixo2016}%
  \BibitemOpen
  \bibfield  {author} {\bibinfo {author} {\bibfnamefont {S.}~\bibnamefont
  {{Boixo}}}, \bibinfo {author} {\bibfnamefont {S.~V.}\ \bibnamefont
  {{Isakov}}}, \bibinfo {author} {\bibfnamefont {V.~N.}\ \bibnamefont
  {{Smelyanskiy}}}, \bibinfo {author} {\bibfnamefont {R.}~\bibnamefont
  {{Babbush}}}, \bibinfo {author} {\bibfnamefont {N.}~\bibnamefont {{Ding}}},
  \bibinfo {author} {\bibfnamefont {Z.}~\bibnamefont {{Jiang}}}, \bibinfo
  {author} {\bibfnamefont {M.~J.}\ \bibnamefont {{Bremner}}}, \bibinfo {author}
  {\bibfnamefont {J.~M.}\ \bibnamefont {{Martinis}}},\ and\ \bibinfo {author}
  {\bibfnamefont {H.}~\bibnamefont {{Neven}}},\ }\bibfield  {title} {\bibinfo
  {title} {{Characterizing Quantum Supremacy in Near-Term Devices}},\
  }\href@noop {} {\bibfield  {journal} {\bibinfo  {journal} {Nature Physics}\
  }\textbf {\bibinfo {volume} {14}},\ \bibinfo {pages} {595–600} (\bibinfo
  {year} {2018})},\ \Eprint {https://arxiv.org/abs/1608.00263}
  {arXiv:1608.00263 [quant-ph]} \BibitemShut {NoStop}%
\bibitem [{\citenamefont {Bouland}\ \emph {et~al.}(2019)\citenamefont
  {Bouland}, \citenamefont {Fefferman}, \citenamefont {Nirkhe},\ and\
  \citenamefont {Vazirani}}]{bouland_complexity_2019}%
  \BibitemOpen
  \bibfield  {author} {\bibinfo {author} {\bibfnamefont {A.}~\bibnamefont
  {Bouland}}, \bibinfo {author} {\bibfnamefont {B.}~\bibnamefont {Fefferman}},
  \bibinfo {author} {\bibfnamefont {C.}~\bibnamefont {Nirkhe}},\ and\ \bibinfo
  {author} {\bibfnamefont {U.}~\bibnamefont {Vazirani}},\ }\bibfield  {title}
  {\bibinfo {title} {On the complexity and verification of quantum random
  circuit sampling},\ }\href {https://doi.org/10.1038/s41567-018-0318-2}
  {\bibfield  {journal} {\bibinfo  {journal} {Nature Physics}\ }\textbf
  {\bibinfo {volume} {15}},\ \bibinfo {pages} {159} (\bibinfo {year}
  {2019})}\BibitemShut {NoStop}%
\bibitem [{\citenamefont {Fefferman}\ and\ \citenamefont
  {Umans}(2016)}]{fefferman2016fourier}%
  \BibitemOpen
  \bibfield  {author} {\bibinfo {author} {\bibfnamefont {B.}~\bibnamefont
  {Fefferman}}\ and\ \bibinfo {author} {\bibfnamefont {C.}~\bibnamefont
  {Umans}},\ }\bibfield  {title} {\bibinfo {title} {{On the Power of Quantum
  Fourier Sampling}},\ }in\ \href {https://doi.org/10.4230/LIPIcs.TQC.2016.1}
  {\emph {\bibinfo {booktitle} {11th Conference on the Theory of Quantum
  Computation, Communication and Cryptography (TQC 2016)}}},\ \bibinfo {series}
  {Leibniz International Proceedings in Informatics (LIPIcs)}, Vol.~\bibinfo
  {volume} {61},\ \bibinfo {editor} {edited by\ \bibinfo {editor}
  {\bibfnamefont {A.}~\bibnamefont {Broadbent}}}\ (\bibinfo  {publisher}
  {Schloss Dagstuhl--Leibniz-Zentrum fuer Informatik},\ \bibinfo {address}
  {Dagstuhl, Germany},\ \bibinfo {year} {2016})\ pp.\ \bibinfo {pages}
  {1:1--1:19}\BibitemShut {NoStop}%
\bibitem [{\citenamefont {Morimae}(2017)}]{Morimae2017}%
  \BibitemOpen
  \bibfield  {author} {\bibinfo {author} {\bibfnamefont {T.}~\bibnamefont
  {Morimae}},\ }\bibfield  {title} {\bibinfo {title} {Hardness of classically
  sampling the one-clean-qubit model with constant total variation distance
  error},\ }\href {https://doi.org/10.1103/PhysRevA.96.040302} {\bibfield
  {journal} {\bibinfo  {journal} {Phys. Rev. A}\ }\textbf {\bibinfo {volume}
  {96}},\ \bibinfo {pages} {040302} (\bibinfo {year} {2017})}\BibitemShut
  {NoStop}%
\bibitem [{\citenamefont {Bermejo-Vega}\ \emph {et~al.}(2018)\citenamefont
  {Bermejo-Vega}, \citenamefont {Hangleiter}, \citenamefont {Schwarz},
  \citenamefont {Raussendorf},\ and\ \citenamefont
  {Eisert}}]{bermejo-vega_architectures_2018}%
  \BibitemOpen
  \bibfield  {author} {\bibinfo {author} {\bibfnamefont {J.}~\bibnamefont
  {Bermejo-Vega}}, \bibinfo {author} {\bibfnamefont {D.}~\bibnamefont
  {Hangleiter}}, \bibinfo {author} {\bibfnamefont {M.}~\bibnamefont {Schwarz}},
  \bibinfo {author} {\bibfnamefont {R.}~\bibnamefont {Raussendorf}},\ and\
  \bibinfo {author} {\bibfnamefont {J.}~\bibnamefont {Eisert}},\ }\bibfield
  {title} {\bibinfo {title} {Architectures for quantum simulation showing a
  quantum speedup},\ }\href {https://doi.org/10.1103/PhysRevX.8.021010}
  {\bibfield  {journal} {\bibinfo  {journal} {Phys. Rev. X}\ }\textbf {\bibinfo
  {volume} {8}},\ \bibinfo {pages} {021010} (\bibinfo {year} {2018})},\
  \bibinfo {note} {arXiv: 1703.00466}\BibitemShut {NoStop}%
\bibitem [{\citenamefont {Hamilton}\ \emph {et~al.}(2017)\citenamefont
  {Hamilton}, \citenamefont {Kruse}, \citenamefont {Sansoni}, \citenamefont
  {Barkhofen}, \citenamefont {Silberhorn},\ and\ \citenamefont
  {Jex}}]{hamilton2017gaussian}%
  \BibitemOpen
  \bibfield  {author} {\bibinfo {author} {\bibfnamefont {C.~S.}\ \bibnamefont
  {Hamilton}}, \bibinfo {author} {\bibfnamefont {R.}~\bibnamefont {Kruse}},
  \bibinfo {author} {\bibfnamefont {L.}~\bibnamefont {Sansoni}}, \bibinfo
  {author} {\bibfnamefont {S.}~\bibnamefont {Barkhofen}}, \bibinfo {author}
  {\bibfnamefont {C.}~\bibnamefont {Silberhorn}},\ and\ \bibinfo {author}
  {\bibfnamefont {I.}~\bibnamefont {Jex}},\ }\bibfield  {title} {\bibinfo
  {title} {Gaussian boson sampling},\ }\href
  {https://doi.org/10.1103/PhysRevLett.119.170501} {\bibfield  {journal}
  {\bibinfo  {journal} {Phys. Rev. Lett.}\ }\textbf {\bibinfo {volume} {119}},\
  \bibinfo {pages} {170501} (\bibinfo {year} {2017})}\BibitemShut {NoStop}%
\bibitem [{\citenamefont {Lund}\ \emph {et~al.}(2014)\citenamefont {Lund},
  \citenamefont {Laing}, \citenamefont {Rahimi-Keshari}, \citenamefont
  {Rudolph}, \citenamefont {O’Brien},\ and\ \citenamefont
  {Ralph}}]{lund2014boson}%
  \BibitemOpen
  \bibfield  {author} {\bibinfo {author} {\bibfnamefont {A.~P.}\ \bibnamefont
  {Lund}}, \bibinfo {author} {\bibfnamefont {A.}~\bibnamefont {Laing}},
  \bibinfo {author} {\bibfnamefont {S.}~\bibnamefont {Rahimi-Keshari}},
  \bibinfo {author} {\bibfnamefont {T.}~\bibnamefont {Rudolph}}, \bibinfo
  {author} {\bibfnamefont {J.~L.}\ \bibnamefont {O’Brien}},\ and\ \bibinfo
  {author} {\bibfnamefont {T.~C.}\ \bibnamefont {Ralph}},\ }\bibfield  {title}
  {\bibinfo {title} {Boson sampling from a gaussian state},\ }\href@noop {}
  {\bibfield  {journal} {\bibinfo  {journal} {Physical review letters}\
  }\textbf {\bibinfo {volume} {113}},\ \bibinfo {pages} {100502} (\bibinfo
  {year} {2014})}\BibitemShut {NoStop}%
\bibitem [{\citenamefont {{Haferkamp}}\ \emph {et~al.}(2019)\citenamefont
  {{Haferkamp}}, \citenamefont {{Hangleiter}}, \citenamefont {{Bouland}},
  \citenamefont {{Fefferman}}, \citenamefont {{Eisert}},\ and\ \citenamefont
  {{Bermejo-Vega}}}]{haferkamp_closing_2019}%
  \BibitemOpen
  \bibfield  {author} {\bibinfo {author} {\bibfnamefont {J.}~\bibnamefont
  {{Haferkamp}}}, \bibinfo {author} {\bibfnamefont {D.}~\bibnamefont
  {{Hangleiter}}}, \bibinfo {author} {\bibfnamefont {A.}~\bibnamefont
  {{Bouland}}}, \bibinfo {author} {\bibfnamefont {B.}~\bibnamefont
  {{Fefferman}}}, \bibinfo {author} {\bibfnamefont {J.}~\bibnamefont
  {{Eisert}}},\ and\ \bibinfo {author} {\bibfnamefont {J.}~\bibnamefont
  {{Bermejo-Vega}}},\ }\bibfield  {title} {\bibinfo {title} {{Closing gaps of a
  quantum advantage with short-time Hamiltonian dynamics}},\ }\href@noop {}
  {\bibfield  {journal} {\bibinfo  {journal} {arXiv e-prints}\ ,\ \bibinfo
  {eid} {arXiv:1908.08069}} (\bibinfo {year} {2019})},\ \Eprint
  {https://arxiv.org/abs/1908.08069} {arXiv:1908.08069 [quant-ph]} \BibitemShut
  {NoStop}%
\bibitem [{\citenamefont {Clements}\ \emph {et~al.}(2016)\citenamefont
  {Clements}, \citenamefont {Humphreys}, \citenamefont {Metcalf}, \citenamefont
  {Kolthammer},\ and\ \citenamefont {Walmsley}}]{ReckUniform2016}%
  \BibitemOpen
  \bibfield  {author} {\bibinfo {author} {\bibfnamefont {W.~R.}\ \bibnamefont
  {Clements}}, \bibinfo {author} {\bibfnamefont {P.~C.}\ \bibnamefont
  {Humphreys}}, \bibinfo {author} {\bibfnamefont {B.~J.}\ \bibnamefont
  {Metcalf}}, \bibinfo {author} {\bibfnamefont {W.~S.}\ \bibnamefont
  {Kolthammer}},\ and\ \bibinfo {author} {\bibfnamefont {I.~A.}\ \bibnamefont
  {Walmsley}},\ }\bibfield  {title} {\bibinfo {title} {Optimal design for
  universal multiport interferometers},\ }\href
  {https://doi.org/10.1364/OPTICA.3.001460} {\bibfield  {journal} {\bibinfo
  {journal} {Optica}\ }\textbf {\bibinfo {volume} {3}},\ \bibinfo {pages}
  {1460} (\bibinfo {year} {2016})}\BibitemShut {NoStop}%
\bibitem [{\citenamefont {Reck}\ \emph {et~al.}(1994)\citenamefont {Reck},
  \citenamefont {Zeilinger}, \citenamefont {Bernstein},\ and\ \citenamefont
  {Bertani}}]{Reck1994}%
  \BibitemOpen
  \bibfield  {author} {\bibinfo {author} {\bibfnamefont {M.}~\bibnamefont
  {Reck}}, \bibinfo {author} {\bibfnamefont {A.}~\bibnamefont {Zeilinger}},
  \bibinfo {author} {\bibfnamefont {H.~J.}\ \bibnamefont {Bernstein}},\ and\
  \bibinfo {author} {\bibfnamefont {P.}~\bibnamefont {Bertani}},\ }\bibfield
  {title} {\bibinfo {title} {Experimental realization of any discrete unitary
  operator},\ }\href {https://doi.org/10.1103/PhysRevLett.73.58} {\bibfield
  {journal} {\bibinfo  {journal} {Phys. Rev. Lett.}\ }\textbf {\bibinfo
  {volume} {73}},\ \bibinfo {pages} {58} (\bibinfo {year} {1994})}\BibitemShut
  {NoStop}%
\bibitem [{\citenamefont {{Arute}}\ \emph {et~al.}(2019)\citenamefont
  {{Arute}}, \citenamefont {{Arya}}, \citenamefont {{Babbush}}, \citenamefont
  {{Bacon}}, \citenamefont {{Bardin}}, \citenamefont {{Barends}}, \citenamefont
  {{Biswas}}, \citenamefont {{Boixo}}, \citenamefont {{Brandao}}, \citenamefont
  {{Buell}} \emph {et~al.}}]{Suprem2019}%
  \BibitemOpen
  \bibfield  {author} {\bibinfo {author} {\bibfnamefont {F.}~\bibnamefont
  {{Arute}}}, \bibinfo {author} {\bibfnamefont {K.}~\bibnamefont {{Arya}}},
  \bibinfo {author} {\bibfnamefont {R.}~\bibnamefont {{Babbush}}}, \bibinfo
  {author} {\bibfnamefont {D.}~\bibnamefont {{Bacon}}}, \bibinfo {author}
  {\bibfnamefont {J.~C.}\ \bibnamefont {{Bardin}}}, \bibinfo {author}
  {\bibfnamefont {R.}~\bibnamefont {{Barends}}}, \bibinfo {author}
  {\bibfnamefont {R.}~\bibnamefont {{Biswas}}}, \bibinfo {author}
  {\bibfnamefont {S.}~\bibnamefont {{Boixo}}}, \bibinfo {author} {\bibfnamefont
  {F.~G.~S.~L.}\ \bibnamefont {{Brandao}}}, \bibinfo {author} {\bibfnamefont
  {D.~A.}\ \bibnamefont {{Buell}}}, \emph {et~al.},\ }\bibfield  {title}
  {\bibinfo {title} {{Quantum supremacy using a programmable superconducting
  processor}},\ }\href {https://doi.org/10.1038/s41586-019-1666-5} {\bibfield
  {journal} {\bibinfo  {journal} {\nat}\ }\textbf {\bibinfo {volume} {574}},\
  \bibinfo {pages} {505} (\bibinfo {year} {2019})}\BibitemShut {NoStop}%
\bibitem [{\citenamefont {Hangleiter}\ \emph
  {et~al.}(2018{\natexlab{a}})\citenamefont {Hangleiter}, \citenamefont
  {Bermejo-Vega}, \citenamefont {Schwarz},\ and\ \citenamefont
  {Eisert}}]{hangleiter_anticoncentration_2018}%
  \BibitemOpen
  \bibfield  {author} {\bibinfo {author} {\bibfnamefont {D.}~\bibnamefont
  {Hangleiter}}, \bibinfo {author} {\bibfnamefont {J.}~\bibnamefont
  {Bermejo-Vega}}, \bibinfo {author} {\bibfnamefont {M.}~\bibnamefont
  {Schwarz}},\ and\ \bibinfo {author} {\bibfnamefont {J.}~\bibnamefont
  {Eisert}},\ }\bibfield  {title} {\bibinfo {title} {Anticoncentration theorems
  for schemes showing a quantum speedup},\ }\href
  {https://doi.org/10.22331/q-2018-05-22-65} {\bibfield  {journal} {\bibinfo
  {journal} {{Quantum}}\ }\textbf {\bibinfo {volume} {2}},\ \bibinfo {pages}
  {65} (\bibinfo {year} {2018}{\natexlab{a}})}\BibitemShut {NoStop}%
\bibitem [{\citenamefont {{Harrow}}\ and\ \citenamefont
  {{Mehraban}}(2018)}]{Harrow2018}%
  \BibitemOpen
  \bibfield  {author} {\bibinfo {author} {\bibfnamefont {A.}~\bibnamefont
  {{Harrow}}}\ and\ \bibinfo {author} {\bibfnamefont {S.}~\bibnamefont
  {{Mehraban}}},\ }\bibfield  {title} {\bibinfo {title} {{Approximate unitary
  $t$-designs by short random quantum circuits using nearest-neighbor and
  long-range gates}},\ }\href@noop {} {\bibfield  {journal} {\bibinfo
  {journal} {arXiv e-prints}\ ,\ \bibinfo {eid} {arXiv:1809.06957}} (\bibinfo
  {year} {2018})},\ \Eprint {https://arxiv.org/abs/1809.06957}
  {arXiv:1809.06957 [quant-ph]} \BibitemShut {NoStop}%
\bibitem [{\citenamefont {Brod}\ \emph {et~al.}(2019)\citenamefont {Brod},
  \citenamefont {Galv{\~a}o}, \citenamefont {Crespi}, \citenamefont {Osellame},
  \citenamefont {Spagnolo},\ and\ \citenamefont
  {Sciarrino}}]{brod2019photonic}%
  \BibitemOpen
  \bibfield  {author} {\bibinfo {author} {\bibfnamefont {D.~J.}\ \bibnamefont
  {Brod}}, \bibinfo {author} {\bibfnamefont {E.~F.}\ \bibnamefont
  {Galv{\~a}o}}, \bibinfo {author} {\bibfnamefont {A.}~\bibnamefont {Crespi}},
  \bibinfo {author} {\bibfnamefont {R.}~\bibnamefont {Osellame}}, \bibinfo
  {author} {\bibfnamefont {N.}~\bibnamefont {Spagnolo}},\ and\ \bibinfo
  {author} {\bibfnamefont {F.}~\bibnamefont {Sciarrino}},\ }\bibfield  {title}
  {\bibinfo {title} {Photonic implementation of boson sampling: a review},\
  }\href@noop {} {\bibfield  {journal} {\bibinfo  {journal} {Advanced
  Photonics}\ }\textbf {\bibinfo {volume} {1}},\ \bibinfo {pages} {034001}
  (\bibinfo {year} {2019})}\BibitemShut {NoStop}%
\bibitem [{\citenamefont {Wang}\ \emph {et~al.}(2019)\citenamefont {Wang},
  \citenamefont {Qin}, \citenamefont {Ding}, \citenamefont {Chen},
  \citenamefont {Chen}, \citenamefont {You}, \citenamefont {He}, \citenamefont
  {Jiang}, \citenamefont {You}, \citenamefont {Wang} \emph
  {et~al.}}]{wang2019boson}%
  \BibitemOpen
  \bibfield  {author} {\bibinfo {author} {\bibfnamefont {H.}~\bibnamefont
  {Wang}}, \bibinfo {author} {\bibfnamefont {J.}~\bibnamefont {Qin}}, \bibinfo
  {author} {\bibfnamefont {X.}~\bibnamefont {Ding}}, \bibinfo {author}
  {\bibfnamefont {M.-C.}\ \bibnamefont {Chen}}, \bibinfo {author}
  {\bibfnamefont {S.}~\bibnamefont {Chen}}, \bibinfo {author} {\bibfnamefont
  {X.}~\bibnamefont {You}}, \bibinfo {author} {\bibfnamefont {Y.-M.}\
  \bibnamefont {He}}, \bibinfo {author} {\bibfnamefont {X.}~\bibnamefont
  {Jiang}}, \bibinfo {author} {\bibfnamefont {L.}~\bibnamefont {You}}, \bibinfo
  {author} {\bibfnamefont {Z.}~\bibnamefont {Wang}}, \emph {et~al.},\
  }\bibfield  {title} {\bibinfo {title} {Boson sampling with 20 input photons
  and a 60-mode interferometer in a $10^{14}$-dimensional hilbert space},\
  }\href@noop {} {\bibfield  {journal} {\bibinfo  {journal} {Physical review
  letters}\ }\textbf {\bibinfo {volume} {123}},\ \bibinfo {pages} {250503}
  (\bibinfo {year} {2019})}\BibitemShut {NoStop}%
\bibitem [{\citenamefont {Zhong}\ \emph {et~al.}(2020)\citenamefont {Zhong},
  \citenamefont {Wang}, \citenamefont {Deng}, \citenamefont {Chen},
  \citenamefont {Peng}, \citenamefont {Luo}, \citenamefont {Qin}, \citenamefont
  {Wu}, \citenamefont {Ding}, \citenamefont {Hu} \emph
  {et~al.}}]{GaussBSExperiment2020}%
  \BibitemOpen
  \bibfield  {author} {\bibinfo {author} {\bibfnamefont {H.-S.}\ \bibnamefont
  {Zhong}}, \bibinfo {author} {\bibfnamefont {H.}~\bibnamefont {Wang}},
  \bibinfo {author} {\bibfnamefont {Y.-H.}\ \bibnamefont {Deng}}, \bibinfo
  {author} {\bibfnamefont {M.-C.}\ \bibnamefont {Chen}}, \bibinfo {author}
  {\bibfnamefont {L.-C.}\ \bibnamefont {Peng}}, \bibinfo {author}
  {\bibfnamefont {Y.-H.}\ \bibnamefont {Luo}}, \bibinfo {author} {\bibfnamefont
  {J.}~\bibnamefont {Qin}}, \bibinfo {author} {\bibfnamefont {D.}~\bibnamefont
  {Wu}}, \bibinfo {author} {\bibfnamefont {X.}~\bibnamefont {Ding}}, \bibinfo
  {author} {\bibfnamefont {Y.}~\bibnamefont {Hu}}, \emph {et~al.},\ }\bibfield
  {title} {\bibinfo {title} {Quantum computational advantage using photons},\
  }\href {https://doi.org/10.1126/science.abe8770} {\bibfield  {journal}
  {\bibinfo  {journal} {Science}\ }\textbf {\bibinfo {volume} {370}},\ \bibinfo
  {pages} {1460} (\bibinfo {year} {2020})}\BibitemShut {NoStop}%
\bibitem [{\citenamefont {Terhal}\ and\ \citenamefont
  {DiVincenzo}(2002)}]{terhal_classical_2002}%
  \BibitemOpen
  \bibfield  {author} {\bibinfo {author} {\bibfnamefont {B.~M.}\ \bibnamefont
  {Terhal}}\ and\ \bibinfo {author} {\bibfnamefont {D.~P.}\ \bibnamefont
  {DiVincenzo}},\ }\bibfield  {title} {\bibinfo {title} {Classical simulation
  of noninteracting-fermion quantum circuits},\ }\href
  {https://doi.org/10.1103/PhysRevA.65.032325} {\bibfield  {journal} {\bibinfo
  {journal} {Phys. Rev. A}\ }\textbf {\bibinfo {volume} {65}},\ \bibinfo
  {pages} {032325} (\bibinfo {year} {2002})}\BibitemShut {NoStop}%
\bibitem [{\citenamefont {{Knill}}(2001)}]{knill_fermionic_2001}%
  \BibitemOpen
  \bibfield  {author} {\bibinfo {author} {\bibfnamefont {E.}~\bibnamefont
  {{Knill}}},\ }\bibfield  {title} {\bibinfo {title} {{Fermionic Linear Optics
  and Matchgates}},\ }\href@noop {} {\bibfield  {journal} {\bibinfo  {journal}
  {arXiv e-prints}\ ,\ \bibinfo {eid} {quant-ph/0108033}} (\bibinfo {year}
  {2001})},\ \Eprint {https://arxiv.org/abs/quant-ph/0108033}
  {arXiv:quant-ph/0108033 [quant-ph]} \BibitemShut {NoStop}%
\bibitem [{\citenamefont {Valiant}(2002)}]{valiant_quantum_2002}%
  \BibitemOpen
  \bibfield  {author} {\bibinfo {author} {\bibfnamefont {L.}~\bibnamefont
  {Valiant}},\ }\bibfield  {title} {\bibinfo {title} {Quantum {Circuits} {That}
  {Can} {Be} {Simulated} {Classically} in {Polynomial} {Time}},\ }\href
  {https://doi.org/10.1137/S0097539700377025} {\bibfield  {journal} {\bibinfo
  {journal} {SIAM J. Comput.}\ }\textbf {\bibinfo {volume} {31}},\ \bibinfo
  {pages} {1229} (\bibinfo {year} {2002})}\BibitemShut {NoStop}%
\bibitem [{\citenamefont {Bocquillon}\ \emph {et~al.}(2014)\citenamefont
  {Bocquillon}, \citenamefont {Freulon}, \citenamefont {Parmentier},
  \citenamefont {Berroir}, \citenamefont {Pla{\c{c}}ais}, \citenamefont {Wahl},
  \citenamefont {Rech}, \citenamefont {Jonckheere}, \citenamefont {Martin},
  \citenamefont {Grenier} \emph {et~al.}}]{bocquillon2014electron}%
  \BibitemOpen
  \bibfield  {author} {\bibinfo {author} {\bibfnamefont {E.}~\bibnamefont
  {Bocquillon}}, \bibinfo {author} {\bibfnamefont {V.}~\bibnamefont {Freulon}},
  \bibinfo {author} {\bibfnamefont {F.~D.}\ \bibnamefont {Parmentier}},
  \bibinfo {author} {\bibfnamefont {J.-M.}\ \bibnamefont {Berroir}}, \bibinfo
  {author} {\bibfnamefont {B.}~\bibnamefont {Pla{\c{c}}ais}}, \bibinfo {author}
  {\bibfnamefont {C.}~\bibnamefont {Wahl}}, \bibinfo {author} {\bibfnamefont
  {J.}~\bibnamefont {Rech}}, \bibinfo {author} {\bibfnamefont {T.}~\bibnamefont
  {Jonckheere}}, \bibinfo {author} {\bibfnamefont {T.}~\bibnamefont {Martin}},
  \bibinfo {author} {\bibfnamefont {C.}~\bibnamefont {Grenier}}, \emph
  {et~al.},\ }\bibfield  {title} {\bibinfo {title} {Electron quantum optics in
  ballistic chiral conductors},\ }\href
  {https://doi.org/https://doi.org/10.1002/andp.201300181} {\bibfield
  {journal} {\bibinfo  {journal} {Annalen der Physik}\ }\textbf {\bibinfo
  {volume} {526}},\ \bibinfo {pages} {1} (\bibinfo {year} {2014})}\BibitemShut
  {NoStop}%
\bibitem [{\citenamefont {Arute}\ \emph {et~al.}(2020)\citenamefont {Arute},
  \citenamefont {Arya}, \citenamefont {Babbush}, \citenamefont {Bacon},
  \citenamefont {Bardin}, \citenamefont {Barends}, \citenamefont {Boixo},
  \citenamefont {Broughton}, \citenamefont {Buckley}, \citenamefont {Buell}
  \emph {et~al.}}]{chemistryGOOGLE2020}%
  \BibitemOpen
  \bibfield  {author} {\bibinfo {author} {\bibfnamefont {F.}~\bibnamefont
  {Arute}}, \bibinfo {author} {\bibfnamefont {K.}~\bibnamefont {Arya}},
  \bibinfo {author} {\bibfnamefont {R.}~\bibnamefont {Babbush}}, \bibinfo
  {author} {\bibfnamefont {D.}~\bibnamefont {Bacon}}, \bibinfo {author}
  {\bibfnamefont {J.~C.}\ \bibnamefont {Bardin}}, \bibinfo {author}
  {\bibfnamefont {R.}~\bibnamefont {Barends}}, \bibinfo {author} {\bibfnamefont
  {S.}~\bibnamefont {Boixo}}, \bibinfo {author} {\bibfnamefont
  {M.}~\bibnamefont {Broughton}}, \bibinfo {author} {\bibfnamefont {B.~B.}\
  \bibnamefont {Buckley}}, \bibinfo {author} {\bibfnamefont {D.~A.}\
  \bibnamefont {Buell}}, \emph {et~al.},\ }\bibfield  {title} {\bibinfo {title}
  {Hartree-fock on a superconducting qubit quantum computer},\ }\href
  {https://doi.org/10.1126/science.abb9811} {\bibfield  {journal} {\bibinfo
  {journal} {Science}\ }\textbf {\bibinfo {volume} {369}},\ \bibinfo {pages}
  {1084} (\bibinfo {year} {2020})}\BibitemShut {NoStop}%
\bibitem [{\citenamefont {Kivlichan}\ \emph {et~al.}(2018)\citenamefont
  {Kivlichan}, \citenamefont {McClean}, \citenamefont {Wiebe}, \citenamefont
  {Gidney}, \citenamefont {Aspuru-Guzik}, \citenamefont {Chan},\ and\
  \citenamefont {Babbush}}]{fermPASSlayout2018}%
  \BibitemOpen
  \bibfield  {author} {\bibinfo {author} {\bibfnamefont {I.~D.}\ \bibnamefont
  {Kivlichan}}, \bibinfo {author} {\bibfnamefont {J.}~\bibnamefont {McClean}},
  \bibinfo {author} {\bibfnamefont {N.}~\bibnamefont {Wiebe}}, \bibinfo
  {author} {\bibfnamefont {C.}~\bibnamefont {Gidney}}, \bibinfo {author}
  {\bibfnamefont {A.}~\bibnamefont {Aspuru-Guzik}}, \bibinfo {author}
  {\bibfnamefont {G.~K.-L.}\ \bibnamefont {Chan}},\ and\ \bibinfo {author}
  {\bibfnamefont {R.}~\bibnamefont {Babbush}},\ }\bibfield  {title} {\bibinfo
  {title} {Quantum simulation of electronic structure with linear depth and
  connectivity},\ }\href {https://doi.org/10.1103/PhysRevLett.120.110501}
  {\bibfield  {journal} {\bibinfo  {journal} {Phys. Rev. Lett.}\ }\textbf
  {\bibinfo {volume} {120}},\ \bibinfo {pages} {110501} (\bibinfo {year}
  {2018})}\BibitemShut {NoStop}%
\bibitem [{\citenamefont {Jiang}\ \emph {et~al.}(2018)\citenamefont {Jiang},
  \citenamefont {Sung}, \citenamefont {Kechedzhi}, \citenamefont
  {Smelyanskiy},\ and\ \citenamefont {Boixo}}]{fermACTpassLAYOUT2018}%
  \BibitemOpen
  \bibfield  {author} {\bibinfo {author} {\bibfnamefont {Z.}~\bibnamefont
  {Jiang}}, \bibinfo {author} {\bibfnamefont {K.~J.}\ \bibnamefont {Sung}},
  \bibinfo {author} {\bibfnamefont {K.}~\bibnamefont {Kechedzhi}}, \bibinfo
  {author} {\bibfnamefont {V.~N.}\ \bibnamefont {Smelyanskiy}},\ and\ \bibinfo
  {author} {\bibfnamefont {S.}~\bibnamefont {Boixo}},\ }\bibfield  {title}
  {\bibinfo {title} {Quantum algorithms to simulate many-body physics of
  correlated fermions},\ }\href
  {https://doi.org/10.1103/PhysRevApplied.9.044036} {\bibfield  {journal}
  {\bibinfo  {journal} {Phys. Rev. Applied}\ }\textbf {\bibinfo {volume} {9}},\
  \bibinfo {pages} {044036} (\bibinfo {year} {2018})}\BibitemShut {NoStop}%
\bibitem [{\citenamefont {Dallaire-Demers}\ \emph {et~al.}(2019)\citenamefont
  {Dallaire-Demers}, \citenamefont {Romero}, \citenamefont {Veis},
  \citenamefont {Sim},\ and\ \citenamefont
  {Aspuru-Guzik}}]{fermGAUSSlayout2019}%
  \BibitemOpen
  \bibfield  {author} {\bibinfo {author} {\bibfnamefont {P.-L.}\ \bibnamefont
  {Dallaire-Demers}}, \bibinfo {author} {\bibfnamefont {J.}~\bibnamefont
  {Romero}}, \bibinfo {author} {\bibfnamefont {L.}~\bibnamefont {Veis}},
  \bibinfo {author} {\bibfnamefont {S.}~\bibnamefont {Sim}},\ and\ \bibinfo
  {author} {\bibfnamefont {A.}~\bibnamefont {Aspuru-Guzik}},\ }\bibfield
  {title} {\bibinfo {title} {Low-depth circuit ansatz for preparing correlated
  fermionic states on a quantum computer},\ }\href
  {https://doi.org/10.1088/2058-9565/ab3951} {\bibfield  {journal} {\bibinfo
  {journal} {Quantum Science and Technology}\ }\textbf {\bibinfo {volume}
  {4}},\ \bibinfo {pages} {045005} (\bibinfo {year} {2019})}\BibitemShut
  {NoStop}%
\bibitem [{\citenamefont {Foxen}\ \emph {et~al.}(2020)\citenamefont {Foxen},
  \citenamefont {Neill}, \citenamefont {Dunsworth}, \citenamefont {Roushan},
  \citenamefont {Chiaro}, \citenamefont {Megrant}, \citenamefont {Kelly},
  \citenamefont {Chen}, \citenamefont {Satzinger}, \citenamefont {Barends}
  \emph {et~al.}}]{ContGatesets2020}%
  \BibitemOpen
  \bibfield  {author} {\bibinfo {author} {\bibfnamefont {B.}~\bibnamefont
  {Foxen}}, \bibinfo {author} {\bibfnamefont {C.}~\bibnamefont {Neill}},
  \bibinfo {author} {\bibfnamefont {A.}~\bibnamefont {Dunsworth}}, \bibinfo
  {author} {\bibfnamefont {P.}~\bibnamefont {Roushan}}, \bibinfo {author}
  {\bibfnamefont {B.}~\bibnamefont {Chiaro}}, \bibinfo {author} {\bibfnamefont
  {A.}~\bibnamefont {Megrant}}, \bibinfo {author} {\bibfnamefont
  {J.}~\bibnamefont {Kelly}}, \bibinfo {author} {\bibfnamefont
  {Z.}~\bibnamefont {Chen}}, \bibinfo {author} {\bibfnamefont {K.}~\bibnamefont
  {Satzinger}}, \bibinfo {author} {\bibfnamefont {R.}~\bibnamefont {Barends}},
  \emph {et~al.},\ }\bibfield  {title} {\bibinfo {title} {Demonstrating a
  continuous set of two-qubit gates for near-term quantum algorithms},\ }\href
  {https://doi.org/10.1103/PhysRevLett.125.120504} {\bibfield  {journal}
  {\bibinfo  {journal} {Phys. Rev. Lett.}\ }\textbf {\bibinfo {volume} {125}},\
  \bibinfo {pages} {120504} (\bibinfo {year} {2020})}\BibitemShut {NoStop}%
\bibitem [{\citenamefont {{Dalzell}}\ \emph {et~al.}(2020)\citenamefont
  {{Dalzell}}, \citenamefont {{Hunter-Jones}},\ and\ \citenamefont
  {{Brand{\~a}o}}}]{logANTICONCENTRATION}%
  \BibitemOpen
  \bibfield  {author} {\bibinfo {author} {\bibfnamefont {A.~M.}\ \bibnamefont
  {{Dalzell}}}, \bibinfo {author} {\bibfnamefont {N.}~\bibnamefont
  {{Hunter-Jones}}},\ and\ \bibinfo {author} {\bibfnamefont {F.~G.~S.~L.}\
  \bibnamefont {{Brand{\~a}o}}},\ }\bibfield  {title} {\bibinfo {title}
  {{Random quantum circuits anti-concentrate in log depth}},\ }\href@noop {}
  {\bibfield  {journal} {\bibinfo  {journal} {arXiv e-prints}\ ,\ \bibinfo
  {eid} {arXiv:2011.12277}} (\bibinfo {year} {2020})},\ \Eprint
  {https://arxiv.org/abs/2011.12277} {arXiv:2011.12277 [quant-ph]} \BibitemShut
  {NoStop}%
\bibitem [{\citenamefont {Arora}\ and\ \citenamefont
  {Barak}(2009)}]{aroraComplexity}%
  \BibitemOpen
  \bibfield  {author} {\bibinfo {author} {\bibfnamefont {S.}~\bibnamefont
  {Arora}}\ and\ \bibinfo {author} {\bibfnamefont {B.}~\bibnamefont {Barak}},\
  }\href@noop {} {\emph {\bibinfo {title} {Computational Complexity: A Modern
  Approach}}},\ \bibinfo {edition} {1st}\ ed.\ (\bibinfo  {publisher}
  {Cambridge University Press},\ \bibinfo {address} {USA},\ \bibinfo {year}
  {2009})\BibitemShut {NoStop}%
\bibitem [{\citenamefont {Terhal}\ and\ \citenamefont
  {DiVincenzo}(2004)}]{terhal2004adptive}%
  \BibitemOpen
  \bibfield  {author} {\bibinfo {author} {\bibfnamefont {B.~M.}\ \bibnamefont
  {Terhal}}\ and\ \bibinfo {author} {\bibfnamefont {D.~P.}\ \bibnamefont
  {DiVincenzo}},\ }\bibfield  {title} {\bibinfo {title} {Adaptive quantum
  computation, constant depth quantum circuits and {Arthur}-{Merlin} games},\
  }\href {https://doi.org/10.26421/QIC4.2-5} {\bibfield  {journal} {\bibinfo
  {journal} {Quantum Information \& Computation}\ }\textbf {\bibinfo {volume}
  {4}},\ \bibinfo {pages} {134} (\bibinfo {year} {2004})}\BibitemShut {NoStop}%
\bibitem [{\citenamefont {Fenner}\ \emph {et~al.}(1999)\citenamefont {Fenner},
  \citenamefont {Green}, \citenamefont {Homer},\ and\ \citenamefont
  {Pruim}}]{fenner1999determining}%
  \BibitemOpen
  \bibfield  {author} {\bibinfo {author} {\bibfnamefont {S.}~\bibnamefont
  {Fenner}}, \bibinfo {author} {\bibfnamefont {F.}~\bibnamefont {Green}},
  \bibinfo {author} {\bibfnamefont {S.}~\bibnamefont {Homer}},\ and\ \bibinfo
  {author} {\bibfnamefont {R.}~\bibnamefont {Pruim}},\ }\bibfield  {title}
  {\bibinfo {title} {Determining acceptance possibility for a quantum
  computation is hard for the polynomial hierarchy},\ }\href@noop {} {\bibfield
   {journal} {\bibinfo  {journal} {Proceedings of the Royal Society of London.
  Series A: Mathematical, Physical and Engineering Sciences}\ }\textbf
  {\bibinfo {volume} {455}},\ \bibinfo {pages} {3953} (\bibinfo {year}
  {1999})}\BibitemShut {NoStop}%
\bibitem [{\citenamefont {{Farhi}}\ and\ \citenamefont
  {{Harrow}}(2016)}]{farhi_QAOA_2016}%
  \BibitemOpen
  \bibfield  {author} {\bibinfo {author} {\bibfnamefont {E.}~\bibnamefont
  {{Farhi}}}\ and\ \bibinfo {author} {\bibfnamefont {A.~W.}\ \bibnamefont
  {{Harrow}}},\ }\bibfield  {title} {\bibinfo {title} {{Quantum Supremacy
  through the Quantum Approximate Optimization Algorithm}},\ }\href@noop {}
  {\bibfield  {journal} {\bibinfo  {journal} {arXiv e-prints}\ ,\ \bibinfo
  {eid} {arXiv:1602.07674}} (\bibinfo {year} {2016})},\ \Eprint
  {https://arxiv.org/abs/1602.07674} {arXiv:1602.07674 [quant-ph]} \BibitemShut
  {NoStop}%
\bibitem [{\citenamefont {Brod}(2015)}]{brod_complexity_2015}%
  \BibitemOpen
  \bibfield  {author} {\bibinfo {author} {\bibfnamefont {D.~J.}\ \bibnamefont
  {Brod}},\ }\bibfield  {title} {\bibinfo {title} {Complexity of simulating
  constant-depth bosonsampling},\ }\href
  {https://doi.org/10.1103/PhysRevA.91.042316} {\bibfield  {journal} {\bibinfo
  {journal} {Phys. Rev. A}\ }\textbf {\bibinfo {volume} {91}},\ \bibinfo
  {pages} {042316} (\bibinfo {year} {2015})}\BibitemShut {NoStop}%
\bibitem [{\citenamefont {Bouland}\ \emph {et~al.}(2018)\citenamefont
  {Bouland}, \citenamefont {Fitzsimons},\ and\ \citenamefont
  {Koh}}]{bouland_conjugated_2018}%
  \BibitemOpen
  \bibfield  {author} {\bibinfo {author} {\bibfnamefont {A.}~\bibnamefont
  {Bouland}}, \bibinfo {author} {\bibfnamefont {J.~F.}\ \bibnamefont
  {Fitzsimons}},\ and\ \bibinfo {author} {\bibfnamefont {D.~E.}\ \bibnamefont
  {Koh}},\ }\bibfield  {title} {\bibinfo {title} {Complexity {Classification}
  of {Conjugated} {Clifford} {Circuits}},\ }in\ \href
  {https://doi.org/10.4230/LIPIcs.CCC.2018.21} {\emph {\bibinfo {booktitle}
  {33rd {Computational} {Complexity} {Conference} ({CCC} 2018)}}},\ \bibinfo
  {series} {Leibniz {International} {Proceedings} in {Informatics} ({LIPIcs})},
  Vol.\ \bibinfo {volume} {102},\ \bibinfo {editor} {edited by\ \bibinfo
  {editor} {\bibfnamefont {R.~A.}\ \bibnamefont {Servedio}}}\ (\bibinfo
  {publisher} {Schloss Dagstuhl–Leibniz-Zentrum fuer Informatik},\ \bibinfo
  {address} {Dagstuhl, Germany},\ \bibinfo {year} {2018})\ pp.\ \bibinfo
  {pages} {21:1--21:25},\ \bibinfo {note} {iSSN: 1868-8969}\BibitemShut
  {NoStop}%
\bibitem [{\citenamefont {Gao}\ \emph {et~al.}(2017)\citenamefont {Gao},
  \citenamefont {Wang},\ and\ \citenamefont {Duan}}]{gao_ising_2017}%
  \BibitemOpen
  \bibfield  {author} {\bibinfo {author} {\bibfnamefont {X.}~\bibnamefont
  {Gao}}, \bibinfo {author} {\bibfnamefont {S.-T.}\ \bibnamefont {Wang}},\ and\
  \bibinfo {author} {\bibfnamefont {L.-M.}\ \bibnamefont {Duan}},\ }\bibfield
  {title} {\bibinfo {title} {Quantum {Supremacy} for {Simulating} a
  {Translation}-{Invariant} {Ising} {Spin} {Model}},\ }\href
  {https://doi.org/10.1103/PhysRevLett.118.040502} {\bibfield  {journal}
  {\bibinfo  {journal} {Phys. Rev. Lett.}\ }\textbf {\bibinfo {volume} {118}},\
  \bibinfo {pages} {040502} (\bibinfo {year} {2017})},\ \bibinfo {note}
  {publisher: American Physical Society}\BibitemShut {NoStop}%
\bibitem [{\citenamefont {Pashayan}\ \emph {et~al.}(2020)\citenamefont
  {Pashayan}, \citenamefont {Bartlett},\ and\ \citenamefont
  {Gross}}]{pashayan_estimation_2020}%
  \BibitemOpen
  \bibfield  {author} {\bibinfo {author} {\bibfnamefont {H.}~\bibnamefont
  {Pashayan}}, \bibinfo {author} {\bibfnamefont {S.~D.}\ \bibnamefont
  {Bartlett}},\ and\ \bibinfo {author} {\bibfnamefont {D.}~\bibnamefont
  {Gross}},\ }\bibfield  {title} {\bibinfo {title} {From estimation of quantum
  probabilities to simulation of quantum circuits},\ }\href
  {https://doi.org/10.22331/q-2020-01-13-223} {\bibfield  {journal} {\bibinfo
  {journal} {{Quantum}}\ }\textbf {\bibinfo {volume} {4}},\ \bibinfo {pages}
  {223} (\bibinfo {year} {2020})}\BibitemShut {NoStop}%
\bibitem [{\citenamefont {Stockmeyer}(1985)}]{Stockmeyer1985}%
  \BibitemOpen
  \bibfield  {author} {\bibinfo {author} {\bibfnamefont {L.}~\bibnamefont
  {Stockmeyer}},\ }\bibfield  {title} {\bibinfo {title} {On approximation
  algorithms for \#{P}},\ }\href
  {https://doi.org/https://doi.org/10.1137/0214060} {\bibfield  {journal}
  {\bibinfo  {journal} {SIAM J. Comput.}\ }\textbf {\bibinfo {volume} {14}},\
  \bibinfo {pages} {849} (\bibinfo {year} {1985})}\BibitemShut {NoStop}%
\bibitem [{\citenamefont {Hangleiter}\ \emph
  {et~al.}(2018{\natexlab{b}})\citenamefont {Hangleiter}, \citenamefont
  {Bermejo-Vega}, \citenamefont {Schwarz},\ and\ \citenamefont
  {Eisert}}]{Hangleiter2018}%
  \BibitemOpen
  \bibfield  {author} {\bibinfo {author} {\bibfnamefont {D.}~\bibnamefont
  {Hangleiter}}, \bibinfo {author} {\bibfnamefont {J.}~\bibnamefont
  {Bermejo-Vega}}, \bibinfo {author} {\bibfnamefont {M.}~\bibnamefont
  {Schwarz}},\ and\ \bibinfo {author} {\bibfnamefont {J.}~\bibnamefont
  {Eisert}},\ }\bibfield  {title} {\bibinfo {title} {Anticoncentration theorems
  for schemes showing a quantum speedup},\ }\href
  {https://doi.org/10.22331/q-2018-05-22-65} {\bibfield  {journal} {\bibinfo
  {journal} {{Quantum}}\ }\textbf {\bibinfo {volume} {2}},\ \bibinfo {pages}
  {65} (\bibinfo {year} {2018}{\natexlab{b}})}\BibitemShut {NoStop}%
\bibitem [{\citenamefont {Harrow}\ and\ \citenamefont
  {Low}(2009)}]{Harrow2009design}%
  \BibitemOpen
  \bibfield  {author} {\bibinfo {author} {\bibfnamefont {A.~W.}\ \bibnamefont
  {Harrow}}\ and\ \bibinfo {author} {\bibfnamefont {R.~A.}\ \bibnamefont
  {Low}},\ }\bibfield  {title} {\bibinfo {title} {Random quantum circuits are
  approximate 2-designs},\ }\href {https://doi.org/10.1007/s00220-009-0873-6}
  {\bibfield  {journal} {\bibinfo  {journal} {Communications in Mathematical
  Physics}\ }\textbf {\bibinfo {volume} {291}},\ \bibinfo {pages} {257}
  (\bibinfo {year} {2009})}\BibitemShut {NoStop}%
\bibitem [{\citenamefont {{Wu}}\ \emph {et~al.}(2021)\citenamefont {{Wu}},
  \citenamefont {{Bao}}, \citenamefont {{Cao}}, \citenamefont {{Chen}},
  \citenamefont {{Chen}}, \citenamefont {{Chen}}, \citenamefont {{Chung}},
  \citenamefont {{Deng}}, \citenamefont {{Du}}, \citenamefont {{Fan}},
  \citenamefont {{Gong}}, \citenamefont {{Guo}}, \citenamefont {{Guo}},
  \citenamefont {{Guo}}, \citenamefont {{Han}}, \citenamefont {{Hong}},
  \citenamefont {{Huang}}, \citenamefont {{Huo}}, \citenamefont {{Li}},
  \citenamefont {{Li}}, \citenamefont {{Li}}, \citenamefont {{Li}},
  \citenamefont {{Liang}}, \citenamefont {{Lin}}, \citenamefont {{Lin}},
  \citenamefont {{Qian}}, \citenamefont {{Qiao}}, \citenamefont {{Rong}},
  \citenamefont {{Su}}, \citenamefont {{Sun}}, \citenamefont {{Wang}},
  \citenamefont {{Wang}}, \citenamefont {{Wu}}, \citenamefont {{Xu}},
  \citenamefont {{Yan}}, \citenamefont {{Yang}}, \citenamefont {{Yang}},
  \citenamefont {{Ye}}, \citenamefont {{Yin}}, \citenamefont {{Ying}},
  \citenamefont {{Yu}}, \citenamefont {{Zha}}, \citenamefont {{Zhang}},
  \citenamefont {{Zhang}}, \citenamefont {{Zhang}}, \citenamefont {{Zhang}},
  \citenamefont {{Zhao}}, \citenamefont {{Zhao}}, \citenamefont {{Zhou}},
  \citenamefont {{Zhu}}, \citenamefont {{Lu}}, \citenamefont {{Peng}},
  \citenamefont {{Zhu}},\ and\ \citenamefont {{Pan}}}]{Wu2021}%
  \BibitemOpen
  \bibfield  {author} {\bibinfo {author} {\bibfnamefont {Y.}~\bibnamefont
  {{Wu}}}, \bibinfo {author} {\bibfnamefont {W.-S.}\ \bibnamefont {{Bao}}},
  \bibinfo {author} {\bibfnamefont {S.}~\bibnamefont {{Cao}}}, \bibinfo
  {author} {\bibfnamefont {F.}~\bibnamefont {{Chen}}}, \bibinfo {author}
  {\bibfnamefont {M.-C.}\ \bibnamefont {{Chen}}}, \bibinfo {author}
  {\bibfnamefont {X.}~\bibnamefont {{Chen}}}, \bibinfo {author} {\bibfnamefont
  {T.-H.}\ \bibnamefont {{Chung}}}, \bibinfo {author} {\bibfnamefont
  {H.}~\bibnamefont {{Deng}}}, \bibinfo {author} {\bibfnamefont
  {Y.}~\bibnamefont {{Du}}}, \bibinfo {author} {\bibfnamefont {D.}~\bibnamefont
  {{Fan}}}, \bibinfo {author} {\bibfnamefont {M.}~\bibnamefont {{Gong}}},
  \bibinfo {author} {\bibfnamefont {C.}~\bibnamefont {{Guo}}}, \bibinfo
  {author} {\bibfnamefont {C.}~\bibnamefont {{Guo}}}, \bibinfo {author}
  {\bibfnamefont {S.}~\bibnamefont {{Guo}}}, \bibinfo {author} {\bibfnamefont
  {L.}~\bibnamefont {{Han}}}, \bibinfo {author} {\bibfnamefont
  {L.}~\bibnamefont {{Hong}}}, \bibinfo {author} {\bibfnamefont {H.-L.}\
  \bibnamefont {{Huang}}}, \bibinfo {author} {\bibfnamefont {Y.-H.}\
  \bibnamefont {{Huo}}}, \bibinfo {author} {\bibfnamefont {L.}~\bibnamefont
  {{Li}}}, \bibinfo {author} {\bibfnamefont {N.}~\bibnamefont {{Li}}}, \bibinfo
  {author} {\bibfnamefont {S.}~\bibnamefont {{Li}}}, \bibinfo {author}
  {\bibfnamefont {Y.}~\bibnamefont {{Li}}}, \bibinfo {author} {\bibfnamefont
  {F.}~\bibnamefont {{Liang}}}, \bibinfo {author} {\bibfnamefont
  {C.}~\bibnamefont {{Lin}}}, \bibinfo {author} {\bibfnamefont
  {J.}~\bibnamefont {{Lin}}}, \bibinfo {author} {\bibfnamefont
  {H.}~\bibnamefont {{Qian}}}, \bibinfo {author} {\bibfnamefont
  {D.}~\bibnamefont {{Qiao}}}, \bibinfo {author} {\bibfnamefont
  {H.}~\bibnamefont {{Rong}}}, \bibinfo {author} {\bibfnamefont
  {H.}~\bibnamefont {{Su}}}, \bibinfo {author} {\bibfnamefont {L.}~\bibnamefont
  {{Sun}}}, \bibinfo {author} {\bibfnamefont {L.}~\bibnamefont {{Wang}}},
  \bibinfo {author} {\bibfnamefont {S.}~\bibnamefont {{Wang}}}, \bibinfo
  {author} {\bibfnamefont {D.}~\bibnamefont {{Wu}}}, \bibinfo {author}
  {\bibfnamefont {Y.}~\bibnamefont {{Xu}}}, \bibinfo {author} {\bibfnamefont
  {K.}~\bibnamefont {{Yan}}}, \bibinfo {author} {\bibfnamefont
  {W.}~\bibnamefont {{Yang}}}, \bibinfo {author} {\bibfnamefont
  {Y.}~\bibnamefont {{Yang}}}, \bibinfo {author} {\bibfnamefont
  {Y.}~\bibnamefont {{Ye}}}, \bibinfo {author} {\bibfnamefont {J.}~\bibnamefont
  {{Yin}}}, \bibinfo {author} {\bibfnamefont {C.}~\bibnamefont {{Ying}}},
  \bibinfo {author} {\bibfnamefont {J.}~\bibnamefont {{Yu}}}, \bibinfo {author}
  {\bibfnamefont {C.}~\bibnamefont {{Zha}}}, \bibinfo {author} {\bibfnamefont
  {C.}~\bibnamefont {{Zhang}}}, \bibinfo {author} {\bibfnamefont
  {H.}~\bibnamefont {{Zhang}}}, \bibinfo {author} {\bibfnamefont
  {K.}~\bibnamefont {{Zhang}}}, \bibinfo {author} {\bibfnamefont
  {Y.}~\bibnamefont {{Zhang}}}, \bibinfo {author} {\bibfnamefont
  {H.}~\bibnamefont {{Zhao}}}, \bibinfo {author} {\bibfnamefont
  {Y.}~\bibnamefont {{Zhao}}}, \bibinfo {author} {\bibfnamefont
  {L.}~\bibnamefont {{Zhou}}}, \bibinfo {author} {\bibfnamefont
  {Q.}~\bibnamefont {{Zhu}}}, \bibinfo {author} {\bibfnamefont {C.-Y.}\
  \bibnamefont {{Lu}}}, \bibinfo {author} {\bibfnamefont {C.-Z.}\ \bibnamefont
  {{Peng}}}, \bibinfo {author} {\bibfnamefont {X.}~\bibnamefont {{Zhu}}},\ and\
  \bibinfo {author} {\bibfnamefont {J.-W.}\ \bibnamefont {{Pan}}},\ }\bibfield
  {title} {\bibinfo {title} {{Strong quantum computational advantage using a
  superconducting quantum processor}},\ }\href@noop {} {\bibfield  {journal}
  {\bibinfo  {journal} {arXiv e-prints}\ ,\ \bibinfo {eid} {arXiv:2106.14734}}
  (\bibinfo {year} {2021})},\ \Eprint {https://arxiv.org/abs/2106.14734}
  {arXiv:2106.14734 [quant-ph]} \BibitemShut {NoStop}%
\bibitem [{\citenamefont {{Bouland}}\ \emph {et~al.}(2021)\citenamefont
  {{Bouland}}, \citenamefont {{Fefferman}}, \citenamefont {{Landau}},\ and\
  \citenamefont {{Liu}}}]{Bouland2021}%
  \BibitemOpen
  \bibfield  {author} {\bibinfo {author} {\bibfnamefont {A.}~\bibnamefont
  {{Bouland}}}, \bibinfo {author} {\bibfnamefont {B.}~\bibnamefont
  {{Fefferman}}}, \bibinfo {author} {\bibfnamefont {Z.}~\bibnamefont
  {{Landau}}},\ and\ \bibinfo {author} {\bibfnamefont {Y.}~\bibnamefont
  {{Liu}}},\ }\bibfield  {title} {\bibinfo {title} {{Noise and the frontier of
  quantum supremacy}},\ }\href@noop {} {\bibfield  {journal} {\bibinfo
  {journal} {arXiv e-prints}\ ,\ \bibinfo {eid} {arXiv:2102.01738}} (\bibinfo
  {year} {2021})},\ \Eprint {https://arxiv.org/abs/2102.01738}
  {arXiv:2102.01738 [quant-ph]} \BibitemShut {NoStop}%
\bibitem [{\citenamefont {{Zhong}}\ \emph {et~al.}(2021)\citenamefont
  {{Zhong}}, \citenamefont {{Deng}}, \citenamefont {{Qin}}, \citenamefont
  {{Wang}}, \citenamefont {{Chen}}, \citenamefont {{Peng}}, \citenamefont
  {{Luo}}, \citenamefont {{Wu}}, \citenamefont {{Gong}}, \citenamefont {{Su}},
  \citenamefont {{Hu}}, \citenamefont {{Hu}}, \citenamefont {{Yang}},
  \citenamefont {{Zhang}}, \citenamefont {{Li}}, \citenamefont {{Li}},
  \citenamefont {{Jiang}}, \citenamefont {{Gan}}, \citenamefont {{Yang}},
  \citenamefont {{You}}, \citenamefont {{Wang}}, \citenamefont {{Li}},
  \citenamefont {{Liu}}, \citenamefont {{Renema}}, \citenamefont {{Lu}},\ and\
  \citenamefont {{Pan}}}]{GaussBSExperiment2021}%
  \BibitemOpen
  \bibfield  {author} {\bibinfo {author} {\bibfnamefont {H.-S.}\ \bibnamefont
  {{Zhong}}}, \bibinfo {author} {\bibfnamefont {Y.-H.}\ \bibnamefont {{Deng}}},
  \bibinfo {author} {\bibfnamefont {J.}~\bibnamefont {{Qin}}}, \bibinfo
  {author} {\bibfnamefont {H.}~\bibnamefont {{Wang}}}, \bibinfo {author}
  {\bibfnamefont {M.-C.}\ \bibnamefont {{Chen}}}, \bibinfo {author}
  {\bibfnamefont {L.-C.}\ \bibnamefont {{Peng}}}, \bibinfo {author}
  {\bibfnamefont {Y.-H.}\ \bibnamefont {{Luo}}}, \bibinfo {author}
  {\bibfnamefont {D.}~\bibnamefont {{Wu}}}, \bibinfo {author} {\bibfnamefont
  {S.-Q.}\ \bibnamefont {{Gong}}}, \bibinfo {author} {\bibfnamefont
  {H.}~\bibnamefont {{Su}}}, \bibinfo {author} {\bibfnamefont {Y.}~\bibnamefont
  {{Hu}}}, \bibinfo {author} {\bibfnamefont {P.}~\bibnamefont {{Hu}}}, \bibinfo
  {author} {\bibfnamefont {X.-Y.}\ \bibnamefont {{Yang}}}, \bibinfo {author}
  {\bibfnamefont {W.-J.}\ \bibnamefont {{Zhang}}}, \bibinfo {author}
  {\bibfnamefont {H.}~\bibnamefont {{Li}}}, \bibinfo {author} {\bibfnamefont
  {Y.}~\bibnamefont {{Li}}}, \bibinfo {author} {\bibfnamefont {X.}~\bibnamefont
  {{Jiang}}}, \bibinfo {author} {\bibfnamefont {L.}~\bibnamefont {{Gan}}},
  \bibinfo {author} {\bibfnamefont {G.}~\bibnamefont {{Yang}}}, \bibinfo
  {author} {\bibfnamefont {L.}~\bibnamefont {{You}}}, \bibinfo {author}
  {\bibfnamefont {Z.}~\bibnamefont {{Wang}}}, \bibinfo {author} {\bibfnamefont
  {L.}~\bibnamefont {{Li}}}, \bibinfo {author} {\bibfnamefont {N.-L.}\
  \bibnamefont {{Liu}}}, \bibinfo {author} {\bibfnamefont {J.}~\bibnamefont
  {{Renema}}}, \bibinfo {author} {\bibfnamefont {C.-Y.}\ \bibnamefont {{Lu}}},\
  and\ \bibinfo {author} {\bibfnamefont {J.-W.}\ \bibnamefont {{Pan}}},\
  }\bibfield  {title} {\bibinfo {title} {{Phase-Programmable Gaussian Boson
  Sampling Using Stimulated Squeezed Light}},\ }\href@noop {} {\bibfield
  {journal} {\bibinfo  {journal} {arXiv e-prints}\ ,\ \bibinfo {eid}
  {arXiv:2106.15534}} (\bibinfo {year} {2021})},\ \Eprint
  {https://arxiv.org/abs/2106.15534} {arXiv:2106.15534 [quant-ph]} \BibitemShut
  {NoStop}%
\bibitem [{\citenamefont {Deshpande}\ \emph {et~al.}(2021)\citenamefont
  {Deshpande}, \citenamefont {Mehta}, \citenamefont {Vincent}, \citenamefont
  {Quesada}, \citenamefont {Hinsche}, \citenamefont {Ioannou}, \citenamefont
  {Madsen}, \citenamefont {Lavoie}, \citenamefont {Qi}, \citenamefont {Eisert}
  \emph {et~al.}}]{deshpande2021GBS}%
  \BibitemOpen
  \bibfield  {author} {\bibinfo {author} {\bibfnamefont {A.}~\bibnamefont
  {Deshpande}}, \bibinfo {author} {\bibfnamefont {A.}~\bibnamefont {Mehta}},
  \bibinfo {author} {\bibfnamefont {T.}~\bibnamefont {Vincent}}, \bibinfo
  {author} {\bibfnamefont {N.}~\bibnamefont {Quesada}}, \bibinfo {author}
  {\bibfnamefont {M.}~\bibnamefont {Hinsche}}, \bibinfo {author} {\bibfnamefont
  {M.}~\bibnamefont {Ioannou}}, \bibinfo {author} {\bibfnamefont
  {L.}~\bibnamefont {Madsen}}, \bibinfo {author} {\bibfnamefont
  {J.}~\bibnamefont {Lavoie}}, \bibinfo {author} {\bibfnamefont
  {H.}~\bibnamefont {Qi}}, \bibinfo {author} {\bibfnamefont {J.}~\bibnamefont
  {Eisert}}, \emph {et~al.},\ }\bibfield  {title} {\bibinfo {title} {Quantum
  computational supremacy via high-dimensional gaussian boson sampling},\
  }\href@noop {} {\bibfield  {journal} {\bibinfo  {journal} {arXiv preprint}\
  ,\ \bibinfo {pages} {arXiv:2102.12474}} (\bibinfo {year} {2021})},\ \Eprint
  {https://arxiv.org/abs/2102.12474} {arXiv:2102.12474 [quant-ph]} \BibitemShut
  {NoStop}%
\bibitem [{\citenamefont {Shepherd}\ and\ \citenamefont
  {Bremner}(2009)}]{Shepherd2009}%
  \BibitemOpen
  \bibfield  {author} {\bibinfo {author} {\bibfnamefont {D.}~\bibnamefont
  {Shepherd}}\ and\ \bibinfo {author} {\bibfnamefont {M.~J.}\ \bibnamefont
  {Bremner}},\ }\bibfield  {title} {\bibinfo {title} {Temporally unstructured
  quantum computation},\ }\href {https://doi.org/10.1098/rspa.2008.0443}
  {\bibfield  {journal} {\bibinfo  {journal} {Proceedings of the Royal Society
  A: Mathematical, Physical and Engineering Sciences}\ }\textbf {\bibinfo
  {volume} {465}},\ \bibinfo {pages} {1413} (\bibinfo {year}
  {2009})}\BibitemShut {NoStop}%
\bibitem [{\citenamefont {Bravyi}\ and\ \citenamefont
  {Kitaev}(2002)}]{bravyi_fermionic_2002}%
  \BibitemOpen
  \bibfield  {author} {\bibinfo {author} {\bibfnamefont {S.~B.}\ \bibnamefont
  {Bravyi}}\ and\ \bibinfo {author} {\bibfnamefont {A.~Y.}\ \bibnamefont
  {Kitaev}},\ }\bibfield  {title} {\bibinfo {title} {Fermionic {Quantum}
  {Computation}},\ }\href {https://doi.org/10.1006/aphy.2002.6254} {\bibfield
  {journal} {\bibinfo  {journal} {Annals of Physics}\ }\textbf {\bibinfo
  {volume} {298}},\ \bibinfo {pages} {210} (\bibinfo {year}
  {2002})}\BibitemShut {NoStop}%
\bibitem [{\citenamefont {Ivanov}(2017)}]{Ivanov2017}%
  \BibitemOpen
  \bibfield  {author} {\bibinfo {author} {\bibfnamefont {D.~A.}\ \bibnamefont
  {Ivanov}},\ }\bibfield  {title} {\bibinfo {title} {Computational complexity
  of exterior products and multiparticle amplitudes of noninteracting fermions
  in entangled states},\ }\href {https://doi.org/10.1103/PhysRevA.96.012322}
  {\bibfield  {journal} {\bibinfo  {journal} {Phys. Rev. A}\ }\textbf {\bibinfo
  {volume} {96}},\ \bibinfo {pages} {012322} (\bibinfo {year}
  {2017})}\BibitemShut {NoStop}%
\bibitem [{\citenamefont {Ivanov}\ and\ \citenamefont
  {Gurvits}(2020)}]{Ivanov2020}%
  \BibitemOpen
  \bibfield  {author} {\bibinfo {author} {\bibfnamefont {D.~A.}\ \bibnamefont
  {Ivanov}}\ and\ \bibinfo {author} {\bibfnamefont {L.}~\bibnamefont
  {Gurvits}},\ }\bibfield  {title} {\bibinfo {title} {Complexity of full
  counting statistics of free quantum particles in product states},\ }\href
  {https://doi.org/10.1103/PhysRevA.101.012303} {\bibfield  {journal} {\bibinfo
   {journal} {Phys. Rev. A}\ }\textbf {\bibinfo {volume} {101}},\ \bibinfo
  {pages} {012303} (\bibinfo {year} {2020})}\BibitemShut {NoStop}%
\bibitem [{\citenamefont {Bravyi}(2006)}]{bravyi_universal_2006}%
  \BibitemOpen
  \bibfield  {author} {\bibinfo {author} {\bibfnamefont {S.}~\bibnamefont
  {Bravyi}},\ }\bibfield  {title} {\bibinfo {title} {Universal quantum
  computation with the $\nu=5/2$ fractional quantum {Hall} state},\ }\href
  {https://doi.org/10.1103/PhysRevA.73.042313} {\bibfield  {journal} {\bibinfo
  {journal} {Phys. Rev. A}\ }\textbf {\bibinfo {volume} {73}},\ \bibinfo
  {pages} {042313} (\bibinfo {year} {2006})}\BibitemShut {NoStop}%
\bibitem [{\citenamefont {Hebenstreit}\ \emph {et~al.}(2019)\citenamefont
  {Hebenstreit}, \citenamefont {Jozsa}, \citenamefont {Kraus}, \citenamefont
  {Strelchuk},\ and\ \citenamefont {Yoganathan}}]{hebenstreit2019}%
  \BibitemOpen
  \bibfield  {author} {\bibinfo {author} {\bibfnamefont {M.}~\bibnamefont
  {Hebenstreit}}, \bibinfo {author} {\bibfnamefont {R.}~\bibnamefont {Jozsa}},
  \bibinfo {author} {\bibfnamefont {B.}~\bibnamefont {Kraus}}, \bibinfo
  {author} {\bibfnamefont {S.}~\bibnamefont {Strelchuk}},\ and\ \bibinfo
  {author} {\bibfnamefont {M.}~\bibnamefont {Yoganathan}},\ }\bibfield  {title}
  {\bibinfo {title} {All pure fermionic non-gaussian states are magic states
  for matchgate computations},\ }\href
  {https://doi.org/10.1103/PhysRevLett.123.080503} {\bibfield  {journal}
  {\bibinfo  {journal} {Phys. Rev. Lett.}\ }\textbf {\bibinfo {volume} {123}},\
  \bibinfo {pages} {080503} (\bibinfo {year} {2019})}\BibitemShut {NoStop}%
\bibitem [{\citenamefont {Rahimi-Keshari}\ \emph {et~al.}(2016)\citenamefont
  {Rahimi-Keshari}, \citenamefont {Ralph},\ and\ \citenamefont
  {Caves}}]{Rahimi-Keshari2016}%
  \BibitemOpen
  \bibfield  {author} {\bibinfo {author} {\bibfnamefont {S.}~\bibnamefont
  {Rahimi-Keshari}}, \bibinfo {author} {\bibfnamefont {T.~C.}\ \bibnamefont
  {Ralph}},\ and\ \bibinfo {author} {\bibfnamefont {C.~M.}\ \bibnamefont
  {Caves}},\ }\bibfield  {title} {\bibinfo {title} {Sufficient conditions for
  efficient classical simulation of quantum optics},\ }\href
  {https://doi.org/10.1103/PhysRevX.6.021039} {\bibfield  {journal} {\bibinfo
  {journal} {Phys. Rev. X}\ }\textbf {\bibinfo {volume} {6}},\ \bibinfo {pages}
  {021039} (\bibinfo {year} {2016})}\BibitemShut {NoStop}%
\bibitem [{\citenamefont {Hangleiter}\ \emph {et~al.}(2019)\citenamefont
  {Hangleiter}, \citenamefont {Kliesch}, \citenamefont {Eisert},\ and\
  \citenamefont {Gogolin}}]{hangleiter2019sample}%
  \BibitemOpen
  \bibfield  {author} {\bibinfo {author} {\bibfnamefont {D.}~\bibnamefont
  {Hangleiter}}, \bibinfo {author} {\bibfnamefont {M.}~\bibnamefont {Kliesch}},
  \bibinfo {author} {\bibfnamefont {J.}~\bibnamefont {Eisert}},\ and\ \bibinfo
  {author} {\bibfnamefont {C.}~\bibnamefont {Gogolin}},\ }\bibfield  {title}
  {\bibinfo {title} {Sample complexity of device-independently certified
  “quantum supremacy”},\ }\href@noop {} {\bibfield  {journal} {\bibinfo
  {journal} {Physical review letters}\ }\textbf {\bibinfo {volume} {122}},\
  \bibinfo {pages} {210502} (\bibinfo {year} {2019})}\BibitemShut {NoStop}%
\bibitem [{\citenamefont {Aaronson}\ and\ \citenamefont
  {Arkhipov}(2014)}]{aaronson2014uniform}%
  \BibitemOpen
  \bibfield  {author} {\bibinfo {author} {\bibfnamefont {S.}~\bibnamefont
  {Aaronson}}\ and\ \bibinfo {author} {\bibfnamefont {A.}~\bibnamefont
  {Arkhipov}},\ }\bibfield  {title} {\bibinfo {title} {Bosonsampling is far
  from uniform},\ }\href@noop {} {\bibfield  {journal} {\bibinfo  {journal}
  {Quantum Information \& Computation}\ }\textbf {\bibinfo {volume} {14}},\
  \bibinfo {pages} {1383} (\bibinfo {year} {2014})}\BibitemShut {NoStop}%
\bibitem [{\citenamefont {Chabaud}\ \emph {et~al.}(2020)\citenamefont
  {Chabaud}, \citenamefont {Grosshans}, \citenamefont {Kashefi},\ and\
  \citenamefont {Markham}}]{chabaud2020efficient}%
  \BibitemOpen
  \bibfield  {author} {\bibinfo {author} {\bibfnamefont {U.}~\bibnamefont
  {Chabaud}}, \bibinfo {author} {\bibfnamefont {F.}~\bibnamefont {Grosshans}},
  \bibinfo {author} {\bibfnamefont {E.}~\bibnamefont {Kashefi}},\ and\ \bibinfo
  {author} {\bibfnamefont {D.}~\bibnamefont {Markham}},\ }\bibfield  {title}
  {\bibinfo {title} {Efficient verification of boson sampling},\ }\href@noop {}
  {\bibfield  {journal} {\bibinfo  {journal} {arXiv preprint arXiv:2006.03520}\
  } (\bibinfo {year} {2020})}\BibitemShut {NoStop}%
\bibitem [{\citenamefont {{Oszmaniec}}\ and\ \citenamefont
  {{Zimbor{\'a}s}}(2017)}]{OZ2017}%
  \BibitemOpen
  \bibfield  {author} {\bibinfo {author} {\bibfnamefont {M.}~\bibnamefont
  {{Oszmaniec}}}\ and\ \bibinfo {author} {\bibfnamefont {Z.}~\bibnamefont
  {{Zimbor{\'a}s}}},\ }\bibfield  {title} {\bibinfo {title} {{Universal
  Extensions of Restricted Classes of Quantum Operations}},\ }\href
  {https://doi.org/10.1103/PhysRevLett.119.220502} {\bibfield  {journal}
  {\bibinfo  {journal} {\prl}\ }\textbf {\bibinfo {volume} {119}},\ \bibinfo
  {eid} {220502} (\bibinfo {year} {2017})},\ \Eprint
  {https://arxiv.org/abs/1705.11188} {arXiv:1705.11188 [quant-ph]} \BibitemShut
  {NoStop}%
\bibitem [{\citenamefont {Bravyi}(2005)}]{bravyi_lagrangian_2004}%
  \BibitemOpen
  \bibfield  {author} {\bibinfo {author} {\bibfnamefont {S.}~\bibnamefont
  {Bravyi}},\ }\bibfield  {title} {\bibinfo {title} {Lagrangian representation
  for fermionic linear optics},\ }\href@noop {} {\bibfield  {journal} {\bibinfo
   {journal} {Quantum Info. Comput.}\ }\textbf {\bibinfo {volume} {5}},\
  \bibinfo {pages} {216–238} (\bibinfo {year} {2005})}\BibitemShut {NoStop}%
\bibitem [{\citenamefont {{Bravyi}}\ and\ \citenamefont
  {{Koenig}}(2011)}]{DissipativeFLO}%
  \BibitemOpen
  \bibfield  {author} {\bibinfo {author} {\bibfnamefont {S.}~\bibnamefont
  {{Bravyi}}}\ and\ \bibinfo {author} {\bibfnamefont {R.}~\bibnamefont
  {{Koenig}}},\ }\bibfield  {title} {\bibinfo {title} {{Classical simulation of
  dissipative fermionic linear optics}},\ }\href@noop {} {\bibfield  {journal}
  {\bibinfo  {journal} {arXiv e-prints}\ ,\ \bibinfo {eid} {arXiv:1112.2184}}
  (\bibinfo {year} {2011})},\ \Eprint {https://arxiv.org/abs/1112.2184}
  {arXiv:1112.2184 [quant-ph]} \BibitemShut {NoStop}%
\bibitem [{\citenamefont {Jozsa}\ and\ \citenamefont
  {Miyake}(2008)}]{Jozsa2008}%
  \BibitemOpen
  \bibfield  {author} {\bibinfo {author} {\bibfnamefont {R.}~\bibnamefont
  {Jozsa}}\ and\ \bibinfo {author} {\bibfnamefont {A.}~\bibnamefont {Miyake}},\
  }\bibfield  {title} {\bibinfo {title} {Matchgates and classical simulation of
  quantum circuits},\ }\href {https://doi.org/10.1098/rspa.2008.0189}
  {\bibfield  {journal} {\bibinfo  {journal} {Proceedings of the Royal Society
  A: Mathematical, Physical and Engineering Sciences}\ }\textbf {\bibinfo
  {volume} {464}},\ \bibinfo {pages} {3089} (\bibinfo {year}
  {2008})}\BibitemShut {NoStop}%
\bibitem [{\citenamefont {Brod}\ and\ \citenamefont
  {Galvão}(2012)}]{brod_geometries_2012}%
  \BibitemOpen
  \bibfield  {author} {\bibinfo {author} {\bibfnamefont {D.~J.}\ \bibnamefont
  {Brod}}\ and\ \bibinfo {author} {\bibfnamefont {E.~F.}\ \bibnamefont
  {Galvão}},\ }\bibfield  {title} {\bibinfo {title} {Geometries for universal
  quantum computation with matchgates},\ }\href
  {https://doi.org/10.1103/PhysRevA.86.052307} {\bibfield  {journal} {\bibinfo
  {journal} {Phys. Rev. A}\ }\textbf {\bibinfo {volume} {86}},\ \bibinfo
  {pages} {052307} (\bibinfo {year} {2012})}\BibitemShut {NoStop}%
\bibitem [{\citenamefont {Brod}(2016)}]{brod_efficient_2016}%
  \BibitemOpen
  \bibfield  {author} {\bibinfo {author} {\bibfnamefont {D.~J.}\ \bibnamefont
  {Brod}},\ }\bibfield  {title} {\bibinfo {title} {Efficient classical
  simulation of matchgate circuits with generalized inputs and measurements},\
  }\href {https://doi.org/10.1103/PhysRevA.93.062332} {\bibfield  {journal}
  {\bibinfo  {journal} {Phys. Rev. A}\ }\textbf {\bibinfo {volume} {93}},\
  \bibinfo {pages} {062332} (\bibinfo {year} {2016})}\BibitemShut {NoStop}%
\bibitem [{\citenamefont {Bravyi}\ and\ \citenamefont
  {Kitaev}(2005)}]{bravyi2005magic}%
  \BibitemOpen
  \bibfield  {author} {\bibinfo {author} {\bibfnamefont {S.}~\bibnamefont
  {Bravyi}}\ and\ \bibinfo {author} {\bibfnamefont {A.}~\bibnamefont
  {Kitaev}},\ }\bibfield  {title} {\bibinfo {title} {Universal quantum
  computation with ideal clifford gates and noisy ancillas},\ }\href@noop {}
  {\bibfield  {journal} {\bibinfo  {journal} {Physical Review A}\ }\textbf
  {\bibinfo {volume} {71}},\ \bibinfo {pages} {022316} (\bibinfo {year}
  {2005})}\BibitemShut {NoStop}%
\bibitem [{\citenamefont {{Bravyi}}(2005)}]{bravyi2005capacity}%
  \BibitemOpen
  \bibfield  {author} {\bibinfo {author} {\bibfnamefont {S.}~\bibnamefont
  {{Bravyi}}},\ }\bibfield  {title} {\bibinfo {title} {{Classical capacity of
  fermionic product channels}},\ }\href@noop {} {\bibfield  {journal} {\bibinfo
   {journal} {arXiv e-prints}\ ,\ \bibinfo {eid} {quant-ph/0507282}} (\bibinfo
  {year} {2005})},\ \Eprint {https://arxiv.org/abs/quant-ph/0507282}
  {arXiv:quant-ph/0507282 [quant-ph]} \BibitemShut {NoStop}%
\bibitem [{\citenamefont {de~Melo}\ \emph {et~al.}(2013)\citenamefont
  {de~Melo}, \citenamefont {{\'{C}}wikli{\'{n}}ski},\ and\ \citenamefont
  {Terhal}}]{melo_power_2013}%
  \BibitemOpen
  \bibfield  {author} {\bibinfo {author} {\bibfnamefont {F.}~\bibnamefont
  {de~Melo}}, \bibinfo {author} {\bibfnamefont {P.}~\bibnamefont
  {{\'{C}}wikli{\'{n}}ski}},\ and\ \bibinfo {author} {\bibfnamefont {B.~M.}\
  \bibnamefont {Terhal}},\ }\bibfield  {title} {\bibinfo {title} {The power of
  noisy fermionic quantum computation},\ }\href
  {https://doi.org/10.1088/1367-2630/15/1/013015} {\bibfield  {journal}
  {\bibinfo  {journal} {New Journal of Physics}\ }\textbf {\bibinfo {volume}
  {15}},\ \bibinfo {pages} {013015} (\bibinfo {year} {2013})}\BibitemShut
  {NoStop}%
\bibitem [{\citenamefont {Oszmaniec}\ \emph {et~al.}(2014)\citenamefont
  {Oszmaniec}, \citenamefont {Gutt},\ and\ \citenamefont
  {Kuś}}]{oszmaniec_classical_2014}%
  \BibitemOpen
  \bibfield  {author} {\bibinfo {author} {\bibfnamefont {M.}~\bibnamefont
  {Oszmaniec}}, \bibinfo {author} {\bibfnamefont {J.}~\bibnamefont {Gutt}},\
  and\ \bibinfo {author} {\bibfnamefont {M.}~\bibnamefont {Kuś}},\ }\bibfield
  {title} {\bibinfo {title} {Classical simulation of fermionic linear optics
  augmented with noisy ancillas},\ }\href
  {https://doi.org/10.1103/PhysRevA.90.020302} {\bibfield  {journal} {\bibinfo
  {journal} {Phys. Rev. A}\ }\textbf {\bibinfo {volume} {90}},\ \bibinfo
  {pages} {020302} (\bibinfo {year} {2014})}\BibitemShut {NoStop}%
\bibitem [{\citenamefont {Brod}\ and\ \citenamefont
  {Galvão}(2011)}]{brod_extending_2011}%
  \BibitemOpen
  \bibfield  {author} {\bibinfo {author} {\bibfnamefont {D.~J.}\ \bibnamefont
  {Brod}}\ and\ \bibinfo {author} {\bibfnamefont {E.~F.}\ \bibnamefont
  {Galvão}},\ }\bibfield  {title} {\bibinfo {title} {Extending matchgates into
  universal quantum computation},\ }\href
  {https://doi.org/10.1103/PhysRevA.84.022310} {\bibfield  {journal} {\bibinfo
  {journal} {Phys. Rev. A}\ }\textbf {\bibinfo {volume} {84}},\ \bibinfo
  {pages} {022310} (\bibinfo {year} {2011})}\BibitemShut {NoStop}%
\bibitem [{\citenamefont {Zimbor{\'a}s}\ \emph {et~al.}(2014)\citenamefont
  {Zimbor{\'a}s}, \citenamefont {Zeier}, \citenamefont {Keyl},\ and\
  \citenamefont {Schulte-Herbr{\"u}ggen}}]{zimboras2014}%
  \BibitemOpen
  \bibfield  {author} {\bibinfo {author} {\bibfnamefont {Z.}~\bibnamefont
  {Zimbor{\'a}s}}, \bibinfo {author} {\bibfnamefont {R.}~\bibnamefont {Zeier}},
  \bibinfo {author} {\bibfnamefont {M.}~\bibnamefont {Keyl}},\ and\ \bibinfo
  {author} {\bibfnamefont {T.}~\bibnamefont {Schulte-Herbr{\"u}ggen}},\
  }\bibfield  {title} {\bibinfo {title} {A dynamic systems approach to fermions
  and their relation to spins},\ }\href@noop {} {\bibfield  {journal} {\bibinfo
   {journal} {EPJ Quantum Technology}\ }\textbf {\bibinfo {volume} {1}},\
  \bibinfo {pages} {11} (\bibinfo {year} {2014})}\BibitemShut {NoStop}%
\bibitem [{\citenamefont {Beenakker}\ \emph {et~al.}(2004)\citenamefont
  {Beenakker}, \citenamefont {DiVincenzo}, \citenamefont {Emary},\ and\
  \citenamefont {Kindermann}}]{beenakker2004charge}%
  \BibitemOpen
  \bibfield  {author} {\bibinfo {author} {\bibfnamefont {C.}~\bibnamefont
  {Beenakker}}, \bibinfo {author} {\bibfnamefont {D.}~\bibnamefont
  {DiVincenzo}}, \bibinfo {author} {\bibfnamefont {C.}~\bibnamefont {Emary}},\
  and\ \bibinfo {author} {\bibfnamefont {M.}~\bibnamefont {Kindermann}},\
  }\bibfield  {title} {\bibinfo {title} {Charge detection enables free-electron
  quantum computation},\ }\href@noop {} {\bibfield  {journal} {\bibinfo
  {journal} {Physical review letters}\ }\textbf {\bibinfo {volume} {93}},\
  \bibinfo {pages} {020501} (\bibinfo {year} {2004})}\BibitemShut {NoStop}%
\bibitem [{\citenamefont {{Hebenstreit}}\ \emph {et~al.}(2020)\citenamefont
  {{Hebenstreit}}, \citenamefont {{Jozsa}}, \citenamefont {{Kraus}},\ and\
  \citenamefont {{Strelchuk}}}]{hebenstreit_computational_2020}%
  \BibitemOpen
  \bibfield  {author} {\bibinfo {author} {\bibfnamefont {M.}~\bibnamefont
  {{Hebenstreit}}}, \bibinfo {author} {\bibfnamefont {R.}~\bibnamefont
  {{Jozsa}}}, \bibinfo {author} {\bibfnamefont {B.}~\bibnamefont {{Kraus}}},\
  and\ \bibinfo {author} {\bibfnamefont {S.}~\bibnamefont {{Strelchuk}}},\
  }\bibfield  {title} {\bibinfo {title} {{Computational power of matchgates
  with supplementary resources}},\ }\href@noop {} {\bibfield  {journal}
  {\bibinfo  {journal} {arXiv e-prints}\ ,\ \bibinfo {eid} {arXiv:2007.08231}}
  (\bibinfo {year} {2020})},\ \Eprint {https://arxiv.org/abs/2007.08231}
  {arXiv:2007.08231 [quant-ph]} \BibitemShut {NoStop}%
\bibitem [{\citenamefont {Arkhipov}(2015)}]{arkhipov2015bosonsampling}%
  \BibitemOpen
  \bibfield  {author} {\bibinfo {author} {\bibfnamefont {A.}~\bibnamefont
  {Arkhipov}},\ }\bibfield  {title} {\bibinfo {title} {Bosonsampling is robust
  against small errors in the network matrix},\ }\href@noop {} {\bibfield
  {journal} {\bibinfo  {journal} {Physical Review A}\ }\textbf {\bibinfo
  {volume} {92}},\ \bibinfo {pages} {062326} (\bibinfo {year}
  {2015})}\BibitemShut {NoStop}%
\bibitem [{\citenamefont {Goodman}\ and\ \citenamefont
  {Wallach}(2009{\natexlab{a}})}]{WalachBook}%
  \BibitemOpen
  \bibfield  {author} {\bibinfo {author} {\bibfnamefont {R.~W.}\ \bibnamefont
  {Goodman}}\ and\ \bibinfo {author} {\bibfnamefont {N.~R.}\ \bibnamefont
  {Wallach}},\ }\href {https://cds.cern.ch/record/1412002} {\emph {\bibinfo
  {title} {{Symmetry, Representations, and Invariants}}}},\ Graduate Texts in
  Mathematics\ (\bibinfo  {publisher} {Springer},\ \bibinfo {address}
  {Dordrecht},\ \bibinfo {year} {2009})\BibitemShut {NoStop}%
\bibitem [{\citenamefont {Hall}(2000)}]{HallGroups}%
  \BibitemOpen
  \bibfield  {author} {\bibinfo {author} {\bibfnamefont {B.~C.}\ \bibnamefont
  {Hall}},\ }\href@noop {} {\emph {\bibinfo {title} {{L}ie Groups, Lie
  Algebras, and Representations: An Elementary Introduction}}}\ (\bibinfo
  {publisher} {Springer},\ \bibinfo {year} {2000})\BibitemShut {NoStop}%
\bibitem [{\citenamefont {{Kondo}}\ \emph {et~al.}(2021)\citenamefont
  {{Kondo}}, \citenamefont {{Mori}},\ and\ \citenamefont
  {{Movassagh}}}]{Kondo2021_robustness}%
  \BibitemOpen
  \bibfield  {author} {\bibinfo {author} {\bibfnamefont {Y.}~\bibnamefont
  {{Kondo}}}, \bibinfo {author} {\bibfnamefont {R.}~\bibnamefont {{Mori}}},\
  and\ \bibinfo {author} {\bibfnamefont {R.}~\bibnamefont {{Movassagh}}},\
  }\bibfield  {title} {\bibinfo {title} {{Improved robustness of quantum
  supremacy for random circuit sampling}},\ }\href@noop {} {\bibfield
  {journal} {\bibinfo  {journal} {arXiv e-prints}\ ,\ \bibinfo {eid}
  {arXiv:2102.01960}} (\bibinfo {year} {2021})},\ \Eprint
  {https://arxiv.org/abs/2102.01960} {arXiv:2102.01960 [quant-ph]} \BibitemShut
  {NoStop}%
\bibitem [{\citenamefont {Lobino}\ \emph {et~al.}(2008)\citenamefont {Lobino},
  \citenamefont {Korystov}, \citenamefont {Kupchak}, \citenamefont {Figueroa},
  \citenamefont {Sanders},\ and\ \citenamefont {Lvovsky}}]{Lobino2008}%
  \BibitemOpen
  \bibfield  {author} {\bibinfo {author} {\bibfnamefont {M.}~\bibnamefont
  {Lobino}}, \bibinfo {author} {\bibfnamefont {D.}~\bibnamefont {Korystov}},
  \bibinfo {author} {\bibfnamefont {C.}~\bibnamefont {Kupchak}}, \bibinfo
  {author} {\bibfnamefont {E.}~\bibnamefont {Figueroa}}, \bibinfo {author}
  {\bibfnamefont {B.~C.}\ \bibnamefont {Sanders}},\ and\ \bibinfo {author}
  {\bibfnamefont {A.~I.}\ \bibnamefont {Lvovsky}},\ }\bibfield  {title}
  {\bibinfo {title} {Complete characterization of quantum-optical processes},\
  }\href {https://doi.org/10.1126/science.1162086} {\bibfield  {journal}
  {\bibinfo  {journal} {Science}\ }\textbf {\bibinfo {volume} {322}},\ \bibinfo
  {pages} {563} (\bibinfo {year} {2008})}\BibitemShut {NoStop}%
\bibitem [{\citenamefont {Rahimi-Keshari}\ \emph {et~al.}(2011)\citenamefont
  {Rahimi-Keshari}, \citenamefont {Scherer}, \citenamefont {Mann},
  \citenamefont {Rezakhani}, \citenamefont {Lvovsky},\ and\ \citenamefont
  {Sanders}}]{OptQPT2011}%
  \BibitemOpen
  \bibfield  {author} {\bibinfo {author} {\bibfnamefont {S.}~\bibnamefont
  {Rahimi-Keshari}}, \bibinfo {author} {\bibfnamefont {A.}~\bibnamefont
  {Scherer}}, \bibinfo {author} {\bibfnamefont {A.}~\bibnamefont {Mann}},
  \bibinfo {author} {\bibfnamefont {A.~T.}\ \bibnamefont {Rezakhani}}, \bibinfo
  {author} {\bibfnamefont {A.~I.}\ \bibnamefont {Lvovsky}},\ and\ \bibinfo
  {author} {\bibfnamefont {B.~C.}\ \bibnamefont {Sanders}},\ }\bibfield
  {title} {\bibinfo {title} {Quantum process tomography with coherent states},\
  }\href {https://doi.org/10.1088/1367-2630/13/1/013006} {\bibfield  {journal}
  {\bibinfo  {journal} {New Journal of Physics}\ }\textbf {\bibinfo {volume}
  {13}},\ \bibinfo {pages} {013006} (\bibinfo {year} {2011})}\BibitemShut
  {NoStop}%
\bibitem [{\citenamefont {Rahimi-Keshari}\ \emph {et~al.}(2013)\citenamefont
  {Rahimi-Keshari}, \citenamefont {Broome}, \citenamefont {Fickler},
  \citenamefont {Fedrizzi}, \citenamefont {Ralph},\ and\ \citenamefont
  {White}}]{LinOptTom2013}%
  \BibitemOpen
  \bibfield  {author} {\bibinfo {author} {\bibfnamefont {S.}~\bibnamefont
  {Rahimi-Keshari}}, \bibinfo {author} {\bibfnamefont {M.~A.}\ \bibnamefont
  {Broome}}, \bibinfo {author} {\bibfnamefont {R.}~\bibnamefont {Fickler}},
  \bibinfo {author} {\bibfnamefont {A.}~\bibnamefont {Fedrizzi}}, \bibinfo
  {author} {\bibfnamefont {T.~C.}\ \bibnamefont {Ralph}},\ and\ \bibinfo
  {author} {\bibfnamefont {A.~G.}\ \bibnamefont {White}},\ }\bibfield  {title}
  {\bibinfo {title} {Direct characterization of linear-optical networks},\
  }\href {https://doi.org/10.1364/OE.21.013450} {\bibfield  {journal} {\bibinfo
   {journal} {Opt. Express}\ }\textbf {\bibinfo {volume} {21}},\ \bibinfo
  {pages} {13450} (\bibinfo {year} {2013})}\BibitemShut {NoStop}%
\bibitem [{\citenamefont {Schuch}\ and\ \citenamefont
  {Siewert}(2003)}]{schuch2003natural}%
  \BibitemOpen
  \bibfield  {author} {\bibinfo {author} {\bibfnamefont {N.}~\bibnamefont
  {Schuch}}\ and\ \bibinfo {author} {\bibfnamefont {J.}~\bibnamefont
  {Siewert}},\ }\bibfield  {title} {\bibinfo {title} {Natural two-qubit gate
  for quantum computation using the xy interaction},\ }\href@noop {} {\bibfield
   {journal} {\bibinfo  {journal} {Physical Review A}\ }\textbf {\bibinfo
  {volume} {67}},\ \bibinfo {pages} {032301} (\bibinfo {year}
  {2003})}\BibitemShut {NoStop}%
\bibitem [{\citenamefont {Aleksandrowicz}\ \emph {et~al.}(2019)\citenamefont
  {Aleksandrowicz}, \citenamefont {Alexander}, \citenamefont {Barkoutsos},
  \citenamefont {Bello}, \citenamefont {Ben-Haim},\ and\ \citenamefont
  {Bucher}}]{Qiskit}%
  \BibitemOpen
  \bibfield  {author} {\bibinfo {author} {\bibfnamefont {G.}~\bibnamefont
  {Aleksandrowicz}}, \bibinfo {author} {\bibfnamefont {T.}~\bibnamefont
  {Alexander}}, \bibinfo {author} {\bibfnamefont {P.}~\bibnamefont
  {Barkoutsos}}, \bibinfo {author} {\bibfnamefont {L.}~\bibnamefont {Bello}},
  \bibinfo {author} {\bibfnamefont {Y.}~\bibnamefont {Ben-Haim}},\ and\
  \bibinfo {author} {\bibnamefont {Bucher}},\ }\href
  {https://doi.org/10.5281/ZENODO.2562111} {\bibinfo {title} {Qiskit: An
  open-source framework for quantum computing}} (\bibinfo {year}
  {2019})\BibitemShut {NoStop}%
\bibitem [{\citenamefont {Arrazola}\ and\ \citenamefont
  {Bromley}(2018)}]{arrazola2018using}%
  \BibitemOpen
  \bibfield  {author} {\bibinfo {author} {\bibfnamefont {J.~M.}\ \bibnamefont
  {Arrazola}}\ and\ \bibinfo {author} {\bibfnamefont {T.~R.}\ \bibnamefont
  {Bromley}},\ }\bibfield  {title} {\bibinfo {title} {Using gaussian boson
  sampling to find dense subgraphs},\ }\href@noop {} {\bibfield  {journal}
  {\bibinfo  {journal} {Physical review letters}\ }\textbf {\bibinfo {volume}
  {121}},\ \bibinfo {pages} {030503} (\bibinfo {year} {2018})}\BibitemShut
  {NoStop}%
\bibitem [{\citenamefont {Arrazola}\ \emph {et~al.}(2018)\citenamefont
  {Arrazola}, \citenamefont {Bromley},\ and\ \citenamefont
  {Rebentrost}}]{arrazola2018quantum}%
  \BibitemOpen
  \bibfield  {author} {\bibinfo {author} {\bibfnamefont {J.~M.}\ \bibnamefont
  {Arrazola}}, \bibinfo {author} {\bibfnamefont {T.~R.}\ \bibnamefont
  {Bromley}},\ and\ \bibinfo {author} {\bibfnamefont {P.}~\bibnamefont
  {Rebentrost}},\ }\bibfield  {title} {\bibinfo {title} {Quantum approximate
  optimization with gaussian boson sampling},\ }\href@noop {} {\bibfield
  {journal} {\bibinfo  {journal} {Physical Review A}\ }\textbf {\bibinfo
  {volume} {98}},\ \bibinfo {pages} {012322} (\bibinfo {year}
  {2018})}\BibitemShut {NoStop}%
\bibitem [{\citenamefont {Huh}\ \emph {et~al.}(2015)\citenamefont {Huh},
  \citenamefont {Guerreschi}, \citenamefont {Peropadre}, \citenamefont
  {McClean},\ and\ \citenamefont {Aspuru-Guzik}}]{huh2015boson}%
  \BibitemOpen
  \bibfield  {author} {\bibinfo {author} {\bibfnamefont {J.}~\bibnamefont
  {Huh}}, \bibinfo {author} {\bibfnamefont {G.~G.}\ \bibnamefont {Guerreschi}},
  \bibinfo {author} {\bibfnamefont {B.}~\bibnamefont {Peropadre}}, \bibinfo
  {author} {\bibfnamefont {J.~R.}\ \bibnamefont {McClean}},\ and\ \bibinfo
  {author} {\bibfnamefont {A.}~\bibnamefont {Aspuru-Guzik}},\ }\bibfield
  {title} {\bibinfo {title} {Boson sampling for molecular vibronic spectra},\
  }\href@noop {} {\bibfield  {journal} {\bibinfo  {journal} {Nature Photonics}\
  }\textbf {\bibinfo {volume} {9}},\ \bibinfo {pages} {615} (\bibinfo {year}
  {2015})}\BibitemShut {NoStop}%
\bibitem [{\citenamefont {Huh}\ and\ \citenamefont
  {Yung}(2017)}]{huh2017vibronic}%
  \BibitemOpen
  \bibfield  {author} {\bibinfo {author} {\bibfnamefont {J.}~\bibnamefont
  {Huh}}\ and\ \bibinfo {author} {\bibfnamefont {M.-H.}\ \bibnamefont {Yung}},\
  }\bibfield  {title} {\bibinfo {title} {Vibronic boson sampling: generalized
  gaussian boson sampling for molecular vibronic spectra at finite
  temperature},\ }\href@noop {} {\bibfield  {journal} {\bibinfo  {journal}
  {Scientific reports}\ }\textbf {\bibinfo {volume} {7}},\ \bibinfo {pages} {1}
  (\bibinfo {year} {2017})}\BibitemShut {NoStop}%
\bibitem [{\citenamefont {Banchi}\ \emph {et~al.}(2020)\citenamefont {Banchi},
  \citenamefont {Fingerhuth}, \citenamefont {Babej}, \citenamefont {Ing},\ and\
  \citenamefont {Arrazola}}]{banchi2020molecular}%
  \BibitemOpen
  \bibfield  {author} {\bibinfo {author} {\bibfnamefont {L.}~\bibnamefont
  {Banchi}}, \bibinfo {author} {\bibfnamefont {M.}~\bibnamefont {Fingerhuth}},
  \bibinfo {author} {\bibfnamefont {T.}~\bibnamefont {Babej}}, \bibinfo
  {author} {\bibfnamefont {C.}~\bibnamefont {Ing}},\ and\ \bibinfo {author}
  {\bibfnamefont {J.~M.}\ \bibnamefont {Arrazola}},\ }\bibfield  {title}
  {\bibinfo {title} {Molecular docking with gaussian boson sampling},\
  }\href@noop {} {\bibfield  {journal} {\bibinfo  {journal} {Science Advances}\
  }\textbf {\bibinfo {volume} {6}},\ \bibinfo {pages} {eaax1950} (\bibinfo
  {year} {2020})}\BibitemShut {NoStop}%
\bibitem [{\citenamefont {Schuld}\ \emph {et~al.}(2020)\citenamefont {Schuld},
  \citenamefont {Br{\'a}dler}, \citenamefont {Israel}, \citenamefont {Su},\
  and\ \citenamefont {Gupt}}]{schuld2020measuring}%
  \BibitemOpen
  \bibfield  {author} {\bibinfo {author} {\bibfnamefont {M.}~\bibnamefont
  {Schuld}}, \bibinfo {author} {\bibfnamefont {K.}~\bibnamefont {Br{\'a}dler}},
  \bibinfo {author} {\bibfnamefont {R.}~\bibnamefont {Israel}}, \bibinfo
  {author} {\bibfnamefont {D.}~\bibnamefont {Su}},\ and\ \bibinfo {author}
  {\bibfnamefont {B.}~\bibnamefont {Gupt}},\ }\bibfield  {title} {\bibinfo
  {title} {Measuring the similarity of graphs with a gaussian boson sampler},\
  }\href@noop {} {\bibfield  {journal} {\bibinfo  {journal} {Physical Review
  A}\ }\textbf {\bibinfo {volume} {101}},\ \bibinfo {pages} {032314} (\bibinfo
  {year} {2020})}\BibitemShut {NoStop}%
\bibitem [{\citenamefont {Ikai}(2011)}]{ikai2011theory}%
  \BibitemOpen
  \bibfield  {author} {\bibinfo {author} {\bibfnamefont {H.}~\bibnamefont
  {Ikai}},\ }\bibfield  {title} {\bibinfo {title} {On the theory of pfaffians
  based on exponential maps in exterior algebras},\ }\href@noop {} {\bibfield
  {journal} {\bibinfo  {journal} {Linear algebra and its applications}\
  }\textbf {\bibinfo {volume} {434}},\ \bibinfo {pages} {1094} (\bibinfo {year}
  {2011})}\BibitemShut {NoStop}%
\bibitem [{\citenamefont {Moylett}\ and\ \citenamefont
  {Turner}(2018)}]{moylett2018quantum}%
  \BibitemOpen
  \bibfield  {author} {\bibinfo {author} {\bibfnamefont {A.~E.}\ \bibnamefont
  {Moylett}}\ and\ \bibinfo {author} {\bibfnamefont {P.~S.}\ \bibnamefont
  {Turner}},\ }\bibfield  {title} {\bibinfo {title} {Quantum simulation of
  partially distinguishable boson sampling},\ }\href@noop {} {\bibfield
  {journal} {\bibinfo  {journal} {Physical Review A}\ }\textbf {\bibinfo
  {volume} {97}},\ \bibinfo {pages} {062329} (\bibinfo {year}
  {2018})}\BibitemShut {NoStop}%
\bibitem [{\citenamefont {{Napp}}\ \emph {et~al.}(2019)\citenamefont {{Napp}},
  \citenamefont {{La Placa}}, \citenamefont {{Dalzell}}, \citenamefont
  {{Brandao}},\ and\ \citenamefont {{Harrow}}}]{Napp2020}%
  \BibitemOpen
  \bibfield  {author} {\bibinfo {author} {\bibfnamefont {J.}~\bibnamefont
  {{Napp}}}, \bibinfo {author} {\bibfnamefont {R.~L.}\ \bibnamefont {{La
  Placa}}}, \bibinfo {author} {\bibfnamefont {A.~M.}\ \bibnamefont
  {{Dalzell}}}, \bibinfo {author} {\bibfnamefont {F.~G.~S.~L.}\ \bibnamefont
  {{Brandao}}},\ and\ \bibinfo {author} {\bibfnamefont {A.~W.}\ \bibnamefont
  {{Harrow}}},\ }\bibfield  {title} {\bibinfo {title} {{Efficient classical
  simulation of random shallow 2D quantum circuits}},\ }\href@noop {}
  {\bibfield  {journal} {\bibinfo  {journal} {arXiv e-prints}\ ,\ \bibinfo
  {eid} {arXiv:2001.00021}} (\bibinfo {year} {2019})},\ \Eprint
  {https://arxiv.org/abs/2001.00021} {arXiv:2001.00021 [quant-ph]} \BibitemShut
  {NoStop}%
\bibitem [{\citenamefont {Helsen}\ \emph {et~al.}(2020)\citenamefont {Helsen},
  \citenamefont {Nezami}, \citenamefont {Reagor},\ and\ \citenamefont
  {Walter}}]{helsen2020matchgateRB}%
  \BibitemOpen
  \bibfield  {author} {\bibinfo {author} {\bibfnamefont {J.}~\bibnamefont
  {Helsen}}, \bibinfo {author} {\bibfnamefont {S.}~\bibnamefont {Nezami}},
  \bibinfo {author} {\bibfnamefont {M.}~\bibnamefont {Reagor}},\ and\ \bibinfo
  {author} {\bibfnamefont {M.}~\bibnamefont {Walter}},\ }\bibfield  {title}
  {\bibinfo {title} {Matchgate benchmarking: Scalable benchmarking of a
  continuous family of many-qubit gates},\ }\href@noop {} {\bibfield  {journal}
  {\bibinfo  {journal} {arXiv preprint arXiv:2011.13048}\ } (\bibinfo {year}
  {2020})}\BibitemShut {NoStop}%
\bibitem [{\citenamefont {Lichtenstein}(1982)}]{WLichtenstein1982}%
  \BibitemOpen
  \bibfield  {author} {\bibinfo {author} {\bibfnamefont {W.}~\bibnamefont
  {Lichtenstein}},\ }\bibfield  {title} {\bibinfo {title} {A system of quadrics
  describing the orbit of the highest weight vector},\ }\href
  {http://www.jstor.org/stable/2044044} {\bibfield  {journal} {\bibinfo
  {journal} {Proceedings of the American Mathematical Society}\ }\textbf
  {\bibinfo {volume} {84}},\ \bibinfo {pages} {605} (\bibinfo {year}
  {1982})}\BibitemShut {NoStop}%
\bibitem [{\citenamefont {Ku\ifmmode~\acute{s}\else \'{s}\fi{}}\ and\
  \citenamefont {Bengtsson}(2009)}]{KusBengtsson2009}%
  \BibitemOpen
  \bibfield  {author} {\bibinfo {author} {\bibfnamefont {M.}~\bibnamefont
  {Ku\ifmmode~\acute{s}\else \'{s}\fi{}}}\ and\ \bibinfo {author}
  {\bibfnamefont {I.}~\bibnamefont {Bengtsson}},\ }\bibfield  {title} {\bibinfo
  {title} {``classical'' quantum states},\ }\href
  {https://doi.org/10.1103/PhysRevA.80.022319} {\bibfield  {journal} {\bibinfo
  {journal} {Phys. Rev. A}\ }\textbf {\bibinfo {volume} {80}},\ \bibinfo
  {pages} {022319} (\bibinfo {year} {2009})}\BibitemShut {NoStop}%
\bibitem [{\citenamefont {Oszmaniec}\ and\ \citenamefont
  {Ku\ifmmode~\acute{s}\else \'{s}\fi{}}(2013)}]{MOuniversalFRAME}%
  \BibitemOpen
  \bibfield  {author} {\bibinfo {author} {\bibfnamefont {M.}~\bibnamefont
  {Oszmaniec}}\ and\ \bibinfo {author} {\bibfnamefont {M.}~\bibnamefont
  {Ku\ifmmode~\acute{s}\else \'{s}\fi{}}},\ }\bibfield  {title} {\bibinfo
  {title} {Universal framework for entanglement detection},\ }\href
  {https://doi.org/10.1103/PhysRevA.88.052328} {\bibfield  {journal} {\bibinfo
  {journal} {Phys. Rev. A}\ }\textbf {\bibinfo {volume} {88}},\ \bibinfo
  {pages} {052328} (\bibinfo {year} {2013})}\BibitemShut {NoStop}%
\bibitem [{\citenamefont {Oszmaniec}\ and\ \citenamefont
  {Ku\ifmmode~\acute{s}\else \'{s}\fi{}}(2014)}]{TypicalityMO2014}%
  \BibitemOpen
  \bibfield  {author} {\bibinfo {author} {\bibfnamefont {M.}~\bibnamefont
  {Oszmaniec}}\ and\ \bibinfo {author} {\bibfnamefont {M.}~\bibnamefont
  {Ku\ifmmode~\acute{s}\else \'{s}\fi{}}},\ }\bibfield  {title} {\bibinfo
  {title} {Fraction of isospectral states exhibiting quantum correlations},\
  }\href {https://doi.org/10.1103/PhysRevA.90.010302} {\bibfield  {journal}
  {\bibinfo  {journal} {Phys. Rev. A}\ }\textbf {\bibinfo {volume} {90}},\
  \bibinfo {pages} {010302} (\bibinfo {year} {2014})}\BibitemShut {NoStop}%
\bibitem [{\citenamefont {{Oszmaniec}}(2014)}]{OszmaniecPhd}%
  \BibitemOpen
  \bibfield  {author} {\bibinfo {author} {\bibfnamefont {M.}~\bibnamefont
  {{Oszmaniec}}},\ }\bibfield  {title} {\bibinfo {title} {{Applications of
  differential geometry and representation theory to description of quantum
  correlations}},\ }\href@noop {} {\bibfield  {journal} {\bibinfo  {journal}
  {arXiv e-prints}\ ,\ \bibinfo {eid} {arXiv:1412.4657}} (\bibinfo {year}
  {2014})},\ \Eprint {https://arxiv.org/abs/1412.4657} {arXiv:1412.4657
  [quant-ph]} \BibitemShut {NoStop}%
\bibitem [{\citenamefont {Lautemann}(1983)}]{lautemann1983}%
  \BibitemOpen
  \bibfield  {author} {\bibinfo {author} {\bibfnamefont {C.}~\bibnamefont
  {Lautemann}},\ }\bibfield  {title} {\bibinfo {title} {Bpp and the polynomial
  hierarchy},\ }\href
  {https://doi.org/https://doi.org/10.1016/0020-0190(83)90044-3} {\bibfield
  {journal} {\bibinfo  {journal} {Information Processing Letters}\ }\textbf
  {\bibinfo {volume} {17}},\ \bibinfo {pages} {215 } (\bibinfo {year}
  {1983})}\BibitemShut {NoStop}%
\bibitem [{\citenamefont {Toda}(1991)}]{Toda1991}%
  \BibitemOpen
  \bibfield  {author} {\bibinfo {author} {\bibfnamefont {S.}~\bibnamefont
  {Toda}},\ }\bibfield  {title} {\bibinfo {title} {{PP} is as hard as the
  polynomial-time hierarchy},\ }\href {https://doi.org/10.1137/0220053}
  {\bibfield  {journal} {\bibinfo  {journal} {{SIAM} Journal on Computing}\
  }\textbf {\bibinfo {volume} {20}},\ \bibinfo {pages} {865} (\bibinfo {year}
  {1991})}\BibitemShut {NoStop}%
\bibitem [{\citenamefont {Goodman}\ and\ \citenamefont
  {Wallach}(2009{\natexlab{b}})}]{GW}%
  \BibitemOpen
  \bibfield  {author} {\bibinfo {author} {\bibfnamefont {R.}~\bibnamefont
  {Goodman}}\ and\ \bibinfo {author} {\bibfnamefont {N.~R.}\ \bibnamefont
  {Wallach}},\ }\href@noop {} {\emph {\bibinfo {title} {Symmetry,
  Representations, and Invariants}}}\ (\bibinfo  {publisher} {Springer},\
  \bibinfo {year} {2009})\BibitemShut {NoStop}%
\bibitem [{\citenamefont {Aubrun}\ and\ \citenamefont
  {Szarek}(2017)}]{SzarekBook}%
  \BibitemOpen
  \bibfield  {author} {\bibinfo {author} {\bibfnamefont {G.}~\bibnamefont
  {Aubrun}}\ and\ \bibinfo {author} {\bibfnamefont {S.~J.}\ \bibnamefont
  {Szarek}},\ }\href {https://cds.cern.ch/record/2296008} {\emph {\bibinfo
  {title} {{Alice and Bob meet Banach: the interface of asymptotic geometric
  analysis and quantum information theory}}}},\ Mathematical surveys and
  monographs\ (\bibinfo  {publisher} {American Mathematical Society},\ \bibinfo
  {address} {Providence, RI},\ \bibinfo {year} {2017})\BibitemShut {NoStop}%
\bibitem [{\citenamefont {Paturi}(1992)}]{Paturi1992}%
  \BibitemOpen
  \bibfield  {author} {\bibinfo {author} {\bibfnamefont {R.}~\bibnamefont
  {Paturi}},\ }\bibfield  {title} {\bibinfo {title} {On the degree of
  polynomials that approximate symmetric boolean functions (preliminary
  version)},\ }in\ \href {https://doi.org/10.1145/129712.129758} {\emph
  {\bibinfo {booktitle} {Proceedings of the Twenty-Fourth Annual ACM Symposium
  on Theory of Computing}}},\ \bibinfo {series and number} {STOC '92}\
  (\bibinfo  {publisher} {Association for Computing Machinery},\ \bibinfo
  {address} {New York, NY, USA},\ \bibinfo {year} {1992})\ p.\ \bibinfo {pages}
  {468–474}\BibitemShut {NoStop}%
\bibitem [{\citenamefont {Coppersmith}\ and\ \citenamefont
  {Rivlin}(1992)}]{PolyGrowth}%
  \BibitemOpen
  \bibfield  {author} {\bibinfo {author} {\bibfnamefont {D.}~\bibnamefont
  {Coppersmith}}\ and\ \bibinfo {author} {\bibfnamefont {T.~J.}\ \bibnamefont
  {Rivlin}},\ }\bibfield  {title} {\bibinfo {title} {The growth of polynomials
  bounded at equally spaced points},\ }\href {https://doi.org/10.1137/0523054}
  {\bibfield  {journal} {\bibinfo  {journal} {SIAM Journal on Mathematical
  Analysis}\ }\textbf {\bibinfo {volume} {23}},\ \bibinfo {pages} {970}
  (\bibinfo {year} {1992})},\ \Eprint
  {https://arxiv.org/abs/https://doi.org/10.1137/0523054}
  {https://doi.org/10.1137/0523054} \BibitemShut {NoStop}%
\bibitem [{\citenamefont {Rakhmanov}(2007)}]{PolyGrowth2}%
  \BibitemOpen
  \bibfield  {author} {\bibinfo {author} {\bibfnamefont {E.~A.}\ \bibnamefont
  {Rakhmanov}},\ }\bibfield  {title} {\bibinfo {title} {Bounds for polynomials
  with a unit discrete norm},\ }\href {http://www.jstor.org/stable/20160024}
  {\bibfield  {journal} {\bibinfo  {journal} {Annals of Mathematics}\ }\textbf
  {\bibinfo {volume} {165}},\ \bibinfo {pages} {55} (\bibinfo {year}
  {2007})}\BibitemShut {NoStop}%
\bibitem [{\citenamefont {Dyer}\ \emph {et~al.}(2000)\citenamefont {Dyer},
  \citenamefont {Goldberg}, \citenamefont {Greenhill},\ and\ \citenamefont
  {Jerrum}}]{Dyer2000}%
  \BibitemOpen
  \bibfield  {author} {\bibinfo {author} {\bibfnamefont {M.~E.}\ \bibnamefont
  {Dyer}}, \bibinfo {author} {\bibfnamefont {L.~A.}\ \bibnamefont {Goldberg}},
  \bibinfo {author} {\bibfnamefont {C.~S.}\ \bibnamefont {Greenhill}},\ and\
  \bibinfo {author} {\bibfnamefont {M.}~\bibnamefont {Jerrum}},\ }\bibfield
  {title} {\bibinfo {title} {On the relative complexity of approximate counting
  problems},\ }in\ \href@noop {} {\emph {\bibinfo {booktitle} {Proceedings of
  the Third International Workshop on Approximation Algorithms for
  Combinatorial Optimization}}},\ \bibinfo {series and number} {APPROX '00}\
  (\bibinfo  {publisher} {Springer-Verlag},\ \bibinfo {address} {Berlin,
  Heidelberg},\ \bibinfo {year} {2000})\ p.\ \bibinfo {pages}
  {108–119}\BibitemShut {NoStop}%
\bibitem [{\citenamefont {Eisert}\ \emph {et~al.}(2020)\citenamefont {Eisert},
  \citenamefont {Hangleiter}, \citenamefont {Walk}, \citenamefont {Roth},
  \citenamefont {Markham}, \citenamefont {Parekh}, \citenamefont {Chabaud},\
  and\ \citenamefont {Kashefi}}]{eisert2020quantum}%
  \BibitemOpen
  \bibfield  {author} {\bibinfo {author} {\bibfnamefont {J.}~\bibnamefont
  {Eisert}}, \bibinfo {author} {\bibfnamefont {D.}~\bibnamefont {Hangleiter}},
  \bibinfo {author} {\bibfnamefont {N.}~\bibnamefont {Walk}}, \bibinfo {author}
  {\bibfnamefont {I.}~\bibnamefont {Roth}}, \bibinfo {author} {\bibfnamefont
  {D.}~\bibnamefont {Markham}}, \bibinfo {author} {\bibfnamefont
  {R.}~\bibnamefont {Parekh}}, \bibinfo {author} {\bibfnamefont
  {U.}~\bibnamefont {Chabaud}},\ and\ \bibinfo {author} {\bibfnamefont
  {E.}~\bibnamefont {Kashefi}},\ }\bibfield  {title} {\bibinfo {title} {Quantum
  certification and benchmarking},\ }\href@noop {} {\bibfield  {journal}
  {\bibinfo  {journal} {Nature Reviews Physics}\ }\textbf {\bibinfo {volume}
  {2}},\ \bibinfo {pages} {382} (\bibinfo {year} {2020})}\BibitemShut {NoStop}%
\bibitem [{\citenamefont {Mathias}(1993)}]{PolarStab}%
  \BibitemOpen
  \bibfield  {author} {\bibinfo {author} {\bibfnamefont {R.}~\bibnamefont
  {Mathias}},\ }\bibfield  {title} {\bibinfo {title} {Perturbation bounds for
  the polar decomposition},\ }\href {https://doi.org/10.1137/0614041}
  {\bibfield  {journal} {\bibinfo  {journal} {SIAM Journal on Matrix Analysis
  and Applications}\ }\textbf {\bibinfo {volume} {14}},\ \bibinfo {pages} {588}
  (\bibinfo {year} {1993})},\ \Eprint
  {https://arxiv.org/abs/https://doi.org/10.1137/0614041}
  {https://doi.org/10.1137/0614041} \BibitemShut {NoStop}%
\bibitem [{\citenamefont {Tropp}(2015)}]{TroopMatrixConcentration}%
  \BibitemOpen
  \bibfield  {author} {\bibinfo {author} {\bibfnamefont {J.~A.}\ \bibnamefont
  {Tropp}},\ }\bibfield  {title} {\bibinfo {title} {An introduction to matrix
  concentration inequalities},\ }\href {https://doi.org/10.1561/2200000048}
  {\bibfield  {journal} {\bibinfo  {journal} {Found. Trends Mach. Learn.}\
  }\textbf {\bibinfo {volume} {8}},\ \bibinfo {pages} {1–230} (\bibinfo
  {year} {2015})}\BibitemShut {NoStop}%
\bibitem [{\citenamefont {{Rains}}(2000)}]{Reins2000}%
  \BibitemOpen
  \bibfield  {author} {\bibinfo {author} {\bibfnamefont {E.~M.}\ \bibnamefont
  {{Rains}}},\ }\bibfield  {title} {\bibinfo {title} {Polynomial invariants of
  quantum codes},\ }\href {https://doi.org/10.1109/18.817508} {\bibfield
  {journal} {\bibinfo  {journal} {IEEE Transactions on Information Theory}\
  }\textbf {\bibinfo {volume} {46}},\ \bibinfo {pages} {54} (\bibinfo {year}
  {2000})}\BibitemShut {NoStop}%
\bibitem [{\citenamefont {Elvang}\ \emph {et~al.}(2005)\citenamefont {Elvang},
  \citenamefont {Cvitanović},\ and\ \citenamefont {Kennedy}}]{DiagYoung2005}%
  \BibitemOpen
  \bibfield  {author} {\bibinfo {author} {\bibfnamefont {H.}~\bibnamefont
  {Elvang}}, \bibinfo {author} {\bibfnamefont {P.}~\bibnamefont
  {Cvitanović}},\ and\ \bibinfo {author} {\bibfnamefont {A.~D.}\ \bibnamefont
  {Kennedy}},\ }\bibfield  {title} {\bibinfo {title} {Diagrammatic young
  projection operators for u(n)},\ }\href {https://doi.org/10.1063/1.1832753}
  {\bibfield  {journal} {\bibinfo  {journal} {Journal of Mathematical Physics}\
  }\textbf {\bibinfo {volume} {46}},\ \bibinfo {pages} {043501} (\bibinfo
  {year} {2005})},\ \Eprint
  {https://arxiv.org/abs/https://doi.org/10.1063/1.1832753}
  {https://doi.org/10.1063/1.1832753} \BibitemShut {NoStop}%
\bibitem [{\citenamefont {Coleman}\ and\ \citenamefont
  {Yukalov}(2000)}]{coleman2000reduced}%
  \BibitemOpen
  \bibfield  {author} {\bibinfo {author} {\bibfnamefont {A.~J.}\ \bibnamefont
  {Coleman}}\ and\ \bibinfo {author} {\bibfnamefont {V.~I.}\ \bibnamefont
  {Yukalov}},\ }\href@noop {} {\emph {\bibinfo {title} {Reduced density
  matrices: Coulson’s challenge}}},\ Vol.~\bibinfo {volume} {72}\ (\bibinfo
  {publisher} {Springer Science \& Business Media},\ \bibinfo {year}
  {2000})\BibitemShut {NoStop}%
\bibitem [{\citenamefont {MacWilliams}\ and\ \citenamefont
  {Sloane}(1983)}]{MacWilliams}%
  \BibitemOpen
  \bibfield  {author} {\bibinfo {author} {\bibfnamefont {F.}~\bibnamefont
  {MacWilliams}}\ and\ \bibinfo {author} {\bibfnamefont {N.}~\bibnamefont
  {Sloane}},\ }\href@noop {} {\emph {\bibinfo {title} {The Theory of
  Error-Correcting Codes}}}\ (\bibinfo  {publisher} {North-Holland Pub. Co.},\
  \bibinfo {year} {1983})\BibitemShut {NoStop}%
\bibitem [{\citenamefont {{Oszmaniec}}\ \emph {et~al.}(2020)\citenamefont
  {{Oszmaniec}}, \citenamefont {{Sawicki}},\ and\ \citenamefont
  {{Horodecki}}}]{oszmaniec2020epsilon}%
  \BibitemOpen
  \bibfield  {author} {\bibinfo {author} {\bibfnamefont {M.}~\bibnamefont
  {{Oszmaniec}}}, \bibinfo {author} {\bibfnamefont {A.}~\bibnamefont
  {{Sawicki}}},\ and\ \bibinfo {author} {\bibfnamefont {M.}~\bibnamefont
  {{Horodecki}}},\ }\bibfield  {title} {\bibinfo {title} {{Epsilon-nets,
  unitary designs and random quantum circuits}},\ }\href@noop {} {\bibfield
  {journal} {\bibinfo  {journal} {arXiv e-prints}\ ,\ \bibinfo {eid}
  {arXiv:2007.10885}} (\bibinfo {year} {2020})},\ \Eprint
  {https://arxiv.org/abs/2007.10885} {arXiv:2007.10885 [quant-ph]} \BibitemShut
  {NoStop}%
\bibitem [{\citenamefont {Fuchs}\ and\ \citenamefont
  {Schweigert}(2003)}]{fuchs2003}%
  \BibitemOpen
  \bibfield  {author} {\bibinfo {author} {\bibfnamefont {J.}~\bibnamefont
  {Fuchs}}\ and\ \bibinfo {author} {\bibfnamefont {C.}~\bibnamefont
  {Schweigert}},\ }\href@noop {} {\emph {\bibinfo {title} {Symmetries, Lie
  algebras and representations: A graduate course for physicists}}}\ (\bibinfo
  {publisher} {Cambridge University Press},\ \bibinfo {year}
  {2003})\BibitemShut {NoStop}%
\end{thebibliography}%

\newpage
\onecolumngrid
\part*{Appendix}

\appendix

We collect here technical results that are used in the main part of the paper. Some of the results stated here can be of independent interest for further works on quantum information processing with fermions.

\section{Decomposition of passive and active FLO unitaries into two-qubit gates}
\label{app:decomp}

Here we provide the derivation of the decomposition of arbitrary passive and active FLO gates into two-qubit gates with layouts depicted in Fig.~\ref{fig:brick_triangle}, which was also studied in Refs.~\cite{fermPASSlayout2018,
fermACTpassLAYOUT2018, fermGAUSSlayout2019}. For  passive bosonic linear optics the analogous decompositions were discussed in Refs.~\cite{Reck1994, ReckUniform2016}.
The way to obtain these results is to consider the standard decomposition of  $\U(d)$ and $\SO(2d)$ elements into so-called Givens rotations and then apply the appropriate FLO representations $\Pipas$ and $\Piact$ on this decomposition, as we will explain below. For simplicity, we  will assume that $d$ is even, which is also the relevant case for our paper.

The (nearest neighbor) Givens rotations $G^{k}(\alpha, \varphi) \in \U(d)$ ($k=1, \ldots, d{-}1$) have the form   
\begin{equation}
G^{(k)}(\alpha, \varphi) = 
       \scalebox{0.9}{ $
       \begin{bmatrix}   1   & \cdots &    0   &  0   & \cdots &    0   \\
                      \vdots & \ddots & \vdots &     \vdots &        & \vdots \\
                         0   & \cdots &    e^{i \varphi}\cos(\alpha)   &   -\sin(\alpha)   &
                         \cdots &   0 \\ 0 & \cdots & e^{i\varphi}\sin(\alpha)   &   \phantom{-}\cos(\alpha)   & \cdots &    0   \\
                      \vdots &     &    \vdots &        \vdots & \ddots & \vdots \\
                         0   & \cdots &    0   &     0   & \cdots &    1
       \end{bmatrix}$} \, , \;
       \end{equation}
where only the $2 \times 2$ block consisting of the entries with row and column indices $k$ and $k+1$  are non-trivial. A general element $U \in \U(d)$ can then be decomposed into Givens rotations in different ways, we will consider two of these (also discussed in \cite{ReckUniform2016}). In the first decomposition one applies  alternatingly ($d$ number of times) a series of Givens rotations $G^{(k)}$ with odd and even $k$ indices, and finally a diagonal unitary $T= \mathrm{diag}(e^{i\kappa_1},e^{i \kappa_2},\dots,e^{i\kappa_d})$ , i.e.,
\begin{align}
     U &= T B_{d/2} A_{d/2}  \ldots B_{2} A_{2} B_{1} A_{1} 
     \,  
     ,    \label{eq:app_brick_1}\\
     A_{j}&= \prod_{\substack{ k \in [d-1] \\ k \; \; \mathrm{odd}}} G^{(k)}(\alpha_{(k,j)}, \varphi_{(k,j)}) \, ,  \; \;  B_{j}= \prod_{{\substack{ k \in [d-1] \\ k \; \; \mathrm{even}}}} G^{(k)}(\beta_{(k,j)}, \nu_{(k,j)})
     , \; \; \;  j \in [d/2]. 
\end{align}
Also in the second decomposition one applies alternatingly  a series of Givens rotations $G^{(k)}$ with odd and even $k$ indices, however, this time there is $2d{-}1$ such layers and in the $\ell$th layer there are only Givens rotations up to index $(d {-} |\ell {-} d|)$, and finally there is again a diagonal unitary $T'= \mathrm{diag}(e^{i\kappa'_1},e^{i \kappa'_2},\dots,e^{i\kappa'_d})$, i.e.,
\begin{align}
     U &=  T' A'_{d} B'_{d-1}  \ldots  B'_{2}A'_{2} B'_{1} A'_{1} 
     \,  
     ,    \label{eq:app_triangle_1}\\
     A'_{j}&= \prod^{d{-}|2j{-}d|}_{\substack{ k =1 \\ k \; \; \mathrm{odd}}} G^{(k)}(\gamma_{(k,j)}, \tau_{(k,j)}) \, ,  \; \;  B'_{j}= \prod^{d{-}|2j{-}d|}_{{\substack{ k \in [d-1] \\ k \; \; \mathrm{even}}}} G^{(k)}(\delta_{(k,j)}, \sigma_{(k,j)})
     , \; \; \;  j \in [d/2]. 
\end{align}
Note that the both decompositions use the same number of elementary Givens rotations.

Now, given an arbitrary passive FLO transformation $V=\Pipas (U)$, with $U \in \U(d)$, we can use the fact that  $\Pipas$ is a representation (and thus a homomorphism) and apply it to the decompositions of Eqs.~\eqref{eq:app_brick_1} and \eqref{eq:app_triangle_1}. We obtain
\begin{align}
     &V = \Pipas(U) = \Pipas(T)  L_{d/2} K_{d/2}  \ldots  L_{2} K_{2} L_{1} K_{1} 
     \,  
     ,    \label{eq:abs_brickwall} \\
     &K_{j}= \prod_{\substack{ k \in [d-1] \\ k \; \; \mathrm{odd}}} \Pipas(G^{(k)}(\alpha_{(k,j)}), \varphi_{(k,j)}) \, ,  \; \;  L_{j}= \prod_{{\substack{ k \in [d-1] \\ k \; \; \mathrm{even}}}} \Pipas(G^{(k)}(\beta_{(k,j)}, \nu_{(k,j)}))
     , \; \; \;  j \in [d/2], 
\end{align}
and similarly
\begin{align}
     & V = \Pipas(U)=  \Pipas(T') K'_{d} L'_{d-1}  \ldots  L'_{2}K'_{2} L'_{1} K'_{1} 
     \,  
     ,    \label{eq:abs_triangle} \\
     &K'_{j}= \prod^{d{-}|2j{-}d|}_{\substack{ k =1 \\ k \; \; \mathrm{odd}}} \Pipas(G^{(k)}(\gamma_{(k,j)}, \tau_{(k,j)})) \, ,  \; \;  L'_{j}= \prod^{d{-}|2j{-}d|}_{{\substack{ k \in [d-1] \\ k \; \; \mathrm{even}}}} \Pipas(G^{(k)}(\delta_{(k,j)}, \sigma_{(k,j)}))
     , \; \; \;  j \in [d/2]. 
\end{align}
Using the definition of $\Pipas$ and the Jordan-Wigner correspondence between fermions and qubits systems, we have that  
\begin{align} 
   \Pipas(\mathrm{diag}(e^{i\alpha},e^{i \alpha_2},\dots,e^{i\alpha_d}))&= \e^{i \alpha_1 Z } \otimes \e^{i \alpha_1 Z } \otimes \cdots \otimes  \e^{i \alpha_d Z }, \\
    \Pipas(G^{(k)}(\alpha_1, \alpha_2)) &= \id^{\otimes k-1 } \otimes (\e^{-i \alpha_1 Z/2} \otimes \e^{i \alpha_1 Z/2}) \; \e^{i \alpha_2 (X \otimes X + Y \otimes Y)/2} \otimes \,  \id^{\otimes d -k-1 }
\end{align}
Thus, Eq.~\eqref{eq:abs_brickwall} and Eq.~\eqref{eq:abs_triangle} provides exactly the brickwall and triangle decomposition of Fig.~\ref{fig:brick_triangle}.

\begin{figure}
    \centering
    \includegraphics[scale=0.4]{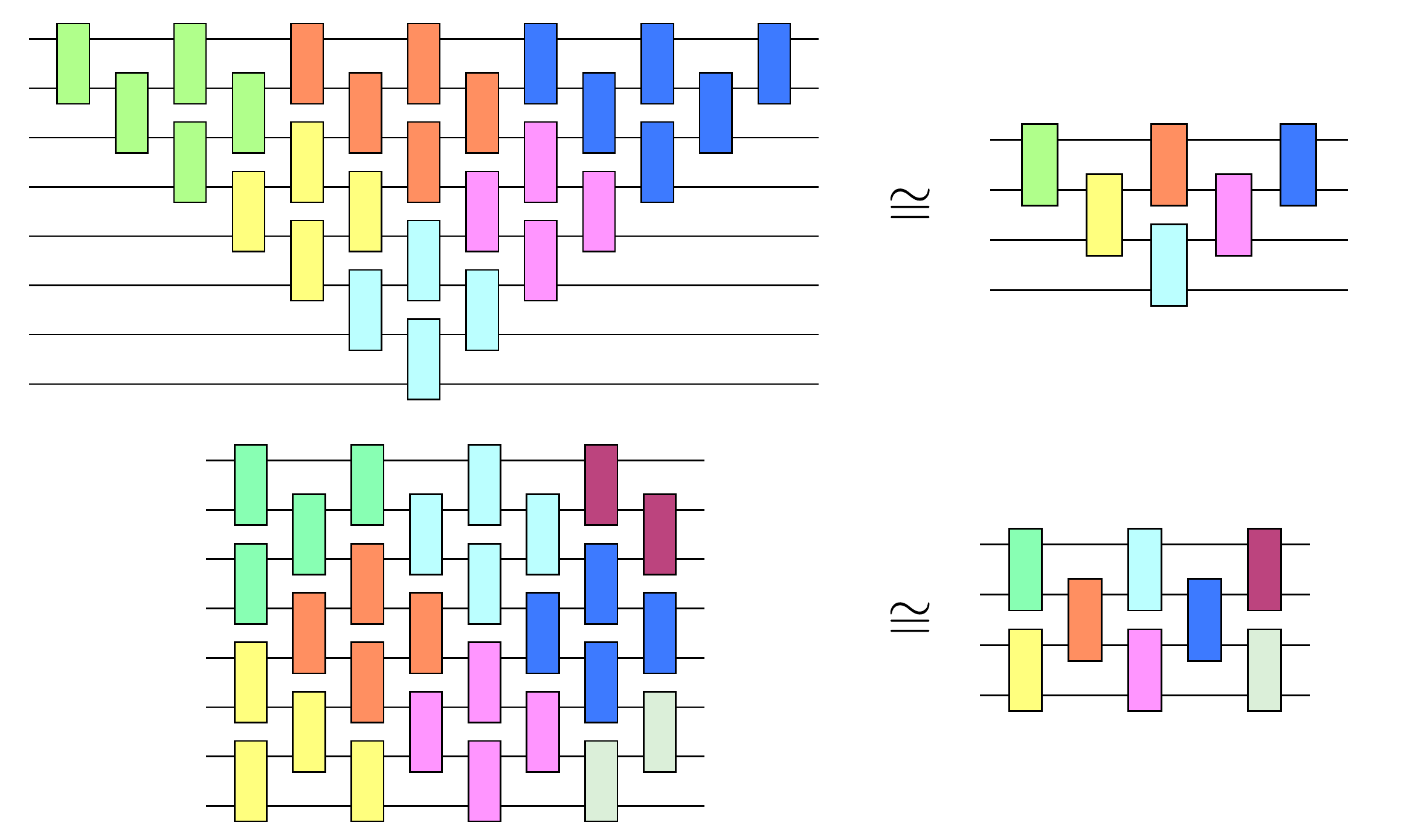}
    \caption{Decomposition of an arbitrary $V=\Piact(O)$ using  majorana-line (left) and qubit-line (right) circuit pictures. The represented Givens rotations can be merged (identical colors depicting the merged rotations) giving rise to a layout of Fig.~\ref{fig:brick_triangle}, with two-qubit gates of type $D_{\mathrm{act}}(\{ \beta_i \})$.}
    \label{fig:coarse-graining}
\end{figure}

Let us now turn to the decomposition of an arbitrary active FLO gate $V=\Piact(O)$ ($O \in \SO(2d)$). An orthogonal matrix $O$ can be decomposed into a sequence of real Givens rotations $G^{(k)}(\alpha):=G^{(k)}(\alpha, 0) \in SO(2d)$ analogously to the  decompositions of a unitary (Eqs.~\eqref{eq:app_brick_1} and \eqref{eq:app_triangle_1}). One can apply alternatingly ($d$ number of times) a series of real Givens rotations $G^{(k)}$ with odd and even $k$ indices, and finally a diagonal orthogonal matrix $S= \mathrm{diag}(s_1,s_2 ,\dots,s_d)$ (with $s_i \in \{1, -1\}$ and $\prod_{i=1}^{2d}=1$), i.e.,
\begin{align}
     U &= S D_{d/2} C_{d/2}  \ldots D_{2} C_{2} D_{1} C_{1} 
     \,  
     ,    \label{eq:app_brick_2}\\
     C_{j}&= \prod_{\substack{ k \in [d-1] \\ k \; \; \mathrm{odd}}} G^{(k)}(\alpha_{(k,j)}) \, ,  \; \;  D_{j}= \prod_{{\substack{ k \in [d-1] \\ k \; \; \mathrm{even}}}} G^{(k)}(\beta_{(k,j)})
     , \; \; \;  j \in [d/2]. 
\end{align}
Alternatively, one can apply alternatingly  a series of Givens rotations $G^{(k)}$ with odd and even $k$ indices with $2d{-}1$ layers and in the $\ell$th layer there are only real Givens rotations up to index $(d {-} |\ell {-} d|)$, and finally there is again a diagonal matrix with signs $S'= \mathrm{diag}(s'_1,s'_2,\dots, s'_d)$, i.e.,
\begin{align}
     U &=  S' C'_{d} D'_{d-1}  \ldots  D'_{2}C'_{2} D'_{1} C'_{1} 
     \,  
     ,    \label{eq:app_triangle_2}\\
     C'_{j}&= \prod^{d{-}|2j{-}d|}_{\substack{ k =1 \\ k \; \; \mathrm{odd}}} G^{(k)}(\gamma_{(k,j)}) \, ,  \; \;  D'_{j}= \prod^{d{-}|2j{-}d|}_{{\substack{ k \in [d-1] \\ k \; \; \mathrm{even}}}} G^{(k)}(\delta_{(k,j)})
     , \; \; \;  j \in [d/2]. 
\end{align}

Given an arbitrary active FLO transformation $V=\Piact (O)$, with $O \in \SO(2d)$, we can use the fact that  $\Piact$ is a projective representation (and thus a projective homomorphism) and apply it to the decompositions of Eqs.~\eqref{eq:app_brick_2} and \eqref{eq:app_triangle_2}, obtaining upto an irrelevant signs $\sigma, \sigma' \in \{1,-1 \}$ that
\begin{align}
     &V = \Piact(O) = \sigma \Piact(S)  F_{d/2} E_{d/2}  \ldots  F_{2} E_{2} F_{1} E_{1} 
     \,  
     ,    \label{eq:abs_brickwall} \\
     &E_{j}= \prod_{\substack{ k \in [d-1] \\ k \; \; \mathrm{odd}}} \Piact(G^{(k)}(\alpha_{(k,j)})) \, ,  \; \;  F_{j}= \prod_{{\substack{ k \in [d-1] \\ k \; \; \mathrm{even}}}} \Piact(G^{(k)}(\beta_{(k,j)}))
     , \; \; \;  j \in [d/2], 
\end{align}
and 
\begin{align}
     & V = \Piact(O)=  \sigma' \Pipas(S') F'_{d} E'_{d-1}  \ldots  F'_{2}E'_{2} F'_{1} E'_{1} 
     \,  
     ,    \label{eq:abs_triangle} \\
     &E'_{j}= \prod^{d{-}|2j{-}d|}_{\substack{ k =1 \\ k \; \; \mathrm{odd}}} \Piact(G^{(k)}(\gamma_{(k,j)})) \, ,  \; \;  L'_{j}= \prod^{d{-}|2j{-}d|}_{{\substack{ k \in [d-1] \\ k \; \; \mathrm{even}}}} \Piact(G^{(k)}(\delta_{(k,j)}))
     , \; \; \;  j \in [d/2]. 
\end{align}
Using the definition of $\Pipas$ and the Jordan-Wigner correspondence between fermions and qubits systems, we have that
\begin{align}
    \Piact(S) = \pm \, X^{s_{1}} Y^{s_{2}} \otimes X^{s_3} Y^{s_4} \otimes \ldots \otimes X^{s_{2d{-}1}} Y^{s_{2d}},
\end{align}
and 
\begin{align}
    \Pipas(G^{(k)}(\alpha)) = \e^{ - \alpha  \, m_{k} m_{k+1} }
    =
    \begin{cases}
     \id^{\otimes(\ell{-}1)} \otimes \e^{ i \alpha Z_{\ell}} \otimes \id^{\otimes(d{-}\ell)} \; \; &\text{if} \;  k=2\ell \; \text{is even} \\ 
      \id^{\otimes(\ell{-}1} \otimes \e^{  i \alpha X_{\ell} X_{\ell+1}} \otimes \id^{\otimes(d{-}\ell{-}1)} \; \; &\text{if} \; k=2\ell+1 \; \text{is odd} \\ 
    \end{cases}
\end{align}
Thus, the circuits would resemble the  brickwall and layouts, however with depths $2d$ and $(4d{-}1)$ on $2d$ majorana lines and not of depth $d$ and $2d-1$ on $d$ qubit lines, see Fig.~\ref{fig:coarse-graining}. 
(In circuits with majorana lines the lines represent individual operators and gates between two majorana lines are unitaries that is composed only of the corresponding two majorana operators \cite{bravyi_fermionic_2002}.)   However, we can make some simplifications by merging gates as shown in Fig.~\ref{fig:coarse-graining}: in the middle of the circuit we can merge 4 two-qubit gates (corresponding to 4 Givens rotations) of the form $\e^{i \alpha_1 X \otimes X} \, (\e^{i \alpha_2 Z} \otimes \e^{i \alpha_3 Z}) \, \e^{i \alpha_4 X \otimes X}$ and these are equal to gates of the form  $D_{\mathrm{act}}(\{ \beta_i \}) = 
     (\e^{i \beta_5  Z/2} \otimes \e^{i \beta_6  Z/2}) \, \e^{i (\beta_3 X \otimes X + \beta_4 Y \otimes Y)/2} \, (\e^{i \beta_1  Z/2} \otimes \e^{i \beta_2  Z/2})$ , where the $\beta_i$'s has to be chosen to satisfy
\begin{align}
   \cos (\alpha_1 + \alpha_4) \cos(\alpha_2-\alpha_3) = \cos(\theta_2) \cos(\theta_1 + \theta_3), \; \; 
\sin (\alpha_1 + \alpha_4)
\cos (\alpha_2 - \alpha_3) = \sin (\theta_2) \cos(\theta_1 - \theta_3), \\
\cos (\alpha_1 - \alpha_4) \sin(\alpha_2-\alpha_3) = \cos(\theta_2) \sin (\theta_1 + \theta_3), \; \; \cos (\alpha_1 - \alpha_4) \sin(\alpha_2 + \alpha_3)=\cos(\theta_5) \sin (\theta_4 + \theta_6) , \\
  \cos (\alpha_1 + \alpha_4) \cos(\alpha_2 + \alpha_3) = \cos(\theta_5) \cos(\theta_4 + \theta_6), \; \; 
\sin (\alpha_1 + \alpha_4)
\cos (\alpha_2+\alpha_3) = \sin (\theta_5) \cos(\theta_4 - \theta_6),
\end{align}
where we used the notations $\theta_1= \beta_1 - \beta_2$, $\theta_2= \beta_3- \beta_4$, $\theta_3=\beta_5 - \beta_6$, $\theta_4 = \beta_1 + \beta_2 $, $\theta_5= \beta_3+ \beta_4$, 
$\theta_6=\beta_5 + \beta_6$.
At the edges of the the circuit we may just either have to join additional local $Z$-rotations to the merged gates, thus it can be again expressed as $D_{\mathrm{act}}(\{ \beta_i \})$, or it is already of the form of $D_{\mathrm{act}}(\{ \beta_i \})$. In this way, we obtain exactly the brickwall and triangle decomposition of Fig.~\ref{fig:brick_triangle} with two-qubit gates of the form of $D_{\mathrm{act}}(\{ \beta_i \})$.


\section{Proof of Theorem \ref{th:SAMPLhar}}\label{app:hardnessOFsampling}

\begin{proof}
We will consider in parallel active and passive FLO circuits. For pasive FLO we have $\H=\Hpas$ and $\nu=\nupas$ while for active FLO we have $\Hact$ and $\nu=\nuact$. With the fixed input $\ket{\inpsi}=\psiQUAD^{\otimes N}$ ,
we write $p_{\x}(V) = \abs{\av{\x|V|\inpsi}}^2$ for the probability of outcome $\x$ (we assume  that $\ket{x}\in\H$), and $p(V)$ for the output probability distribution of a circuit $V$. 
Suppose that there exists a classical sampler $\mathcal{C}$ that performs Fermion Sampling for a fixed but arbitrary FLO circuit $V$, and denote by $q(V)$ the distribution from which $\mathcal{C}$ samples.
Then for a given $\x$, by Stockmeyer's approximate counting algorithm \cite{Stockmeyer1985}, a $\mathrm{BPP^{NP}}$ machine with an oracle access to $\mathcal{C}$ can produce a multiplicative estimates $\tilde{q}_\x(V)$ of $q_{\x}(V)$ such that
\begin{align}
    \abs{q_{\x}(V) - \tilde{q}_{\x}(V)} \le \frac{q_{\x}}{\poly{N}}
\end{align}
for every $\x$. 
We will show that $\tilde{q}_{\x}(V)$ is also close to $p_{\x}(V)$ for most $\x$ and $V$ that anti-concentrate.
Judiciously applying the triangle inequality, we have that
\begin{align}
\begin{split}
	|p_{\x}&(V)-\tilde{q}_{\x}(V)| \\
	&\le |p_{\x}(V)-q_{\x}(V)| + |q_{\x}(V)-\tilde{q}_{\x}(V)|  
\end{split}	
	\\
		&\le |p_{\x}(V)-q_{\x}(V)| + \frac{q_{\x}(V)}{\poly{N}} \\
		&\le |p_{\x}(V)-q_{\x}(V)| + \frac{|p_{\x}(V) - q_{\x}(V)| + p_{\x}(V)}{\poly{N}} \\
		&= \frac{p_{\x}(V)}{\poly{N}} + |p_{\x}(V)-q_{\x}(V)| \left(1 + \frac{1}{\poly{N}}\right)\label{eq:stockmeyer2}
\end{align}
Given that the distributions $p(V)$ and $q(V)$ are $\epsilon$-close in the $l_1$ norm, particular probabilities $p_{\x}(V)$ and $q_{\x}(V)$ must be exponentially close for most $\x$. 
This statement is made precise using Markov's inequality:
for a nonnegative random variable $X$ and $a>0$,
\begin{equation}
{\rm Pr}(X\ge a) \le \frac{\mathbb{E}X}{a}.
\end{equation}
Setting $X=|p_{\x}(V)-q_{\x}(V)|$ and  $a=\epsilon/(|\H|\delta)$, (the probability is over the outcomes $\x$ which is distributed uniformly over $\H$, see Remark \ref{rem:circuit-outcome-joint-prob})
\begin{equation}
\begin{split}
\Pr_{\x\sim\mathrm{unif}(\H)} &\left( |p_{\x}(V)-q_{\x}(V)|  
\ge \frac{\epsilon}{|\H|\delta} \right)\\
&\le \frac{\E_{\x\sim\mathrm{unif}(\H)}(|p_{\x}(V)-q_{\x}(V)|) |\H| \delta}{\epsilon} \le \delta.  
\end{split}
\end{equation}
Combining the probability bound with the inequality \eqref{eq:stockmeyer2}, we have that with probability at least $1-\delta$ over random $\x\sim\mathrm{unif}(\H)$
 we have 
\begin{equation}\label{eq:stock1}
|p_{\x}(V)-\tilde{q}_{\x}(V)| 
	 < \frac{p_{\x}(V)}{\poly{N}} + \frac{\epsilon}{|\H|\delta} \left(1 + \frac{1}{\poly{N}}\right)\ .
\end{equation}

To turn the above additive upper bound to a multiplicative one, we use the anticoncentration property (Theorem \ref{thm:anticoncentrationFLO}), which let us replace $1/|\H|$ by an upper bound $p_{\x}(V)/\alpha$ with probability $(1-\alpha)^2/C$.

In order to do so, we must consider the joint probability of $(V,\x)$ as described in Remark \ref{rem:circuit-outcome-joint-prob}. 
Let $A$ be the event that $p_{\x}(V)$ and $q_{\x}(V)$ for a fixed $V$ are exponential close due to Markov's inequality, and $B$ be the event that the distribution $p(V)$ anticoncentrates.
The probability of both ``good events" happening is lower bounded by $\Pr(A\cap B) \ge \max\{0,\Pr(A)+\Pr(B)-1 \}$. That is, if we denote by $\A(V,\x)$ an event that
\begin{equation}
\begin{split}
    |p_{\x}(V)&-\tilde{q}_{\x}(V)| \\
    &<
	p_{\x}(V)\left[
	 \frac{1}{\poly{N}} + \frac{\epsilon}{\alpha\delta} \left(1 + \frac{1}{\poly{N}}\right) \right],
\end{split}
\end{equation}
we have that
\begin{equation}\label{eq:stock2}
	\Pr_{V\sim\nu,\x\sim \mathrm{unif}(\H)} \left[\A(V,\x)\right]
	 > \frac{(1-\alpha)^2}{C}-\delta,
\end{equation}
which can be simplified by using the hiding property described in Lemma \ref{lem:hidingPROP}. The property implies that  $p_{x}(V)=p_{\x_0}(V_\x)$ and $\tilde{q}_\x (V)=\tilde{q}_{\x_0} (V_\x)$ so that
    \begin{equation}
    \Pr_{V\sim\nu,\x\sim  \mathrm{unif}(\H)}\left[\A(V,\x)\right] 
    =\underset{\x\sim\mathrm{unif}(\H)}{\E} \left(\Pr_{V\sim\nu}\left[\A(V_\x,\x_0)\right] \right).
    \end{equation}
Moreover from the invariance of the Haar measure it follows that for every $\ket{\x}\in\H$,    $V_\x$ is distributed in the same way as $V$.  Consequently,
\begin{equation} 
\begin{split}
 \underset{\x\sim\mathrm{unif}(\H)}{\E} \! \! \left(\Pr_{V\sim\nu}\left[\A(V_\x,\x_0)\right]
\right) \! 
&= \! \! \! \! \underset{\x\sim\mathrm{unif}(\H)}{\E} \! \!
\left(\Pr_{V\sim\nu}\left[\A(V,\x_0)\right] \right) \\
&= \Pr_{V\sim\nu}\left[\A(V,\x_0)\right].
\end{split}
\end{equation}
We finally obtain that for every $\x_0$,
\begin{align}\label{eq:stock3}
	\Pr_{V\sim\nu} \left\{|p_{\x_o}(V)-\tilde{q}_{\x_o}(V)|  < 
	p_{\x_o}(V)\left[
	 \frac{1}{\poly{N}} + \frac{\epsilon}{\alpha\delta} \left(1 + \frac{1}{\poly{N}}\right) \right]\right \}
	 > \frac{(1-\alpha)^2}{C}-\delta.
\end{align}

Following \cite{Bremner2011} and requiring a constant $\epsilon$ and relative error $\epsilon/(\alpha\delta)$ (See Remark \ref{rem:constants-in-stockmeyer}) we may set, for instance,
\begin{align}
    \alpha = \frac{1}{2}, && \delta = \frac{(1-\alpha)^2}{2C} = \frac{1}{8C}, && \epsilon = \frac{\alpha\delta}{4} = \frac{1}{64C},
\end{align}
Stockmeyer's algorithm is able to output  $(1/4+o(1), 1/(8C))$-multiplicative approximates of the output probabilities for $1/(8C)$ fraction of the (passive or active, with constant $C_{\mathrm{pas}}$ or $C_{\mathrm{act}}$ respectively) FLO circuits $V$ if there is a classical machine that approximately sample from $p_{\x}(V)$ for any  FLO circuit $V$ within the $l_1$ distance $1/(64C)$. 

\end{proof}

\begin{rem}\label{rem:constants-in-stockmeyer}
    Of the three parameters $\epsilon,\delta$, and $\alpha$, the $l_1$-distance $\epsilon$ and the relative error $\epsilon/(\alpha\delta)$ are typically assumed to be constant ($1/4+o(1)$ for the latter) in quantum advantage proposals \cite{Bremner2016,gao_ising_2017,bermejo-vega_architectures_2018,bouland_conjugated_2018}. In which case $\delta$ is also a constant, and then one optimizes for the constant $\alpha$. However, one may allow $\epsilon$ to decay inverse polynomially in the size of the system while retaining a sensible notion of simulation by the sampling task \cite{Aaronson2013,pashayan_estimation_2020}. Doing so allows a more plausible (weaker) average-case hardness assumption but the sampling task becomes more demanding.
\end{rem}

\section{TV distance between Haar measure and its Cayley path deformations}\label{app:TVbound}

In this part we prove Lemma \ref{lem:TVdistanceGroup} that upper bounds total variation distance between the Haar mesures $\mu_G$ and thier deformed versions $\mu^\theta_G$, for $G=\U(d)$ and $G=\SO(2d)$. In what follows we will use the notions and notation established in Section  \ref{sec:cayley}

{\bf Lemma} {\it
Let $G$ be equal to $\U(d)$ or  $SO(2d)$. Let $g_0\in G$ be a fixed element in $G$. Let $g\sim \mu_G$ an let $g_\theta=g_0 F_\theta(g)$, for $\theta\in[0,1]$ and $F_\theta:G\rightarrow G$ defined in \eqref{eq:deformTRUE}. Let now $\mu_G^\theta$ denotes the induced measure according to which $g_\theta$? is distributed. Assume furthermore that $\theta\leq 1-\Delta$, for $\Delta>0$. We then have }
\begin{equation}
\begin{split}
    \norm{\mu_{\U(d)} -\mu^\theta_{\U(d)}}_{\mathrm{TVD}} &\leq   d^2 \Delta/2\  ,\ \\ \norm{\mu_{\SO(2d)} -\mu^\theta_{\SO(2d)}}_{\mathrm{TVD}} &\leq  d^2 \Delta/2 \ .
\end{split}
\end{equation}

\begin{proof}
By the block diagonalization previously discussed, the TVD can be computed in terms of an integral on maximal torus $\mathbb{T}$ of $G$:
\begin{align}
	\norm{\mu_G - \mu_G^{\theta}}_{\mathrm{TVD}}
	=\frac{1}{2} \int_{\mathbb{T}} d\bVar \abs{\mu_G(\bVar) - \mu_G^{\theta}(\bVar)}.
\end{align}
$\mu_G$ is the distribution for generic $g\in G$ of the ``generalized eigenvalues" $\bPh \coloneqq (\phi_1,\dots,\phi_d)$, quantities that are invariant under conjugation by any element of $G$, and it is given by the celebrated Weyl's integration formulas:

\begin{fact}[Weyl's integration formula for $\U(d)$ and $\SO(d)$]\label{thm:weyl} 
\begin{equation}\label{eq:weyl}
\begin{split}
	\mu_{\U(d)}(\bPh) &= 
	\frac{1}{d!(2\pi)^d} \prod_{1\le j,k\le d} \left|e^{i\phi_k} - e^{i\phi_j}\right|^2   \\
	\mu_{\SO(2d)}(\bPh) &= 
	\frac{2}{d!(2\pi)^d} \prod_{1\le j,k\le d} 4(\cos(\phi_k) - \cos(\phi_j))^2 
\end{split}
\end{equation}
\end{fact}

To upper bound the TVD,
we use the fact that the Haar measure $\mu_G(\bPh)$ is induced from $\mu_G^{\theta}(\bVar)$ by the inverse map $F^{-1}_{\theta}$ to compute $\mu_G(\bVar)$ by Fact~\ref{fact:transport-measure} below.
In particular, we will show that the two measures $\mu_G$ and $\mu_G^{\theta}$ expressed in the same coordinates $\bVar$ are proportional to each other and bound the proportionality constant.

\begin{fact}[Transport of measure]\label{fact:transport-measure}
Let $M$ and $N$ be $d$-dimensional smooth manifolds  with local coordinates $\bPh = (\phi_1,\phi_2,\dots,\phi_d)$ and $\bVar = (\varphi_1,\varphi_2,\dots)$, $\mu$ a measure on $M$, and $F:M\to N$ a smooth map. Then 
\begin{align}
    \tilde{\mu} = \mu \circ F^{-1}
\end{align}
is a measure on $N$ transported under $F$, where $F^{-1}$ denotes the pre-image of $F$.
In particular, for any measurable set $A \subset N$,
\begin{align}\label{eq:pushforward}
    \tilde{\mu}(A) = \mu(F^{-1}(A)) \coloneqq \int_{F^{-1}(A)} \mu(\bPh) d\bPh.
\end{align}
Explicitly, since the manifolds are locally Euclidean, $\tilde{\mu}(A)$ has an expression in terms of the Jacobian:
\begin{align}\label{eq:induced-measure}
    \tilde{\mu}(A) = \int_A \mu(F^{-1}(d\bVar)) \abs{DF^{-1}(\bVar)},
\end{align}
where $\abs{DF^{-1}(\bVar)}$ is the Jacobian, which
by the inverse function theorem
\begin{align}
    \abs{DF^{-1}(\bVar)} = \abs{DF(F^{-1}(\bVar))}^{-1}.
\end{align}
\end{fact}

\begin{rem}
Since $F^{-1}(\bVar)=\bPh$, the formula can be interpreted as a change of variable from $\bVar$ to $\bPh$. In our case,  $F^{-1}_{\theta}: \mathbb{T}(\bPh) \to \mathbb{T}(\bVar)$ plays the role of $F$. \eqref{eq:change-of-variable} is precisely the change of variable induced by $F^{-1}_{\theta}$.
\end{rem}

As an intermediate step, let us derive the explicitly change-of-variable formula from $\bVar$ to $\bPh$.
By applying to 
\begin{align}\label{eq:rescaling-map}
    \exp(\varphi
    _j\oX_j) = F_{\theta}\ (\exp(\phi_j \oX_j)).
\end{align}
(recalling that $\oX$ are generators of the maximal torus for $\U(d)$ and $\SO(2d)$)
the identity
\begin{align}\label{eq:tan} 
	f^{-1} (\exp(\phi_j \oX_j)) &= -\tan(\phi_j/2)\oX_j,
\end{align}
which can be verified by explicitly computing the Cayley transform \eqref{eq:cayley-transform} of at most a $2\times 2$ matrix, we obtain the change-of-variable formula:
\begin{align}\label{eq:change-of-variable}
	\varphi &= 2\tan^{-1}[\theta\tan(\phi/2)].
\end{align}

Now we compute and bound \eqref{eq:induced-measure} 
\begin{align}
    \mu_G(A) = \int_A \mu_G (d \bPh =F_{\theta}(d\bVar))
    \abs{DF^{-1}_{\theta}(\bPh = F_{\theta}(\bVar))}^{-1}.
\end{align}
Throughout the proof, we set $\theta = 1-\Delta$ and notice that the final upper bound on the TVD would still hold for $\theta \le 1-\Delta$.
The change-of-variable formula \eqref{eq:change-of-variable} directly gives the element of the (diagonal) Jacobian
\begin{align}
	\abs{\partial_{\phi}\varphi}^{-1} = 
	\frac{\cos^2(\phi/2) + (1-\Delta)^2\sin^2(\phi/2)}{1-\Delta} = \frac{1-\Delta(2-\Delta) \sin^2(\phi/2)}{1-\Delta},
\end{align}
which attains the minimum when $\sin^2(\phi/2)=1$ and the maximum when $\sin^2(\phi/2)=0$. 
Thus, we have the following bound on the Jacobian for both the passive and active cases
\begin{align}
    &1-\Delta \le \abs{\partial_{\phi}\varphi}^{-1} \le 
    \frac{1}{1-\Delta}, \\
	(1-\Delta)^d \le  &\abs{DF^{-1}_{\theta}(\bPh = F_{\theta}(\bVar))}^{-1} \le \frac{1}{(1-\Delta)^d}.
\end{align}

At last, to bound the TVD, we express the measures $\mu_G$ and $\mu_G^{\theta}$ in the same coordinates $\bVar$.
For the case of passive FLO, this can be done by directly applying the inverse of the deformation map \eqref{eq:deform2} to each group element $e^{i\phi_j}$, $j\in [d]$
\begin{align}
    F^{-1}_{1-\Delta}(e^{i\varphi_j})
    = \frac{\Delta + (\Delta-2)e^{i\varphi_j}}{\Delta(e^{i\varphi_j}+1) - 2}.
\end{align}
As a result,
\begin{align}
	|e^{i\phi_k} &- e^{i\phi_j}| 
	= 
	\abs{F^{-1}_{1-\Delta}(e^{i\varphi_k}) - F^{-1}_{1-\Delta}(e^{i\varphi_j})} 
	= 
	\frac{(1-\Delta)\abs{e^{i\varphi_k} - e^{i\varphi_j}}}
	{\abs{ 1-\frac{\Delta}{2}(e^{i\varphi_j}+1)}
	\abs{ 1-\frac{\Delta}{2}(e^{i\varphi_k}+1)}}
	\eqqcolon \Gamma_{\mathrm{pas}}\abs{e^{i\varphi_k} - e^{i\varphi_j}},
\end{align}
which implies, via the Weyl's formula \eqref{eq:weyl} that the two measures are proportional:
\begin{align}
    \mu_{\U(d)}(\bVar) = \Gamma_{\mathrm{pas}}^{d(d-1)/2} \mu_{\U(d)}^{\theta}(\bVar).
\end{align}
The proportionality constant $\Gamma_{\mathrm{pas}}$ attains the maximum value when $e^{i\varphi_j} = e^{i\varphi_k} = 1$ and the minimum value when $e^{i\varphi_j} = e^{i\varphi_k} = -1$, giving the following bound
\begin{align}
   &(1-\Delta)^2  \le \Gamma_{\mathrm{pas}} \le \frac{1}{(1-\Delta)^2},
\end{align}
which leads to the bound of the TVD stated in the lemma:
\begin{align}
	\norm{\mu_{\U(d)} - \mu_{\U(d)}^{\theta}}_\mathrm{TVD}
	&\le \frac{1}{2}\abs{1 - (1-\Delta)^{d^2}} \le \frac{d^2 \Delta}{2};
\end{align}
the inequality in the last line can be proved by induction on $d^2 \ge 1$.

Turning to the case of active FLO, the change-of-variable formula \eqref{eq:change-of-variable} implies that for any $j,k \in [d]$,
\begin{align}
\begin{split}
	\cos\phi_k - \cos\phi_j &=  \frac{(1-\Delta)^2 - \tan^2(\varphi_k/2)}{(1-\Delta)^2 + \tan^2(\varphi_k/2)} - \frac{(1-\Delta)^2 - \tan^2(\varphi_j/2)}{(1-\Delta)^2 + \tan^2(\varphi_j/2)} 
\end{split}\\
	&= 
	\frac{(1-\Delta)^2(\cos \varphi_k - \cos \varphi_j)}
	{\left( 1-\Delta(1-\frac{\Delta}{2})(1+\cos \varphi_j) \right)
	\left( 1-\Delta(1-\frac{\Delta}{2})(1+\cos \varphi_k) \right)}
	 \\
	&\eqqcolon \Gamma_{\mathrm{act}}(\cos \varphi_k - \cos \varphi_j),
\end{align}
where we have used $\cos(2\theta) = (1-\tan^2 \theta)(1+\tan^2 \theta)^{-1}$ in the first line and $\tan^2 \theta = (1-\cos(2\theta))(1+\cos(2\theta))^{-1}$ in the second line.
Thus we have from the Weyl's formula \eqref{eq:weyl} that
\begin{align}
    \mu_{\SO(2d)}(\bVar) 
    = \Gamma_{\mathrm{act}}^{d(d-1)/2} \mu_{\SO(2d)}^{\theta}(\bVar).
\end{align}
The proportionality constant $\Gamma_{\mathrm{act}}$  attains the maximum value when $\cos \varphi_j = 1$ and the minimum value when $\cos \varphi_j=-1$, giving the bound
\begin{align}
(1-\Delta)^2 \le \Gamma_{\mathrm{act}}
	\le \frac{1}{(1-\Delta)^2}.
\end{align}
Therefore, the TVD is bounded in a similar manner to the passive case.
\begin{align}
	\norm{\mu_{\SO(2d)} - \mu_{\SO(2d)}^{\theta}}_{\mathrm{TVD}}
	&\le \frac{1}{2}\abs{1 - (1-\Delta)^{d^2}} \le \frac{d^2 \Delta}{2}. \label{eq:finite-N-scale}
\end{align}
\end{proof}

\section{Polynomials associated to probabilities in FLO circuits}\label{sec:degree}

In this section, we give the degrees of matrix polynomials associated to fermionic representations of $G=\U(d)$ and $G=\SO(2d)$. These polynomials, when evaluated on the Cayley path $g_\theta$ in the appropriate group (see Eq.~\eqref{eq:deformTRUE}), give rise polynomials and rational functions $\theta$ for the outcome probabilities $p_\x(\Pi(g_\theta),\inpsi)=|\bra{\x}\Pi(g) \ket{\inpsi}|^2 $ in our quantum advantage schemes.  The explicit form of these polynomials will be used in Section \ref{sec:worst-to-avg-reduction} when discussing worst-to-average-case reductions.

We start with discussing the passive FLO case and then the active FLO case. 

It will be useful to introduce the following notation. Given a $d \times d$ matrix $M$ and two subsets of indices $\X, \Y \subset [d]$ with cardinality $n$, where  $\X = \{ a_1, a_2, \ldots a_n \}$ ($a_i < a_j$ if $i<j$) and $\Y=\{b_1, b_2, \ldots, b_n \}$ ( $b_i < b_j$ if $i<j$),
we define $M_{\X,\Y}$ as the $n \times n$ matrix with entries
\begin{equation}
    (M_{\X, \Y} )_{k,\ell} = M_{a_k, b_\ell} \, , \; \; \;  \; k,\ell=1, \ldots n.
\end{equation}

\begin{lem} \label{lem:Udet}
Given two Fock basis states $\ket{\X}, \ket{\Y}  \in \bigwedge^n(\C^d)$ and a $U \in \U(d)$, the the amplitude between $\ket{\X}$ and $\Pipas(U)\ket{\Y}$
is provided by the expression  
\begin{equation}
    \bra{\X} \Pipas(U)\ket{\Y} = \det (U_{\X, \Y}).
\end{equation}
\end{lem}
\begin{proof}
Let  $\X = \{ a_1, a_2, \ldots a_n \}$ ($a_i < a_j$ if $i<j$) and $\Y=\{b_1, b_2, \ldots, b_n \}$ ($b_i < b_j$ if $i<j$). By definition we have that 
\begin{equation}
\begin{split}
    \Pipas(U) \ket {\Y} &= U^{\otimes n} \ket{b_1} \wedge \ket{b_2} \wedge \cdots \wedge \ket{b_n}\\
    &= \ket{\xi_1} \wedge \ket{\xi_2} \wedge \cdots \wedge \ket{\xi_n}
\end{split}
\end{equation}
where
\begin{equation}
    \ket{\xi_\ell}= U\ket{b_\ell} = \sum_{j=1}^d U_{j, b_\ell} \ket{j} \, , \; \; \; \ell=1, \ldots, n.
\end{equation}
Using the last two equations and  Eq.~\eqref{eq:det-inner}, we can deduce that
\begin{equation}
\begin{split}
    \bra{\X} \Pipas(U)\ket{\Y} &= \det(C) \, , 
    \\
    C_{k,\ell} = \braket{a_k}{\xi_\ell} = \bra{a_k}  \sum_{j=1}^d U_{j, b_\ell} \ket{j} &= U_{a_k, b_\ell}= (U_{\X, \Y})_{k, \ell}
\end{split}
\end{equation}
which proves the statement.
\end{proof}
This lemma allows us to directly obtain the following result:

\begin{prop} [Degrees of polynomials describing probabilities associated to  passive FLO circuits.]\label{prop:passive_degree}
Consider a state $ \ket{\Psi} \in \bigwedge^n(\C^d)$. For an arbitrary $U \in \U(d)$ the outcome probability
 $ p_{\x}( \Pipas(U),\Psi) =|\bra{\x}\Pipas(U)\ket{\Psi}|^2$ is a  degree $2n$ homogeneous polynomial in the entries of $U$ and $U^{\dagger}$.
\end{prop}
\begin{proof}
One can expand the vector  $ \ket{\Psi}  $ in terms of the Fock basis states belonging to $\bigwedge^n(\C^d)$ as
\begin{equation}
    \ket{\Psi} = \sum_{\substack{\Y \subset [d] \\ |\Y|=n}} c_{\Y} \ket{\Y}.
\end{equation}
Let $\X \subset [d]$ denote the set of indices corresponding to $\x$ as an indicator function (i.e., $\ket{\x} = \ket{\X}$). Using Lemma~\ref{lem:Udet}, we can write the relevant amplitude as
\begin{equation}
\begin{split}
   \bra{\x}\Pipas(U)\ket{\Psi} &= \sum_{\substack{\Y \subset [d] \\ |\Y|=n}} c_{\Y}  \bra{\X} \Pipas(U)\ket{\Y} \\
   &=
   \sum_{\substack{\Y \subset [d] \\ |\Y|=n}} c_{\Y} \det (U_{\X, \Y})
   .
\end{split}
\end{equation}
As each term in the sum is a determinant of a $n \times n$ submatrix of $U$, this expression gives a homogeneous polynomial of the entries of $U$ of order $n$. This in turn directly implies that $ p_{\x}( \Pipas(U),\Psi) =|\bra{\x}\Pipas(U)\ket{\Psi}|^2$ is a degree $2n$ polynomial in the entries of $U$ and $U^{\dagger}$.
\end{proof}

\begin{lem}[Polynomial for output amplitude of passive FLO \cite{Ivanov2017}]\label{lem:ivanov_amplitude}
Consider the input state $\ket{\inpsi}=\psiQUAD^{\ot N}\in \bigwedge^{2N}(\C^{4N})$. For an arbitrary $U \in \U(4N)$ the outcome amplitude is given by 

\begin{equation}\label{eq:probamplitudePAS}
\begin{split}
     \bra{\x}&\Pipas(U)\ket{\Psi} =\frac{1}{\sqrt{2^N}} \sum_{(y_1,\cdots,y_N)\in\lbrace0,1 \rbrace^N} \\
     & \times \det (U^T_{\lbrace 2y_1+1,2y_1+2,\cdots,2y_N+4N-3,2y_N+4N-2 \rbrace,\X})\ ,
\end{split}
\end{equation}
  where $U^T_{\lbrace 2y_1+1,2y_1+2,\cdots,2y_N+4N-3,2y_N+4N-2 \rbrace,\X}$ indicates the transpose of $U$ with the rows not indexed by $\lbrace 2y_1+1,2y_1+2,\cdots,2y_N+4N-3,2y_N+4N-2 \rbrace$ and columns not indexed by $\X$. Note that this is a degree $N$ polynomial in the entries of $U$.
\end{lem}
\begin{proof}

To derive this polynomial, we rewrite the input fermionic magic state $\ket{\inpsi}$ as in Eq.~\eqref{eq:pureSTATEdecomp}.

\begin{align}
    \ket{\inpsi}&=\frac{1}{\sqrt{2^N}}\sum_{\Y\in\Cstate} \ket{\Y}\\ 
    &=\frac{1}{\sqrt{2^N}} \sum_{\y \in \lbrace0,1 \rbrace^N} \ket{\Y_\y},
\end{align}
 where $\Cstate$ consists of subsets labelled by bitstrings (for more detail, see paragraph after Eq.~\eqref{eq:pureSTATEdecomp}). Using the expression for the output amplitude in Proposition \ref{prop:passive_degree} we write

\begin{align}
     \bra{\x}&\Pipas(U)\ket{\Psi} = \frac{1}{\sqrt{2^N}} \sum_{\y \in \lbrace0,1 \rbrace^N} \det (U_{\X, \Y_\y})\label{eq:rederived-discriminant1}\\
\begin{split}
     &=\frac{1}{\sqrt{2^N}} \sum_{(y_1,\cdots,y_N)\in\lbrace0,1 \rbrace^N} \\
     & \times \det (U^T_{\lbrace 2y_1+1,2y_1+2,\cdots,2y_N+4N-3,2y_N+4N-2 \rbrace,\X})\ ,\label{eq:rederived-discriminant2}
\end{split}
\end{align}
where in the last line we have replaced the definition of $\Y_y$ and also used the fact that the determinant is invariant under the transpose.
\end{proof}

The expression in Eq.~\eqref{eq:probamplitudePAS} can be rewritten as a mixed discriminant 
\begin{equation}
D_{2,2}(v_1, \ldots, v_{4N})=\frac{1}{\sqrt{2^N}}\sum_{i_k=0,1, \atop k=1,\ldots, N}
\det\begin{bmatrix} 
v_{2i_1+1} \\ v_{2i_1+2} \\ v_{2i_2+5} \\ v_{2i_2+6} \\ \vdots \\
v_{2i_N+4N-3} \\ v_{2i_N+4N-2}
\end{bmatrix}\, ,
\label{eq:mixed-discriminant}
\end{equation}

here $v_k$ correspond to the rows of the matrix $U^T$ in Eq.~\eqref{eq:rederived-discriminant2} with the columns not indexed by $\x$ removed. This polynomial over entries of matrices of size $2N\times N$ was found to be $\sharP$-hard in the general case \cite{Ivanov2017}. This was proven by reducing the computation of the permanent of a weighted adjacency matrix to these polynomials of a transformed adjacency matrix with polynomial overhead.
\begin{rem}
For the hardness of sampling what is actually required is the $\sharP$-hardness of computing the square of the amplitude. In \cite{Ivanov2017} the permanents used only involved positive numbers and thus there is no issue in establishing $\sharP$-hardness for the probabilities.
\end{rem}

Next we turn to studying the output probabilities after an active FLO evolution. It will be useful to introduce the following notation: given a set of (majorana) indices $\A = \{a_1, a_2, \ldots a_k \} \subset [2d]$ (with $a_i < a_j$ if $i<j$), we define
\begin{equation} \label{eq:major-mon}
    m_{\A} = 
    m_{a_1}m_{a_2} \cdots m_{a_k}.
\end{equation}
These majorana monomials define an orthogonal (but not orthonormal) basis in the space of operators with respect to the Hilbert-Schmidt scalar product
\begin{equation} \label{eq:major_ort}
\tr(m_{\A} m_{\B}^\dagger)= (-1)^{f(|\B|)} \tr(m_{\A} m_{\B}) = \delta_{\A, \B}   \, \frac{1}{2^d},    
\end{equation}
where $f(n)=1$ if $(n \mod 4) \in \{ 2,3 \} $ and $f(n)=0$ otherwise.  

Consider a subset $\A = {a_1, a_2 \ldots a_k} \subset [2d]$ (with $a_i < a_j$ if $i<j$), then from Eq.~\eqref{eq:mode_transform}   and the majorana anticommutation relations it follows for $O \in \SO(2d)$ that   
\begin{equation} \label{eq:monomial_evol}
\begin{split}
    &\Piact(O) m_{\A} \Piact(O)^\dagger \\
    &= \sum_{b_1, \ldots b_k =1 }^{d} \epsilon_{b_1, b_2, \ldots, b_k } O_{a_1, b_1} O_{a_2, b_2} \cdots O_{a_k, b_k}  m_{\{b_1, \ldots b_k\}} .
\end{split}
\end{equation}

\begin{prop} [Degrees of polynomials describing probabilities associated to  active FLO circuits.]\label{prop:active_degree}
Consider a state $ \ket{\Psi} \in $. For an arbitrary $O \in \SO(2d)$ the outcome probability
 $ p_{\x}( \Piact(O),\Psi) =|\bra{\x}\Piact(O)\ket{\Psi}|^2$ is a degree $d$ polynomial in the entries of $O$. 
\end{prop}

\begin{proof}
Let us consider the expansion of $\ketbra{\x}{\x}$ and $\Psi$ in terms of majorana monomials
\begin{equation} 
\begin{split}
    \ketbra{\x}{\x} &= \sum_{\A \subset [d]} (a_{\A} \, m_{\A} +  b_{\A} \, Q \, m_{\A}) \ , \\
    \Psi&=\sum_{\B \subset [d]} (c_{\B} \, m_{\B} + d_{\B} \, Q \, m_{\B}) .
\end{split}
\end{equation}
Using this and Eqs.~\eqref{eq:major_ort} and \eqref{eq:monomial_evol} can now write the outcome probability as

\begin{align}
    p_{\x}( \Piact(O),\Psi) &= \tr \left(\ketbra{\x}{\x}  \Piact(O) \Psi \Piact(O)^\dagger \right) \nonumber \\
    &= \sum_{\A, \B \subset [d]} \left( a_{\A} c_{\B} \tr \left(m_{\A} \Piact(O) m_{\B} \Piact(O)^\dagger \right) + b_{\A} d_{\B} \tr \left(Qm_{\A} \Piact(O) Qm_{\B} \Piact(O)^\dagger \right) \right) \nonumber \\
    &= \sum_{k=0}^d \sum_{\substack{\A, \B \subset [d]\\ |\A|=|\B|=k}}  w_{\A, \B} \sum_{\ell_1, \ldots \ell_k =1 }^{d} \epsilon_{\ell_1, \ell_2, \ldots, \ell_k } \delta_{\A, \{\ell_1, \ldots , \ell_k \}} \,  O_{b_1, \ell_1} O_{b_2, \ell_2} \cdots O_{b_k, \ell_k} ,
\end{align}

where $ w_{\A, \B} = \frac{(-1)^{f(|A|)}}{2^{d}}(a_{\A} c_{\B}  + (-1)^k b_{\A} d_{\B})$.
Since each term in the sum is a degree $d$ or less polynomial in the entries of $O$ the theorem is proved.
\end{proof}

\begin{Definition}[Degree of rational functions] 
Let $P(\theta),Q(\theta)$ be polynomials of degree $d_1$ and $d_2$ respectively. Let  $R(\theta)=\frac{P(\theta)}{Q(\theta)}$ be the corresponding rational function. Assume that that $P$ and $Q$ do not have non-constant polynomial divisors.  Then, we define rational degree of $R$ as the pair  $\deg(R)=(d_1,d_2)$
\end{Definition}

The following results states that FLO circuit representations of elements of the appropriate symmetry group $G$, when evaluated on Cayley paths, give rise to outcome probabilities that are rational functions of low degree (in number of modes $d$ and number of particles $n$).

\begin{lem}[Degrees of rational functions describing probabilities associated to interpolation of FLO circuits]
\label{lem:degreesOFRAT}
Let $G$ be equal to $\U(d)$ or  $\SO(2d)$. Let $g_0 , g\in G$ be a fixed elements of the group $G$. Consider a rational path in the group defined by interpolation via Cayley path 
\begin{equation}
    g_\theta=g_0 F_\theta (g)\ ,\ \theta\in[0,1]\ .
\end{equation}
 Let now $\Pi:G\rightarrow\U(\H)$  be the appropriate representation of $G$ describing appropriate class of FLO circuits ($G=\U(d)$, $\Pi=\Pipas$,  $\H=\bigwedge^n(\C^d)$ for passive FLO and $G=\SO(2d)$, $\Pi=\Piact$, $\H=\Hfock^+(\C^d)$ for active FLO). Let us fix $\ket{\Psi} \in \H$ and a Fock state $\ket{\x}\in \H$. Then the outcome probability
\begin{equation}
     R_{g_0,g}(\theta) =\tr(\kb{\x}{\x} \Pi(g_\theta) \rho \Pi(g_\theta)^\dag)
\end{equation}
viewed as a function of parameter $\theta$ is a rational function of degrees 
\begin{equation}
\begin{split}
\text{Passive FLO: }&\ \    \deg (R_{g_0,g}) =  (2d n, 2d n)\ , \\
\text{Active FLO: }&\ \    \deg (R_{g_0,g}) =  (2d^2,2d^2)\ 
\end{split}
\end{equation}
Moreover the denominator of the rational functions are given by 
\begin{equation}\label{eq:denominatorsCIRCUITS}
\begin{split}
\text{Passive FLO: }&\ \     Q_g(\theta) = \prod_{j=1}^d (1+\theta^2\tan^2(\phi_j/2))^{n}  \ , \\
\text{Active FLO: }&\  Q_g(\theta)= \prod_{j=1}^d (1+\theta^2\tan^2(\phi_j/2))^{d}\ ,
\end{split}
\end{equation}
where $\phi_j$ , $j\in[d]$ are phases of generalized eigenvalues of matrix  $g$ belonging to the suitable group $G$ and thus $Q_g(\theta)$ can be efficiently computed (see Section \ref{sec:cayley}).

\end{lem}

\begin{proof}
We begin by proving the passive FLO case. Recall from Eq.~\eqref{eq:RATpoly_passive} that $g_\theta$ was expressed as a matrix with entries of degree $(d,d)$ on $\theta$. By virtue of Proposition \ref{prop:passive_degree}, we know that  $p_{\x}( \Pipas(g_\theta),\Psi)= R_{g_0,g}(\theta)$ is a polynomial of degree $2n$ on the entries of $g_\theta$ which immediately implies the degree on $\theta$ is $\deg (R_{g_0,g}) =  (2d n, 2d n)$. The denominator of the rational functions in $g_\theta$ is given by Eq.~\eqref{eq:Q-U(d)}, from the expression for the amplitude in Proposition \ref{prop:passive_degree}, we know that the denominator in $R_{g_0,g}$ must be of the form  $\abs{\Q_g(\theta)^n}^2$ which gives the result form $\prod_{j=1}^d (1+\theta^2\tan^2(\phi_j/2))^{n}$ .

For the active case, we obtain from Eq.~\eqref{eq:RATpoly_active} that $g_\theta$ is a matrix with entries that are polynomials of degree $(2d,2d)$. Then by Proposition \ref{prop:active_degree}, $p_{\x}( \Piact(g_\theta),\Psi)$ is of degree $d$ on the entries of $g_\theta$ implying $\deg (R_{g_0,g}) =  (2d^2,2d^2)$. The denominator $ Q_g(\theta)$ is obtained by noting that the expression in Eq.~\eqref{eq:Q-SO(2d)} for $\Q_g(\theta)$ appears as the denominator in each entry of $g_\theta$ and by Proposition \ref{prop:active_degree} the degree on this denominator is $d$, thus proving the result.
\end{proof}

\section{Details of computations for anticoncentration} \label{sec:Projectors}

\subsection{Passive FLO} 

In this part we give detailed computations related to establishing upper bound in Eq. \eqref{eq:EPupperbound}
for the case of passive FLO:
\begin{equation}\label{eq:EPupperboundPASAPP}
 \tr(\Pfer \inpsi\ot \inpsi ) \leq \frac{\Cpas}{N}\  ,\ \text{for } \Cpas = \Cpasval
\end{equation}
The prove of the above inequality is split into three parts. First, in Lemma \ref{lem:PROJpassFERMpurities} we give an explicit form of $\Pfer$. Second, in Lemma  \ref{lem:passPROJfinal} we find an upper bound on $\tr(\Pfer \inpsi\ot \inpsi)$ via combinatorial expression that can be efficiently computed for any fixed value of $N$. Finally, in Lemma \ref{lem:pasAnti-computation} given in Part \ref{app:comp} of the Appendix we prove an upper bound to the said combinatorial expression which yields Eq. \eqref{eq:EPupperboundPASAPP}.

\begin{lem}[Projector for passive fermionic linear optics]\label{lem:PROJpassFERMpurities}
Let $\bigwedge^n \left(\C^d \right)  $ be a fermionic $n$-particle representation of $\U(d)$ ($d \geq n$). 
Let $\Pfer$ be the projector onto a unique irreducible representation $\tilde{\Hfer}\subset \bigwedge^n \left(\C^d \right) \otimes \bigwedge^n \left(\C^d \right)$  of $\U(d)$ such that $\ket{\n} \ot \ket{\n}\in \tilde{\Hfer}$, where $\ket{\n}$ is a $n$-particle Fock state.
Then, for any $\rho\in\D\left(\bigwedge^n(\C^d) \right)$ we have 
\begin{equation}
\tr(\Pfer \rho\otimes\rho ) = \frac{1}{n+1}\sum_{k=0}^{n} {n \choose k} \tr(\rho_k^2)\ ,
\end{equation}
where $\rho_k = \tr_{n-k}(\rho)$ is a $k$-particle reduction of $\rho$.  Moreover, the dimension of $\tilde{\Hfer}$ equals
\begin{equation}\label{eq:passiveDIM}
    |\tilde{\Hfer}|=\binom{d}{n}^2 \frac{d+1}{(d-n+1)(n+1)}\ .
\end{equation}
\end{lem}

\begin{proof}
We  consider $\bigwedge^n(\C^d)$ as anti-symmetric subspace of the Hilbert space of $n$ distinguishable partices: $\bigwedge^n(\C^d)\subset (\C^d)^{\ot n}$ with $d$ dimensional single particle Hilbert spaces. Therefore also  $\bigwedge^n(\C^d) \otimes \bigwedge^n(\C^d)$ can be considered as a subspace of $2n$ distinguishable particles:
\begin{equation}
\bigwedge^n(\C^d) \otimes \bigwedge^n(\C^d) \subset (\C^d)^{\ot n} \ot (\C^d)^{\ot n}\ .
\end{equation}
Let us now label particles entering the first factor of the latter tensor product by $1,\ldots,n$ and by $1',\ldots,n' $ particles entering the second factor. In \cite{MOuniversalFRAME} it was proven that
\begin{equation}
\Pfer = \frac{2^n}{n+1}  \P^{\lbrace 1,\ldots,n \rbrace}_\mathrm{asym} \P^{\lbrace 1',\ldots,n'\rbrace}_\mathrm{asym} \left(\prod^n_{k=1}  \P^{k,k'}_{\mathrm{sym}}\right) \P^{\lbrace 1,\ldots,n \rbrace}_\mathrm{asym} \P^{\lbrace 1',\ldots,n'\rbrace}_\mathrm{asym} \ . 
\end{equation} 
In the above $\P^{k,k'}_{\mathrm{sym}}=\frac{1}{2}(\I\ot\I + \SS^{k,k'})$ is the projector onto a subspace of $(\C^d)^{\ot n} \ot (\C^d)^{\ot n}$ which is symmetric upon interchange of particles $k$ and $k'$ (by $\SS^{k,k'}$ we denote the unitary operator that swaps particles $k$ and $k'$).
Moreover,  $\P^{\mathcal{A}}_\mathrm{asym}$ denotes the projector onto a subspace which is anti-symmetric under exchange of particles in a subset $\mathcal{A}$. Now for $\rho\in\D\left(\bigwedge^n(\C^d) \right)$ we have
\begin{equation}
 \P^{\lbrace 1,\ldots,n \rbrace}_\mathrm{asym} \P^{\lbrace 1',\ldots,n'\rbrace}_\mathrm{asym} \rho \otimes \rho =\rho \otimes \rho
\end{equation}
and therefore 
\begin{equation}\label{eq:beforePROD}
\tr(\Pfer \rho\otimes\rho )= \frac{2^n}{n+1} \tr\left[\left(\prod^n_{k=1}  \P^{k,k'}_{sym}\right) \rho \ot \rho  \right]\ . 
\end{equation}
Using the definition of $\P^{k,k'}_{sym}$ we get the expansion
\begin{equation}\label{eq:prodEXPAND}
\prod^n_{k=1}  \P^{k,k'}_{sym}= \frac{1}{2^n} \sum_{\X \subset [d]} \prod_{i\in \X} \SS^{i,i'} \ ,
\end{equation}
where the summation is over subsets $X$ of $[d]=\lbrace1,\ldots,d\rbrace$. Using a well-known connection between partial swaps and purites of reduced density matrices (see for example \cite{Reins2000}):  
\begin{equation}
\tr\left(\prod_{i\in \X} \SS^{i,i'} \rho\ot \rho\right)= \tr(\rho_{X}^2)\ , 
\end{equation}
where $\rho_X = \tr_{[d]\setminus \X}(\rho)$ is the reduction of $\rho$ to particles in $X$. From the symmetry of $\rho$ we have $\tr(\rho_{X}^2) = \tr(\rho_{k}^2)$, where $k=|X|$ (size of the set $X$). Inserting this into \eqref{eq:prodEXPAND} and \eqref{eq:beforePROD} we finally obtain
\begin{equation}
\tr(\Pfer \rho\otimes\rho )=\frac{1}{n+1}\sum_{k=0}^{n} {n \choose k} \tr(\rho_k^2)\ .
\end{equation}
The formula for the dimension \eqref{eq:passiveDIM} follows from the fact that the Hilbert space $\tilde{\Hfer}$ is a carrier space of an irreducible representation of $\U(d)$ labelled by a Young diagram having two columns each of which has $n$ rows. Formulas for dimensions of such irreducible representations are known (see for example \cite{DiagYoung2005}) and were used previously in the context of detection of mixed states that cannot be decomposed as a convex combination of Slater determinants \cite{TypicalityMO2014}.
\end{proof}

Note that for all $n$-particle pure states we have $\tr(\rho_k^2)=\tr(\rho_{n-k}^2)$. This observation gives us the following    
\begin{corr}\label{corr:simplifiedPASSIVEflo}
Let $\Psi\in\D\left(\bigwedge^{2m}(\C^d)\right)$ be a pure state. Let $\Pfer$ be defined as in Lemma \ref{lem:PROJpassFERMpurities}. We then have
\begin{equation}
\tr(\Pfer \Psi\otimes\Psi ) = \frac{1}{2m+1}\left[ 2 \sum_{k=0}^{m-1} {2m \choose k} \tr(\rho_{k}^2) +{2m \choose m} \tr(\rho_m^2) \right]\ .
\end{equation}  
\end{corr}

We now proceed with some further technical results which will allow us to compute $\tr(\Pfer \inrho \ot \inrho)$. 

For a set of indices $\X =\{x_1, x_2, \ldots, x_n \} \subset [d]$ where $x_i < x_j$ if $i < j$, and a subset of it $\S=\{ x_{\ell_1}, x_{\ell_2}, \ldots , x_{\ell_k} \} \subset \X$  it will be useful to introduce the following sign
\begin{equation}
    (-1)^{J(\X, \S)} \ , \; \textrm{where} \; \;  J(\X, \S) =  \ell_1 + \ell_2 + \ldots + \ell_k + \frac{k(k-1)}{2} \, .
\end{equation}
This notation allows us to express in a compact way the following matrix element: for any two Fock basis states $\ket{\X}, \ket{\Y} \in \bigwedge^{n}\C^d$ belonging to the index sets $\X, \Y \subset [d]$, we have that 
\begin{equation} \label{eq:matrix_element}
\bra{\X}  f_{s_1}^{\dagger}  f_{s_2}^{\dagger}  \cdots  f_{s_k}^{\dagger} f_{q_k}  \cdots  f_{q_2}f_{q_1}  \ket{\Y} = 
\begin{cases}
\delta_{ \X \setminus \S,  \Y \setminus \Q} \,  \epsilon_{s_1, ... , s_k} \epsilon_{q_1,..., q_k} (-1)^{J(\X, \S) + J(\Y, \Q)}\; \; \textrm{if} \; \; \S \subset \X  \, \textrm{and}  \; \Q \subset \Y, 
\\
0 \; \;  \textrm{else},
\end{cases}
\end{equation}
where $\S= \{s_1, s_2\ldots s_k \}$ and $\Q=\{q_1,q_2, \ldots q_k \}$.

\begin{prop}\label{prop:PARTtraceSLATER}
Let  $\ket{\X},\ket{\Y} \in \bigwedge^n(\C^d)$ be a fermionic $n$-particle Fock states corresponding to $n$-element subsets $\X,\Y\subset [d]$ (cf. notation introduced in Section~\ref{sec:notations}), then for any $k=0,\ldots,n$ we have
\begin{equation}\label{eq:partTRACEmatrixELEM}
\tr_k(\ketbra{\X}{Y})=\frac{1}{{n \choose k}} \sum_{\S\in {\X\cap \Y \choose k}} (-1)^{J(\X, \S) + J(\Y, \S)} \ketbra{\X\setminus \S}{\Y\setminus \S}\ .
\end{equation} 
Note that the notation used in the  above expression implies $\tr_k(\ketbra{\X}{\Y})=0$ if $|\X\cap\Y|<k$.
\end{prop} 

\begin{proof}
For any two states $\ket{\Psi}, \ket{\Phi} \in \bigwedge^n \C^d \subset (\C^d)^{\otimes n}$, the $k$-fold partial trace (wrt the tensor product structure) results in an operator 
$O=\tr_k(\ketbra{\Psi}{\Phi}) \in \mathcal{B}(\bigwedge^{\ell} \C^d) \subset \mathcal{B}((\C^d)^{\otimes \ell})$ (with $\ell=n-k$) that 
has the following matrix elements \cite{coleman2000reduced}:
\begin{align}
    \bra{v_1} \otimes \bra{v_2} \otimes \cdots \bra{v_{\ell}} \, O \, \ket{w_1} \otimes  \ket{w_2} \otimes \cdots \otimes \ket{w_{\ell}} 
    = \frac{1}{{n \choose k}} \bra{\Phi}  (f_1^{\dagger})^{v_1} (f_2^{\dagger})^{v_2} \cdots  (f_{\ell}^{\dagger})^{v_{\ell}} (f_{\ell}^{\phantom{\dagger}})^{w_\ell} \cdots (f_2^{\phantom{\dagger}})^{w_2} (f_1^{\phantom{\dagger}})^{w_1}  \ket{\Psi}\,.
\end{align}
Inserting in this equation the $\ket{\Psi}=\ket{\X}$ and $\ket{\Phi}= \ket{\Y}$ and using Eq.~\eqref{eq:matrix_element}, we get that 
\begin{align}
 & \bra{v_1} \otimes \bra{v_2} \otimes \cdots \bra{v_{\ell}} \, O \, \ket{w_1} \otimes  \ket{w_2} \otimes \cdots \otimes \ket{w_{\ell}} = \nonumber \\
& \qquad \qquad \qquad \begin{cases}
{n \choose k}^{-1} \, \delta_{ \Y \setminus \A,  \X \setminus \B} \,  \epsilon_{v_1, ... , v_k} \epsilon_{w_1,..., w_k} (-1)^{J(\Y, \A) + J(\X, \B)}\; \; \textrm{if} \; \; \A \subset \Y  \,  \textrm{and } \; \B \subset \X, 
\\
0 \; \;  \textrm{else},
\end{cases} \label{eq:matrix_el_1}
\end{align}
where $\v = (v_1, \ldots, v_\ell) $ and $\w=(w_1, \ldots, w_\ell)$ are the indicator bit strings of the sets $\A$ and $\B$, respectively.
Now considering also the following matrix entries
\begin{align}
   & \bra{v_1} \otimes \bra{v_2} \otimes \cdots \bra{v_{\ell}} \, \Big( \frac{1}{{n \choose k}} \sum_{\S\in {\X\cap \Y \choose k}} (-1)^{J(\X, \S) + J(\Y, \S)} \ketbra{\X\setminus \S}{\Y\setminus \S}  \Big) \, \ket{w_1} \otimes  \ket{w_2} \otimes \cdots \otimes \ket{w_{\ell}} = \nonumber \\
   & \tr \Big( \frac{1}{{n \choose k}} \sum_{\S\in {\X\cap \Y \choose k}} (-1)^{J(\X, \S) + J(\Y, \S)} \ketbra{\X\setminus \S}{\Y\setminus \S}  (f_1^{\dagger})^{v_1}  \cdots  (f_{\ell}^{\dagger})^{v_{\ell}} (f_{\ell}^{\phantom{\dagger}})^{w_\ell} \cdots  (f_1^{\phantom{\dagger}})^{w_1} \Big) = \nonumber \\
  &  \frac{1}{{n \choose k}} \sum_{\S\in {\X\cap \Y \choose k}} (-1)^{J(\X, \S) + J(\Y, \S)} \bra{\Y\setminus \S}  (f_1^{\dagger})^{v_1}  \cdots  (f_{\ell}^{\dagger})^{v_{\ell}} (f_{\ell}^{\phantom{\dagger}})^{w_\ell} \cdots  (f_1^{\phantom{\dagger}})^{w_1}  \ket{\X\setminus \S}= \nonumber \\
   &  \frac{1}{{n \choose k}} \sum_{\S\in {\X\cap \Y \choose k}} (-1)^{J(\X, \S) + J(\Y, \S)}  \,    \delta_{\Y \setminus \S, \A} \delta_{\X \setminus \S, \B} \, \epsilon_{v_1, ... , v_k} \epsilon_{w_1,..., w_k}= \nonumber \\
&    
\frac{1}{{n \choose k}} \sum_{\S\in {\X\cap \Y \choose k}} (-1)^{J(\X, \S) + J(\Y, \S)}  \,    \delta_{\Y \setminus \A, \S} \delta_{\Y \setminus \B, \S} \, \epsilon_{v_1, ... , v_k} \epsilon_{w_1,..., w_k}= \nonumber \\
&
\qquad \qquad \qquad \qquad \qquad 
\begin{cases}
{n \choose k}^{-1} \, \delta_{ \Y \setminus \A,  \X \setminus \B} \,  \epsilon_{v_1, ... , v_k} \epsilon_{w_1,..., w_k} (-1)^{J(\Y, \A) + J(\X, \B)}\; \; \textrm{if} \; \; \A \subset \Y  \,  \textrm{and } \; \B \subset \X, \\
0, \; \; \; \textrm{else},
\end{cases} \label{eq:matrix_el_2}
\end{align}
where we have used that $(-1)^{J(\Y, \A) + J(\X, \B)} = (-1)^{J(\Y, \S) + J(\X, \S)} $, which follows from the fact that $(-1)^{J(\X, \S)}= (-1)^{J(\X, \X \setminus \S) + |\X|\cdot |\S| }$.
Thus, the matrix elements of Eq.~\eqref{eq:matrix_el_1} and Eq.~\eqref{eq:matrix_el_2} coincide, which proves the propositions.
\end{proof}

We introduce the convenient notation for $\ket{\inpsi}$: 
\begin{equation}\label{eq:pureSTATEdecomp}
    \ket{\inpsi}=\frac{1}{\sqrt{2^N}}\sum_{\X\in\Cstate} \ket{\X}\ ,
\end{equation}
where $\Cstate$ is a collection of subsets of $[4N]$ that appear in the decomposition of $\ket{\inpsi}$. Note that from the definition of $\ket{\inpsi}$ 
it follows that subsets are labelled by bitstrings $\x=(x_1,\ldots, x_N)$, where $x_i\in\lbrace 0,1\rbrace$ labels which pair of the neighbouring physical modes are occupied in a given quadropule of modes. For $N=2$ we have four possible subsets belonging to $\Cstate$
\begin{equation}\label{eq:expQUADpassive}
    \X_{00}=\lbrace 1,2,5,6\rbrace \ ,\   \X_{01}=\lbrace 1,2,7,8\rbrace \ ,\ \X_{10}=\lbrace 3,4,5,6\rbrace ,\ \X_{11}=\lbrace 3,4,7,8\rbrace\ .
\end{equation}
For general $N$ the collection $\Cstate$ consists of the following subsets labelled by bitstrings $\x$
\begin{equation}
    \X_\x =\lbrace 1+2x_1,2+2x_1,5+2x_2,6+2x_2,\ldots, 4i-3+2x_i,4i-2+2x_i,\ldots,4N-3+2x_N,4N-2+2x_N \rbrace\ .   
\end{equation}

The formula from Lemma \ref{prop:PARTtraceSLATER} allows us to obtain bounds for the purities of reduced density matrices of $\inpsi$.

\begin{prop}[Bounds on purites of reduced density matrices of $\inrho$]\label{prop:puritiesINP}
Consider the setting of this paper, i.e., $d=4N$ and $n=2N$, where $N$ is the number of quadruples used in our quantum advantage proposal. Let $\inpsi\in \D(\Hfer)$  be the input state. Then, for $k=0,\ldots,N$ we have
\begin{equation}
\tr\left[\tr_{k} (\inpsi)^2   \right]\leq \frac{1}{{2N \choose k}^2} \sum_{l=0}^{\lfloor k/2 \rfloor} \frac{N!}{l!(k-2l)!(N-k+l)!} 
\end{equation}  
\end{prop}
\begin{proof}
We  use the decomposition of the state vector $\ket{\inpsi}$ given in Eq. \eqref{eq:pureSTATEdecomp} and obtain
\begin{equation}
\inpsi = \frac{1}{2^N} \sum_{\X,\Y\in \Cstate} \ketbra{\X}{\Y}\ .
\end{equation}
Employing \eqref{eq:partTRACEmatrixELEM} and denoting ${J(\X, \S) + J(\Y, \S)}=K(\X,\Y,\S)$ we obtain (remember that $n=2N$)
\begin{equation}\label{eq:partialTRACEexpansion}
\tr_{k} (\inpsi) = \frac{1}{2^N {2N \choose k} } \sum_{\X,\Y\in \Cstate} \sum_{\S\in {\X\cap\Y \choose k}} (-1)^{K(\X,\Y,\S)} \ketbra{\X\setminus \S}{\Y\setminus \S}\ .
\end{equation}
By reordering the sum we obtain 
\begin{equation}\label{eq:semiNightmare}
  \tr_{k} (\inpsi) = \frac{1}{2^N {2N \choose k} }   \sum_{\X',\Y'\in \binom{[4N]}{2N-k}} \ketbra{\X'}{\Y'} \sum_{ \substack{\S \in \binom{[4N]}{k}\ ,\ \X,\Y\in\Cstate \\ \text{s.t. } \X\setminus\S =\X',   \Y\setminus\S=\X'     }  }  (-1)^{K(\X,\Y,\S)}\ ,
\end{equation}
where the second sum is a combinatorial term that gives a coefficient which which particular operator $\ketbra{\X'}{\Y'}$ appears. Crucially, operators $\ketbra{\X'}{\Y'}$, $|\X'|=|\Y'|=2N-k$ are orthonormal with respect to the Hilbert-Schmidt inner product in $\mathcal{B}(\bigwedge^{2N-k}(\C^{4N}))$. Therefore in order to bound purity of $\tr_{k} (\inrho) $ it suffices to count the number of terms in the second sum in \eqref{eq:semiNightmare}: 
\begin{equation}\label{eq:purityUP1}
   \tr\left[\tr_{k} (\inpsi)^2\right] \leq \frac{1}{2^{2N} {2N \choose k}^2 }  \sum_{\X',\Y'\in \binom{[4N]}{2N-k}} \N(\X',\Y')^2\ ,
\end{equation}
where
\begin{equation}\label{eq:normCONSTpass}
  \N(\X',\Y')=  \left|\SET{(\S,\X,\Y)}{\S \in \binom{[4N]}{k}\ ,\ \X,\Y\in\Cstate ,\  \X\setminus\S =\X',\   \Y\setminus\S=\X' }\right|\ .
\end{equation}
In what follows, in order to make our considerations less abstract, we shall refer to lements of subsets involved as "particles". To compute $N(\X',\Y')$ we note that $\X',\Y'$ for which $\N(\X',\Y')\neq0$ must arise from substracting from $\X\in \Cstate$ particles occupying subset $\S$. Since particles corresponding to $\X\in\Cstate$ occupy only two out of four possible modes in every quadropule of modes in a "binary fashion" (See Eq. \eqref{eq:expQUADpassive}), This imposes constraints on the possible configurations of particles from $\X'$ in every quadropole. Specifically, consider the quadropule of physical modes $\A=\{1,2,3,4\}$. Let $\X'_\A=\X'\cap\A$. We have seven possibilities for the set $\X'_\A$:
\begin{equation}\label{eq:occCONS1}
    \X'^{\text{N1}}_\A = \{1,2\}\ ,\  \X'^{\text{N2}}_\A = \{3,4\}\ ,\ \X'^{\text{B}}_\A = \emptyset\ ,
\end{equation}
\begin{equation}\label{eq:occCONS2}
      \X'^{\text{F1}}_\A =  \{1\}\ ,\  \X'^{\text{F2}}_\A =  \{2\}\ ,\   \X'^{\text{F3}}_\A =  \{3\}\ ,\  \X'^{\text{F4}}_\A =  \{4\}\ . 
\end{equation}
All other forms of $\X\cap\A$ yield $\N(\X',\Y')=0$. Under the condition that $\X'$ originates from $\X\in\Cstate$ these configurations impose conditions on possible arrangement of lost particles in quadruple $\A$, denoted by $\S_\A=\S\cap\A $:
\begin{equation}
    \S_\A(N1)=\S_\A(N2)=\emptyset\  ,\  \S_\A(B)=  \{1,2\}\ \text{or}\  \S_\A(B)=  \{3,4\}\ , 
\end{equation}
\begin{equation}
     \S_\A(F1) =  \{2\}\ ,\  \S_\A(F2) =  \{1\}\ ,\   \S_\A(F3) =  \{4\}\ ,\  \S_\A(F4) =  \{3\}\ . 
\end{equation}
This motivates us to introduce \emph{type} of quadruples of $\X'$ whose names are motivated by types of constraints the impose on $\S\cap\A$: 
\begin{equation}
    T_\A(\X')=\begin{cases}
       \mathrm{NULL}\  &\text{iff} \ \X'_\A=\{1,2\}$ or $\X'_\A=\{3,4\} \\
       \mathrm{BINARY}\  &\text{iff}\  \X'_\A=\emptyset  \\
       \mathrm{FIXED}\  &\text{iff}\   \X'_\A\in\{\{1\},\{2\},\{3\},\{4\}  \}
    \end{cases}\ .
\end{equation}
We repeat the same procedure for other quadruples $\{5,6,7,8\}$, $\{9,10,11,12\}$,  etc. To a given $\X'$ we then associate "pattern of types": 
\begin{equation}\label{eq:patTYPES}
    \X' \longmapsto \L(\X') = \left(\lnull[\X']\ ,\ \lbin[\X']\ ,\ \lfix[\X']\right)\ ,
\end{equation}
that lists the number of quadruples of different types in $\X'$. This pattern gives us the number of $k$-element subsets  $\S\in\binom{[4N]}{k}$ contributing to $\N(\X',\Y')$ (cf. \eqref{eq:normCONSTpass}). From the considerations given previously $\N_\S(\X')=2^{\lbin[\X']}$ different $\S$ that contribute. Let us chose $\Y'$ that is compatible with the pattern of lost particles in $\X'$ is the sense that $\X'\cap \Y'=\emptyset$ and $\Y'$ follows the general constrains of occupations in each quadruples  described previously (like the ones stated in Eq.\eqref{eq:occCONS1} and Eq.\eqref{eq:occCONS2}). Since for fixed $\X',\Y'$ subset $\S$ uniquely specifies $\X,\Y\in\Cstate$, we finally get 
\begin{equation}
    \N(\X',\Y') = 2^{\lbin[\X']}\ .
\end{equation}
It is now easy to see that, for $\X'$ characterized by particular $\L(\X')$, there are exactly
\begin{equation}
    \N_{\mathrm{comp}}(\X')=2^{\lnull[\X']}
\end{equation}
different compatible sets $\Y'$. In fact, compatible $\Y'$ necessarily satisfy $\L(\Y')=\L(\X')$. Finally,  simple counting argument shows that there are
\begin{equation}
    \N(\L)=2^{2\lfix+\lnull}\frac{N!}{\lnull!\lbin!\lfix!} 
\end{equation}
different subsets $\X'$ that have the "pattern type" $\L=(\lnull,\lbin,\lfix)$. Hence the contribution in  from subsets $\X',\Y'$ of "pattern type" $\L=(\lnull,\lbin,\lfix)$ to the sum in Eq.\eqref{eq:purityUP1} is equals
\begin{equation}
   \N(\X',\Y')^2 \N_{\mathrm{comp}}(\X')  \N(\L) =  4^{\lnull+\lbin+\lfix}   \frac{N!}{\lnull!\lbin!\lfix!} =2^{2N} \frac{N!}{\lnull!\lbin!\lfix!}\ .
\end{equation}
Parameters $\lnull,\lbin,\lfix$ are not independent because of identities: $N=\lnull+\lbin+\lfix$ (this one we already used implicitly) and $2\lbin+\lfix=2N-k$. Choosing $\lbin$ as an independent parameter applying the above considerations to \eqref{eq:purityUP1} we finally obtain
\begin{equation}\label{eq:purityUP2}
   \tr\left[\tr_{k} (\inpsi)^2\right] \leq \frac{1}{{2N \choose k}^2 }  \sum_{\lbin=0}^{\lfloor\frac{k}{2}\rfloor} \frac{N!}{(N-k+\lbin)!(k-2\lbin)!\lbin!}\ ,
\end{equation}
where summation range for $\lbin$ comes from its definition as the number of quadropules in $\X'$ that are are left without particles.
\end{proof}

Combining Corollary \ref{corr:simplifiedPASSIVEflo} and Proposition \ref{prop:puritiesINP} we obtain explicit upper bound for the expectation value of the projector $\Pfer$

\begin{lem}\label{lem:passPROJfinal}
Consider the setting of our quantum advantage proposal, i.e., $d=4N$ and $n=2N$.  Let $\inpsi\in \D\left( \bigwedge^{2N}(\C^{4N})\right)$. Let $\Pfer$ be defined as in Lemma \ref{lem:PROJpassFERMpurities}. 
We then have 
\begin{equation}\label{eq:passFINALexpression}
    \tr\left(\Pfer \inpsi\ot\inpsi  \right)\leq \frac{1}{2N+1}\left[ 2 \sum_{k=0}^{N-1}  {2N \choose k} \tr(\rho_{k}^2) +{2N \choose N} \tr(\rho_N^2) \right]\ , 
\end{equation}
where 
\begin{align}\label{eq:passive-purity}
    \tr(\rho_k^2) &= \frac{1}{{2N \choose k}^2} \sum_{l=0}^{\lfloor k/2 \rfloor} \frac{N!}{l!(k-2l)!(N-k+l)!}\ . 
\end{align}
\end{lem}

\subsection{Active FLO}

We give here computations related to establishing upper bound in Eq. \eqref{eq:EPupperbound}
for the case of active FLO:
\begin{equation}\label{eq:upperboundACTapp}
 \tr(\Pflo \inpsi\ot \inpsi ) \leq \frac{\Cact}{\sqrt{\pi N}}\  ,\ \text{for } \Cact = \Cactval
\end{equation}
Similarly to the case of passive FLO the proof  divided into three parts. First,  in Lemma \ref{lem:PROJactive} we give an explicit form of $\Pflo$. Second, in Lemma  \ref{lem:actPROJfinal} we find an upper bound on $\tr(\Pflo \inpsi\ot \inpsi)$ via combinatorial expression that can be efficiently computed for any fixed value of $N$. Finally, in Lemma \ref{lem:actAnti-computation} given in Part \ref{app:comp} of the Appendix we prove an upper bound to the said  expression which yields Eq. \eqref{eq:upperboundACTapp}.

Recall that by $m_i$, $i=1,\ldots,2d$ we denoted the standard majorana operators in the $d$ mode Fermionic Fock space $\Hfock(\C^d)$ (cf. Section \ref{sec:notations}). The fermionic parity operator is given by $Q=i^d \prod_{i=1}^{2d} m_i$.

\begin{lem}[Projector for active fermionic linear optics]\label{lem:PROJactive}
Let $\Hact=\Hfock^+\left(\C^d \right)$ be the positive parity subspace of Fock space corresponding to $d$ fermionic modes. Let $\Pflo$ be the projector onto a unique irreducible representation $\tHact\subset \Hact \ot \Hact$  of $\SO(2d)$ such that $\ket{\Phi} \ot \ket{\Phi}\in \tHact$, where $\Phi$ are arbitrary pure positive parity Gaussian states. We then have
\begin{equation}\label{eq:PactProj}
    \Pflo= \P_+ \ot \P_+ \P_0 \P_+ \ot \P_+ \ , 
\end{equation}
where $\P_+= \frac{1}{2}(\I+Q)$ is the orthogonal projector onto $\Hfock^+(\C^d)\subset\Hfock(\C^d) $ and
\begin{equation}\label{eq:explicitFLO}
     \P_0 = \frac{1}{2^{2d}} \sum_{p=0}^{d} C_p \sum_{\X\in\binom{[2d]}{2p}}  \prod_{i\in\X} m_i \ot m_i\ .
\end{equation}
The numbers  $C_p$  satisfy $C_p=(-1)^d C_{d-p}$ and for $p\leq \lfloor d/2 \rfloor$ we have
\begin{equation}\label{eq:activeC_k}
		C_p    =  (-1)^{p}  \frac{(2p)!(2d-2p)!}{(d!)^2} 
		{d \choose p} 
\end{equation}
Moreover, the dimension $\tHact$ equals
\begin{equation}\label{eq:dimHact}
    |\tHact|=\frac{1}{2}\binom{2d}{d}\ .
\end{equation}
\end{lem}
\begin{proof}
The result follows from the characterization of pure fermionic Gaussian states given in Corollary 1 in: \cite{melo_power_2013} which states that a pure state $\Psi$ is a pure Fermionic Gaussian state if and only if $\Lambda \ket{\Psi} \ot \ket{\Psi}=0$, where $\Lambda$ is the operator acting on $\Hfree(\C^d) \ot \Hfree(\C^d)$ introduced previously by Bravyi in \cite{bravyi_lagrangian_2004}
\begin{equation}\label{eq:LAMBDAdef}
    \Lambda=\sum_{i=1}^{2d} m_i \ot m_i\ .
\end{equation}
Equivalently,  $\Psi$ is  pure Fermionic Gaussian state iff $\P^\Lambda_0 \ket{\Psi}\ot \ket{\Psi}=0$,  where $\P^\Lambda_0$ is the projector onto zero eigenspace of $\Lambda$. Since we are interested in Gaussian states having positive parity ($Q\ket{\Psi}=\ket{\Psi}$), and operators $Q\ot\I $, $\I\ot Q$ commute with $\Lambda$, we get the following equivalence
\begin{equation}
    \Psi\ \text{is pure positive parity fermionic Gaussian state} \Longleftrightarrow \P_+ \ot \P_+ \P^\Lambda_0 \P_+ \ot \P_+ \ket{\Psi}\ot \ket{\Psi} =0\ .  
\end{equation}
We now show that $\P^\Lambda_0= \P_0$. The operator $\Lambda$ from \eqref{eq:LAMBDAdef} is a sum of $2d$ commuting hermitian operators $m_i\ot m_i$ which satisfy $(m_i\ot m_i)^2=\id\ot \id$. A one dimensional projector onto a joint eigenspace of $M_i$ corresponding to eigenvalues $\mu_i$, $i\in[2d]$ reads is given by
\begin{equation}
    \P_{\muu} = \frac{1}{2^{2d}} \prod_{i=1}^{2d}\left(\id\ot\id + \mu_i m_i \ot m_i \right)\ ,
\end{equation}
where $\muu= (\mu_1,\mu_2,\ldots,\mu_{2d})\in\{-1,1\}^{2d}$. Any arrangement of eigenvalues $\muu$ corresponds to eigenvalue $\lambda=\sum_{i=1}^{2d} \mu_i$. Consequent, projector onto eigenspace zero of $\Lambda$ reads
\begin{equation}\label{eq:PlambdaSUM1}
    \P^{\Lambda}_0 = \sum_{\substack{\muu\in\{-1,1\}^{2d}  \\ \sum_{i=1}^{2d} \mu_i =0  }} \P_{\muu}  \ . 
\end{equation}
Expanding each of the projectors $ \P_{\muu}$ into sum of products of Majorana monomials gives
\begin{equation}
    \P_{\muu} = \frac{1}{2^{2d}} \sum_{k=0}^{2d} \sum_{\X\in \binom{[2d]}{k}} {\muu}^{\X} \prod_{i\in \X} m_i \ot m_i \ ,    
\end{equation}
where we have defined ${\muu}^{\X}= \prod_{i\in \X}\mu_i$. Inserting this expression to \eqref{eq:PlambdaSUM1} gives
\begin{equation}\label{eq:PlambdaSUM2}
     \P^{\Lambda}_0 = \frac{1}{2^{2d}} \sum_{k=0}^{2d} \sum_{\X\in \binom{[2d]}{k}} A_\X \prod_{i\in \X} m_i \ot m_i \ ,
\end{equation}
with
\begin{equation}
    A_\X = \sum_{\substack{\muu\in\{-1,1\}^{2d}  \\ \sum_{i=1}^{2d} \mu_i =0  }} {\muu}^{\X}\ .
\end{equation}
Every $\muu\in\{-1,1\}^{2d}$ can be identified with a subset $\Y_{\muu}\subset [2d]$ defined by $\Y_{\muu}=\SET{i}{\mu_i=-1}$. Under this identification $\muu^{\X}=(-1)^{|\X\cap \Y_{\muu}|} $. Consequently we obtain
\begin{equation}
     A_\X = \sum_{\Y\in\binom{[2d]}{d}} (-1)^{|\X\cap \Y|} \ . 
\end{equation}
Let us first observe that because $ (-1)^{|\X\cap \Y|}=(-1)^d (-1)^{|\bar{\X}\cap \Y|}$, for $\bar{\X}=[2d]\setminus \X$ and $|Y|=d$, we have $C_\X=(-1)^d C_{\bar{\X}}$.  Assuming $|X|=k\leq d$ we get
\begin{equation}
     A_\X  = \sum_{l=0}^k (-1)^l \sum_{\substack{\Y\in\binom{[2d]}{d} \\ |\X\cap \Y|=l  }}  =  \sum_{l=0}^k (-l)^l {k \choose l}{2d-k \choose d-l}\ ,
\end{equation}
where in to get the second equality we counted the number of sets $\Y\in\binom{[2d]}{d}$ satisfying $|\X\cap \Y|=l$, where $|X|=k\leq d$. Since we $A_\X$ depends only on $|\X|$ we will sue the notation denoting $A_\X=A_{|\X|}$. USing simple algebra we obtain  
\begin{align}
	A_k = \sum_{l=0}^k (-1)^l {k \choose l}{2d-k \choose d-l} 
	&= \frac{k!(2d-k)!}{(d!)^2} \sum_{l=0}^k (-1)^l
	{d \choose l}{d \choose k-l}  \ .
\end{align}
This can be further simplified using the identity 
\begin{equation}
    \sum_{l=0}^k (-1)^l {d \choose l}{d \choose k-l} =  \begin{cases} (-1)^{k/2}
		{d \choose k/2} &\mbox{if } k\ \text{is even}  \\
0 & \mbox{if } k\ \text{is odd}  \end{cases} \ .
\end{equation}
Denoting $A_{2p}= C_p$ and using $C_\X=(-1)^d C_{\bar{\X}}$ we observe that $C_{d-k}=C_k$. Inserting the expression for $A_\X$ to \eqref{eq:PlambdaSUM2} we finally obtain the desired result:
\begin{equation}
     \P^\Lambda_0 = \frac{1}{2^{2d}} \sum_{p=0}^{d} C_p \sum_{\X\in\binom{[2d]}{2p}}  \prod_{i\in\X} m_i \ot m_i\ .
\end{equation}
where  $C_p$  satisfy $C_p=(-1)^d C_{d-p}$ and for $p\leq \lfloor d/2 \rfloor$ 
\begin{equation}
		C_p    =  (-1)^{p}  \frac{(2p)!(2d-2p)!}{(d!)^2} 
		{d \choose p} \ . 
\end{equation}
We thus established $ \Pflo= \P_+ \ot \P_+ \P_0 \P_+ \ot \P_+ $. The dimension of the subspace on which $\Pflo$ projects, $|\tHact|$,  can be now computed as $\tr(\Pfer)$ by using standard algebraic properties of Majorana operators.

\end{proof}

In order to proof the following lemma we use explicit form of $\Pfer$ to compute $\tr(\Pfer \inpsi \ot \inpsi)$. 

\begin{lem}\label{lem:actPROJfinal}
Consider the setting of our quantum advantage proposal, i.e., $d=4N$ and $n=2N$.  Let $\inpsi \in \D\left( \Hfock^+(\C^{4N})\right)$. Let $\Pflo$ be a projector specified in Lemma \ref{lem:PROJactive}. 
We then have 
\begin{equation}\label{eq:actFINALexpression}
    \tr(\Pflo \inpsi \ot \inpsi) = \frac{1}{2^{8N}} \left[
    2 \sum_q^{N-1} C_{2q}
    \sum_{l=0}^{\lfloor \frac{q}{2} \rfloor} \frac{N!}{l!(q-2l)!(N-q+l)!} 14^{q-2l} + C_{2N} \sum_{l=0}^{N} \frac{N!}{(l!)^2(4N-2l)!} 14^{N-2l}
    \right].
\end{equation}
where 
\begin{equation}
    C_{2q}= \frac{(4q)!(8N-4q)!}{((4N)!)^2}  {4N \choose 2q} \ .
\end{equation}
\end{lem}

\begin{proof}

We start be observing that due to FLO invariance of $\Pflo$ we have $\tr(\Pflo \inpsi \ot \inpsi)=\tr(\Pflo \inpsi' \ot \inpsi')$, where $\inpsi'=V\inpsi V^\dag$, for $V\in\Gact$. Note that by applying $V=\prod_{1=1}^{N}m_{3i-1} m_{4i-i}$ we can transform $\ket{\inpsi}=\psiQUAD^{\ot N}$ (recall that $\psiQUAD=\frac{1}{\sqrt{2}}(\ket{0011}+\ket{1100})$ into $\ket{\inpsi'}=\ket{a_8}^{\ot N}$, where 
$\ket{a_8}=\frac{1}{\sqrt{2}}(\ket{0000}+\ket{1111})$ is the state that was used considered previously by Bravyi in the context of magic state injection for model of computation based on Ising anyons \cite{bravyi_universal_2006} (see also \cite{melo_power_2013,oszmaniec_classical_2014}). The state $\ketbra{a^\B_8}{a^\B_8}$ on  octet  of normally-ordered Majorana modes denoted by $\B\subset[2d]$  can be decomposed using Majorana monomials
\begin{equation}\label{eq:A8inSUBSET}
    \ketbra{a^\B_8}{a^\B_8}=\frac{1}{2^4} \left(\mathbb{I}+Q^\B + A^{\B}_1 + A^{\B}_2 +\ldots + A^{\B}_{14} \right)\ , 
\end{equation}
where $Q^{\B}=\prod_{i\in B} m_i$ and operators $ A^{\B}_i$, $i=1,\ldots,14$ are quartic (i.e. fourth other)  Majorana monomials  supported on modes belonging to $\B$ ant satisfying $(A^\B_i)^2=\id$. We will not need explicit form of $\ketbra{a_8}{a_8}$ but it can be found in the works cited above). The algebraic framework of Majorana fermion operators allows us to write the equivalent input state $\ketbra{\inpsi'}{\inpsi'}$ as a product (in a standard operator sense) of states $\ketbra{a_8}{a_8}$ supported on disjoint octets of modes
\begin{equation}\label{eq:prodA8}
    \ketbra{\inpsi'}{\inpsi'}= \prod_{i=1}^{N}  \ketbra{a^{\B_i}_8}{a^{\B_i}_8}\ ,
\end{equation}
where $\B_1=\{1,2,\ldots,8\}$, $\B_1=\{1,2,\ldots,8\}$,$\B_2=\{9,10,\ldots,16\}$, etc. We proceed similarly as in the proof of Proposition \ref{prop:puritiesINP} and expand the above expressions into product of majorana monomials and obtain 
\begin{equation}\label{eq:PrimeState}
    \ketbra{\inpsi'}{\inpsi'} =\frac{1}{2^{4N}} \sum_{\X\in\Ca} (-1)^{F(\X)}  \prod_{i\in \X} m_i\ ,
\end{equation}
where $\Ca$ is a collection of subsets of $8N$ Majorana modes  modes that appear in the product expansion of $\ketbra{\inpsi'}{\inpsi'}$ and $(-1)^{F(\X)}$ a sign possibly depending on a subset $\X$. Because $\ket{\inpsi'}\in\Hfock^+(\C^d)$ and the form projector $\Pflo$ (cf. Eq. \eqref{eq:PactProj}) we have $\tr\left( \ketbra{\inpsi'}{\inpsi'}^{\ot 2} \Pflo\right)= \tr\left(\ketbra{\inpsi'}{\inpsi'}^{\ot 2} \P_0\right)$, where $\P_0$ is given in \eqref{eq:explicitFLO}. Combining \eqref{eq:PrimeState} with  \eqref{eq:explicitFLO} gives
\begin{align}
\tr\left(\ketbra{\inpsi'}{\inpsi'}^{\ot 2} \P_0\right) = \frac{1}{2^{16N}} \sum_{p=0}^{4N} C_p  \sum_{\Z\in\binom{[8N]}{2p}}   \sum_{\X,\Y\in\Ca} (-1)^{F(\X)+F(\Y)} \tr\left[  \prod_{i\in \X} m_j \ot \prod_{k\in \Y} m_i  \prod_{k\in\Z} m_k \ot m_k \right]\ .
\end{align}
 Using 
 \begin{equation}
 (-1)^{F(\X)+F(\Y)} \tr\left[  \prod_{i\in \X} m_j \ot \prod_{k\in \Y} m_i  \prod_{k\in\Z} m_k \ot m_k \right] = 2^{8 N} \delta_{\X,\Y} \delta_{\Z,\X}   
 \end{equation}
 we obtain 
 \begin{equation}
 \tr\left(\ketbra{\inpsi'}{\inpsi'}^{\ot 2} \P_0\right)  = \frac{1}{2^{8N}} \sum_{p=0}^{4N} C_p  \sum_{\X\in\binom{[8N]}{2p}\cap \Ca }\ . 
 \end{equation}
Recall that from the definition of $\Ca$, this collection of subsets of $[8N]$ consists only on subsets that are have cardinality divisible by $4$. Therefore the above can be described equivalently by  
  \begin{equation}
 \tr\left(\ketbra{\inpsi'}{\inpsi'}^{\ot 2} \P_0\right)  = \frac{1}{2^{8N}} \sum_{q=0}^{2N} C_{2q}  \sum_{\X\in\binom{[8N]}{4q}\cap \Ca }\ . 
 \end{equation}
Moreover, form $Q \ketbra{\inpsi'}{\inpsi'} =\ketbra{\inpsi'}{\inpsi'}$ we get  $\bar{\X}\in\Ca$ if and only if $\X\in \Ca$ and consequently
\begin{equation}
\sum_{\X\in\binom{[8N]}{4q}\cap \Ca } = \sum_{\X\in\binom{[8N]}{8N-4q}\cap \Ca }\ .
\end{equation} 
Using this and the property $C_{2q}= C_{8N-2q}$ we finally get
\begin{equation}\label{eq:simplifiedACTOVE}
 \tr\left(\ketbra{\inpsi'}{\inpsi'}^{\ot 2} \P_0\right)  = \frac{1}{2^{8N}} \sum_{q=0}^{N-1} 2\left( C_{2q} \sum_{\X\in\binom{[8N]}{4q}\cap \Ca }\right) +  
     \left(C_{2N} \frac{1}{2^{8N}} \sum_{\X\in\binom{[8N]}{4N}\cap \Ca } \right)   \ .
\end{equation}
Therefore, we have reduced the problem of computing $ \tr\left(\ketbra{\inpsi'}{\inpsi'}^{\ot 2} \P_0\right)$ (equal to $ \tr\left(\ketbra{\inpsi}{\inpsi}^{\ot 2} \P_0\right)$ ) to the problem of counting different sets of cardinality $4q$ ($q=0,1,\ldots,N$) one can find in $\Ca$. This problem can be tackled using similar technique to the one used in the proof of Proposition \ref{prop:puritiesINP} i.e. by introducing \emph{pattern of types} of subsets in $\Ca$. We have $8N$ Majorana modes in total. In what follows we shall refer to "standard octets" as $N$ disjoint octets on which states $\ket{a_8^{\B_i}}$ are supported in Eq. \eqref{eq:prodA8}.   A subset $\X\in\Ca$ satisfying $|\X|=4q$ ($q\leq N$) can be characterized, in analogy to  \eqref{eq:patTYPES}, by pattern of types, i.e. a triple  
 \begin{equation}
    \X \longmapsto \L(\X) = \left(\loct[\X]\ ,\ \lempty[\X]\ ,\ \lquad[\X]\right)\ ,
\end{equation}
where $\loct[\X]$ counts the number of standard octets contained in $\X$,  $\lempty[\X]$ counts how many standard octets are not populated by elements of  $\X$, and finally $\lquad[\X]$ is the number of octets in which $\X$ intersects only in four elements (note that from the construction of $\Ca$ and due to specific form of the state $\ketbra{a^\B_8}{a^\B_8}$ in Eq. \eqref{eq:A8inSUBSET} these are the only possibilities). With these concepts counting of sets $\X\in\Ca$ of carnality $4q$ can be done analogously as in Proposition \ref{prop:puritiesINP} i.e by counting how many sets of different "pattern of types" $\L(\X)$ of given carnality are there. The final results reads

\begin{equation}
\left|\SET{\X}{\X\in\binom{[8N]}{4q}\ ,\ \X\in \Ca  }\right| =\sum_{l=0}^{\lfloor \frac{q}{2} \rfloor} \frac{N!}{l!(q-2l)!(N-q+l)!} 14^{q-2l}\ , 
\end{equation}
where $l$ labels the number of possible "fully occupied" standard octets in $\X$ of carnality $4q$. The term $14^{q-2l}$ appears because for the said value of fully occupied octets there are necessarily $q-2l$ octets of quartic type, and every such octet there is exactly $14$ possibilities.   We conclude the proof by using the above identity in Eq. \eqref{eq:simplifiedACTOVE} and employing the explicit formula for $C_{2q}$ from \eqref{eq:activeC_k}.

\end{proof}

\subsection{Computation of the sums} \label{app:comp}

In this part we prove the bounds on the combinatorial sums appearing in Lemma \ref{lem:passPROJfinal}, Lemma \ref{lem:actPROJfinal}. This ultimately proves anticoncentration bounds for passive and active FLO circuits initialized in magic input states $\inpsi$ in Theorem \ref{thm:anticoncentrationFLO}.

Our general strategy for the analytical part will be based on the following tight inequalities satisfied by binomial and trinomial coefficients. 

\begin{lem}[Bounds for binomial and trinomial coefficients]\label{lem:binBounds}
Let $n,k$ be a natural numbers such that $k\in\{1,\ldots,n-1\}$. Let $x=\frac{k}{n}$. Then we have
\begin{equation}\label{eq:binBOUND}
    c\cdot \sqrt{\frac{n}{k(n-k)}} \exp\left(n\, h(x)\right) \leq \binom{n}{k} \leq C\cdot \sqrt{\frac{n}{k(n-k)}} \exp\left(n\, h(x)\right) \ ,
\end{equation}
where $c=\frac{1}{2\sqrt{2}}$, $C=\frac{1}{\sqrt{2\pi}}$, and $h(x)=-x\log(x) -(1-x)\log(1-x)$ is the binary entropy. 

Moreover, let $k,l,m$ be nonzero natural numbers such that $k+l+m=n$. Let $x=\frac{k}{n}, y=\frac{l}{n},z=\frac{m}{n}$. Then we have
\begin{equation}\label{eq:triBOUND}
     a\, \sqrt{\frac{n}{k\cdot l \cdot m}} \exp\left(n\, h(x,y,z)\right) \leq \binom{n}{k,l,m} \leq A\,  \sqrt{\frac{n}{k\cdot l \cdot m}} \exp\left(n\, h(x,y,z)\right)\ ,
\end{equation}
where $a=\frac{1}{8}$, $A=\frac{1}{2\pi}$ and $h(x,y,z)=-x \log (x) -y \log(y)-z \log(z) $ is the entropy of three-outcome probability distribution.
\end{lem}
The inequality \eqref{eq:binBOUND} can be found in Lemma 7 in Chapter 10 of \cite{MacWilliams} while \eqref{eq:triBOUND} follows from it due to identity $\binom{n}{k,l,m}=\binom{n}{k}\binom{l+m}{m}$.

We first consider the case of passive FLO. We observe that from Eq. \eqref{eq:passFINALexpression} it follows that 
\begin{equation}\label{eq:pasFINcombinatorics}
     \tr\left(\Pfer \inpsi\ot\inpsi  \right) \leq
     \frac{2}{2N+1}\sum_{k=0}^N \sum_{l=0}^{\lfloor k/2 \rfloor} \frac{\binom{N}{l,k-2l,N-k+l}}{\binom{2N}{k}}\ . 
\end{equation}

\begin{lem}\label{lem:pasAnti-computation}
Consider the setting of our quantum advantage proposal, i.e., $d=4N$ and $n=2N$.  Let $\inpsi\in \D\left( \bigwedge^{2N}(\C^{4N})\right)$.  Let $\Pfer$ be defined as in Lemma \ref{lem:PROJpassFERMpurities}. We  then have
\begin{align}
    \tr\left(\Pfer \inpsi\ot\inpsi  \right) \leq \frac{\Cpas}{N},\ \text{for }  \Cpas= \Cpasval\ .
\end{align}
\end{lem}
\begin{proof}
Let us denote 
\begin{equation}
    f_N(k,l)\coloneqq \frac{\binom{N}{l,k-2l,N-k+l}}{\binom{2N}{k}}\ .
\end{equation}
From \eqref{eq:pasFINcombinatorics} it follows that 
\begin{align}\label{eq:pasFIRSTbound}
     \tr\left(\Pfer \inpsi\ot\inpsi  \right) \leq
     \frac{2}{2N+1}\sum_{k=0}^N \sum_{l=0}^{\lfloor k/2 \rfloor} f_N(k,l) \leq \frac{1}{N}\left(\A_{k=0} + \A_{l=0} +\A_{k=2l}+\A_{gen}  \right)\ , 
\end{align}
where 
\begin{align}
    \A_{k=0} &= f_N(0,0)=1\ , \\
    \A_{l=0} &=\sum_{k=1}^{N} \frac{\binom{N}{k}}{\binom{2N}{k}}\ , \\
    \A_{k=2l} &= \sum_{\substack {k>1 \\ k \text{ even}}}^N \frac{\binom{N}{k/2}}{\binom{2N}{k}}\ , \\
    \A_{gen} &= \sum_{k=1}^N \sum_{l=1}^{l<k/2}  f_N(k,l)\ .
\end{align}
We upper bound each term above separately (except for the trivial case of $\A_{k=0}$). 
The following analytical proof for the bound requires $N \ge 130$. In particular, the bound for \eqref{eq:ALfinbound} $\A_{l=0}$ is valid for $N \ge 40$, and the bound \eqref{eq:AgenBound} for $\A_{gen}$ is valid for $N \ge 130$. At the end of the proof, we show in Fig.~\ref{fig:pasFINbound} that the bound also holds for all smaller values of $N$.

\textbf{Upper bound on } $\A_{l=0}$.  
In this case, we derive a bound valid for $N > 40$.
The bounds from Lemma \ref{lem:binBounds} gives
\begin{equation}
    \A_{l=0} \leq \frac{1}{\binom{2N}{N}} + \frac{C}{c} \sum_{k=1}^{N-1}  \sqrt{\frac{2N-k}{2(N-k)}}  \exp\left[N \left\{h(k/N)- 2 h(k/2N)\right\} \right]\ . 
\end{equation}
We use now the inequality $h(x)-2h(x/2) \leq -\frac{2}{3} x$, valid for $x\in[0,1]$ to obtain
\begin{equation}
    \A_{l=0} \leq \frac{1}{\binom{2N}{N}}  +\frac{C}{c} \sum_{k=1}^{N-1} \sqrt{\frac{2N-k}{2(N-k)}}  
    \exp\left(-\frac{2k}{3} \right)\ .
\end{equation}
We then apply the bound $\binom{2N}{N}\geq c 2^{2N}\sqrt{2/N}$ and divide the sum over $k$ into two parts 
\begin{equation}
    \A_{l=0} \leq \frac{\sqrt{N}}{\sqrt{2} c} 2^{-2N} + \frac{C}{c}\left( \sum_{k=1}^{k\leq1/2 N }  \sqrt{\frac{2N-k}{2(N-k)}}   \exp\left(-\frac{2k}{3} \right) +   \sum_{k>1/2 N}^{N-1 }  \sqrt{\frac{2N-k}{2(N-k)}}  \exp\left(-\frac{2k}{3} \right) \right)
\end{equation}
For $k\leq N/2$ we have $\sqrt{\frac{2N-k}{2(N-k)}} \leq \sqrt{3/2}$ and therefore 
\begin{equation}\label{eq:Al}
    \A_{l=0} \leq \frac{\sqrt{N}}{\sqrt{2} c}2^{-2N} + \frac{C}{c} \left(\sqrt{\frac{3}{2}}\frac{1}{e^{2/3}-1} + \frac{N^{3/2}}{2} \exp\left(-\frac{N}{3}\right)  \right)\ ,
\end{equation}
where we have utilized the expression for the sum of geometric progression and the upper bound $\sqrt{\frac{2N-k}{2(N-k)}}\leq \sqrt{N}$, valid for $k\leq N-1$ . Using  expression \eqref{eq:Al} it is easy to verify that for $N>40$ we have
\begin{equation}\label{eq:ALfinbound}
    \A_{l=0} \leq \frac{3}{2}\ .
\end{equation}

\textbf{Upper bound on } $\A_{k=2l}$. Estimates for binomials from Lemma \ref{lem:binBounds} yield
\begin{equation}
    \A_{k=2l} \leq \frac{C\sqrt{2}}{c} \sum_{\substack {k>1 \\ k \text{ even}}}^N  \exp\left[-N h(k/2N)\right]\ .
\end{equation}
Concavity of binary entropy $h(\cdot)$ implies that for $x\in[0,1]$ we have $\frac{\log(2)}{2} x \leq h(x/2) $ and consequently
\begin{equation}
    \A_{k=2l} \leq \frac{C\sqrt{2}}{c} \sum_{\substack {k>1 \\ k \text{ even}}}^N  \exp\left(- \frac{k\log(2)}{2}\right)=\frac{C\sqrt{2}}{c} \sum_{p=1}^{\lfloor N/2 \rfloor}  2^{-p}\ .
\end{equation}
The sum of the geometric series in the above expression is upper bounded by $1$ and therefore
\begin{equation}\label{eq:Akfinbound}
     \A_{k=2l} \leq \frac{C\sqrt{2}}{c} \leq \frac{8}{5}\ .
\end{equation}

\textbf{Upper bound on } $\A_{gen}$. 
In the following proof, we require that $N \ge 130$.
For the generic points in the sum \eqref{eq:pasFINcombinatorics} inequalities from Lemma \ref{lem:binBounds} give
\begin{equation}\label{eq:AgenFirst}
    \A_{gen} \leq \frac{A}{\sqrt{2}c}  \sum_{k=1}^N \sum_{l=1}^{l<k/2}  \sqrt{\frac{k(2N-k)}{l(k-2l)(N-k+l)}}\, 
    \exp\left( N \left\{h\left[x_l,y_k-2x_l,1-y_k+x_l\right] -2 h\left[y_k /2\right]\right\} \right)\ ,
\end{equation}
where $x_l=l/N$, $y_k=k/N$. Note that $k=1$ and $k=2$ are implicitly excluded from the above sum because of the constraints on $l$ and hence
\begin{equation}\label{eq:AgenSecond}
    \A_{gen} \leq \frac{A}{\sqrt{2}c}  \sum_{k=3}^N \sum_{l=1}^{l<k/2}  \sqrt{\frac{k(2N-k)}{l(k-2l)(N-k+l)}}\, 
    \exp\left( N \left\{h\left[x_l,y_k-2x_l,1-y_k+x_l\right] -2 h\left[y_k /2\right]\right\} \right)\ .
\end{equation}
In order to upper bound the expression we maximize the function 
\begin{equation}\label{eq:F-passive}
    F(x,y)= h\left(x,y -2x,1-y+x\right) -2 h\left(y /2\right)
\end{equation}
over $x\in[0,y/2]$, for fixed value of $y\in[0,1]$. Looking for critical points reduces the problem to solving quadratic equation which gives a unique solution in the interval $[0,y/2]$:
\begin{equation}\label{eq:xopt-passive}
    x_{opt}(y)= \frac{1}{6}\left( 1+3y -\sqrt{1+6y-3y^2}\right)\ .
\end{equation}

\begin{figure}
    \centering
    \includegraphics[scale=0.45]{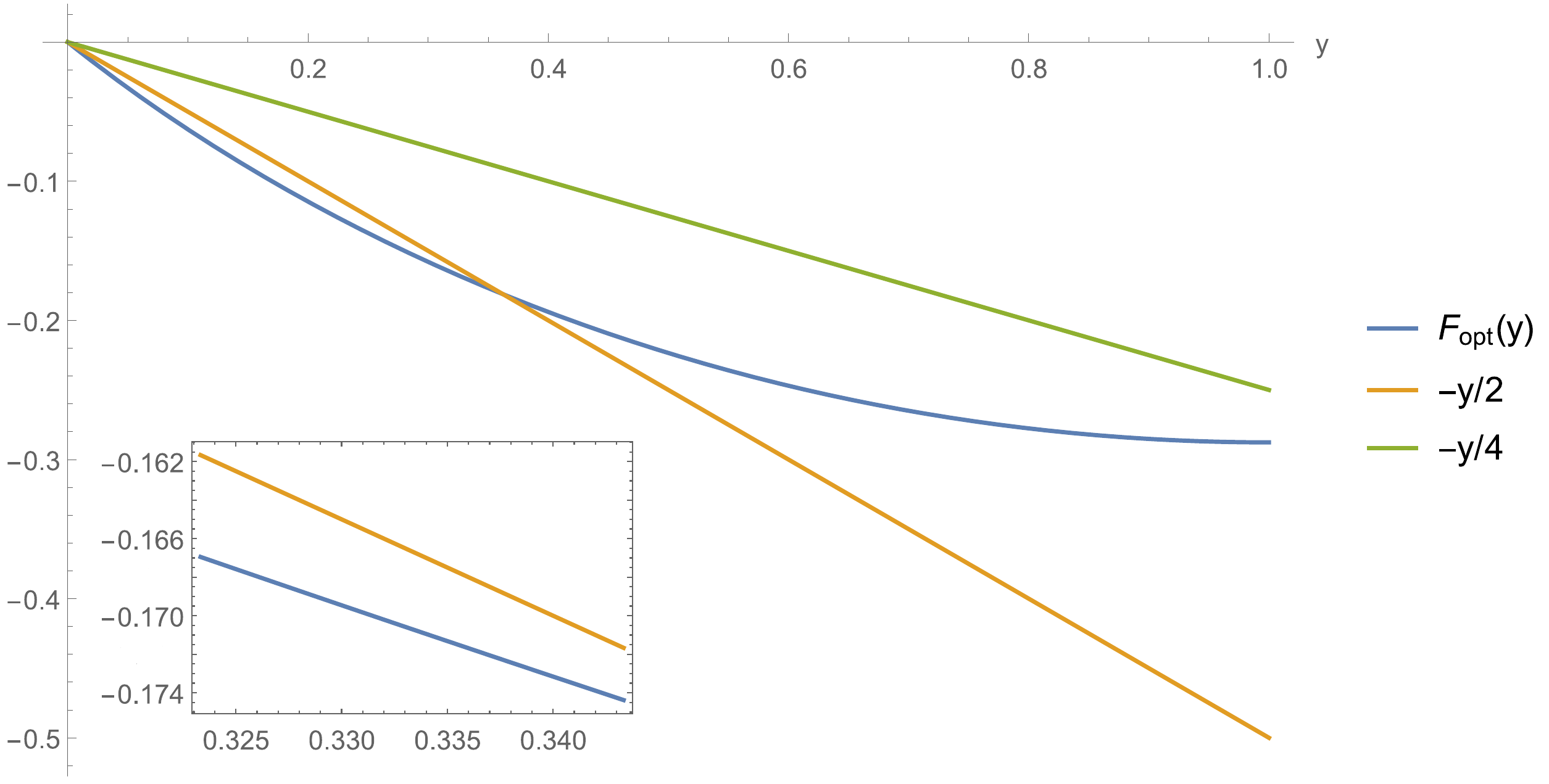}
    \caption{Function $F_{opt}(y)\coloneqq F(x_{opt}(y),y)$ where $F$ is defined in Eq.~\eqref{eq:F-passive} and $x_{opt}(y)$ is given in Eq.~\eqref{eq:xopt-passive}. The function is bounded by $-y/3$ in the interval $[0,1/3]$ and by $-y/4$ in the interval $[1/3,1]$. The inset plot shows that the inequality is also valid near $y=1/3$.}
    \label{fig:passive-gen-bound}
\end{figure}

Crucially, the function $F_{opt}(y)\coloneqq F(x_{opt}(y),y)$ is a continuous function of parameter $y$, which is also analytic in the interior the interval $(0,1)$. Moreover, $F_{opt}(y)$ satisfies (see Fig.~\ref{fig:passive-gen-bound}):
\begin{equation}
F_{opt}(y)\leq - \frac{1}{2}y\ \ \text{for}\ y\in\left[0,1/3\right]\ ,\     F_{opt}(y) \leq - \frac{1}{4}y \ \ \text{for}\ y\in\left[0,1\right]\ .
\end{equation}
 It follows that
\begin{align}\label{eq:casesPAS}
    N \left(h\left[x_l,y_k-2x_l,1-y_k+x_l\right] -2 h\left[y_k /2\right]\right) &\leq 
    - \frac{1}{2} k\enspace \text{for}\ \  1\leq k\leq N/3\ , \\
    N \left(h\left[x_l,y_k-2x_l,1-y_k+x_l\right] -2 h\left[y_k /2\right]\right) &\leq  - \frac{1}{4} k\enspace \text{for}\ \ 1\leq k\leq N. \
\end{align}
Moreover, for integer $l$ satisfying $1\leq l < k/2$ we have $l(k-2l)\geq (k-2)/2$ and consequently for $k\geq 3$ we have $\frac{k}{l(k-2l)}\leq \frac{2 k}{k-2} \leq 6$. As a result we have
\begin{equation}\label{eq:interPAS}
    \sum_{l=1}^{l<k/2}  \sqrt{\frac{k(2N-k)}{l(k-2l)(N-k+l)}} \leq 
    \frac{\sqrt{6}k}{2} \sqrt{\frac{2N-k}{N-k+1}}\ .
\end{equation}
Inserting \eqref{eq:casesPAS}  and  \eqref{eq:interPAS} into Eq. \eqref{eq:AgenSecond} gives
\begin{equation}
   \A_{gen} \leq \frac{\sqrt{3}A}{2 c} \left( \sum_{k=3}^{k\leq N/3}  \sqrt{\frac{2N-k}{N-k+1}}\, k\, \exp\left(-\frac{k}{2}\right) + \sum_{k>N/3 }^{N}  \sqrt{\frac{2N-k}{N-k+1}}\, k\, \exp\left(-\frac{k}{4}\right) \right) 
\end{equation}
Observing that for $k\leq N/3$ we have $\sqrt{\frac{2N-k}{N-k+1}}\leq\sqrt{\frac{5}{2}}$, while and for general $k\leq N$ $\sqrt{\frac{2N-k}{N-k+1}}\leq \sqrt{N}$, we obtain
\begin{equation}
   \A_{gen} \leq \frac{\sqrt{3}A}{2 c} \left( \sqrt{\frac{5}{2}}  \sum_{k=3}^{k\leq N/3}   k\, \exp\left(-\frac{k}{2}\right) + \frac{2N^{\frac{3}{2}}}{3} \exp\left(-\frac{N}{12}\right) \right)\ . 
\end{equation}
We bound the first summand as follows
\begin{equation}
    \sum_{k=3}^{k\leq N/3}   k\, \exp\left(-\frac{k}{2}\right) \leq \sum_{k=3}^{\infty}   k\, \exp\left(-\frac{k}{2}\right) = \frac{3\sqrt{e} -2}{(\sqrt{e}-1)^2 e}\ .
\end{equation}
This finally gives us
\begin{equation}
    \A_{gen} \leq \frac{\sqrt{15}A}{2\sqrt{2} c}  \frac{3\sqrt{e} -2}{(\sqrt{e}-1)^2 e}+ \frac{A}{\sqrt{3}c}N^{\frac{3}{2}} \exp\left(-\frac{N}{12}\right)\ .
\end{equation}
Using the above expression we get that for $N\geq 130$ we have
\begin{equation}\label{eq:AgenBound}
    \A_{gen} \leq \frac{8}{5}\ .
\end{equation}
Finally, combining bounds  \eqref{eq:ALfinbound}, \eqref{eq:Akfinbound} and \eqref{eq:AgenBound} together with $\A_{k=0}=1$  we see that for $N\geq 130$,
\begin{equation}
    \A_{k=0} + \A_{l=0} +\A_{k=2l}+\A_{gen} \leq 5.7\ .
\end{equation}
Inserting this into the bound \eqref{eq:pasFIRSTbound} proves the lemma for $N \ge 130$. For $N \le 130$, the validity of the bound can be verified numerically as shown in Fig.~\ref{fig:pasFINbound}, which completes the proof.

\begin{figure}
    \centering
    \includegraphics[scale=0.75]{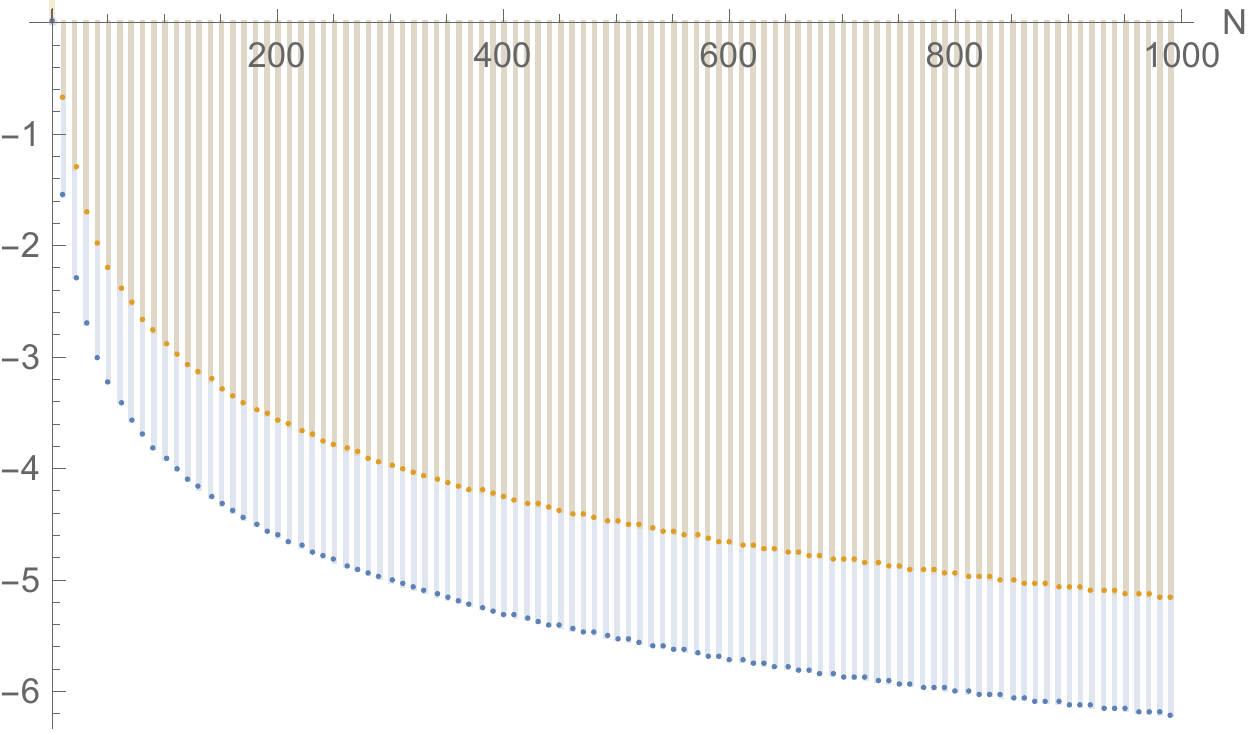}
    \caption{Plots of the logarithm of the expression \eqref{eq:pasFINcombinatorics} (blue) and $\log(\Cpas/N) = \log(\Cpasval/N)$ (orange), which constitutes a valid upper bound for all $N \le 1000$.} 
    \label{fig:pasFINbound}
\end{figure}

\end{proof}

Analogously for the active FLO case,
Eq. \eqref{eq:actFINALexpression} implies that
\begin{equation}\label{eq:ACTcombinatorics}
    \tr(\Pflo \inpsi \ot \inpsi) 
    \le \frac{\binom{8N}{4N}}{2^{8N-1}}
    \sum_{q=0}^{N}
    \sum_{l=0}^{\lfloor \frac{q}{2} \rfloor} 
    \frac{\binom{4N}{2q} \binom{N}{l,q-2l,N-q+l}}{\binom{8N}{4q}}   14^{q-2l}.
\end{equation}
\begin{lem}\label{lem:actAnti-computation}
Consider the setting of our quantum advantage proposal, i.e., $d=4N$ and $n=2N$.  Let $\inpsi\in \D\left( \bigwedge^{2N}(\C^{4N})\right)$.  Let $\Pflo$ be defined as in Lemma \ref{lem:PROJactive}. We  then have
\begin{align}\label{eq:actMAINBOUND}
    \tr\left(\Pflo \inpsi\ot\inpsi  \right) \leq \frac{\Cact}{\sqrt{\pi N}},\ \text{for }  \Cact= \Cactval\ .
\end{align}
\end{lem}

\begin{proof}
Our proof strategy is analogous to the one used in the case of passive FLO. Let us denote
\begin{align}
    g_N(q,l) \coloneqq \frac{\binom{4N}{2q} \binom{N}{l,q-2l,N-q+l}}{\binom{8N}{4q}}   14^{q-2l}\ .
\end{align}
It follows from \eqref{eq:ACTcombinatorics} and the entropic bound for binomial coefficients in Lemma \ref{lem:binBounds}, 
\begin{align}
     \frac{\binom{8N}{4N}}{2^{8N-1}} \le \frac{1}{\sqrt{\pi N}}\ ,
\end{align}
that
\begin{align}\label{eq:actFIRSTbound}
    \tr(\Pflo \inpsi \ot \inpsi) 
    \le \frac{1}{\sqrt{\pi N}}
    \sum_{q=0}^{N}
    \sum_{l=0}^{\lfloor \frac{q}{2} \rfloor} 
    g_N(q,l)
    \le \frac{1}{\sqrt{\pi N}}
    (\B_{q=0} + \B_{l=0} + \B_{q=2l} + \B_{gen})\ ,
\end{align}
where
\begin{align}
    \B_{q=0} &= g_{N}(0,0) = 1,\  \\
    \B_{l=0} &= \sum_{q=1}^N \frac{\binom{4N}{2q}\binom{N}{q}}{\binom{8N}{4q}} 14^q,\ \\
    \B_{q=2l} &= \sum_{\substack{q>1 \\ q\, \mathrm{even}}}^N \frac{\binom{4N}{2q}\binom{N}{q/2}}{\binom{8N}{4q}},\ \\ 
    \B_{gen} &= \sum_{q=1}^N \sum_{l=1}^{l<q/2} g_N(q,l)\ .
\end{align}
We upper bound each term above separately (except for the trivial case of $\B_{q=0}$). 
The following analytical proof for the bound requires $N \ge 7000$. In particular, the bound for \eqref{eq:ALfinbound} $\B_{l=0}$ is valid for $N \ge 1000$, and the bound \eqref{eq:AgenBound} for $\B_{gen}$ is valid for $N \ge 7000$.   At the end of the proof, we show in Fig.~\ref{fig:actFINbound} that the bound \eqref{eq:actMAINBOUND} also holds for all smaller values of $N\leq 7000$ by numerically evaluating right-hand side of  \eqref{eq:ACTcombinatorics}.

\textbf{Upper bound on } $\B_{l=0}$. 
For this term, we require that $N \ge 1000$.
The entropic bound in Lemma \ref{lem:binBounds} implies that
\begin{align}
    \B_{l=0} \le \frac{\binom{4N}{2N}}{\binom{8N}{4N}} 14^N
    + \frac{C^2\sqrt{2}}{c} \sum_{q=1}^{N-1} \sqrt{\frac{N}{ q(N-q)}} \exp[N\{ h(q/N) - 4h(q/2N) + \log(14)q/N \}]\ .
\end{align}

To upper bound the sum, we split the sum into two sums: one from $q=1$ to $q \le N/5$ and another from $q > N/5$ to $q=N-1$, and
upper bound the function
\begin{align}
    H(x) \coloneqq h(x) - 4h(x/2) + x\log(14),
\end{align}
$x \in [0,1]$ in the intervals $[0,1/5]$ and $(1/5,1]$ separately.
In particular, we have that (See also Fig.~\ref{fig:active_l=0})
\begin{figure}
    \centering
    \includegraphics[scale=0.75]{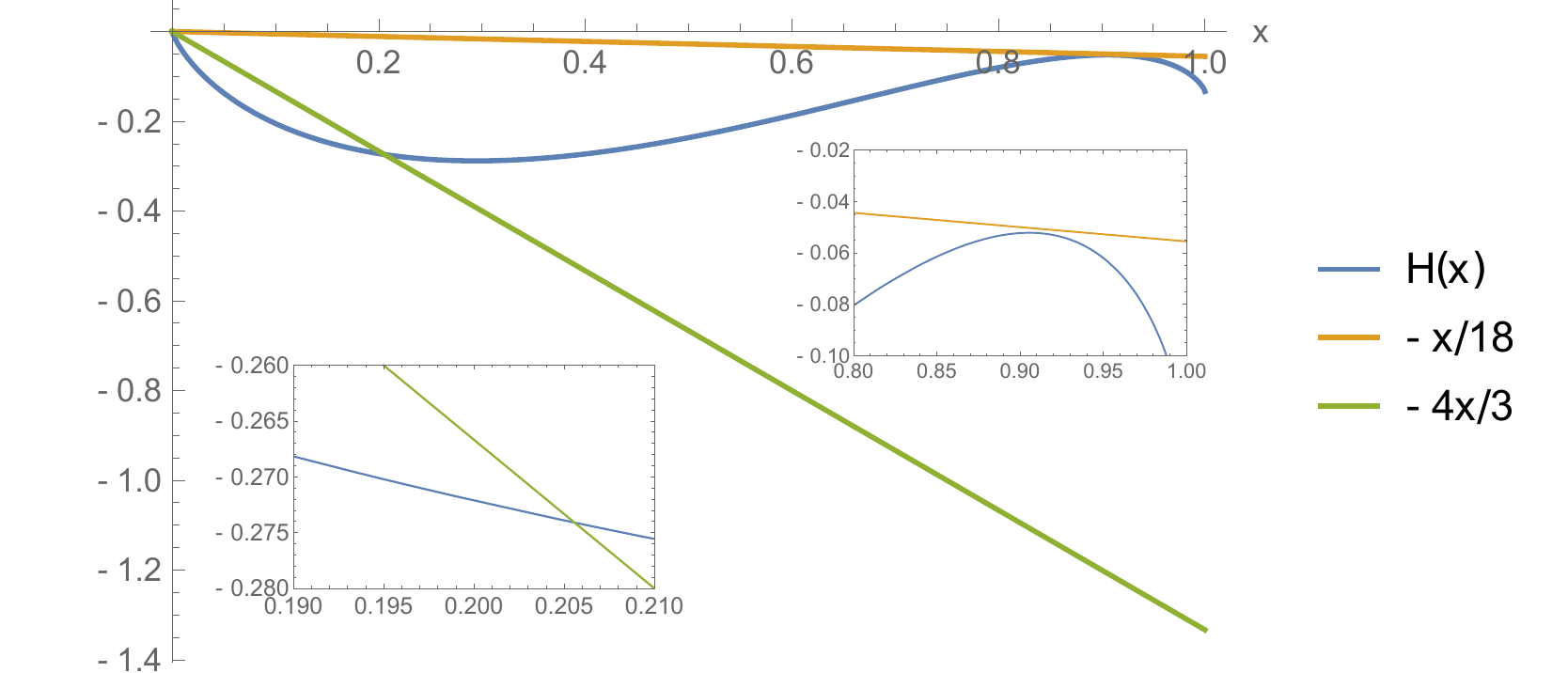}
    \caption{Function $H(x)$ defined in \eqref{eq:active_H}. The function is bounded above by $-4x/3$ in the interval $[0,1/5]$ and by $-x/18$ in the interval $[0,1]$. The inset plot shows the validity of the upper bound in each interval.}
    \label{fig:active_l=0}
\end{figure}
\begin{align}\label{eq:active_H}
    H(x) \le -\frac{4}{3}x\enspace \text{for}\; x\in[0,2/5], && H(x) \le -\frac{1}{18}x \enspace \text{for}\; x\in[0,1] 
\end{align}
Together with the bound $\binom{4N}{2N}/\binom{8N}{4N} \le \frac{C\sqrt{2}}{c} 2^{-4N}$
and $\sqrt{N/(q(N-q))} \le \sqrt{2}$ valid for $N \ge 2$ ( this is  because  $\sqrt{\frac{N}{q(N-q)}}$ is convex for $q\in[1,N-1]$  and thus the expression takes the maximum values at the end points),
we obtain
\begin{align}
    \B_{l=0} &\le \frac{C\sqrt{2}}{c} \left(\frac{14}{16}\right)^N +
    \frac{2C^2}{c}
    \left(
    \sum_{q=1}^{q \le N/5}
    \exp(-4q/3)
    + \sum_{q > N/5}^{N-1}
    \exp(-q/18)
    \right) \\
    &\le \frac{C\sqrt{2}}{c} \left(\frac{14}{16}\right)^N +
    \frac{2C^2}{c}\left(
    \frac{1}{e^{4/3}-1} + \frac{4N}{5}
    \exp \left[-\frac{N}{18\cdot 5}\right]
    \right), \label{eq:B_l=0}
\end{align}
where we have used the sum of the geometric series to arrive at the final expression.
Using the expression \eqref{eq:B_l=0}, it can be verified that 
\begin{align}\label{eq:Bl0bound}
    \B_{l=0} \le \frac{1}{3}
\end{align}
holds for $N \ge 1000$.

\textbf{Upper bound on } $\B_{q=2l}$.
From Lemma \ref{lem:binBounds} we see that
\begin{align}
    \B_{q=2l}  &\le \frac{C^2\sqrt{2}}{c}
    \sum_{\substack{q>1 \\ q\, \mathrm{even}}}^N
     \sqrt{\frac{N}{\frac{q}{2}(N-\frac{q}{2})}} \exp[-3Nh(q/2N)]
\end{align}


Now by concavity of $h(x)$ for $x\in[0,\frac{1}{2}]$ we have $\log(2)x/2 \leq h(x/2)$ for $x\in [0,1]$. Then

\begin{align}
    \B_{q=2l}  &\le \frac{C^2\sqrt{2}}{c}
    \sum_{\substack{q>1 \\ q\, \mathrm{even}}}^N
     \sqrt{\frac{N}{\frac{q}{2}(N-\frac{q}{2})}} \exp[-3q\log(2)/2]\\
     &=\frac{C^2\sqrt{2}}{c} \sum_{p=1}^{\lfloor N/2 \rfloor} \sqrt{\frac{N}{p(N-p)}}2^{-3p}
\end{align}


We can bound $\sqrt{\frac{N}{p(N-p)}}\leq \sqrt{2}$ the same way as in the passive case. Then we obtain
\begin{equation}\label{eq:BlqFIN}
    \B_{q=2l} \leq \frac{2C^2}{7c} \leq 0.13
\end{equation}
where we used that the geometric sum of $2^{-3p}$ is bounded by $1/7$.

\textbf{Upper bound on } $\B_{gen}$.
Following bounds form Lemma \ref{lem:binBounds} and defining $x_l = \frac{l}{N}$ and $y_q=\frac{q}{N}$ we obtain

\begin{align}\label{eq:Bgen}
    \B_{gen} &\leq  \frac{\sqrt{2} C A}{c} \sum_{q=1}^N \sum_{l=1}^{l<q/2} \sqrt{\frac{N}{l(q-2l)(N-q+l)}}\exp[N G(x_l,y_q)]\ ,
\end{align}
where, following the analogous construction in Lemma \ref{lem:pasAnti-computation}, we introduced
\begin{equation} \label{eq:F-active}
    G(x,y)\coloneqq -4 h(y/2) + h(x,y-2x,1-y+x)+ (y-2x)\log(14)\ .
\end{equation}
As in the case of passive FLO, our strategy is to upper bound $G(x,y)$ by a function that allows for analytical treatment. To this end, we first optimize $G(x,y)$ over $x\in[0,y/2]$ for fixed $y\in[0,1]$. Solving for the critical points gives the following optimal solution $x_{opt}\in[0,y/2]$ (at the extremal  points of this interval function $G(x,y)$, treated as a function of $x$ for fixed $y$, takes smaller values)  
\begin{equation}\label{eq:xopt-active}
    x_{opt}(y)= \frac{1}{96}\left( -49+48y +7\sqrt{40-96y+48y^2}\right)\ .
\end{equation}

\begin{figure}
    \centering
    \includegraphics[scale=0.5]{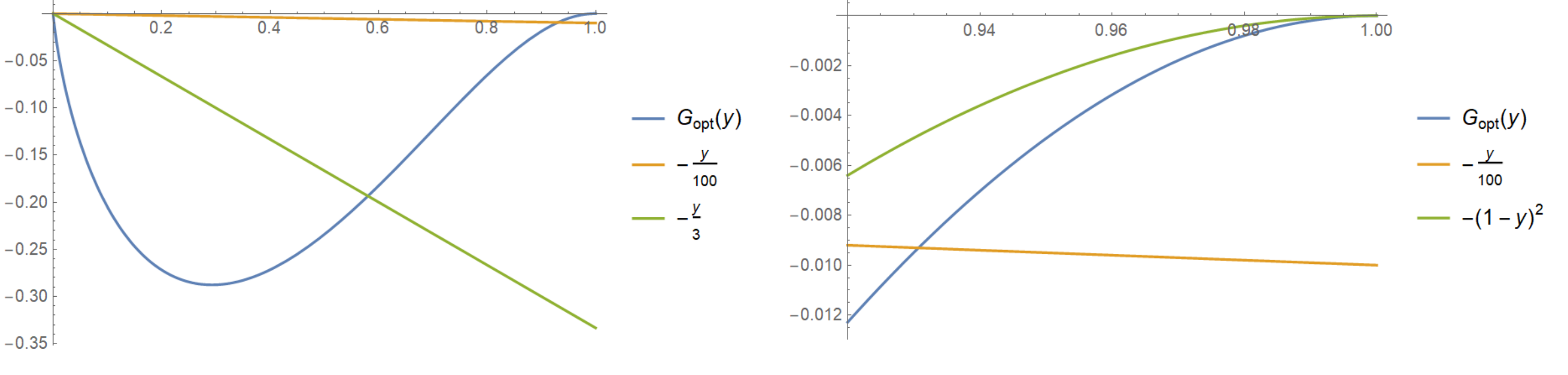}
    \caption{Function $G_{opt}(y)= G(x_{opt}(y),y)$ where $G$ is defined in Eq.~\eqref{eq:F-active} and $x_{opt}(y)$ is defined in Eq.~\eqref{eq:xopt-active}. The function is presented alongside simple analitical lower bounds are valid in specific intervals formulated in Eq. \eqref{eq:intervalsACTIVE}. }
    \label{fig:active-gen-bound}
\end{figure}
The maximum of $G(x,y)$ over $x\in[0,y/2]$, $G_{opt}(y)\coloneqq G(x_{opt}(y),y)$ is a continuous function of $y\in[0,1]$ and also analytic for $y\in(0,1)$. We can bound  $G_{opt}(y)$ in the following way (see Fig.~\ref{fig:active-gen-bound})
\begin{equation}\label{eq:intervalsACTIVE}
G_{opt}(y)\leq -y/3\ \ \text{for}\ y\in\left[0,1/2\right]\ ,\
G_{opt}(y)\leq -y/100\ \ \text{for}\ y\in\left[1/5,0.925\right]\ ,\
G_{opt}(y) \leq -(1-y)^2 \ \ \text{for}\ y\in\left[0.925,1\right]\ .
\end{equation}
We shall need much more refined information about $G(x,y)$ than in the case of analogous considerations for passive FLO. Namely, we will need to control how fast $G(x,y)$ decays as a function of $x-x_{opt}(y)$, for fixed $y$. To this end we compute for $x\in(0,y/2)$, $y\in(0,1)$
\begin{equation}
    \partial^2_x G(x,y) =- \left( \frac{1}{x} +\frac{1}{1-y+x} + \frac{4}{y-2x}\right)\ .
\end{equation}
From the above expression we get\footnote{It is easy to check that $3/2 x_{opt}(y)\leq y/2$. } 
\begin{equation}\label{eq:secondDERbound}
    \partial^2_x G(x,y) \leq - 16\ \text{for } x\in(0,y/2)\ \ \text{and}\ \  \partial^2_x G(x,y) \leq -\frac{2}{3 x_{opt}(y)}\  \text{for } x\in\left[\frac{x_{opt}(y)}{2},\frac{3 x_{opt}(y)}{2}\right]\ .
\end{equation}
Using the analyticity of $G(x,y)$ as a function of $x$ inside the interval $(0,y/2)$, we can Taylor expand it around $x_{opt}(y)$ (for fixed value of $y$):
\begin{equation}
    G(x,y) = G_{opt}(y) +(\partial_x G(x_{opt}(y),y))(x-x_{opt}(y)) +\int_{x_{opt}(y)}^{x}d\tau \partial_{\tau} G(\tau,y) \ .   
\end{equation}
Using the fact that $x_{opt}(y)$ is a critical point and bounds, identity
\begin{equation}
    \partial_{\tau} G(\tau,y) = \int_{x_{opt}(y)}^{\tau}dx \partial^2_x G(x,y) 
\end{equation}
and bounds from Eq.~\eqref{eq:secondDERbound} we get finally get 
\begin{align}
    G(x,y) &\leq G_{opt}(y) - 8 (x-x_{opt}(y))^2 \ &\text{for } & x\in[0,y/2]\ ,\ y\in[0,1]\ , \label{eq:genBOUND} \\
    G(x,y) &\leq G_{opt}(y) - \frac{1}{3x_{opt}(y)} (x-x_{opt}(y))^2 \ & \text{for }& x\in\left[\frac{x_{opt}(y)}{2},\frac{3 x_{opt}(y)}{2}\right]\ ,\ y\in[0,1]\ \label{eq:fineBOUND} . 
\end{align}

Coming back to the bound on $\B_{gen}$ from \eqref{eq:Bgen}, similarly to the case of passive FLO, due to constrains on $l$, the sum appearing in  \eqref{eq:Bgen} effectively starts  from $q=3$. Moreover, we also note that $l(q-2l)\geq(q-2)/2$ and therefore 
\begin{equation}\label{eq:actFRACbound}
    \sqrt{\frac{N}{l(q-2l)(N-q+l)}} \leq \sqrt{\frac{2N}{(q-2)(N-q+l)}} \leq \sqrt{\frac{2N}{N-2}}\ , 
\end{equation}
where in the second inequality we used the fact that $q\in[3,N]$ and $l\geq 1$. 
Using the above and expanding the expression in \eqref{eq:Bgen} in the different intervals defined in \eqref{eq:intervalsACTIVE} we obtain
\begin{align}
    \B_{gen} &\leq \frac{2C A}{c}\sqrt{\frac{N}{N-2}}\left(\sum_{q=3}^{q \leq N/2 } \sum_{l=1}^{l<q/2} \exp[-q/2] + \sum_{q>N/2}^{q<0.925N} \sum_{l=1}^{l<q/2} \exp[-q/100] \right) \label{eq:easySUMS}\\
    &+\frac{\sqrt{2} C A}{c} \sum_{q> 0.925N}^{N} \sum_{l=1}^{l<q/2}  \sqrt{\frac{N}{l(q-2l)(N-q+l)}} \exp[N G(x_l,y_q)]\ \label{eq:diffSUM}.
\end{align}
Two sums from Eq. \eqref{eq:easySUMS} can be handled analogously as in the case of passive FLO:
\begin{align}
    &\frac{2C A}{c}\sqrt{\frac{N}{N-2}}\left(\sum_{q=3}^{q \leq N/5 } \sum_{l=1}^{l<q/2} \exp[-q/3] + \sum_{q>N/2}^{q<0.925N} \sum_{l=1}^{l<q/2} \exp[-q/100] \right)\\
    &\leq \frac{2C A}{c}\sqrt{\frac{N}{N-2}} \left( \sum_{q=3}^{\infty} (q/2) \exp(-q/3) + +(N^{3/2}/2) \exp[-\frac{N}{200}]\right) . \\
    & =  \frac{2C A}{c}\sqrt{\frac{N}{N-2}} \left( \frac{3 e^{1/3} -2}{2 e^{2/3} (e^{1/3}-1)} +(N^{3/2}/4) \exp[-\frac{N}{200}] \right)\ \leq 2\ \label{eq:simACTupp} ,
\end{align}
where the last inequality is valid for $N\geq 1800$. The sum in \eqref{eq:diffSUM} will be analyzed  using inequalities \eqref{eq:genBOUND} and \eqref{eq:fineBOUND}. For fixed $y_q$ (Which corresponds to $q=y_q N$)  we set $l_{opt}(y_q)=x_{opt}(y_q) N$ and divide the range of summation over $l$ in \eqref{eq:diffSUM}  into two parts that corresponds to intervals in bounds \eqref{eq:genBOUND} and \eqref{eq:fineBOUND} respectively :
\begin{align}
    \Lmax_q &=\left\{l\  \left|\ \frac{1}{2}l_{opt} (y_q) \leq  l \leq \frac{3}{2}l_{opt} (y_q)  \right.   \right\}\ , \\
    \Lgen_q &=\left\{l\  \left|\ 1 \leq l < \frac{1}{2}l_{opt} (y_q) \text{ or } \frac{3}{2}l_{opt} (y_q) < l < q/2 \right.   \right\}\ . \\
\end{align}
It is now straightforward to verify that:
\begin{equation}
 \sqrt{\frac{N}{l(q-2l)(N-q+l)}}  \leq      \sqrt{\frac{4N}{l_{opt}(q-3l_{opt})l_{op}}}\ \text{for } l\in\Lmax_q\ ,
\end{equation}
where for clarity we surpassed the dependence of $l_{opt}$ on $q$. Moreover from \eqref{eq:genBOUND} and \eqref{eq:fineBOUND} we get
\begin{align}
   N G(x_l,y_q) & \leq N G_{opt}(y_q)- \frac{(l-l_{opt})^2}{3l_{opt}} \ \text{for } l\in \Lmax_q\ , \\
   N G(x_l,y_q) & \leq N G_{opt}(y_q) -2 x^2_{opt} N \ \text{for } l\in \Lgen_q\ .
\end{align}
Finally, we arrive at the following bound
\begin{align}
    &\frac{\sqrt{2} C A}{c} \sum_{q> 0.925N}^{N}\sum_{l=1}^{l<q/2}  \sqrt{\frac{N}{l(q-2l)(N-q+l)}} \exp[N G(x_l,y_q)] \\ &\leq 
    \frac{\sqrt{2} C A}{c} \sum_{q> 0.925N}^{N}  \exp[N G_{opt}(y_q)] \sqrt{\frac{4N}{l_{opt}(q-3l_{opt})l_{opt}}}   \sum_{l\in\Lmax_q} \exp\left(- \frac{(l-l_{opt})^2}{3l_{opt}}    \right)  \label{eq:maxcontr} \\
    & + \sqrt{\frac{N}{N-2}}\frac{C A}{c} \sum_{q> 0.925N}^{N} q \exp[N G_{opt}(y_q)]  \exp\left(- 2 x^2_{opt} N \right) \ \label{eq:genrest} \ , 
\end{align}
where we used \eqref{eq:actFRACbound} to get \eqref{eq:genrest}. We first analyze the second sum. We using \eqref{eq:intervalsACTIVE} we obtain 
\begin{equation}\label{eq:GaussSUMbound}
    \sum_{q> 0.925N}^{N} q \exp[N G_{opt}(y_q)] \leq N \sum_{q> 0.925N}^{N}  \exp[\frac{(N-q)^2}{N}]\ \leq N (1 +\frac{\sqrt{\pi N}}{2}) \leq  N^{\frac{3}{2}}\ .
\end{equation}
where we used
\begin{equation}
    \sum_{x=0}^\infty \exp(-\frac{x^2}{a}) \leq 1 + \int_0^\infty dx \exp(-\frac{x^2}{a}) = 1+\frac{\sqrt{\pi a}}{2}\ , 
\end{equation}
valid for all $a>0$, and $N\geq 100$. Importantly, for $q>0.925 N$ (which corresponds to $y\geq 0.925$), we have $x_{opt}\geq 0.03$. Using this and assuming $N\geq 7000$, we finally obtain
\begin{equation}\label{eq:nongenSUMfin}
    \sqrt{\frac{N}{N-2}}\frac{C A}{c} \sum_{q> 0.925N}^{N} q \exp[N G_{opt}(y_q)]  \exp\left(- 2 x^2_{opt} N \right) \leq \sqrt{\frac{N}{N-2}}\frac{C A}{c} N^{\frac{3}{2}} \exp(-\frac{9}{5000} N)\leq 1\ .
\end{equation}
We use similar methods to bound \eqref{eq:maxcontr}. First, we upper bound the exponential sum
\begin{equation}
    \sum_{l\in\Lmax_q} \exp\left(- \frac{(l-l_{opt})^2}{3l_{opt}}\right) \leq 1 + \sqrt{\pi 3 l_{opt}} \leq \frac{10}{3} \sqrt{l_{opt}}\ ,
\end{equation}
which allows estimate
\begin{equation}
    \sqrt{\frac{4N}{l_{opt}(q-3l_{opt})l_{opt}}}   \sum_{l\in\Lmax_q} \exp\left(- \frac{(l-l_{opt})^2}{3l_{opt}}\right)\leq  \frac{10}{3} \sqrt{\frac{4N}{l_{opt}(q-3l_{opt})}} \leq  \frac{10}{3} \sqrt{\frac{4}{(0.03)(0.7 N)}}= \frac{20}{3} \sqrt{\frac{1000}{21 N}}\ ,
\end{equation}
where in the second inequality we used that for $q\geq 0.925 N$ we have $l_{opt}(y_q)\geq 0.03 N$ and $q - 3l_{opt}(y_q) \geq 0.7 N$ .
Inserting thin inequality to \eqref{eq:maxcontr} and again using \eqref{eq:GaussSUMbound} gives that for $N\geq 7000$
\begin{equation}
    \frac{\sqrt{2} C A}{c} \sum_{q> 0.925N}^{N}\sum_{l=1}^{l<q/2}  \sqrt{\frac{N}{l(q-2l)(N-q+l)}} \exp[N G(x_l,y_q)] \leq 1+  \frac{\sqrt{2} C A}{c} \sqrt{\frac{1000}{21}} \leq 12.7\ .
\end{equation}
Combining this estimate with the bound \eqref{eq:simACTupp} and using \eqref{eq:easySUMS}, we finally obtain that for $N\geq 7000$ 
\begin{equation}\label{eq:BgenFINBound}
    \B_{gen} \leq 14.7\ .
\end{equation}

Finally, combining bounds  \eqref{eq:Bl0bound}, \eqref{eq:BlqFIN} and \eqref{eq:BgenFINBound}  together with $\B_{k=0}=1$ in inequality \eqref{eq:actFIRSTbound}  we see that for $N\geq 7000$, 
\begin{equation}
       \tr(\Pflo \inpsi \ot \inpsi) 
    \le \frac{1}{\sqrt{\pi N}}
    (\B_{q=0} + \B_{l=0} + \B_{q=2l} + \B_{gen}) \leq \frac{16.2}{\sqrt{\pi N}}\ .
\end{equation}

For $N \le 7000$, the validity of the bound can be verified numerically as shown in Fig.~\ref{fig:actFINbound}, which completes the proof.

\begin{figure}
    \centering
    \includegraphics[scale=0.75]{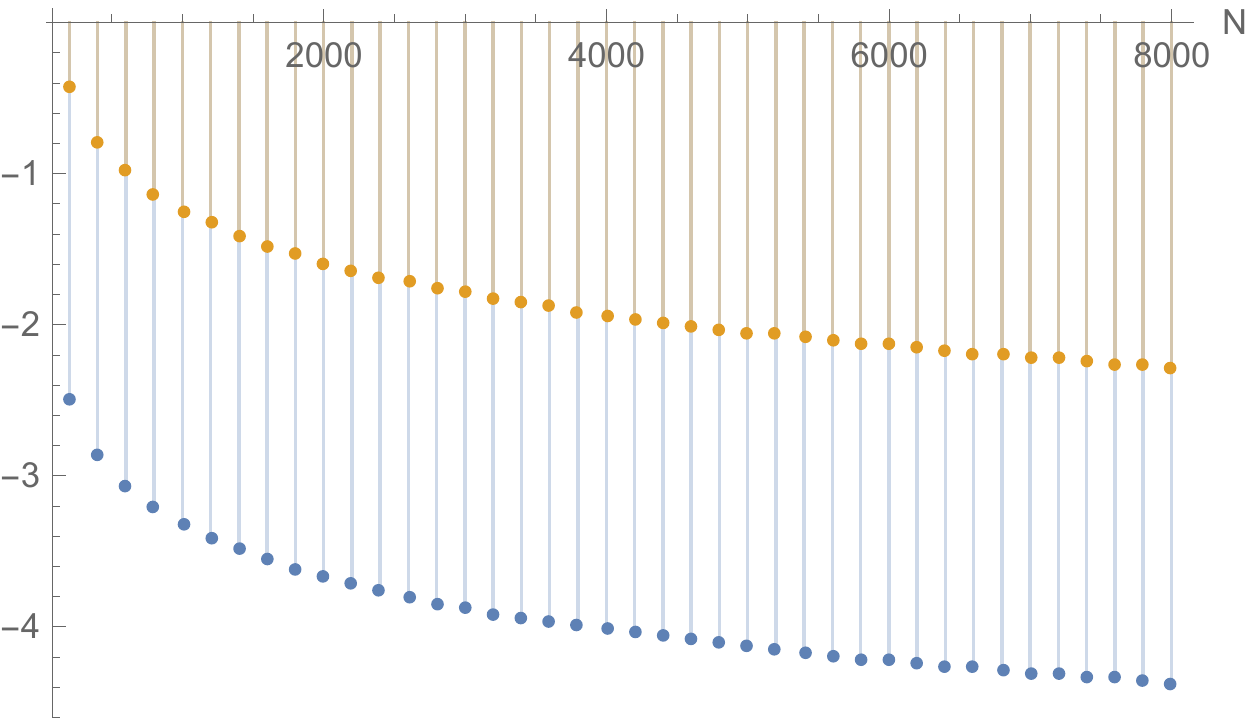}
    \caption{Plots of the logarithm of the the expression \eqref{eq:ACTcombinatorics} (blue) and $\log(\Cact/\sqrt{N}) = \log(\Cactval/\sqrt{\pi N})$ (orange), which is a valid upper bound for all $N \le 7000$.} 
    \label{fig:actFINbound}
\end{figure}

\end{proof}

\section{Efficient tomography of FLO unitaries}\label{app:effTOM}

Here we prove Lemma~\ref{lem:stabAFLO}, which establishes a bound concerning the stability of the active FLO representation which is needed in the efficient tomographic scheme of Section~\ref{sec:cert}. 

{\bf Lemma} (Stability of active FLO representation)
{\it Consider two elements of the  orthogonal group, $O, O' \in \SO(2d)$, and let $V$ and $V'$ be the corresponding active FLO unitaries, i.e., $V = \Piact(O)$ and $V' = \Piact(O')$. Let $\Phi_{V}$ and $\Phi_{V'}$ be the unitary channels defined by $V$ and $V'$ respectively. Then, the following inequality is satisfied}  
\begin{equation}
   \| \Phi_{V} - \Phi_{V'} \|_{\Diamond} 
   \le  2 d \| O - O' \|.
\end{equation}

\begin{proof}
The proof will rely on representation theoretic methods, however, as we have noted, $\Pi_{\mathrm{act}}$ is a projective representation of $\SO(2d)$ and not a proper representation. Instead, we will use  $ \Piact \otimes \Piact$, which is already a proper representation of $\SO(2d)$. Thus, we will bound the diamond norm difference between the unitary channels $\phi_{V \otimes V}$ and $\phi_{V' \otimes V'}$ corresponding to the unitaries  $ V \otimes V =  \Piact \otimes \Piact (O)$ and   $V' \otimes V' =  \Piact \otimes \Piact (O')$, respectively, and then use the inequalities
\begin{equation} \label{eq:appcert}
 \| \Phi_{V} - \Phi_{V'} \|_{\Diamond} \le 
 \| \Phi_{V \otimes V } - \Phi_{V' \otimes V'} \|_{\Diamond} \le 2 \|V \otimes V - V' \otimes V' \|.
\end{equation}
Here the first inequality follows directly from the definition of the diamond norm, while the second is a standard inequality relating the diamond norm distance of unitary channels to the operator norm distance of unitaries (see, e.g., \cite{oszmaniec2020epsilon}).

Thus, our proof strategy will be to upper bound $\|V \otimes V - V' \otimes V' \|=\| \Piact \otimes \Piact (O) -\Piact \otimes \Piact (O')  \| $. For this we use the decomposition of the $\Piact \otimes \Piact$ into subrepresentations in the following way \cite{fuchs2003}:
\begin{equation}
\Piact \otimes \Piact \cong \bigoplus_{s=0}^{d-1} \left(\bigwedge^s\Pi \right)^{\oplus 2} \oplus \bigwedge^{d} \Pi ,
\end{equation}
where $\Pi$ denotes the defining representation of $\SO(2d)$ and  its $\ell$th antisymmetric tensor power $\bigwedge^{\ell} \Pi$  is given by
\begin{align}\label{eq:U}
     \bigwedge^{\ell} \Pi : \SO(2d) &\longrightarrow \U \Big( \bigwedge^{\ell}(\C^{2d}) \Big), \\
    O &\longmapsto  \left. O^{\ot n}\right|_{\bigwedge^{\ell}(\C^d)}.
\end{align}
This decomposition immediately implies that 
\begin{equation}
  \| V\otimes V - V' \otimes V'  \| = \| \Piact \otimes \Piact (O) -\Piact \otimes \Piact (O')  \|  \le \max_{\ell \in [d]} \| \bigwedge^{\ell} \Pi (O) - \bigwedge^{\ell} \Pi (O') \| \le  \max_{\ell \in [d]} \| O^{\otimes \ell} - O'^{\otimes \ell} \|.
\end{equation}
Inserting the above inequality into Eq.~\eqref{eq:appcert} and using that  $\| O^{\otimes \ell} - O'^{\ell} \| \le \ell \| O - O' \| $ (and $\ell \le d$), we obtain  
\begin{equation}
    \| \Phi_{V} - \Phi_{V'} \|_{\Diamond} \le 2d \|  O - O' \|.
\end{equation}
\end{proof}

\section{$\sharP$-Hardness of probabilities in shallow depth active FLO circuits}\label{sec:sharP_shallow}

We argued in section \ref{sec:hardness-sampling} that amplitudes of active FLO circuits are $\sharP$-hard to compute. Here we show that similarly strong simulation (i.e., computing output probabilities) of constant-depth active FLO circuits is hard. It has been proven in previous work \cite{Bremner2011} that under certain conditions, non-universal circuit families of shallow depth are hard to simulate under plausible conjectures which in addition implies that the output probabilities are $\sharP$-hard. In concrete, it is required that the postselected version of the circuit family is universal for quantum computation. This method is not robust as it only shows that exactly computing the output probabilities are hard, nonetheless it may be of interest that such hardness results can be obtained for constant-depth active FLO circuits. The required theorem is as follows


\begin{thm}
    Let $\mathcal{F}$ be a restricted family of quantum circuits. If circuits from $\mathcal{F}$ with the added power of postselection can simulate the output probability distributions of universal quantum circuits with postselection (i.e., $\mathcal{F}$ is universal with postselection) then computing the output probabilities (strong simulation) of circuits in $\mathcal{F}$ is $\sharP$-hard. 
\end{thm}
\begin{proof}
Similar results have been proven in \cite{Aaronson2013,Bremner2011} and later in other works related to active FLO \cite{hebenstreit_computational_2020}. Let $C$ be some circuit with gates from a universal gate set and let $P_C(\y)$ be the output probability of result $\y$. By hypothesis, with the power of postselection we can use a circuit $F$ from $\mathcal{F}$ to simulate $C$ and thus $P_C(\y)=P_F(\y_{*} | 00\cdots 0)=\frac{P_F(\y_{*}00\cdots 0)}{P_F(00\cdots 0)}$, where $\y_{*}$ is potentially a bitstring encoding $\y$ (which will be our case below). This directly implies that if we could compute the output probabilities of $F$ then this would allow for computing the output probabilities of $C$. Since universal circuits are known to include $\sharP$-hard instances, the result follows.
\end{proof}


In what follows, we will always assume that the active FLO circuits are supplied with auxiliary states $\psiQUAD$. Throughout this section we will consider the encoding $\ket{0_L}=\ket{00}$ and $\ket{1_L}=\ket{11}$. To prove that computing the probabilities of shallow depth active FLO circuits is $\sharP$-hard, we prove now Lemma \ref{lem:post-universal}.

\begin{lem}\label{lem:post-universal}
     Constant-depth active FLO circuits supplied with auxiliary states $\psiQUAD$ with the added power of postselection are universal.
\end{lem}

 To prove this, we follow Ref.~\cite{brod_complexity_2015}, which showed similar results in the context of Boson Sampling. The starting point is the brickwork graph state which allows for universal computation on the measurement based quantum computation (MBQC) scheme. We can write the preparation of the brickwork graph state plus measurements on the state as a single circuit with adaptive measurements. If we are given the power to postselect measurements, then the preparation of the graph state requires a constant depth circuit with single qubit gates and $\mathrm{CZ}$ gates. If we can simulate these gates with constant-depth active FLO circuits and postselection, then this would imply Lemma \ref{lem:post-universal}. Using the encoding defined above, we show Theorem \ref{thm:matchgate-universal} which directly implies Lemma \ref{lem:post-universal}.

\begin{thm}\label{thm:matchgate-universal}
Active FLO acting on an initial state consisting of tensor products of $\ket{\Psi_4}$ with the added power of postselection can simulate single qubit gates and $\mathrm{CZ}$ with constant-depth circuits. These simulations are at the logical level using the encoding above. 
\end{thm}
\begin{proof}
As explained before, the circuit induced by the brickwork state with post selection is universal and of constant depth, consisting of single qubit gates and $\mathrm{CZ}$ gates. Using the encoding above we can simulate single qubit gates and $\mathrm{CZ}$ gates in constant depth, then we can simulate the whole universal constant-depth circuit with a circuit from $\mathcal{C}_{act}$ and postselection.

That single qubit gates at the logical level can be implemented with this encoding is already known \cite{bravyi_fermionic_2002}. Implementing $\mathrm{CZ}$ at the logical level will require the use of post selection and the auxiliary states $\psiQUAD$. First, we note that the state $\psiQUAD$ can be transformed into the state $\ket{a_8}=\frac{1}{\sqrt{2}}(\ket{0000}+\ket{1111})$ using only active FLO operations. This was shown previously in the proof of Lemma \ref{lem:actPROJfinal}. Second, in Lemma 1 of \cite{bravyi_universal_2006} it is shown that using a single copy of $\ket{a_8}$ and particle number measurements it is possible to implement a $\mathrm{CZ}$ at the logical level using the same encoding we use here. This two facts together imply that $\mathrm{CZ}$ can be implemented with active FLO circuits supplied by $\psiQUAD$ states and postselection. The auxiliary states can be swapped to the desired position when implementing a gate without incurring on extra negative signs with our encoding since the auxiliary states used are fermionic as for example argued in \cite{hebenstreit2019}.
\end{proof}

\end{document}